\documentclass[a4]{article}

\usepackage[utf8]{inputenc}

\usepackage{caption}
\usepackage{graphicx}
\usepackage{hyperref}

\usepackage{amsmath,amsthm}
\usepackage{amssymb,mathrsfs}
\usepackage{a4wide}
\usepackage{physics}
\usepackage{color}
\usepackage{enumerate}
\usepackage[normalem]{ulem}
\usepackage{cancel}
\usepackage{bbm}
\usepackage{booktabs}
\usepackage{tikz}
\usepackage{nccmath}
\usepackage{mathtools}
\usepackage{relsize}
\usepackage{subcaption}
\usepackage{pgfplots}

\usetikzlibrary{shapes.misc}
\usetikzlibrary{calc}

\tikzset{cross/.style={cross out, draw=black, minimum size=2*(#1-\pgflinewidth), inner sep=0pt, outer sep=0pt},
cross/.default={1pt}}

\pgfplotsset{compat=1.18} 

\mathtoolsset{showonlyrefs=true}

\newcommand{\cL}{\mathcal{L}}
\renewcommand{\d}{\mathrm{d}}
\newcommand{\e}{\mathrm{e}}
\newcommand{\E}{\mathbb{E}}
\renewcommand{\P}{\mathbb{P}}
\renewcommand{\div}{\mathrm{div}}
\newcommand{\R}{\mathbb{R}}
\newcommand{\n}{\mathrm{n}}
\newcommand{\cD}{\mathcal{D}}
\newcommand{\fa}{\mathfrak{a}}
\newcommand{\fb}{\mathfrak{b}}

\newcommand{\DR}{\mathrm{DR}}
\newcommand{\ER}{\mathrm{ER}}
\newcommand{\MT}{\mathrm{MT}}
\newcommand{\reg}{\mathrm{reg}}

\newcommand{\N}{\mathbb N}
\newcommand{\1}{\mathbbm 1}
\newcommand{\supp}{\mathrm{supp}}
\newcommand{\basin}[1]{\mathcal A(#1)}
\newcommand{\testfuncs}{\mathcal C^\infty_{\mathrm c}}
\newcommand{\dsint}{\displaystyle\int}

\newcommand{\Winf}{{{\cW}^{1,\infty}}}
\newcommand{\Id}{\mathrm{Id}}
\newcommand{\epsLimit}[1]{\alpha^{(#1)}}

\newcommand{\hessEigvec}[2]{v^{(#1)}_{#2}} 
\newcommand{\hessEigval}[2]{\nu^{(#1)}_{#2}} 
\newcommand{\Vstar}{V^*}
\newcommand{\Lmu}{L_\beta^2}
\newcommand{\cW}{\mathcal{W}}

\newcommand{\bigo}{\mathcal{O}}
\newcommand{\smallo}{\scalebox{0.7}{$\mathcal O$}}
\DeclareMathOperator{\co}{co}

\newcommand{\corr}{\mathrm{corr}}
\newcommand{\phase}{\mathrm{phase}}

\newtheorem{lemma}{Lemma}

\newtheorem{theorem}{Theorem}
\newtheorem{remark}{Remark}
\newtheorem{corollary}{Corollary}
\newtheorem{proposition}{Proposition}

\newtheorem{hypothesis}{Assumption}
\newtheorem{algorithm}{Algorithm}

\title{Shape optimization of metastable states}
\author{No\'e Blassel$^{1,2}$, Tony Leli\`evre$^{1,2}$, Gabriel Stoltz$^{1,2}$\\
   \small 1: CERMICS, \'Ecole des Ponts, Institut Polytechnique de Paris, Marne-la-Vall\'ee, France  \\
   \small 2: MATHERIALS project-team, Inria Paris, France}

\begin{document}
\maketitle
\begin{abstract}
    The definition of metastable states is a ubiquitous task in the design and analysis of molecular simulations, and is a crucial input in a variety of acceleration methods for the sampling of long configurational trajectories.
    Although standard definitions based on local energy minimization procedures can sometimes be used, these definitions are typically suboptimal, or entirely inadequate when entropic effects are significant, or when the lowest energy barriers are quickly overcome by thermal fluctuations.
    In this work, we propose an approach to the definition of metastable states, based on the shape-optimization of a local separation of timescale metric directly linked to the efficiency of a family of accelerated molecular dynamics algorithms.
    To realize this approach, we derive analytic expressions for shape-variations of Dirichlet eigenvalues for a class of operators associated with reversible elliptic diffusions, and use them to construct a local ascent algorithm, explicitly treating the case of multiple eigenvalues.
    We propose two methods to make our method tractable in high-dimensional systems: one based on dynamical coarse-graining, the other on recently obtained low-temperature shape-sensitive spectral asymptotics.
    We validate our method on a benchmark biomolecular system, showcasing a significant improvement over conventional definitions of metastable states.
\end{abstract}
\tableofcontents

\section{Introduction}
    \label{sec:intro}
    Molecular Dynamics (MD)~\cite{AT17,LM16} is one of the workhorses of modern computational statistical physics, enabling the exploration of complex biomolecular systems at atomistic resolution.
    By numerically integrating equations of motion, MD generates trajectories that sample the system's configuration space according to target statistical ensembles, typically the Boltzmann-Gibbs distribution relevant to canonical (NVT) or isothermal-isobaric (NPT) conditions.
    Understanding phenomena such as protein folding or conformational transitions between functional states hinges on accurately capturing these dynamics over biologically relevant timescales.
    However, the inherent separation of timescales characterizing transitions between metastable states often presents significant computational challenges, motivating the development of enhanced sampling and analysis methodologies to efficiently probe rare events.

    In this work, we are concerned with the definition of these metastable states.
    It is often convenient to associate with a given local minimum of the energy function its basin of attraction for a zero-temperature dynamics. Although this procedure, which provides a natural and numerically convenient definition of metastable states, is often unsatisfactory, for instance in many biological applications where the energy landscape displays many local minima separated by shallow energy barriers.
    In this setting, one seeks alternative, better descriptions, often by replacing the energy with the free energy associated with a given reaction coordinate. In this work, we provide a general and principled approach to define ``good'' metastable states, using techniques of shape optimization originally developed for problems in continuum mechanics.
    More precisely, we optimize the boundary of configurational domains in phase-space, with respect to a certain spectral criterion relating the shape of the domain with so-called quasistationary timescales within the state.
    One of the motivations of this work is to maximize the efficiency of a class of algorithms aimed at sampling long, unbiased molecular trajectories, an example of which is discussed in detail in Appendix~\ref{sec:parrep} below.

    \paragraph{Dynamical setting.}
    To formalize this problem, we first specify the class of models we consider for conformational molecular dynamics, namely reversible elliptic diffusions.
    More precisely, we consider in this work strong solutions to the stochastic differential equation (SDE)
    \begin{equation}
        \label{eq:overdamped_langevin}
        \d X_t = \left[-a(X_t)\nabla V(X_t) +\frac1\beta \div\,a(X_t)\right]\,\d t + \sqrt{\frac2\beta}a(X_t)^{1/2}\,\d W_t,
    \end{equation}
    where~$a:\R^d\to\R^{d\times d}$ is a symmetric positive-definite matrix field,~$\nabla V:\R^d\to\R^d$ is a locally Lipschitz vector field which is the gradient of a potential~$V:\R^d\to\R$,~$\div\,a$ denotes the row-wise divergence operator, and~$W$ is a standard~$d$-dimensional Brownian motion.
    The usefulness of the dynamics~\eqref{eq:overdamped_langevin} comes from the fact that it is reversible, and thus invariant, for the Gibbs probability measure
    \begin{equation}
        \label{eq:gibbs_measure}
        \mu(\d x) = \frac1{\mathcal Z_\beta}\e^{-\beta V(x)}\,\d x,\qquad \mathcal Z_\beta = \int_{\R^d} \e^{-\beta V},
    \end{equation}
    which is the configurational marginal of the canonical (NVT) ensemble at inverse temperature~$\beta=(k_{\mathrm{B}}\theta)^{-1}$ (where~$k_{\mathrm{B}}$ is Boltzmann's constant and~$\theta$ is the temperature)-- provided~$\mathcal Z_\beta$ is finite, which we will always assume.
    As such, it may be used to sample the NVT ensemble. The case~$a=\mathrm{Id}$ corresponds to what is known as the overdamped Langevin equation.
    As all the dynamics~\eqref{eq:overdamped_langevin} sample the same target measure, the free parameter~$a$ can be optimized to accelerate various metrics associated to the efficiency of MCMC samplers, see~\cite{LPRSS24,LSS24,CTZ24}.
    In this work, we consider the problem of sampling trajectories of~\eqref{eq:overdamped_langevin}, with~$a$ fixed.
    The dynamics~\eqref{eq:overdamped_langevin} also arises as the Kramers--Smoluchowski approximation, or so-called overdamped limit, of the kinetic Langevin dynamics, defined by the SDE
    \begin{equation}
        \label{eq:underdamped_langevin}
        \left\{\begin{aligned}
            \d q^\gamma_t &= M^{-1}p^\gamma_t\,\d t,\\
            \d p^\gamma_t &= -\nabla V(q^\gamma_t)\,\d t - \gamma \Gamma(q^\gamma_t)M^{-1}p^\gamma_t\,\d t + \sqrt{\frac{2\gamma}{\beta}}\Sigma(q^\gamma_t)\,\d W^\gamma_t,
        \end{aligned}
        \right.
    \end{equation}
    where the momentum process~$p^\gamma_t$ takes values in~$\R^{d}$,~$W_t^\gamma$ is a standard~$d$-dimensional Brownian motion,~$V$ is as in~\eqref{eq:overdamped_langevin}, and~$M\in\R^{d\times d}$ is a positive-definite mass matrix (typically a diagonal matrix with entries equal to the atomic masses in the system).
    The matrix fields~$\Gamma,\Sigma:\R^d\to \R^{d\times d}$ define fluctuation and dissipation profiles. They are assumed to be non-degenerate, and to satisfy the fluctuation-dissipation condition~$ \Sigma \Sigma^\top  = \Gamma$, which ensures that the Boltzmann--Gibbs distribution with density proportional to~$\e^{-\beta(\frac12p^\top M^{-1}p + V(q))}\,\d p\,\d q$ is invariant under the dynamics.
    The friction parameter~$\gamma>0$ modulates the rate of momentum dissipation, and in this context, the matrix field~$a = \Gamma^{-1}$ arises naturally as the limiting diffusion matrix in the large friction regime.
    More exactly, it can be shown that the finite-time trajectories of the time-rescaled position process~$(q^\gamma_{\gamma t})_{0\leq t\leq T}$ converge to solutions~$(X_t)_{0\leq t \leq T}$ of~\eqref{eq:overdamped_langevin} in the limit~$\gamma\to +\infty$, see for example~\cite{B26}, with~$a=\Gamma^{-1}=\left(\Sigma\Sigma^\top\right)^{-1}$.

    In most MD packages, the Langevin dynamics~\eqref{eq:underdamped_langevin} is implemented with~$\Gamma = M$, in which case~$a=M^{-1}$ in~\eqref{eq:overdamped_langevin}. We therefore use~\eqref{eq:overdamped_langevin} as a model of the underlying underdamped Langevin dynamics with which simulations are typically run, keeping in mind that any timescale inferred at the level of the dynamics~\eqref{eq:overdamped_langevin} should be divided by a factor~$\gamma$ to obtain the corresponding timescale for the underdamped dynamics, in order to account for the rescaling involved in the Kramers--Smoluchowski approximation.

    The infinitesimal generator of the evolution semigroup for the dynamics~\eqref{eq:overdamped_langevin} is the operator
    \begin{equation}
        \label{eq:generator}
        \cL_\beta = \left(-a\nabla V + \frac1\beta\div\,a\right)^\top\nabla + \frac1\beta a:\nabla^2 = \frac1\beta\e^{\beta V}\div\left(\e^{-\beta V} a \nabla \cdot \right).
    \end{equation}
    In an appropriate functional setting (see Section~\ref{subsec:framework} below), it can be shown to be self-adjoint with pure point spectrum.

    \paragraph{Local metastability and quasi-stationary timescales.}
    The main difficulty in sampling long trajectories from the process~\eqref{eq:overdamped_langevin} (as well as from~\eqref{eq:underdamped_langevin}, for that matter) is the phenomenon of metastability, which often arises from the presence of energy wells separated by high-energy barriers (relative to the characteristic thermal fluctuation scale~$\beta^{-1}$), or from entropic traps, see~\cite[Section 1.2.3]{LS16}.
    More generically, this phenomenon can be understood as the presence of subsets of the configuration space in which the dynamics resides for long times before abruptly transitioning and settling in the next metastable state.
    This property is characterized by the    existence of a separation of timescales between intra-state fluctuations and inter-state transitions. In full generality, there may be a hierarchy of timescales, corresponding to states, superstates (energy superbasins), etc.
    In the local approach to metastability, one fixes such a subset~$\Omega\subset\R^d$, and studies local dynamical properties of the system inside~$\Omega$.
    A central object of interest in this study is the quasi-stationary distribution (QSD) of the dynamics inside~$\Omega$, which formalizes the notion of the local equilibrium that the dynamics reaches inside~$\Omega$, provided it remains trapped for a sufficiently long time.
    More formally, the QSD inside~$\Omega$ for the dynamics~\eqref{eq:overdamped_langevin} is defined as a probability measure~$\nu \in \mathcal P_1(\Omega)$ such that, for any~$A\subset \Omega$ measurable,
    \begin{equation}
        \label{eq:qsd}
        \nu(A) = \int_\Omega \mathbb P_x\left(X_t\in A\middle|\,\tau>t\right)\, \nu(\d x),\qquad \tau = \inf\left\{t\geq 0: X_t\not\in \Omega\right\}.
    \end{equation}
    Under mild assumptions on~$\Omega$,~$V$ and~$a$ (see~\cite{LBLLP12} and Assumptions~\eqref{eq:a_ellipticity},~\eqref{eq:coeff_regularity} below), the QSD is unique, and coincides with the Yaglom limit:
    \begin{equation}
        \label{eq:yaglom_limit}
        \nu(A) =\underset{t\to\infty}{\lim}\,\mu_{t,x}(A),\qquad \mu_{t,x}(A) := \P_x\left(X_t\in A\,\middle|\,\tau >t\right),
    \end{equation}
    for an arbitrary initial condition~$x\in \Omega$.

    From this definition alone, it is not entirely clear which domains~$\Omega$ correspond to metastable states.
    A natural albeit informal answer to this question is to require that for most visits in~$\Omega$, convergence to the QSD in~\eqref{eq:yaglom_limit} occurs much faster than the typical metastable exit time~$\E_\nu[\tau]$.
    This definition suggests a quantitative measure of the local metastability of a given domain~$\Omega$, namely the ratio between the metastable exit time and the convergence time to the QSD. Moreover, these timescales can be analyzed by relating them to the eigenvalues of the operator~\eqref{eq:generator}, endowed with Dirichlet boundary conditions on~$\partial\Omega$.
    Indeed, on the one hand it is shown in~\cite[Propositions 2 \& 3]{LBLLP12} that the QSD in~$\Omega$ has an explicit density in terms of the principal Dirichlet eigenfunction~$u_1(\Omega)$ of~$\cL_\beta$ in~$\Omega$:
    \begin{equation}
        \label{eq:qsd_spectral_link}
        \nu(\d x) = Z_{\beta,\Omega}^{-1}\e^{-\beta V(x)}u_{1}(\Omega)(x)\,\d x,\qquad\left\{\begin{aligned}\cL_\beta u_1(\Omega) &= -\lambda_1(\Omega)u_1(\Omega)&\text{ in $\Omega$}, \\ u_1(\Omega) &= 0&\text{on $\partial\Omega$},\end{aligned}\right.
    \end{equation}
    and that the exit time starting from the QSD is an exponential random variable with rate~$\lambda_1(\Omega)$ and independent from the exit point: for all Borel sets~$A\subset\partial\Omega$, it holds
    \begin{equation}
        \label{eq:exit_event}
        \mathbb P_\nu(\tau > t,X_{\tau}\in A) = \e^{-t\lambda_1(\Omega)}\mathbb P_\nu(X_\tau\in A).
    \end{equation}
    In particular, the mean exit time under the QSD (or metastable exit time) is given by~$\E_\nu[\tau] = 1/\lambda_1(\Omega)$.
    In fact, for regular domains, the law of~$X_\tau$ starting under~$\nu$ is also explicit in terms of the normal derivative of the density~$\frac{\d\nu}{\d x}$ on~$\partial\Omega$, see the proof of Proposition 3 in~\cite{LBLLP12}.
    
    Moreover, on the other hand, bounds on the total variation distance between~$\mu_{t,x}$ and~$\nu$ are also available in terms of the spectral gap~$\lambda_{2}(\Omega)-\lambda_{1}(\Omega)$. Namely, a spectral expansion argument (see the proof of~\cite[Theorem~1.1]{SL13}) shows that there exists~$C(x),t(x)>0$ such that
    \begin{equation}
        \label{eq:decorrelation_rate}
        \d_{\mathrm{TV}}\left(\mu_{t,x},\nu\right) \leq C(x)\e^{-t(\lambda_2(\Omega)-\lambda_1(\Omega))},\qquad \forall\,t>t(x),
    \end{equation} 
    where~$\d_{\mathrm{TV}}$ denotes the total variation distance between two probability measures:~$\d_{\mathrm{TV}}(\pi,\rho) = \underset{\|f\|_{\infty}\leq 1}{\sup}\,|\pi(f)-\rho(f)|$.
    The restriction of the estimate~\eqref{eq:decorrelation_rate} to times larger than~$t(x)$ is technical, and is related to the lack of regularity of~$\mu_{0,x}=\delta_x$.
    If one considers initial conditions with sufficient regularity, a similar estimate holds for all~$t>0$.
    It can be shown, e.g. by taking~$X_0\sim C\e^{-\beta V}\left(u_1(\Omega)+\varepsilon u_2(\Omega)\right)\,\d x$ for some appropriate~$C,\varepsilon>0$, that the rate~$\lambda_2(\Omega)-\lambda_1(\Omega)$ in~\eqref{eq:decorrelation_rate} is sharp, and therefore corresponds to the asymptotic rate of convergence of~$\mu_{t,x}$ to~$\nu$.
    
    In view of the above discussion on the exit rate~$\lambda_1(\Omega)$ and the convergence rate to the QSD~$\lambda_2(\Omega)-\lambda_1(\Omega)$, a natural measure of the metastability of the dynamics inside~$\Omega$ is given by the ratio:
    \begin{equation}
        \label{eq:separation_of_timescales}
        N^*(\Omega) = \frac{\lambda_{2}(\Omega)-\lambda_1(\Omega)}{\lambda_1(\Omega)}.
    \end{equation}
    In this work, our aim is to optimize the shape of the domain~$\Omega$ in order to make~$N^*(\Omega)$ as large as possible, see problem~\eqref{eq:shape_optimization_problem} below. The quantity~$N^*(\Omega)$ has been identified in~\cite{V98,PUV15} as a ``scalability metric'' associated with a given definition of metastable state~$\Omega$, which quantifies the efficiency of a class of accelerated MD algorithms, the so-called ``Parallel Replica'' methods.
    We discuss the link between the separation metric~\eqref{eq:separation_of_timescales} and the Parallel Replica method in Appendix~\ref{sec:parrep} below.
    
    Beyond the family of Parallel Replica methods, the other accelerated MD methods developed by Arthur Voter~(see~\cite{V97,SV00}) also rely on definitions of metastable states, and a separation of timescales hypothesis within these states.
    Although our main motivation stems from algorithmic efficiency concerns, we stress that other, more theoretical motivations lead one to consider the problem~\eqref{eq:shape_optimization_problem}.
    It is indeed expected that identifying highly locally metastable domains (in the sense of a large separation of timescales) leads to configurational dynamics amenable to approximation by various simpler, discrete-space dynamics, such as Markov jump processes.  The quantity~\eqref{eq:separation_of_timescales} has for instance been identified as the key approximation parameter in  an approach to reduced-state dynamics using Markov renewal processes~(see~\cite{AJP23}).
    It is therefore of more general interest to investigate how much freedom one has in defining more general states than simple energy basins, and how to ensure a large separation of timescales.
    Let us finally mention that the case~$V=0$, which amounts to maximizing the ratio of the first two Dirichlet eigenvalues of the Laplacian, also arises in the field of spectral geometry as the Payne--Polya--Weinberger conjecture, see~\cite{PPW56,AB92}.
    
    We consider the shape-optimization problem
    \begin{equation}
        \label{eq:shape_optimization_problem}
        \underset{\Omega\in \mathcal S}{\max}\,N^*(\Omega),\qquad \mathcal S=\left\{\Omega\subset \R^d\,\text{ bounded, Lipschitz and connected}\right\}.
    \end{equation}
    The optimization problem as formulated in~\eqref{eq:shape_optimization_problem} is typically not well-posed. Whenever the operator~\eqref{eq:generator} acting on~$L^2(\R^d,\nu(\d x))$ has compact resolvent, a simple argument involving the sequence of domains~$\Omega_n = B_{\R^d}(0,n)$ shows that~$\lambda_1(\Omega_n)\xrightarrow{n\to\infty}0$ and~$\lambda_2(\Omega_n)\xrightarrow{n\to\infty}\lambda_2(\R^d)>0$, so that there is generically no bounded domain maximizing~\eqref{eq:separation_of_timescales}.
    This situation is somewhat standard in the numerical optimization of eigenvalue functionals, and well-posedness is generally only obtained upon imposing normalizing constraints on the design variable, or geometric penalization terms in the objective function. In shape optimization of eigenvalue functionals, these are typically related to the measure of the domain.

    The goal of this work is to develop a method to identify several metastable states, understood as critical shapes of the spectral functional~$N^*$. As such, we are not overly concerned with finding a unique global optimum, nor indeed showing its existence. Rather, we address practical methods to numerically optimize~$N^*$ \textit{locally} around a candidate state definition~$\Omega_0$.
    
    It has been observed (see~\cite{PUV15} or Appendix~\ref{sec:shallow_states} for simple examples of this) that the shape optimization landscape for the separation of timescales displays local maxima around single energy wells (which we define loosely as domains containing a local energy minimum~$z_0$), i.e. domains for which arbitrary perturbations of the boundary locally decrease the separation of timescales, even when these wells are separated by shallow energy barriers.
    The existence of locally optimal domains around wells surrounded by high energy barriers is also supported by theoretical results, see~\cite[Section 3.3]{BLS24} or Section~\ref{sec:semiclassic} below.

    \paragraph{Main contributions of this work.}
    In this work, we introduce a novel and principled approach to the definition of metastable states in MD. In so doing, we make several methodological advances.
\begin{itemize}
    \item{We introduce the spectral criterion~\eqref{eq:separation_of_timescales} and link it to the efficiency of Parallel Replica dynamics.}
    \item{We provide in Theorem~\ref{thm:gateaux_differentiability} and Corollary~\ref{cor:boundary_expression} explicit expressions for shape variations of Dirichlet eigenvalues of a large class of diffusions. These formulas also cover the case of degenerate eigenvalues.}
    \item{We define a robust steepest ascent method (Algorithm~\ref{alg:ascent}) to optimize~$N^*(\Omega)$ in low dimension, taking in particular account the degeneracy of the eigenvalues, and adaptively selecting an ascent direction accordingly.}
    \item{We propose two projection techniques to adapt the algorithm to high-dimensional problems. One is based on a coarse-graining strategy, using a collective variable. The other is based on exact, shape-sensitive spectral asymptotics obtained in the recent work~\cite{BLS24}.}
    \item{We validate our methods with numerical experiments, which demonstrate the interest of the approach on various problems of increasing complexity, including a biomolecular system.}
\end{itemize}

    \paragraph{Outline of the work.}
    In Section~\ref{sec:main_result} we present our main theoretical results,~Theorem~\ref{thm:gateaux_differentiability} and Corollary~\ref{cor:boundary_expression}, which form the basis of our numerical method.
    In Section~\ref{sec:ascent}, we describe the ascent method using the results of Section~\ref{sec:main_result}.
    In Section~\ref{sec:practical_opt}, we discuss two practical methods to approach the shape-optimization problem in high-dimensional systems, which is the standard setting in MD.
    In Section~\ref{sec:numerical}, we present various numerical results to validate our methods.
    Some conclusions and perspectives are gathered in Section~\ref{sec:ccl}. Finally, we conclude this work with two appendices: Appendix~\ref{sec:proof}, in which we give a full proof of Theorem~\ref{thm:gateaux_differentiability}, and Appendix~\ref{sec:parrep}, in which we discuss the relevance to the Parallel Replica algorithm.

    \section{Main results}
    \label{sec:main_result}
    In this section, we present the main theoretical results which form the basis of our optimization method. In Section~\ref{subsec:framework}, we introduce various notation and useful notions. In Section~\ref{subsec:shape_perturbation_formulas}, we state our main result, before proving a reformulation in~\ref{subsec:boundary_form}.

    \subsection{Framework and notation}
    \label{subsec:framework}
        \paragraph{Assumptions on~$V$ and~$a$.}
        We assume that the diffusion matrix~$a$ is locally elliptic: for any compact set~$K\subset \R^d$,
        \begin{equation}
            \label{eq:a_ellipticity}
            \tag{\bf Ell}
            \exists\,\varepsilon_a(K)>0:\qquad u^\top a(x) u \geq \varepsilon_a(K) |u|^2\qquad\forall u\in \R^d,\text{ for almost all~$x\in K$}.
        \end{equation}
        We also assume that~$V$ and~$a$ have locally bounded derivatives up to order~$2$:
        \begin{equation}
            \label{eq:coeff_regularity}
            \tag{\bf Reg}
            V\in \cW_{\mathrm{loc}}^{2,\infty}\left(\R^d\right),\qquad a\in \cW_{\mathrm{loc}}^{2,\infty}(\R^d;\mathcal M_d).
        \end{equation}
        \paragraph{Functional spaces.}
                Throughout this work, we consider the following Hilbert spaces, defined for an open Lipschitz domain~$\Omega\subset \R^d$ by
                    \begin{equation}
                        \label{eq:sobolev_spaces}
                        \begin{aligned}
                        &\Lmu(\Omega) = \left\{u\text{ measurable }\middle|\,\|u\|^2_{\Lmu(\Omega)} :=\int_{\Omega} u^2 \,\e^{-\beta V} < +\infty\right\},\\
                        &H^{k}_\beta(\Omega) = \left\{u\in \Lmu(\Omega)\middle|\,\partial^{\alpha}u\in \Lmu(\Omega),\,\forall\, |\alpha|\leq k\right\},
                        \end{aligned}
                    \end{equation}
                    where~$\partial^\alpha = \partial_{x_1}^{\alpha_1}\dots\partial_{x_d}^{\alpha_d}$ denotes the weak differentiation operator associated to a multi-index~$\alpha = (\alpha_1,\dots,\alpha_d)\in \R^d$. For the flat case (i.e. when~$V\equiv 0$), we simply write~$L^2(\Omega)$ and~$H^k(\Omega)$. As in the flat case,~$H_{0,\beta}^k(\Omega)$ denotes the~$H_\beta^k(\Omega)$-norm closure of~$\testfuncs(\Omega)$.

                    If~$\Omega$ is bounded (which will be the case in the following) and for any~$k\in\N$, the sets~$H^{k}_\beta(\Omega)$ and~$H^{k}(\Omega)$ are equal as Banach spaces, but are endowed with different inner products.
        \paragraph{Lipschitz shape perturbations.}
        For the purpose of studying shape perturbations of eigenvalues, we introduce an appropriate Banach space of deformation fields.
        We denote by~$\Winf(\R^d;\R^d)$ (or simply~$\Winf$) the set of essentially bounded vector fields with essentially bounded weak differential:
        \begin{equation}
            \label{eq:lipschitz_functions}
            \Winf(\R^d;\R^d) = \left\{\theta:\,\R^d\to\R^d\text{ measurable } \middle|\,\left\|\theta\right\|_{\Winf} :=\left\|\theta\right\|_{L^{\infty}\left(\R^d;\R^d\right)} + \left\|\nabla\theta\right\|_{L^{\infty}\left(\R^d;\mathcal M_d\right)} < +\infty\right\},
        \end{equation}
        where~$\R^d$ is endowed with the Euclidean norm and where~$\mathcal M_d$ denotes the space of~$d\times d$ matrices, which is endowed with the induced operator norm.
        For any finite-dimensional vector space~$E$ and~$\theta\in \Winf(\R^d;E)$,~$\theta$ has a Lipschitz-continuous representative (see for example~\cite[Section 5.8.2.b, Theorem 4]{E22}). We will therefore identify throughout this work elements of~$\Winf(\R^d;E)$ with their Lipschitz representatives.
        The normed vector space~$(\Winf,\|\cdot\|_{\Winf})$ is a Banach space, and due to Rademacher's theorem,~$\nabla\theta\in \Winf$ is differentiable almost everywhere.
        We use the convention~$(\nabla \theta)_{ij} = \partial_i \theta_j$, so that~$\nabla \theta = D\theta^\top \in \R^{d\times d}$ is the transpose of the Jacobian matrix.

        The interest of this class of perturbations is the stability of the class of Lipschitz domains under~$\Winf$ shape perturbations, as formalized by the following result.
        \begin{proposition}
            \label{prop:domain_stability}
            Let~$\Omega\subset \R^d$ be a bounded, open Lipschitz domain, and~$k\geq 1$. There exists~$h_0>0$ such that, for all~$\theta\in B_{\Winf}(0,h_0)$,
            \begin{equation}
                \label{eq:stable_perturbation}
                \Omega_\theta := (\Id+\theta)\Omega = \left\{x+\theta(x),x\in\Omega\right\}
            \end{equation}
            is still a bounded, open Lipschitz domain.
        \end{proposition}
        We depict schematically the perturbed domain~\eqref{eq:stable_perturbation} in Figure~\ref{fig:shape_perturbation}.
        The proof of Proposition~\ref{prop:domain_stability} relies on the fact that bounded Lipschitz domains are characterized by a geometric condition in the class of so-called uniform~$\varepsilon$-cone conditions, which is stable under bi-Lipschitz homeomorphisms. We refer to~\cite[Section III]{C75} for a proof of this result.
        Another straightforward but important property of this class of perturbations is that the composition mapping
        \begin{equation}
            \label{eq:composition}
            \left\{\begin{aligned}
                H_0^1(\Omega)&\to H_0^1(\Omega_\theta)\\
                v&\mapsto v\circ \Phi_\theta,
            \end{aligned}\right.
        \end{equation}
        where~$\Phi_\theta(x) = x + \theta(x)$, is a Banach space isomorphism for~$\|\theta\|_{\Winf}$ sufficiently small, with inverse~$v_\theta\mapsto v_\theta\circ \Phi_\theta^{-1}$.
        \begin{figure}
            \center
            \begin{tikzpicture}[x=0.75pt,y=0.75pt,yscale=-0.7,xscale=0.7]
                \draw  [dash pattern={on 1.69pt off 2.76pt}][line width=1.5]  (133.33,43.67) .. controls (171.33,-34.33) and (201.33,62.67) .. (337.33,69.67) .. controls (473.33,76.67) and (526.33,151.67) .. (493.33,229.67) .. controls (460.33,307.67) and (529.33,360.67) .. (475.33,398.67) .. controls (421.33,436.67) and (396.33,276.67) .. (303.33,278.67) .. controls (210.33,280.67) and (220.33,379.17) .. (167.33,373.67) .. controls (114.33,368.17) and (142.33,328.67) .. (109.33,286.67) .. controls (76.33,244.67) and (11.33,245.67) .. (33.33,189.67) .. controls (55.33,133.67) and (95.33,121.67) .. (133.33,43.67) -- cycle ;
                \draw  [line width=1.5]  (133.33,36.67) .. controls (203.33,17.67) and (175.33,91.67) .. (308.33,120.67) .. controls (441.33,149.67) and (548.33,87.67) .. (552.33,167.67) .. controls (556.33,247.67) and (442.33,194.67) .. (462.33,300.67) .. controls (482.33,406.67) and (404.33,323.67) .. (311.33,325.67) .. controls (218.33,327.67) and (216.33,433.17) .. (163.33,427.67) .. controls (110.33,422.17) and (159.33,293.67) .. (126.33,251.67) .. controls (93.33,209.67) and (50.33,227.67) .. (42.33,171.67) .. controls (34.33,115.67) and (63.33,55.67) .. (133.33,36.67) -- cycle ;
                \draw [color={rgb, 255:red, 255; green, 0; blue, 0 }  ,draw opacity=1 ]   (393.33,129.67) -- (398.13,81.66) ;
                \draw [shift={(398.33,79.67)}, rotate = 95.71] [color={rgb, 255:red, 255; green, 0; blue, 0 }  ,draw opacity=1 ][line width=0.75]    (10.93,-3.29) .. controls (6.95,-1.4) and (3.31,-0.3) .. (0,0) .. controls (3.31,0.3) and (6.95,1.4) .. (10.93,3.29)   ;
                \draw [color={rgb, 255:red, 255; green, 0; blue, 0 }  ,draw opacity=1 ]   (552.33,167.67) -- (506.82,172.78) ;
                \draw [shift={(504.83,173)}, rotate = 353.59] [color={rgb, 255:red, 255; green, 0; blue, 0 }  ,draw opacity=1 ][line width=0.75]    (10.93,-3.29) .. controls (6.95,-1.4) and (3.31,-0.3) .. (0,0) .. controls (3.31,0.3) and (6.95,1.4) .. (10.93,3.29)   ;
                \draw [color={rgb, 255:red, 255; green, 0; blue, 0 }  ,draw opacity=1 ]   (308.33,120.67) -- (320.85,70.61) ;
                \draw [shift={(321.33,68.67)}, rotate = 104.04] [color={rgb, 255:red, 255; green, 0; blue, 0 }  ,draw opacity=1 ][line width=0.75]    (10.93,-3.29) .. controls (6.95,-1.4) and (3.31,-0.3) .. (0,0) .. controls (3.31,0.3) and (6.95,1.4) .. (10.93,3.29)   ;
                \draw [color={rgb, 255:red, 253; green, 0; blue, 0 }  ,draw opacity=1 ]   (239.33,95.67) -- (250.29,56.92) ;
                \draw [shift={(250.83,55)}, rotate = 105.79] [color={rgb, 255:red, 253; green, 0; blue, 0 }  ,draw opacity=1 ][line width=0.75]    (10.93,-3.29) .. controls (6.95,-1.4) and (3.31,-0.3) .. (0,0) .. controls (3.31,0.3) and (6.95,1.4) .. (10.93,3.29)   ;
                \draw [color={rgb, 255:red, 255; green, 0; blue, 0 }  ,draw opacity=1 ]   (178.33,43.67) -- (182.96,19.63) ;
                \draw [shift={(183.33,17.67)}, rotate = 100.89] [color={rgb, 255:red, 255; green, 0; blue, 0 }  ,draw opacity=1 ][line width=0.75]    (10.93,-3.29) .. controls (6.95,-1.4) and (3.31,-0.3) .. (0,0) .. controls (3.31,0.3) and (6.95,1.4) .. (10.93,3.29)   ;
                \draw [color={rgb, 255:red, 255; green, 0; blue, 0 }  ,draw opacity=1 ]   (60.33,86.67) -- (89.66,105.9) ;
                \draw [shift={(91.33,107)}, rotate = 213.26] [color={rgb, 255:red, 255; green, 0; blue, 0 }  ,draw opacity=1 ][line width=0.75]    (10.93,-3.29) .. controls (6.95,-1.4) and (3.31,-0.3) .. (0,0) .. controls (3.31,0.3) and (6.95,1.4) .. (10.93,3.29)   ;
                \draw [color={rgb, 255:red, 255; green, 0; blue, 0 }  ,draw opacity=1 ]   (103.33,231.67) -- (84.42,260.99) ;
                \draw [shift={(83.33,262.67)}, rotate = 302.83] [color={rgb, 255:red, 255; green, 0; blue, 0 }  ,draw opacity=1 ][line width=0.75]    (10.93,-3.29) .. controls (6.95,-1.4) and (3.31,-0.3) .. (0,0) .. controls (3.31,0.3) and (6.95,1.4) .. (10.93,3.29)   ;
                \draw [color={rgb, 255:red, 255; green, 0; blue, 0 }  ,draw opacity=1 ]   (163.33,427.67) -- (167.19,375.66) ;
                \draw [shift={(167.33,373.67)}, rotate = 94.24] [color={rgb, 255:red, 255; green, 0; blue, 0 }  ,draw opacity=1 ][line width=0.75]    (10.93,-3.29) .. controls (6.95,-1.4) and (3.31,-0.3) .. (0,0) .. controls (3.31,0.3) and (6.95,1.4) .. (10.93,3.29)   ;
                \draw [color={rgb, 255:red, 255; green, 0; blue, 0 }  ,draw opacity=1 ]   (241.33,356.33) -- (223.08,333.56) ;
                \draw [shift={(221.83,332)}, rotate = 51.29] [color={rgb, 255:red, 255; green, 0; blue, 0 }  ,draw opacity=1 ][line width=0.75]    (10.93,-3.29) .. controls (6.95,-1.4) and (3.31,-0.3) .. (0,0) .. controls (3.31,0.3) and (6.95,1.4) .. (10.93,3.29)   ;
                \draw [color={rgb, 255:red, 255; green, 0; blue, 0 }  ,draw opacity=1 ]   (311.33,325.67) -- (307.05,285.99) ;
                \draw [shift={(306.83,284)}, rotate = 83.84] [color={rgb, 255:red, 255; green, 0; blue, 0 }  ,draw opacity=1 ][line width=0.75]    (10.93,-3.29) .. controls (6.95,-1.4) and (3.31,-0.3) .. (0,0) .. controls (3.31,0.3) and (6.95,1.4) .. (10.93,3.29)   ;
                \draw [color={rgb, 255:red, 255; green, 0; blue, 0 }  ,draw opacity=1 ]   (450.33,353.67) -- (468.61,401.13) ;
                \draw [shift={(469.33,403)}, rotate = 248.94] [color={rgb, 255:red, 255; green, 0; blue, 0 }  ,draw opacity=1 ][line width=0.75]    (10.93,-3.29) .. controls (6.95,-1.4) and (3.31,-0.3) .. (0,0) .. controls (3.31,0.3) and (6.95,1.4) .. (10.93,3.29)   ;
                \draw [color={rgb, 255:red, 255; green, 0; blue, 0 }  ,draw opacity=1 ]   (466.33,322.67) -- (488.34,323.58) ;
                \draw [shift={(490.33,323.67)}, rotate = 182.39] [color={rgb, 255:red, 255; green, 0; blue, 0 }  ,draw opacity=1 ][line width=0.75]    (10.93,-3.29) .. controls (6.95,-1.4) and (3.31,-0.3) .. (0,0) .. controls (3.31,0.3) and (6.95,1.4) .. (10.93,3.29)   ;
                \draw [color={rgb, 255:red, 255; green, 0; blue, 0 }  ,draw opacity=1 ]   (460.33,264.67) -- (481.42,271.08) ;
                \draw [shift={(483.33,271.67)}, rotate = 196.93] [color={rgb, 255:red, 255; green, 0; blue, 0 }  ,draw opacity=1 ][line width=0.75]    (10.93,-3.29) .. controls (6.95,-1.4) and (3.31,-0.3) .. (0,0) .. controls (3.31,0.3) and (6.95,1.4) .. (10.93,3.29)   ;
                \draw [color={rgb, 255:red, 255; green, 0; blue, 0 }  ,draw opacity=1 ]   (105.33,46.67) -- (112.82,75.07) ;
                \draw [shift={(113.33,77)}, rotate = 255.23] [color={rgb, 255:red, 255; green, 0; blue, 0 }  ,draw opacity=1 ][line width=0.75]    (10.93,-3.29) .. controls (6.95,-1.4) and (3.31,-0.3) .. (0,0) .. controls (3.31,0.3) and (6.95,1.4) .. (10.93,3.29)   ;
                \draw [color={rgb, 255:red, 255; green, 0; blue, 0 }  ,draw opacity=1 ]   (41.33,141.67) -- (57.84,142.86) ;
                \draw [shift={(59.83,143)}, rotate = 184.12] [color={rgb, 255:red, 255; green, 0; blue, 0 }  ,draw opacity=1 ][line width=0.75]    (10.93,-3.29) .. controls (6.95,-1.4) and (3.31,-0.3) .. (0,0) .. controls (3.31,0.3) and (6.95,1.4) .. (10.93,3.29)   ;
                \draw [color={rgb, 255:red, 255; green, 0; blue, 0 }  ,draw opacity=1 ]   (54.33,201.67) -- (31.05,215.64) ;
                \draw [shift={(29.33,216.67)}, rotate = 329.04] [color={rgb, 255:red, 255; green, 0; blue, 0 }  ,draw opacity=1 ][line width=0.75]    (10.93,-3.29) .. controls (6.95,-1.4) and (3.31,-0.3) .. (0,0) .. controls (3.31,0.3) and (6.95,1.4) .. (10.93,3.29)   ;

            \draw (45,43.4) node [anchor=north west,scale=1.2][inner sep=0.75pt]    {$\partial \Omega $};
            \draw (289,39.4) node [anchor=north west,scale=1.2][inner sep=0.75pt]    {$\partial \Omega _{\theta }$};
            \draw (283,301.4) node [anchor=north west,scale=1.2][inner sep=0.75pt]    {$\textcolor[rgb]{1,0,0}{\theta }$};
            \end{tikzpicture}
            \caption{\label{fig:shape_perturbation} The standard framework of the Hadamard shape derivative: a reference domain~$\Omega$ is deformed into~$\Omega_\theta$ defined in~\eqref{eq:stable_perturbation} following a perturbation field~$\theta\in\Winf$. Regularity properties of a shape functional~$J(\Omega)$ are studied via those of the map~$\theta\mapsto J(\Omega_\theta)$.}
        \end{figure}
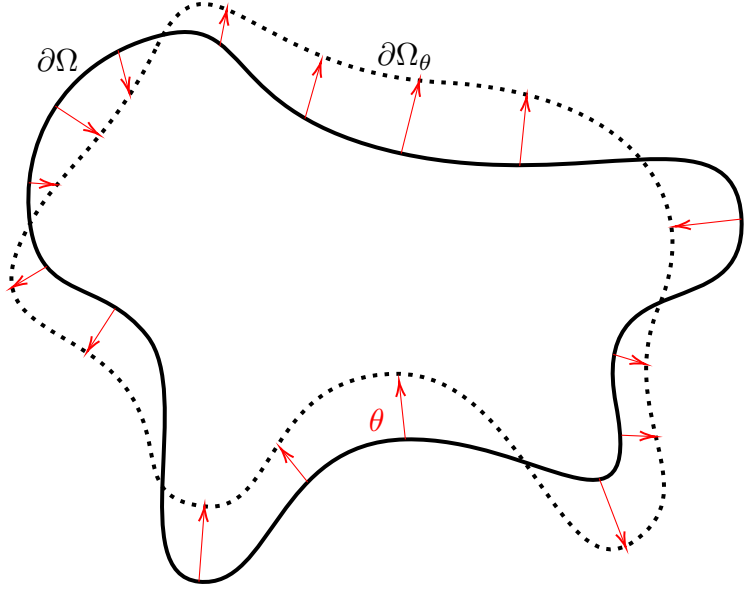
        
        \paragraph{Spectral properties of the Dirichlet generator.}
        We recall that the evolution semigroup associated with the diffusion~\eqref{eq:overdamped_langevin} is generated by the operator~\eqref{eq:generator}.
        Given a bounded open domain~$\Omega$, the Dirichlet realization of~$-\cL_\beta$ on~$L^2_\beta(\Omega)$, also denoted by~$-\cL_\beta$, is defined as the Friedrichs extension (see~\cite{T14}) of the positive quadratic form
        \begin{equation}
            \mathcal C_{\mathrm{c}}^\infty(\Omega)\ni u\mapsto \frac1\beta\int_{\Omega}\nabla u^\top a \nabla u\,\e^{-\beta V}.
        \end{equation}
        It is a self-adjoint operator with domain~$\mathcal D(\cL_\beta) = \left\{u\in H_{0,\beta}^1(\Omega):\, \cL_\beta u\in L^2_\beta(\Omega)\right\}$.
       If~$\Omega$ is a smooth domain, we simply have~$\mathcal D(\cL_\beta)=H_{2,\beta}(\Omega)\cap H_{0,\beta}^1(\Omega)$.

        Since~$\mathcal D(\cL_\beta)\subset H_{0,\beta}^1(\Omega)$ is compactly embedded in~$L^2_\beta(\Omega)$, $-\cL_\beta$ has compact resolvent, and its spectrum is composed of a sequence
        $$0<\lambda_1(\Omega) \leq \lambda_2(\Omega)\dots $$
        of eigenvalues with finite multiplicities tending to~$+\infty$. We enumerate the spectrum with multiplicity, and consider the following normalization for eigenvectors: for any integers~$i,j\geq 1$,
        \begin{equation}
            \label{eq:eigfunc_normalization}
            \int_{\Omega}u_{i}(\Omega)u_{j}(\Omega) \e^{-\beta V} = \delta_{ij},
        \end{equation}
        where for any~$k\geq 1$,~$u_{k}(\Omega)\in L^2_\beta(\Omega)$ satisfies the eigenrelation~$-\cL_\beta u_k(\Omega) = \lambda_k(\Omega)u_k(\Omega)$. It can also be shown that the eigenfunction associated with~$\lambda_1(\Omega)$ is a signed function~$u_1(\Omega)$ (since~$|u_1(\Omega)|\in Q(\cL_\beta)$ is also a minimizer of the quadratic form), which is unique up to normalization, and the Harnack inequality implies that~$u_1(\Omega)$ does not vanish inside~$\Omega$.
        Therefore, the orthogonality constraint~\eqref{eq:eigfunc_normalization} forces the principal eigenvalue to be simple, i.e.~$0<\lambda_1(\Omega)<\lambda_2(\Omega)$. Moreover, one can choose~$u_1(\Omega)$ to be positive in~$\Omega$, which will be our convention throughout this work.

        Precise statements regarding the spectral properties of~$\cL_\beta$ will be given in the proof of Theorem~\ref{thm:gateaux_differentiability} below.

    \paragraph{Shape perturbation analysis.}
        In Section~\ref{sec:main_result}, we derive regularity results (Theorem~\ref{thm:gateaux_differentiability}) for the Dirichlet eigenvalues of the generator~$-\cL_\beta$ with respect to Lipschitz shape perturbations.
        To do so, we adopt the standard framework of shape calculus, considering mappings from perturbations of the domain to eigenvalues
        \begin{equation}
            \theta\mapsto \lambda_k(\Omega_\theta),\qquad \forall\,k\geq 1,
        \end{equation}
        and obtain regularity results with respect to~$\theta\in\Winf$ with explicit first-order formulas.

        To illustrate the main difficulty when dealing with eigenvalues, we consider the following two-dimensional example, which already gives insight into the infinite-dimensional situation.
        Consider the following matrix-valued map~$A:\R^2\to\R^{2\times 2}$ (which depends on two independent parameters, and therefore lies outside the scope of analytic perturbation theory): 
        $$A(\theta) = \begin{pmatrix}
            -\theta_1& \theta_2\\ \theta_2& \theta_1
        \end{pmatrix},\qquad \mathrm{Spec}\,A(\theta) = \left\{\pm \sqrt{\theta_1^2+\theta_2^2}\right\}.$$
        Simple eigenvalues remain Fréchet-differentiable with respect to~$\theta$.
        One does not however have Fréchet differentiability for~degenerate eigenvalues (as~$0$ for~$\theta=0$ above), even if one is free to choose the ordering of the eigenvalues. Indeed, there is no local parametrization of~$\mathrm{Spec}\,A(\theta)$ as the union of two differentiable surfaces in a neighborhood of~$\theta=0$: geometrically, it is a double cone in~$\R^3$ with a vertex at~$\theta=0$.  However, it is simple to see that, for a fixed perturbation direction~$\theta\in \R^2$, the set $\mathrm{Spec}\,A(t\theta)$ may be parametrized as the union of two differentiable graphs, namely~$t\mapsto \pm t|\theta|$, and in this sense the degenerate eigenvalue is Gateaux-differentiable.
        If one moreover orders the eigenvalues, one gets the parametrization~$t\mapsto \pm|t\theta|$, and the eigenvalues are again non-differentiable at~$t=0$ (even in the sense of Gateaux), but only semi-differentiable, with well-defined left and right derivatives. This is simply an artifact of the non-differentiability of the ordering map, which nevertheless is semi-differentiable on the diagonal~$\{(x,y)\in \R^2:\, x=y\}$.
        
        The case of the Dirichlet eigenvalues of~$-\cL_\beta$ is similar.
        Namely, for a degenerate eigenvalue~$\lambda_k(\Omega)$ of multiplicity~$m$ and a fixed perturbation direction~$\theta\in\Winf$, the spectral cluster
        $$\left\{\lambda_{k+\ell}(\Omega_{t\theta}),0\leq \ell < m,\,|t|\text{ small}\right\}$$
         around~$\lambda_k(\Omega)$ depends differentiably on~$t$, in a sense made precise in Theorem~\ref{thm:gateaux_differentiability} below.
        It is also the case that, if~$\lambda_k(\Omega)$ is simple, then~$\theta\mapsto \lambda_k(\Omega_\theta)$ has~$\mathcal C^1(\Winf)$-regularity in a neighborhood of~$0$, a property known as shape-differentiability.
        In both the simple and degenerate cases, explicit formulas~for the directional one-sided derivatives (and thus also the Fréchet derivative in the simple case) of the ordered eigenvalues~$\lambda_{k+\ell}(\Omega_\theta)$ with respect to~$\theta$ at~$\theta=0$ are available for~$0\leq \ell<m$.
        These results justify formal computations~(see~Corollary~\ref{cor:boundary_expression} below), generalizing those of Hadamard~\cite{H08} for the Laplacian, and allowing for the identification of shape-ascent directions for smooth functionals of the Dirichlet spectrum. This forms the crux of our numerical method, see Section~\ref{sec:ascent} below.
        The general strategy we follow was proposed by Haug and Rousselet in~\cite{HR80a,HR80b,R83,HR83} for problems in structural mechanics.
        
        However, besides the fact that the operators we consider here are different from those in~\cite{R83,HR83}, the regularity results we prove are stronger than those derived in~\cite{HR80a,HR80b,R83,HR83} (for instance, we show Fréchet-differentiability of simple eigenvalues in a~$\Winf$-neighborhood of~$\theta=0$). These results require locally uniform-in-$\theta$ estimates throughout the proof, and we therefore give a self-contained derivation.
        
        Let us also mention the books~\cite[Section 5.7]{HP05} and~\cite[Section 2.5]{H06} for a more pedagogical and somewhat less technical approach than our proof in the case of the Laplacian, but which only applies to the case of simple eigenvalues.

    \subsection{Shape perturbation formulas}
    \label{subsec:shape_perturbation_formulas}
    Our main result is the following theorem, which summarizes the regularity properties for the Dirichlet ordered eigenvalue maps~$\theta\mapsto \lambda_k(\Omega_\theta)$, with explicit expressions for the directional derivatives at~$\theta=0$ in terms of a~$L^2_\beta(\Omega)$-orthonormal basis of eigenvectors.
    Crucially, formulas are still available in the case of degenerate eigenvalues.
\begin{theorem}
    \label{thm:gateaux_differentiability}
    Let~$\Omega\subset \R^d$ be a bounded Lipschitz domain, and~$\lambda_k(\Omega)=\lambda_{k+\ell}(\Omega)$ for~$0\leq \ell < m$ be a  multiplicity~$m\geq 1$ eigenvalue for the operator~$-\cL_\beta$ on~$\Omega$ with Dirichlet boundary conditions. Let~$\left(u_k^{(i)}(\Omega)\right)_{1\leq i\leq m}$ be a basis of eigenvectors for the associated invariant subspace of~$L^2_\beta(\Omega)$, satisfying the normalization convention~\eqref{eq:eigfunc_normalization}.
    We recall that, for~$\theta\in\Winf\left(\R^d,\R^d\right)$, the transported domain is denoted $\Omega_\theta = (\Id+\theta)\Omega$. The following properties hold.
        \begin{enumerate}[i)]
        \item{
        The map~$\theta\mapsto \left(\lambda_{k+\ell}(\Omega_\theta)\right)_{0\leq \ell < m}$
        is Lipschitz in a~$\Winf$-neighborhood of~$\theta=0$.
        }
        \item{
        Fix~$\theta\in\Winf\left(\R^d,\R^d\right)$. There exist~$t_\theta>0$ and~$m$ differentiable maps
        \begin{equation}
            (-t_\theta,t_\theta)\ni t \mapsto \mu_\ell(t),\qquad 1\leq \ell\leq m
        \end{equation}
        such that~
        \begin{equation}
            \label{eq:multiset}
            \left\{\mu_\ell(t),\,1\leq \ell\leq m\right\} = \left\{ \lambda_{k+\ell}(\Omega_{t\theta}),\,0\leq \ell<m\right\}
        \end{equation}
        for all~$t\in(-t_\theta,t_\theta)$.

        Moreover, the set~$\left\{\mu_\ell'(0),\,1\leq \ell\leq m\right\}$ of derivatives at~$t=0$ is the spectrum of the symmetric matrix~$M^{\Omega,k}(\theta)$ with entries, for~$1\leq i,j\leq m$:
        \begin{equation}
            \begin{aligned}
                \label{eq:derivative_volume_expression}
            M^{\Omega,k}_{ij}(\theta) = &\frac1\beta\int_\Omega \nabla u_k^{(i)}(\Omega)^\top \left(\nabla a^\top \theta -a\nabla \theta -\nabla\theta^\top a\right)\nabla u_k^{(j)}(\Omega)\e^{-\beta V}
            \\&+ \frac1\beta \int_\Omega \nabla u_k^{(i)}(\Omega)^\top a \nabla u_k^{(j)}(\Omega)\div\left(\theta\e^{-\beta V}\right)\\
            &-\lambda_k(\Omega)\int_\Omega u_k^{(i)}(\Omega)u_k^{(j)}(\Omega)\div\left(\theta\e^{-\beta V}\right).
        \end{aligned}
        \end{equation}
        }
        \item{
            If~$\lambda_k(\Omega)$ is a simple eigenvalue, i.e.~$m=1$, then the map~$\theta\mapsto \lambda_k(\Omega_\theta)$ is~$\mathcal C^1(\Winf;\R)$ in a~$\Winf$-neighborhood of~$\theta=0$.
        }
    \end{enumerate}
    In the expression~\eqref{eq:derivative_volume_expression} above, we use the shorthand~$\nabla a^\top \theta$ for the matrix with entries~$\sum_{\alpha=1}^d\partial_{\alpha}a_{ij}\theta_\alpha$.
\end{theorem}
\begin{remark}
    Note that, from the second item in Theorem~\ref{thm:gateaux_differentiability}, the Gateaux right-derivatives of the ordered eigenvalues can be deduced from the ordering of the eigenvalues of the matrix~$M^{\Omega,k}$ defined in~\eqref{eq:derivative_volume_expression}.
    Namely, for any~$0\leq \ell< m$, the right-derivative~$\frac{\d}{\d t}\lambda_{k+\ell}(\Omega_{t\theta})\big|_{t=0^+}$ is given by the~$\ell$-th smallest eigenvalue of~$M^{\Omega,k}(\theta)$, counted with multiplicity.
    This simply follows by comparing the first-order expansions of the eigenvalues given in~\eqref{eq:multiset}.
    It may happen that~$M^{\Omega,k}(\theta)$ has degenerate eigenvalues, in which case some eigenvalue branches are tangent to one another, and~$\lambda_{k}(\Omega_{t\theta})$ remains degenerate to first-order in~$t$ around~$t=0$.
    Such a situation is depicted in Figure~\ref{fig:multiple_eigenvalues} below.
\end{remark}
As the proof of Theorem~\ref{thm:gateaux_differentiability} is somewhat lengthy, it is postponed to Appendix~\ref{sec:proof} below.
    \begin{figure}
    \center
\begin{tikzpicture}
    \begin{axis}[
        axis lines=center,
        xtick=\empty, ytick=\empty,
        xlabel={$t$},
        ylabel={$\mathrm{Spec(-\cL_\beta(\Omega_{t\theta}))}$},
        domain=-0.25:0.5, 
        samples=100, 
        clip=false,
        width=0.8\textwidth]
        
        \addplot[red, thick] {max(max(x^2-x,x^3-x),2*x-x^2)} node[red, anchor=west] {$\lambda_{k+2}(\Omega_{t\theta})$};
        \addplot[blue, thick] {min(min(x^2-x,x^3-x),2*x-x^2)} node[blue, anchor=west] {$\lambda_{k}(\Omega_{t\theta})$};
        \addplot[brown, thick] {x^3-max(max(x^2-x,x^3-x),2*x-x^2)-min(min(x^2-x,x^3-x),2*x-x^2)} node[brown, anchor=west] {$\lambda_{k+1}(\Omega_{t\theta})$};
        
        \addplot[black, dashed, domain = 0:0.5] {-x};
        \addplot[black, dashed, domain = 0:0.5] {-x};
        \addplot[black, dashed, domain = 0:0.5] {2*x};
    \end{axis}
\end{tikzpicture}
\caption{Directional shape perturbation of the triple Dirichlet eigenvalue~$\lambda_k(\Omega)$ in the direction~$\theta$. The slopes of the Gateaux right-tangents (in black dashed lines) correspond to the eigenvalues of the matrix~$M^{\Omega,k}(\theta)$ (counted with multiplicity). In this case, the bottom eigenvalue has multiplicity two, and two half-tangents coincide.}
\label{fig:multiple_eigenvalues}
\end{figure}
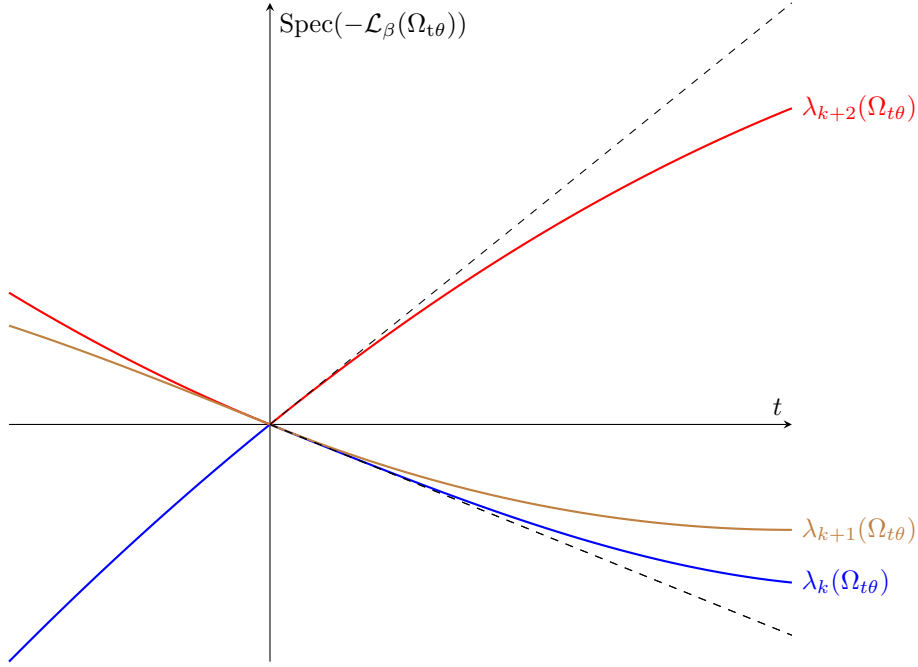

\subsection{Revisiting eigenvalue derivatives as boundary integrals}
\label{subsec:boundary_form}
The next result states that the components of the matrix~\eqref{eq:derivative_volume_expression} defining the directional derivatives of a multiple eigenvalue have a simpler form, provided that the boundary has sufficient regularity.
\begin{corollary}
    \label{cor:boundary_expression}
    Assume~that $\Omega$ is convex or has a~$\mathcal C^{1,1}$ boundary.
    Then the components~\eqref{eq:derivative_volume_expression} can be rewritten as the following boundary integrals for~$1\leq i,j\leq m$:  
    \begin{equation}
        \label{eq:boundary_form}
        M^{\Omega,k}_{ij}(\theta) = -\frac1\beta\int_{\partial\Omega} \frac{\partial u_k^{(i)}}{\partial\n}\frac{\partial u_k^{(j)}}{\partial\n}\left(\n^\top a \n\right) \left(\theta^\top\n\right)\e^{-\beta V},
    \end{equation}
    where~$\n$ denotes the unit outward normal to~$\partial\Omega$ and~$\frac{\partial u}{\partial \n} =\nabla u^\top\n$ denotes the normal derivative.
\end{corollary}
Compared to~\eqref{eq:derivative_volume_expression}, the form~\eqref{eq:boundary_form} is useful from the numerical point of view, since it does not involve any derivative of the diffusion tensor~$a$ or of the perturbation field~$\theta$. As such, it is the one we use for the purpose of numerical shape optimization, see Section~\ref{sec:ascent} below.

\begin{proof}[Proof of Corollary~\ref{cor:boundary_expression}.]
    We fix~$1\leq i,j\leq m$ and for simplicity, we denote by~$u_k^{(i)}(\Omega)=u$,~$u_k^{(j)}(\Omega)=v$ and~$\lambda_k(\Omega)=\lambda$. By standard results of elliptic regularity (see~\cite[Theorems 2.4.2.5 and~3.2.1.3]{G11}), the regularity of~$\partial\Omega$ or the convexity of~$\Omega$ ensure that~$u$ and $v$ belong to~$H^2(\Omega)$, so that~$\nabla u,\nabla v,\theta \in L^2(\partial\Omega)$ by the Sobolev trace theorem,
    with furthermore, since~$u,v\in H_0^1(\Omega)$,
    \begin{equation}
        \label{eq:boundary_gradients}
        \nabla u = \frac{\partial u}{\partial \n}\n,\quad \nabla v = \frac{\partial v}{\partial\n}\n\qquad \text{in } L^2(\partial\Omega)^d,
    \end{equation}
    where~$\nabla u,\nabla v$ are defined in~$L^2(\partial\Omega)$ in the sense of the trace.
    We recall a Green-like identity for~$f\in H^1(\Omega)$ and~$g\in\mathcal D(\cL_\beta)$. In view of the following equality in~$L^1(\Omega)$
    \begin{equation}
        \begin{aligned}
            \frac1\beta\div\left(f\e^{-\beta V}a\nabla g\right) &= \frac1\beta f\div\left(\e^{-\beta V}a\nabla g\right) + \frac1\beta \nabla f^\top a \nabla g\e^{-\beta V}\\
            &= \left(f\cL_\beta g + \frac1\beta\nabla f^\top a \nabla g\right)\e^{-\beta V},
        \end{aligned}
    \end{equation}
    the Green--Ostrogradski formula gives
    \begin{equation}
        \label{eq:green_formula}
        \frac1\beta\int_{\partial \Omega} f\n^\top a\nabla g\,\e^{-\beta V} = \int_{\Omega}f\cL_\beta g\,\e^{-\beta V} + \frac1\beta\int_\Omega\nabla f^\top a \nabla g\,\e^{-\beta V}.
    \end{equation}
    Applying~\eqref{eq:green_formula} with~$f=\theta^\top \nabla u$ and~$g=v$, observing that~$\theta^\top\nabla u\in H^1(\Omega)$ and using~\eqref{eq:boundary_gradients} as well as the eigenrelation~$\cL_\beta v = -\lambda v$, we obtain
    \begin{equation}
        \begin{aligned}
        \frac1\beta\int_{\partial\Omega}\frac{\partial u}{\partial \n}\frac{\partial v}{\partial\n}\n^\top a\n \theta^\top\n\,\e^{-\beta V} &= -\lambda\int_{\Omega}\theta^\top\nabla u v\,\e^{-\beta V} + \frac1\beta\int_{\Omega}\nabla\left(\theta^\top\nabla u\right)^\top a \nabla v\,\e^{-\beta V}\\
        &= -\lambda\int_\Omega \theta^\top \nabla u v\,\e^{-\beta V} + \frac1\beta\int_\Omega \nabla u^\top \nabla\theta^\top a \nabla v\,\e^{-\beta V}+\frac1\beta \int_\Omega \theta^\top \nabla^2 u a \nabla v\,\e^{-\beta V}.
        \end{aligned}
    \end{equation}
    Applying this identity to~\eqref{eq:derivative_volume_expression} twice (exchanging the roles of~$u$ and~$v$ the second time), we get
    \begin{equation}
        \begin{aligned}
            M_{ij}^{\Omega,k}(\theta) &= \frac1\beta\int_\Omega \nabla u^\top \nabla a^\top \theta\nabla v\,\e^{-\beta V}-\frac2\beta\int_{\partial\Omega}\frac{\partial u}{\partial \n}\frac{\partial v}{\partial\n}\n^\top a\n \theta^\top\n\,\e^{-\beta V}\\
            &-\lambda\int_{\Omega}\theta^\top \nabla(uv)\,\e^{-\beta V}-\lambda\int_{\Omega}uv\,\div\left(\theta\e^{-\beta V}\right)\\
            &+\frac1\beta\int_\Omega\theta^\top\left(\nabla^2 u a\nabla v + \nabla^2 v a \nabla u\right)\e^{-\beta V} + \frac1\beta\int_\Omega \nabla u^\top a \nabla v\,\div\left(\theta\e^{-\beta V}\right).
        \end{aligned}
    \end{equation}
    Note that the second line is equal to
    \[-\lambda\int_\Omega \div\left(uv \theta \e^{-\beta V}\right) = 0,\]
    by the Green--Ostrogradski formula and the boundary condition~$u,v\in H_0^1(\Omega)$.
    It then suffices to notice that
    \begin{equation}
        \div\left(\nabla u^\top a \nabla v \theta \e^{-\beta V}\right) = \theta^\top\left(\nabla^2 u a\nabla v + \nabla^2 v a \nabla u\right)\e^{-\beta V} + \nabla u^\top \nabla a^\top \theta \nabla v\,\e^{-\beta V} + \nabla u^\top a \nabla v\,\div\left(\theta\e^{-\beta V}\right),
    \end{equation}
    to conclude that
    \begin{equation}
        \begin{aligned}
            M_{ij}^{\Omega,k}(\theta) &= \frac1\beta \int_{\partial\Omega}\nabla u^\top a \nabla v \theta^\top \n\,\e^{-\beta V} - \frac2\beta \int_{\partial\Omega}\frac{\partial u}{\partial \n}\frac{\partial v}{\partial\n}\n^\top a\n \theta^\top\n\,\e^{-\beta V}\\
            &=-\frac1\beta\int_{\partial\Omega}\frac{\partial u}{\partial \n}\frac{\partial v}{\partial\n}\n^\top a\n \theta^\top\n\,\e^{-\beta V}
        \end{aligned}
    \end{equation}
    as claimed.
\end{proof}

\section{Numerical optimization}
\label{sec:ascent}
Using the results of Section~\ref{sec:main_result}, we describe in this section an ascent algorithm to numerically optimize smooth functionals of the eigenvalues of the Dirichlet generator~$\cL_\beta$.
We first present in Section~\ref{subsec:fem} the discretization procedure used to solve the Dirichlet eigenproblem. In Section~\ref{subsec:local_opt}, we describe the local ascent method we use, and detail the choice of ascent direction in Section~\ref{subsec:ascent_direction}. 

Throughout this section, we fix a smooth function~$J$ of~$k\in \N^*$ ordered Dirichlet eigenvalues, which we seek to maximize:
\begin{equation}
    \label{eq:functional}
        \underset{\Omega\subset\R^d}{\max}\, J\left(\lambda_1(\Omega),\dots,\lambda_{k}(\Omega)\right),\qquad J \in \mathcal C^{\infty}\left(\left(\R_+^*\right)^k,\R\right).
\end{equation}
By an abuse of notation, we also write the shorthands $J(\Omega):=J\left(\lambda_1(\Omega),\dots,\lambda_{k}(\Omega)\right)$,~$\partial_{\lambda_i} J(\Omega):= \partial_{\lambda_i} J(\lambda_1(\Omega),\dots,\lambda_k(\Omega))$ for~$1\leq i\leq k$~and denote by $DJ(\Omega;\theta)$ the Gateaux right-derivative at point~$\theta$ of the map~$\theta\mapsto J\left(\lambda_1(\Omega_\theta),\dots,\lambda_k(\Omega_\theta)\right)$, which exists by the third item in Theorem~\ref{thm:gateaux_differentiability}, or its Fr\'echet derivative whenever it is defined.

\subsection{Finite-element discretization of the eigenproblem}
\label{subsec:fem}
The numerical method we propose is based on a finite-element (FEM) approximation of the spectrum. As such, it is computationally affordable in the low-dimensional setting~$d\leq 3$.
For higher dimensional systems, one may resort to a low-dimensional representation of the dynamics, see Section~\ref{sec:coarse_graining} below where this is illustrated in a case when a good low-dimensional collective variable is available.
\paragraph{Finite-element meshes.}
All the shapes we consider in this work are parametrized by simplicial meshes. A mesh~$\Sigma$ for a given polyhedral domain~$\Omega$ consists for our purposes of the data
\begin{equation}
    \label{eq:mesh}
    \Sigma = (\mathcal V,\mathcal T),\qquad\mathcal V = \left(x_i\right)_{1\leq i\leq N_V},\qquad \mathcal T=\left(T_i\right)_{1\leq i\leq N_T},
\end{equation}
where~$\mathcal V \in \left(\R^d\right)^{N_V}$ is the set of $0$-cells or vertices, and~$\mathcal T \in \left(\mathcal V^{d+1}\right)^{N_T}$ defines the set of~$d$-cells, namely triangles for~$d=2$ or tetrahedra for~$d=3$.
We assume the usual finite-element method (FEM) conditions on~$\mathcal T$:
\begin{equation}
    \label{eq:mesh_conditions}
    \overline{\Omega} = \bigcup_{i=1}^{N_T} \co\,T_i,\qquad \forall\,1\leq i < j \leq N_T,\qquad\, \overset{\circ}{\co}\,T_i\cap \overset{\circ}{\co}\,T_j= \varnothing,
\end{equation}
where~$\co$~(resp.~$\overset{\circ}{\co}$) denotes the closed (resp. open) convex hull. The set~$\co T$ is the (closed) $d$-cell associated with any~$T\in\mathcal T$.

Dirichlet eigenvalues are approximated using the following procedure.
\paragraph{Rayleigh--Ritz approximation of the Dirichlet spectrum.}
    Given a mesh~$\Sigma$, the Rayleigh--Ritz method for the Dirichlet eigenproblem consists in performing the following steps.
    \label{alg:rayleigh_ritz}
    \begin{enumerate}[A.]
        \item{Fix a finite-dimensional subspace~$E_{0}(\Sigma)\subset H_{0}^1(\Omega)$, spanned by a set of basis functions~$\Phi(\Sigma)=\left(\phi_i\right)_{1\leq i\leq d_{\Sigma}}$. A typical choice is the set of~$\P_1$ elements for the interior vertices~$\mathcal V\cap \Omega$.
        Another approach is to take~$E_0(\Sigma) \subset H^1(\Omega)$ and enforce the Dirichlet boundary condition by adding a penalization term to the weak formulation. We use the latter method, which is implemented by default in FreeFem++~\cite{freefem} (with the default value of the penalization parameter).
        }
        \item{Form the matrices
        \begin{equation}
            \label{eq:fem_matrices}
            A(\Sigma) = \left(\int_{\Omega} \nabla\phi_i^\top a \nabla\phi_j\,\e^{-\beta V}\right)_{1\leq i,j\leq d_{\Sigma}},\qquad B(\Sigma) = \left(\int_{\Omega} \phi_i\phi_j\,\e^{-\beta V}\right)_{1\leq i,j\leq d_{\Sigma}}.
        \end{equation}
        In practice the integrals can be restricted to the set~$\supp\,\phi_i \cap \supp\,\phi_j = \cup_{n\in \mathcal N_{ij}}\co\, T_n$, where~$\mathcal N_{ij}$ is a set indexing the cells on which both~$\phi_i$ and~$\phi_j$ are non-zero. Generally, the integrals in~\eqref{eq:fem_matrices} consist in the sum of integrals over only a handful of cells in~$\mathcal T$, which are approximated by quadrature rules.
        The resulting matrices are sparse, which makes the computation of the bottom eigenvalues tractable with iterative methods.
        }
        \item{Solve the generalized eigenvalue problem (e.g. using a Lanczos algorithm) for~$1\leq\ell\leq k$:
        \begin{equation}
            \label{eq:fem_eigenvalues}
            A(\Sigma)w_\ell(\Sigma) = -\lambda_\ell(\Sigma)B(\Sigma)w_\ell(\Sigma),\qquad w_\ell(\Sigma)\in \R^{d_{\Sigma}}.
        \end{equation}
        The Rayleigh--Ritz eigenpair~$\left(\lambda_\ell(\Sigma),w_\ell(\Sigma)^\top \Phi(\Sigma)\right)$ can then be used as an approximation of the Dirichlet eigenpair~$(\lambda_\ell(\Omega),u_\ell(\Omega))$. We denote by~$u_\ell(\Sigma) = w_\ell(\Sigma)^\top \Phi(\Sigma)$ the approximated eigenfunction, and convene that the eigenvalues are listed in increasing order.
        }
    \end{enumerate}

We also select the shape perturbation~$\theta$ (see Figure \ref{fig:shape_perturbation}) in a finite-dimensional space~$W(\Sigma)\subset \Winf(\R^d;\R^d)$. In practice, we take~$W(\Sigma) \subset H^1(\Omega)^d$ to be the finite-dimensional space spanned by the set of~$\P_1$ vector-valued elements associated with~$\Sigma$.

We finally introduce the following notion of numerical degeneracy for Rayleigh--Ritz eigenvalues:
we say that~$\lambda_\ell(\Sigma)$ has~$\varepsilon$-multiplicity~$m \geq 1$ if
    \begin{equation}
        \label{eq:degeneracy}
        \frac{\lambda_{\ell}(\Sigma)-\lambda_{\ell-1}(\Sigma)}{\lambda_{\ell-1}(\Sigma)} > \varepsilon,\qquad \frac{\lambda_{\ell+m-1}(\Sigma)-\lambda_\ell(\Sigma)}{\lambda_\ell(\Sigma)} \leq \varepsilon < \frac{\lambda_{\ell+m}(\Sigma)-\lambda_\ell(\Sigma)}{\lambda_\ell(\Sigma)}.
    \end{equation}

\subsection{Local optimization procedure.}
\label{subsec:local_opt}
The algorithm starts from the choice of some initial mesh-like open domain~$\Omega_0$, with an underlying mesh~$\Sigma_0$.
The ascent algorithm used to solve~\eqref{eq:shape_optimization_problem} takes the following parameters as input.

\noindent
{
\begin{center}
\begin{tabular}{ll}
\toprule
\textbf{Parameter} & \textbf{Description} \\
\midrule
$\Omega_0$, $\Sigma_0 = (\mathcal{V}_0, \mathcal{T}_0)$ & Initial polyhedral domain and its mesh \\
$\varepsilon_{\mathrm{degen}} > 0$ & Degeneracy tolerance parameter \\
$m_{\mathrm{max}} \geq 2$ & Maximal degeneracy rank \\
$\eta_{\mathrm{max}} > 0$ & Maximal step size \\
$0 < \alpha < 1$ & Step size multiplier \\
$\varepsilon_{\mathrm{term}} > 0$ & Termination criterion tolerance \\
$M_{\mathrm{grad}}>0$ & Gradient normalization parameter \\
$N_{\mathrm{search}}>0$ & Number of search points in the degenerate case\\
\bottomrule
\end{tabular}

\noindent
Input parameters for Algorithm~\ref{alg:ascent}.
\end{center}
}

We proceed by iterating the following steps.
\begin{algorithm}[Ascent iteration.]
\label{alg:ascent}
At step~$n\geq 0$:
\begin{enumerate}[A.]
    \item{Approximate the~$k+m_{\mathrm{max}}+1$ first eigenpairs for~$\Sigma_n$ using the finite-element Rayleigh--Ritz procedure from Section~\ref{subsec:fem} above.}
    \item{Identify an ascent direction $\theta_n\in \mathcal W(\Sigma_n)$ such that~$\widehat{DJ}(\Sigma_n;\theta_n)>0$, where
    \begin{equation}
        \label{eq:directional_derivative}
        \widehat{DJ}(\Sigma;\theta) = \nabla J(\lambda_1(\Sigma),\dots,\lambda_k(\Sigma))^\top \widehat{D\Lambda}(\Sigma;\theta),\qquad \widehat{D\Lambda}(\Sigma;\theta) = \left(\widehat{D\lambda_i}(\Sigma;\theta)\right)_{1\leq i\leq k}
    \end{equation}
    and~where~$\widehat{D\lambda_i}(\Sigma;\theta)$ is the approximation of the right-Gateaux derivative of~$\lambda_i(\Omega)$ in the direction~$\theta$ from step A., i.e.
    \begin{equation}
        \label{eq:directional_derivative_approx_simple}
        \widehat{D\lambda_i}(\Sigma;\theta) = 
            -\frac1\beta\int_{\partial\Omega} \left(\frac{\partial u_i(\Sigma)}{\partial\n}\right)^2 \n^\top a\n \e^{-\beta V}\theta^\top\n \text{ if }\lambda_i(\Sigma)\text{ has~$\varepsilon_{\mathrm{degen}}$-multiplicity } 1,
    \end{equation}
    and otherwise is given by the~$\ell$-th smallest eigenvalue of the matrix
    \begin{equation}
        \label{eq:directional_derivative_approx_multiple}
        \left(-\frac1\beta\int_{\partial\Omega} \frac{\partial u_\sigma(\Sigma)}{\partial\n}\frac{\partial u_\tau(\Sigma)}{\partial\n} \n^\top a\n\theta^\top\n\,\e^{-\beta V}\right)_{i-\ell+1\leq \sigma,\tau\leq i-\ell+m}
    \end{equation}
    if~$\lambda_{i-\ell+1}(\Sigma)$ has~$\varepsilon_{\mathrm{degen}}$-multiplicity $m\geq \ell$ for some~$2\leq\ell\leq m_{\max}$. If~$\lambda_{i-\ell+1}(\Sigma)$ has~$\varepsilon_{\mathrm{degen}}$ greater than~$m_{\max}$, the iteration fails. The choice of~$\theta_n$ and its discretization are the crucial features of the algorithm, and are made precise in Section~\ref{subsec:ascent_direction} below.}
    \item{Set the step size~$\eta_n =\eta_{\max}$, and displace the vertices of the mesh via~$\widetilde{ \mathcal V}_{n+1} = \mathcal{V}_n + \eta_n \theta_n(\mathcal V_n)$.
    The geometry of the mesh~$\widetilde{\Sigma}_{n+1}$ is defined by the set of new vertices~$\widetilde{\mathcal V}_{n+1}$, inheriting its combinatorial structure from~$\Sigma_n$.
    If~$\widetilde{\Sigma}_{n+1}$ is a valid mesh for a domain~$\Omega_{n+1}$, i.e. satisfies the FEM conditions~\eqref{eq:mesh_conditions}, set~$\Sigma_{n+1} = \mathcal A\left(\widetilde{\Sigma}_{n+1}\right)$, where~$\mathcal A$ is a local mesh refinement procedure designed to preserve meshing quality, namely the {\fontfamily{ccr}\selectfont adaptmesh} function from FreeFem++.
    Otherwise, set~$\eta_n \leftarrow \alpha\eta_n$ and repeat this step. For the sake of computational efficiency and simplicity, we limit ourselves to a fixed maximal step size $\eta_{\max}$, although various other strategies to select~$\eta_n$ are a classical topic in numerical optimization, see~\cite[Chapter 3]{NW99}.
    }
    \item{Set $n\leftarrow n+1$ and proceed from step A., unless the termination condition
    $$\widehat{DJ}(\Sigma_n;\theta_n)< \varepsilon_{\mathrm{term}}$$
    is met. Other termination criteria are possible and are again a classical topic, see~\cite{NW99}.
    }
\end{enumerate}
\end{algorithm}

\subsection{Choice of ascent directions}
\label{subsec:ascent_direction}
We now detail how to find ascent directions~$\theta_n$ in step B. of Algorithm~\ref{alg:ascent}. Following the standard reading on numerical shape optimization (see for instance~\cite[Section 6.5]{AS07}), we take a ``solve-then-discretize'' approach. We first describe how to identify steepest ascent directions at the continuous level (for both simple and multiple eigenvalues), and then make precise the discretization procedure.
For the purpose of this discussion, we assume to avoid undue technical difficulties that~$\Omega$ is a smooth domain and the coefficients~$a,V$ are smooth, ensuring by elliptic regularity that the Dirichlet eigenfunctions are smooth on~$\overline{\Omega}$, and therefore smooth and bounded on~$\partial\Omega$.

\paragraph{Case of simple eigenvalues.}
We first handle the case where each of the~$\lambda_i(\Omega)$ have multiplicity 1.
In this case, according to Corollary~\ref{cor:boundary_expression}, the differential of~$J$ with respect to the perturbation~$\theta$ can be expressed as a continuous linear form of the normal perturbation~$\theta^\top\n$ on~$\partial\Omega$, i.e.
\begin{equation}
    D J(\Omega;\theta) = \int_{\partial\Omega}\phi_J(\Omega)\theta^\top\n,
\end{equation}
for the scalar-valued map~$\phi_J(\Omega)$ defined on~$\partial\Omega$ by
\begin{equation}
    \label{eq:shape_gradient_explicit}
    \phi_J(\Omega) = -\frac1\beta \left[\sum_{i=1}^k\partial_i J(\lambda_1(\Omega),\dots,\lambda_k(\Omega))\left(\frac{\partial u_i(\Omega)}{\partial\n}\right)^2\right]\n^\top a\n \e^{-\beta V}.
\end{equation}

The vector field~$\phi_J(\Omega) \n$ is therefore the~$L^2(\partial\Omega)$-gradient of~$J$ with respect to~$\theta$, which is why~$\phi_J(\Omega)$ is also called the shape-gradient of~$J$ at~$\Omega$.
A natural approach to shape-optimization is to approximate the $L^2(\partial\Omega)$-gradient flow by an explicit Euler discretization, setting $\widetilde\Omega = (\Id+\eta\theta)\Omega$, where~$\theta$ is chosen so that~$\left.\theta^\top\n\right|_{\partial\Omega}=\phi_J(\Omega)$.
When using mesh-discretizations of~$\Omega$, two difficulties arise with this approach. Firstly, one must specify how to displace the internal vertices of the mesh, or in other words how to extend~$\phi_J(\Omega)\n$ to~$\Omega$.
Secondly, the normal derivative~$\n$ is an irregular field on the boundary of a mesh. In practice, we observe that displacements of the boundary vertices along the mesh normal field leads to rapid collapse in the quality of the boundary mesh, which prevents the naive algorithm from converging.

To overcome both difficulties, a standard approach~(see for instance~\cite[Section 5.2.2]{ADJ21} for an extensive discussion) is to resort to an extension-regularization procedure, seeking a Riesz representative of~$\theta \mapsto D J(\Omega;\theta)$ in a Hilbert space~$\mathcal H(\Omega)\subset L^2(\Omega)$ consisting of more regular shape-perturbations, defined on the whole of~$\Omega$.
To ensure that this is possible,~$\mathcal H(\Omega)$ should be continuously embedded in~$L^2(\partial\Omega)$. 
A standard choice, which we also use in this work, is to take
\begin{equation}
    \label{eq:ext_reg_space}
    \mathcal H(\Omega) = H^1(\Omega)^d,\qquad \left\langle\theta,\psi\right\rangle_{\mathcal H(\Omega)} = \int_{\Omega} \left(\varepsilon_\reg^2\nabla\theta :\nabla\psi + \theta^\top\psi\right),
\end{equation}
where~$\varepsilon_\reg>0$ is a regularization scale, which is chosen of the order of a few cell widths for the underlying mesh. Therefore, the problem of finding a Riesz representative of~$\theta\mapsto DJ(\Omega;\theta)$ amounts to solving the problem
\begin{equation}
    \label{eq:ext_reg_relation}
    \left\langle \theta_\reg,\theta\right\rangle_{\mathcal H(\Omega)} = \int_{\partial\Omega}\phi_J(\Omega)\theta^\top\n,\qquad\forall\,\theta\in \mathcal H(\Omega)
\end{equation}
for~$\theta_\reg\in \mathcal H(\Omega)$.

Solving~\eqref{eq:ext_reg_relation}, and choosing~$\theta=\theta_\reg$, one finds that~$DJ(\Omega;\theta_\reg)=\|\theta_\reg\|_{\mathcal H(\Omega)}^2$, so that~$\theta_\reg$ is indeed a valid ascent direction defined on the whole of~$\Omega$, and moreover $\theta_\reg=0$ if and only if $\Omega$ is a critical shape of $J$. Note that this approach is still valid whenever~$\phi_J(\Omega)\n\in H^{-1/2}(\partial\Omega)^d$ and~$\Omega$ is a Lipschitz domain, since the Sobolev trace theorem then gives the continuity of the trace~$\gamma:\mathcal H(\Omega)\to H^{1/2}(\partial\Omega)^d$. In practice, the problem~\eqref{eq:ext_reg_relation} is solved by a Galerkin method, which we discuss below.

For our choice of~$\mathcal H(\Omega)$, the requirement~\eqref{eq:ext_reg_relation} is the weak formulation of the following Neumann boundary value problem:
\begin{equation}
    \label{eq:ext_reg_pde}
    \left\{
        \begin{aligned}
            -\varepsilon_\reg^2\Delta \theta_\reg + \theta_\reg &= 0\text{ in }\Omega,\\
            \varepsilon_\reg^2\nabla\theta_\reg\n &= \phi_J(\Omega)\n\text{ on }\partial\Omega.
        \end{aligned}
    \right.
\end{equation}
where~$\Delta$ is the component-wise Laplace operator. Let us denote by
\begin{equation}
    \label{eq:ext_reg_op}
    R_{\varepsilon_\reg}:\left\{
    \begin{aligned}
        H^{-1/2}(\partial\Omega)^d &\to \mathcal H(\Omega),\\
        \phi_J(\Omega)\n &\mapsto\,\theta_\reg \quad \text{solution to equation~\eqref{eq:ext_reg_pde}}
    \end{aligned}
    \right.
\end{equation}
the operator which maps the boundary data $\Phi_J(\Omega)\n$ to the solution~$\theta$ of the above Neumann problem.

Various other approaches to the extension-regularization procedure, tailored to preserve mesh quality over many iterations, are sometimes preferred, see for instance~\cite[Section 3.5]{DFOP18} or~\cite[Section 5.2.2]{ADJ21}. They simply correspond to other choices of~$\mathcal H(\Omega)$ and the associated inner product.

\paragraph{Case of multiple eigenvalues.}
The case of multiple eigenvalues is more challenging. To simplify the presentation, and motivated by the maximization of~\eqref{eq:separation_of_timescales}, we focus on the case where~$J$ depends only on the first two Dirichlet eigenvalues, and~$\lambda_2(\Omega)$ has multiplicity~$m = 2$ ($\lambda_1(\Omega)$ is always simple by theory), and~$\partial_{\lambda_2} J(\Omega)\geq 0$. The generalization to more eigenvalues and/or other local monotonicity properties of~$J$ is straightforward, although the computational cost of the method increases with the total multiplicity.

According to the third item in Theorem~\ref{thm:harmonic}, multiple eigenvalues are no longer Fréchet-differentiable, and one therefore loses any natural notion of shape gradient. However, the objective is still directionally differentiable. The natural counterpart to the shape gradient is given by the steepest ascent perturbation
\begin{equation}
    \label{eq:steepest_ascent_problem}
    \theta^*\in\underset{\|\theta\|_{\mathcal H(\Omega)=1}}{\mathrm{Argmax}}\,DJ(\Omega;\theta).
\end{equation}
Note that one seeks a steepest ascent perturbation in the space~$\mathcal H(\Omega)$ of regular perturbations defined in~\eqref{eq:ext_reg_space}. This is done to ultimately preserve mesh quality, just as in the case of simple eigenvalues. It is however not immediately clear that the problem~\eqref{eq:steepest_ascent_problem} is well-posed or tractable. Fortunately, this turns out to be the case in our setting. First, we write
\begin{equation}
    \begin{aligned}
    DJ(\Omega;\theta) &= \partial_{\lambda_1} J(\Omega)D\lambda_1(\Omega;\theta)+\partial_{\lambda_2} J(\Omega)\underset{|u|=1}{\min}\,u^\top M^{\Omega,2}(\theta)u\\
    &= \underset{|u|=1}{\min}\,u^\top\left[\partial_{\lambda_1} J(\Omega)D\lambda_1(\Omega;\theta)\mathrm{I}_{2}+\partial_{\lambda_2} J(\Omega)M^{\Omega,2}(\theta)\right] u,
    \end{aligned}
\end{equation}
using~$\partial_{\lambda_2}J(\Omega)\geq 0$ and the fact that~$D\lambda_2(\Omega;\theta)$ is the smallest eigenvalue of the~$2\times2$ matrix~$M^{\Omega,2}(\theta)$ defined in~\eqref{eq:derivative_volume_expression}. The problem~\eqref{eq:steepest_ascent_problem} is therefore to maximize with respect to~$\theta$ the smallest eigenvalue of the symmetric matrix~$Q^\Omega(\theta)$ whose~$(i,j)$-th component is given by
\begin{equation}
    \label{eq:perturbation_matrix_J_boundary_form}
    Q_{ij}^{\Omega}(\theta) = \langle \phi^{ij}_{J}(\Omega)\n,\theta\rangle_{L^2(\partial\Omega)},\qquad \phi^{ij}_{J}(\Omega) = -\frac1\beta\left[\partial_{\lambda_2} J(\Omega)\frac{\partial u_2^{(i)}}{\partial\n}\frac{\partial u_2^{(j)}}{\partial\n} +\delta_{ij}\partial_{\lambda_1} J(\Omega)\left(\frac{\partial u_1}{\partial\n}\right)^2\right]\n^\top a\n\e^{-\beta V},
\end{equation}
where we write~$u_1=u_1(\Omega)$,~$u_2^{{(i)}}=u_2^{(i)}(\Omega)$ for~$i=1,2$, and use the formula~\eqref{eq:boundary_form}. Crucially, this matrix depends linearly on~$\theta$, although its smallest eigenvalue does not.

By the regularization procedure detailed in the previous paragraph, we may also write
\begin{equation}
    \label{eq:perturbation_matrix_J_volume_form}
    Q_{ij}^{\Omega}(\theta) = \langle R_{\varepsilon_\reg}\phi^{ij}_{J}(\Omega)\n,\theta\rangle_{\mathcal{H}(\Omega)},\qquad \forall\,1\leq i,j\leq 2.
\end{equation}
Let us denote by~$\psi_{ij}:= R_{\varepsilon_\reg}\phi^{ij}_{J}(\Omega)\n\in \mathcal H(\Omega)$,~$G:=\mathrm{Span}_{\mathcal H(\Omega)}\{\psi_{ij},\,1\leq i\leq j\leq 2\}$, and~$\Pi_G$ the~$\mathcal H(\Omega)$-orthogonal projector onto~$G$.

To solve~\eqref{eq:steepest_ascent_problem}, we distinguish between two cases.
\begin{itemize}
    \item{If~$\underset{\|\theta\|_{\mathcal H(\Omega)}=1}{\sup}\, DJ(\Omega;\theta)\leq 0$, then the~$\sup$ is equal to~$0$, and is attained for any~$\theta\in G^\perp$ with unit norm. Note this is a first-order optimality condition: to first order, any shape perturbation can only decrease the value of $J$.}
    \item{If~$\underset{\|\theta\|_{\mathcal H(\Omega)}=1}{\sup}\, DJ(\Omega;\theta)>0$ then by positive homogeneity of the smallest eigenvalue with respect to~$\theta$ and the identity~$Q^\Omega(\theta)=Q^\Omega(\Pi_G\theta)$, we rewrite 
    \begin{equation}
    \underset{\|\theta\|_{\mathcal H(\Omega)=1}}{\mathrm{sup}}\,DJ(\Omega;\theta) = \underset{\|\theta\|_{\mathcal H(\Omega)\leq 1}}{\mathrm{sup}}\,DJ(\Omega;\theta) = \underset{\|\theta\|_{\substack{\mathcal H(\Omega)\leq 1\\\theta \in G}}}{\mathrm{max}}\,DJ(\Omega;\theta).
    \end{equation}}
\end{itemize}
In the second case, the~$\sup$ is replaced by a~$\max$, since the supremum is taken over the compact set~$\overline{B_{\mathcal H(\Omega)}(0,1)}\cap G$. Hence, in both cases, a maximizer for~\eqref{eq:steepest_ascent_problem} is attained.
In fact, we can check that, in the second case, the maximizer is unique, as implied by the following elementary lemma.
\begin{lemma}
    \label{lemma:unique_maximizer}
    Let~$B$ be the closed unit ball in a finite-dimensional Hilbert space $E$, and let $f:B\to\R$ be a concave function such that~$\underset{\theta\in B}{\sup}\, f(\theta)>0$, and is furthermore positively homogeneous of degree~$\alpha>0$. Then, there exists a unique maximizer $\theta^*\in \partial B$ for the problem
    \begin{equation}
        \underset{\theta\in B}{\sup}\,f(\theta).
    \end{equation}
\end{lemma}
\begin{proof}
    Since $B$ is compact, there exists a maximizer. Assume for the sake of contradiction the existence of two distinct maximizers~$\theta_1\neq \theta_2$.
    We let~$\theta = \frac12(\theta_1+\theta_2)$, and note that~$\|\theta\|_E<1$.

    First,
    \begin{equation}
        \begin{aligned}
            f\left(\theta\right)&\geq \frac12\left(f(\theta_1)+f(\theta_2)\right)\\
            &=\max_B\,f>0,
        \end{aligned}
    \end{equation}
    since~$f$ is concave, then
    \begin{equation}
        \begin{aligned}
        f(\theta/\|\theta\|_{E})&=\|\theta\|^{-\alpha}_{E}f(\theta)>f(\theta)\\
        &\geq \max_B\,f
        \end{aligned}
    \end{equation}
    using homogeneity. We have reached a contradiction, therefore there exists a unique maximizer $\theta^*$, which necessarily satisfies~$\|\theta^*\|=1$, since~$f(\theta/\|\theta\|_E)>f(\theta)$ whenever~$f(\theta)>0$, using homogeneity once again.
\end{proof}

In our setting, we let~$E=G$, and notice that, since~$\theta\mapsto u^\top Q^\Omega(\theta)u$ is linear for any~$u\in \R^2$, the map
\begin{equation}
    \theta\mapsto DJ(\Omega;\theta)=\underset{|u|=1}{\min}\,u^\top Q^{\Omega}(\theta)u
\end{equation}
is concave and positively homogeneous of degree~$\alpha=1$. Under the additional assumption
$$\underset{\|\theta\|_{\mathcal H}=1}{\sup}\,DJ(\Omega;\theta)=\underset{\|\theta\|_{\mathcal H}=1,\theta\in G}{\sup}\,DJ(\Omega;\theta)>0,$$
 the conditions of Lemma~\ref{lemma:unique_maximizer} are satisfied, which proves the existence of a unique~$\theta^*$ solving~\eqref{eq:steepest_ascent_problem}.

In practice, finding~$\theta^*$ is tractable by a direct search method. Letting~$g=(g_1,g_2,g_3)\in \mathcal H(\Omega)^3$ be a~$\mathcal H(\Omega)$-orthonormal basis for $G$, obtained by applying a Gram--Schmidt procedure to the~$\{\psi_{ij},1\leq i\leq j\leq 2\}$, the problem~\eqref{eq:steepest_ascent_problem} reduces to an optimization with respect to a parameter~$\alpha$ on the unit sphere~$\mathbb S^2\subset\R^3$. 
If we fail to find~$\alpha\in\mathbb S^2$ such that~$DJ(\Omega;\alpha^\top g)>0$, we deduce that~$\Omega$ satisfies a first-order optimality condition, although this case never came up in our examples.

\begin{remark}
    We note that, even for objectives $J$ involving several eigenvalues with multiplicities greater than $2$, the optimization problem~\eqref{eq:steepest_ascent_problem} can still be reduced to a finite-dimensional optimization problem.
    However, the dimensionality of the problem may be large, and is related to the number of linearly independent components of the perturbation matrix~\eqref{eq:perturbation_matrix_J_volume_form}, namely
    $$ \dim\,G \leq \sum_{j=1}^{\ell} \frac{m_j(m_j+1)}{2},$$
    where~$\ell$ denotes the number of distinct degenerate eigenvalue involved in the definition of~$J$, and the set~$\{m_j,1\leq j\leq \ell\}$ enumerates their respective multiplicities.
    Moreover, the finite-dimensional problem will generally not be concave, in which case the optimum~$\theta^*$ may not be unique, and the problem may be itself hard to solve, especially if~$\dim\, G$ is large.
\end{remark}

\paragraph{Discretization of ascent directions.}
We now explain how we discretize the choice of ascent direction at the~$k$-th iteration of Algorithm~\ref{alg:ascent}. The domain~$\Omega_k$ is approximated by a mesh~$\Sigma_k =(\mathcal V_k,\mathcal T_k)$, and the extension-regularization operator~$R_{\varepsilon_\reg}$ is replaced by a Galerkin approximation $\widehat{R}_{\varepsilon_\reg}$

We consider the subspace $W(\Sigma_k)$ spanned by the basis~$\Theta_k$ of~$\mathbb P_1$ vector-valued elements associated with~$\Sigma_k$, and compute its Gram matrix $G_\reg(\Sigma_k)$ with respect to the~$\mathcal H(\Omega_k)$-inner product~\eqref{eq:ext_reg_space} for the basis~$\Theta_k$. This costly step only needs to be performed once, regardless of the number of extension-regularization calls (which is determined by the degeneracy of the eigenvalues, as~\eqref{eq:perturbation_matrix_J_volume_form} needs to be computed).
For any~$f\in H^{-1/2}(\partial\Omega_k)^d$, we compute the components~$b_f(\Sigma_k)$ of~$\theta\mapsto\langle f,\theta\rangle_{L^2(\partial\Omega_k)^d}$ in the basis~$\Theta_k$, solve~$G_\reg(\Sigma_k)\alpha = b_f(\Sigma_k)$ for~$\alpha\in\R^{|\Theta_k|}$, and take~$\widehat{R}_{\varepsilon_\reg}(f):=\Theta_k^\top\alpha$. In practice, the components of~$b_f(\Sigma_k)$ are further approximated by quadrature rules.
All spectral quantities, namely the eigenvalues $\lambda_j(\Omega)$ and the eigenvectors~$u_j(\Omega)$ for~$1\leq j\leq k$, are replaced by their Rayleigh--Ritz counterparts $\lambda_j(\Sigma)$ and~$u_j(\Sigma)$, as well as the corresponding normal derivatives.

Numerically, exactly degenerate eigenvalues are never encountered. However, when~$\lambda_2$ is almost degenerate, i.e.~$(\lambda_3(\Sigma_k)-\lambda_2(\Sigma_k))/\lambda_2(\Sigma_k) \ll 1$, the displacement in step~$C.$ of the ascent algorithm~\ref{alg:ascent} may lead to the crossing of the eigenvalue branches, in such a way that it leads in fact to a local decrease in the value of~$J$.
This manifests itself through local oscillations in the eigenvalues and objective functions throughout the ascent algorithm, see Figure~\ref{fig:oscillations} below.
This is a well-known problem in the numerical optimization of non-smooth objective functions, and decreasing the step size~$\eta_k$ to ensure local ascent is not a viable solution, as it may lead to very slow convergence to a local minimum, or altogether prevent it.
In the context of numerical optimization of eigenvalues, this behavior has been for example observed in~\cite{CTZ24}, where Nesterov-type acceleration techniques are suggested.

We follow another approach, assuming exact degeneracy when detecting~$\varepsilon_{\mathrm{degen}}$-degeneracy, and choosing an ascent direction within a low-dimensional space of perturbations, according to the analytical prescription of the previous paragraph.
See the work~\cite{DG06} for a closely related method applied to an exactly degenerate eigenvalue problem.

More precisely, in the case of~$\varepsilon_{\mathrm{degen}}$-simple eigenvalues, we set
\begin{equation}
    \theta_k := \widehat{R}_{\varepsilon_\reg}(\phi_J(\Sigma_k)\n) / \max\left(M_{\mathrm{grad}},\|\widehat{R}_{\varepsilon_\reg}(\phi_J(\Sigma_k)\n)\|_{\mathcal H(\Sigma_k)}\right),
\end{equation}
where~$\phi_J(\Sigma_k)$ is obtained by substituting Rayleigh--Ritz approximations in the definition~\eqref{eq:shape_gradient_explicit} of the shape gradient~$\phi_J(\Omega_k)$, and we recall~$M_{\mathrm{grad}}>0$ is a hyperparameter.
In other words, if the shape gradient is larger than~$M_{\mathrm{grad}}$ in the~$\mathcal H(\Sigma_k)$-norm, the ascent perturbation is normalized. This procedure is equivalent to step size adaptation in an explicit Euler discretization of the underlying geometric flow, and corresponds to some time reparameterization (in the limit~$\eta_{\max}\to 0$) of the trajectories generated by Algorithm~\ref{alg:ascent}.
We found this choice convenient to stabilize the numerical flow, since the gradient varies by several orders of magnitude throughout the numerical trajectories for the problem we considered. To ensure convergence near local maxima, this normalization is capped at~$M_{\mathrm{grad}}>0$.

For the case of~$\varepsilon_{\mathrm{degen}}$-degenerate eigenvalues, we first solve
\begin{equation}
    \label{eq:ascent_direction_multiple}
   \forall\,1\leq i,j\leq 2,\qquad\psi_{ij}(\Sigma_k) = \widehat{R}_{\varepsilon_{\mathrm{reg}}}(\phi_{J}^{ij}(\Sigma_k)\n),
\end{equation}
where the~$\phi_{J}^{ij}(\Sigma_k)$ are obtained from~\eqref{eq:perturbation_matrix_J_boundary_form} by substituting Rayleigh--Ritz approximations in place of exact eigenelements. We then apply the Gram--Schmidt algorithm (for the~$\mathcal H(\Sigma_k)$-scalar product~\eqref{eq:ext_reg_space}) to this set of perturbations, yielding a basis
$g(\Sigma_k)=(g_1(\Sigma_k),g_2(\Sigma_k),g_3(\Sigma_k))\in\mathcal H(\Sigma_k)^3$ of regular perturbations defined on~$\Omega_k$. We then solve
\begin{equation}
    \alpha^* = \underset{\alpha\in L_{N_{\mathrm{search}}}}{\max}\,\widehat{DJ}(\Sigma_k;\alpha^\top g(\Sigma_k)),
\end{equation}
where~$\widehat{DJ}(\Sigma_k;\cdot)$ is defined in~\eqref{eq:directional_derivative_approx_multiple}, and~$L_{N_{\mathrm{search}}}\subset\mathbb S^2$ is a set of~$N_{\mathrm{search}}$ points on the sphere. In practice, we use a Fibonacci lattice (see~\cite{G10}), which is simple to implement and distributes points quasi-uniformly.
This optimization step is extremely cheap, after having precomputed the matrix elements
\begin{equation}
    \langle\phi_{J}^{ij}(\Sigma_k)\n,g_k(\Sigma_k)\rangle_{L^2(\partial\Omega_k)},\qquad 1\leq i,j\leq 2,\,1\leq k\leq 3.
\end{equation}
Note that one could use the equivalent volume form~\eqref{eq:perturbation_matrix_J_volume_form}, but since boundary integrals are cheaper to compute and give good results in practice, we work with the latter instead.
After this precomputation step, the cost of evaluating the value of~$DJ(\Omega_k;\alpha^\top g)$ for $\alpha\in \mathbb S^2$ becomes negligible, and one can deduce the optimal perturbation~$\theta^*(\Sigma)=\alpha^{*\top}g(\Sigma)$ at virtually no cost.
We set~$\theta_k = \theta^*(\Sigma)$, which is by construction normalized in~$\mathcal H(\Sigma)$.

It would be of interest to obtain rigorous consistency results in the regimes~$\varepsilon_{\mathrm{degen}}\to 0$ and~$|\mathcal T|,N_{\mathrm{search}}\to+\infty$, as well as proving local convergence results for the algorithm and/or the underlying geometric flow. We leave this delicate question up for future work.

\section{Practical methods for high-dimensional systems}
\label{sec:practical_opt}
Although Theorem~\ref{thm:gateaux_differentiability} is interesting from a theoretical perspective, its applicability to the numerical shape optimization of spectral functionals is limited to settings for which the eigenelements of~$\cL_\beta$ are available.
For high-dimensional systems, which are typical in molecular simulation, this is hardly the case. It is therefore necessary to provide alternative numerical approaches.
In this section, we discuss such methods. The first one, discussed in Section~\ref{sec:coarse_graining}, relies on optimizing the separation of timescales for an effective dynamics through a given collective variable. 
The second one, discussed in Section~\ref{sec:semiclassic}, relies on the optimization of asymptotic expressions derived in the low-temperature regime, in the recent results of~\cite{BLS24}.
\subsection{Coarse graining of dynamical rates}
    \label{sec:coarse_graining}
    In this section, we propose a numerical strategy based on a Galerkin method and Theorem~\ref{thm:gateaux_differentiability}, after projecting the infinitesimal generator onto a collective variable (CV) or reaction coordinate.

    In practical cases from molecular dynamics, the process~\eqref{eq:overdamped_langevin} evolves in a high-dimensional space~$\R^d$ with~$d\gg 1$. In order to interpret trajectories in configurational space, it is often useful to view them through a low-dimensional map~$\xi:\R^d\to\R^m$, also known as a collective variable or reaction coordinate.
    Classical examples include geometric quantities such as dihedral angles, well-chosen interatomic distances, coordination numbers, path collective variables, which all derive from chemical intuition, and thus generally have a good physical interpretation.
    In recent years, machine learning techniques have been applied to the automatic construction of CVs optimized for a variety of purposes, see for instance~\cite{F18,GSal20,C21,GHRCNL21} for a review of recent approaches.

    Here we assume that a collective variable~$\xi$ is given, and consider the new problem of optimizing the effective separation of timescales with respect to a domain defined in collective variable space.
    The effective objective is defined with respect to a surrogate dynamics~(see~\eqref{eq:effective_dynamics}), which is already studied in~\cite{LL10,ZHS16,NKBC21}, although the methodology could in principle be applied to other reduced order models of the dynamics as well (see Remark~\ref{rem:other_eff_diffusion} below).

    \paragraph{Assumptions on the collective variable.}
    From now on, we assume that~$\xi$ is smooth, with~$\nabla \xi$ of full rank~$m$ everywhere. In particular, the Gram matrix~$G_\xi = \nabla\xi^\top\nabla\xi\in \R^{m\times m}$ is everywhere invertible. This condition ensures, by the implicit function theorem, that~$\xi$ foliates~$\R^d$ into a disjoint union of smooth submanifolds, which are given by the level sets~$\Sigma_z :=\xi^{-1}(z)$, for~$z\in\R^m$.
    We denote by~$\mu_z$ the canonical measure conditioned on~$\Sigma_z$. It corresponds to the probability measure defined by
    \begin{equation}
        \label{eq:conditional_measure}
        \mu_z\in \mathcal M_1(\Sigma_z),\qquad\frac{\d \mu_z}{\d \mathcal H_{\Sigma_z}} = \e^{-\beta V}\left(\det G_\xi\right)^{-1/2}\e^{\beta F_\xi(z)},
    \end{equation}
    where~$\mathcal H_{\Sigma_z}$ is the~$(d-m)$-dimensional Hausdorff measure on the submanifold~$\Sigma_z$. The factor~$\e^{-\beta F_\xi(z)}$ is a normalization constant expressed in terms of the free energy~$F_\xi:\R^m\to\R$ defined as
    \begin{equation}
        \label{eq:free_energy}
        F_\xi(z) :=-\frac1\beta\log\dsint_{\Sigma_z}\e^{-\beta V}\left(\det G_\xi\right)^{-1/2}\,\d\mathcal H_{\Sigma_z}.
    \end{equation}

    The collective variable will serve two purposes. Firstly, states will be defined in collective variable space, i.e. by fixing~$\Omega_\xi\subset \R^m$, and considering the preimage~$\xi^{-1}(\Omega_\xi)$. Secondly, the variational principle defining Dirichlet eigenvalues for the generator~$-\cL_\beta$ will be restricted to functions which are only a function of the collective variable~$\xi$. This will define, for each domain~$\Omega_\xi$, a set of Rayleigh--Ritz eigenvalues which will serve as effective eigenvalues associated with~$\xi^{-1}(\Omega_\xi)$.

    Introduce the weighted space~$L_{\beta}^2(\Omega_\xi) = L^2(\Omega_\xi,\e^{-\beta F_\xi(z)}\,\d z)$, and the associated weighted Sobolev spaces as in~\eqref{eq:sobolev_spaces}. We denote, for~$\Omega_\xi\subset\R^m$ and~$\varphi\in H_{0,\beta}^1(\Omega_\xi)$,
    \begin{equation}
        R_\xi\left(\varphi;\Omega_\xi\right) = R\left(\varphi\circ\xi;\xi^{-1}(\Omega_\xi)\right),
    \end{equation}
    where~$R(\cdot;\Omega)$ is the Rayleigh quotient associated with the Dirichlet realization of~$\cL_\beta$ on~$\Omega$, i.e.
    \begin{equation}
        R(\psi;\Omega) = \frac1\beta\frac{\dsint_{\Omega}\nabla\psi^\top a \nabla\psi\e^{-\beta V}}{\dsint_{\Omega}\psi^2\e^{-\beta V}}\qquad\forall\psi\in H_{0,\beta}^1(\Omega).
    \end{equation}
    Then, the coarea formula (see~\cite[Corollary 5.2.6]{KP08}) allows us to write
    \begin{equation}
        \label{eq:coarea_computation}
        \begin{aligned}
            R_\xi\left(\varphi;\Omega_\xi\right) &= \frac1\beta\frac{\dsint_{\xi^{-1}(\Omega_\xi)}\nabla(\varphi\circ\xi)^\top a\nabla(\varphi\circ\xi)\e^{-\beta V}}{\dsint_{\xi^{-1}(\Omega_\xi)}(\varphi\circ\xi)^2\e^{-\beta V}}\\
            &=\frac1\beta\frac{\dsint_{\Omega_\xi}\dsint_{\Sigma_z}{\left[\nabla\varphi\circ\xi\right]^\top\nabla\xi^\top a\nabla\xi\left[\nabla\varphi\circ\xi\right]}\,\e^{-\beta F_\xi\circ\xi}\,\d\mu_z\,\d z}{\dsint_{\Omega_\xi}\dsint_{\Sigma_z}(\varphi\circ \xi)^2\,\e^{-\beta F_\xi\circ \xi}\,\d\mu_z\,\d z}\\
            &=\frac1\beta\frac{\dsint_{\Omega_\xi}\nabla\varphi^\top a_\xi \nabla \varphi\,\e^{-\beta F_\xi}}{\dsint_{\Omega_\xi}\varphi^2\,\e^{-\beta F_\xi}}\\
        \end{aligned}
    \end{equation}
    where~$a_\xi$ denotes the symmetric, positive-definite matrix-valued map
    \begin{equation}
        \label{eq:diffusion_tensor}
        a_\xi(z) = \dsint_{\Sigma_z}{\nabla\xi^\top a\nabla\xi}\d\mu_z \in \R^{m\times m}.
    \end{equation}
    In practice, both~$F_\xi$ and~$a_\xi$ have to be estimated from samples of the reference dynamics. See~\cite{LRS10} for a review of methods to estimate~$F_\xi$, and~\eqref{eq:approx_diffusion_tensor} in Section~\ref{subsec:diala} below for a sample-based approximation of~$a_\xi$.

    It follows that~$R_\xi$, which we interpret as a family of coarse-grained Rayleigh quotients on the lower-dimensional space~$\R^m$, has the same basic structure as~$R$.
    Indeed, it corresponds to the family of Dirichlet Rayleigh quotients associated with a reversible diffusion on~$\R^m$ of the form~\eqref{eq:overdamped_langevin}, where the potential~$V$ and diffusion matrix~$a$ have been replaced by their lower-dimensional analogs defined in terms of conditional expectations with respect to the reference dynamics:
    \begin{equation}
        \label{eq:effective_dynamics}
        \d Z_t^\xi = \left(-a_\xi(Z_t^\xi)\nabla F_\xi(Z^\xi_t) + \frac1\beta\div\,a_\xi(Z_t^\xi)\right)\,\d t + \sqrt{\frac{2}\beta}a_\xi(Z_t^\xi)^{1/2}\,\d B_t,
    \end{equation}
    where~$B$ is a~$m$-dimensional standard Brownian motion. The dynamics~\eqref{eq:effective_dynamics} can be understood as a Markovian model for the dynamics of~$\xi(X_t)$, which is also reversible with respect to the Gibbs measure associated with the free energy.
    We refer to~\cite{LL10,ZHS16} for additional details on the mathematical properties of the effective dynamics.

    A natural question is whether one can hope to approximate the true dynamical rates with those predicted by the effective dynamics~\eqref{eq:effective_dynamics}.
    The answer has practical implications, since in the case where~$m$ is sufficiently low-dimensional, the eigenvalue problem associated with~$R_\xi$ becomes numerically tractable, and one may then optimize the separation of timescales~$N^*(\Omega)$ with respect to domains~$\Omega$ defined in terms of the CV.
    It should be noted that it is anyway common practice to define configurational states in terms of a CV.

    These considerations motivate the following Galerkin approach, already discussed in~\cite[Section 3.3.2]{ZHS16} for the case~$\Omega_\xi=\R^m$.
    We introduce the following linear subspace
    \begin{equation}
        \mathcal \mathcal V_\xi = \{\varphi\circ \xi,\varphi\in H_{0,\beta}^1(\Omega_\xi)\} \subset H_{0,\beta}^1\left(\xi^{-1}(\Omega_\xi)\right),
    \end{equation}
    and define the local coarse-grained rates
    \begin{equation}
        \label{eq:variational_principle_xi}
        \lambda^\xi_k(\Omega_\xi) := \underset{E_\xi}{\min}\,\underset{\varphi\in E_\xi}{\max}\,R_\xi(\varphi,\varphi;\Omega_\xi) = \underset{E}{\min}\,\underset{\varphi\in E}{\max}\,R\left(\varphi,\varphi;\xi^{-1}(\Omega_\xi)\right),
    \end{equation}
    where~$E_\xi$ ranges over the set of~$k$-dimensional subspaces of~$H_{0,\beta}^1(\Omega_\xi)$ in the first equality, and~$E$ ranges over the set of~$k$-dimensional subspaces of~$\mathcal E_\xi$ in the second.
    In other words,~$\lambda^\xi_k$ is the~$k$-th eigenvalue of the following operator acting on the weighted space~$L^2(\Omega_\xi,\e^{-\beta F_\xi}\,\d z)$ with Dirichlet boundary conditions:
    \[-\cL^\xi_\beta\varphi = -\frac{1}{\beta}\e^{\beta F_\xi}\div\left(\e^{-\beta F_\xi}a_\xi\nabla\varphi\right).\]
    It follows easily from the Courant--Fischer principle that~$\lambda^\xi_k(\Omega_\xi)\geq \lambda_k(\xi^{-1}(\Omega_\xi))$, and moreover that if~$\{u_1(\xi^{-1}(\Omega_\xi)),\dots,u_k(\xi^{-1}(\Omega_\xi))\}\subset \mathcal V_\xi$ for some~$k\leq m$, it holds that~$\lambda_k^\xi(\Omega_\xi) = \lambda_k(\xi^{-1}(\Omega_\xi))$.
    Thus, the dynamical rates associated with the effective dynamics will systematically overestimate the true rates. However, these will still be accurate if the Dirichlet eigenfunctions for~$\cL_\beta$ on~$\xi^{-1}(\Omega_\xi)$ can be well approximated in the class~$\mathcal V_\xi$.

    More precisely, we have the following result, adapted from~\cite[Proposition 5]{ZHS16}.
    \begin{proposition}
        \label{prop:rates_estimates}
        Let~$k\geq 1$ and~$\lambda^\xi_k$ (respectively, ~$\lambda_k$) be the~$k$-th principal eigenvalue of~$-\cL^\xi_\beta$ (resp. $-\cL_\beta$) in~$\Omega_\xi$ (resp.~$\xi^{-1}(\Omega_\xi)$), with associated eigenfunction~$u_k$ (resp. $u_k^\xi$), with the normalization~\eqref{eq:eigfunc_normalization}.
        Then,
        \begin{equation}
            \label{eq:rates_inequalities}
            \lambda_k \leq \lambda^\xi_k \leq \lambda_k + \frac1\beta\int_{\xi^{-1}(\Omega_\xi)}\nabla\left[u_k-u_k^\xi\circ\xi\right]^\top a\nabla\left[u_k-u_k^\xi\circ\xi\right]\e^{-\beta V}.
        \end{equation}
    \end{proposition}
    The proof of Proposition~\ref{prop:rates_estimates} is a straightforward adaptation of~\cite[Proposition 5]{ZHS16} to the case of absorbing Dirichlet boundary conditions on~$\partial\Omega_\xi$ and is therefore omitted.

    A useful corollary of Theorem~\ref{thm:gateaux_differentiability} is the following result.
    \begin{proposition}
        \label{prop:directional_derivative_xi}
        Let~$\Omega_\xi\subset\R^m$ be a bounded open domain which is convex or has a~$\mathcal C^{1,1}$ boundary.
        Assume that~$\xi$ is such that~Assumptions~\eqref{eq:a_ellipticity} and~\eqref{eq:coeff_regularity} are satisfied with~$d=m$,~$V=F_\xi$ and~$a=a_\xi$.
        Let~$\lambda_k^\xi=\lambda_k^\xi(\Omega_\xi)$ be an eigenvalue for~$\cL_\beta^\xi$ of multiplicity~$m_k^\xi\geq 1$, satisfying the normalization
        \[\int_{\Omega_\xi} u_k^{(i),\xi}(\Omega_\xi)u_k^{(j),\xi}(\Omega_\xi)\e^{-\beta F_\xi} = \delta_{ij},\qquad 1\leq i,j\leq m_k^\xi,\]
        where the~$u_k^{(i),\xi}(\Omega_\xi)$ for~$1\leq i\leq m_k^\xi$ are a basis of corresponding eigenvectors in~$L^2(\Omega_\xi,\e^{-\beta F_\xi})$.
        Then, for~$\theta\in \Winf(\R^m;\R^m)$ and~$0\leq \ell<m$, the map~$t\mapsto \lambda_{k+\ell}^{\xi}((\Id + t\theta)\Omega_\xi)$ is semi-differentiable at~$t=0$, and the right-differential is the~$(\ell+1)$-th smallest eigenvalue of the matrix
        \begin{equation}
        \label{eq:shape_derivative_xi}
        M_{ij}^{\xi}(\theta) = -\frac1\beta\int_{\partial \Omega_\xi}\frac{\partial u_k^{(i),\xi}(\Omega_\xi)}{\partial\n}\frac{\partial u_k^{(j),\xi}(\Omega_\xi)}{\partial\n} n^\top a_\xi\n\theta^\top\n\,\e^{-\beta F_\xi}\qquad 1\leq i,j\leq m.
        \end{equation}
    \end{proposition}
    \begin{proof}
        The result is a direct application of Theorem~\ref{thm:gateaux_differentiability} and Corollary~\ref{cor:boundary_expression}.
    \end{proof}
    We discuss sufficient conditions for the assumptions of Proposition~\ref{prop:directional_derivative_xi} in Appendix~\ref{sec:tech_reg} below.

    Proposition~\ref{prop:directional_derivative_xi} suggests a practical approach to the shape optimization of spectral functionals~$\mathcal F(\lambda_1(\Omega),\dots,\lambda_k(\Omega))$ in a high-dimensional setting, replacing the original objective with the coarse-grained objective~$\mathcal F(\lambda_1^\xi(\Omega_\xi),\dots,\lambda^\xi_k(\Omega_\xi))$.
    The computational implementation of this approach however requires access to the free-energy~$F_\xi$ and the matrix~$a_\xi$, for which a number of sampling methods are available, see~\cite{LRS10} for an overview.
    Due to the approximation error in~\eqref{eq:rates_inequalities}, we cannot expect the resulting shapes to be optimal for the original objective in the class of domains defined in CV space. They can nevertheless be used as input in acceleration methods such as ParRep, since this algorithm is dynamically unbiased by construction (in the limit of long decorrelation times).

    \begin{remark}
        \label{rem:other_eff_diffusion}
        The quality of the approximation~\eqref{eq:rates_inequalities} is quite sensitive to the choice of collective variable~$\xi$,
        and so, for a poor choice of~$\xi$, the effective dynamics~\eqref{eq:effective_dynamics} and its associated eigenvalues may give little insight into the original timescales (see Section~\ref{subsec:num_coarse_graining} below for an example).

        However, one could in principle apply the same methodology to other reversible, elliptic diffusions in~$\R^m$ apart from~\eqref{eq:effective_dynamics}, designed to better replicate the dynamical properties of~$\xi(X_t)$.
        In particular, instead of directly measuring~$F_\xi$ and~$a_\xi$ using thermodynamic averages, one can use a parametric approach to fit drift and diffusion coefficients of a dynamical model in~$\R^m$ directly from trajectories of $\xi(X_t)$ in CV space, see for instance~\cite{KPPK15}.
        This option has the advantage of being available even when the underlying dynamics in configurational space is not of the form~\eqref{eq:overdamped_langevin}, as long as the model enforces the form of a reversible diffusion~\eqref{eq:overdamped_langevin} in~$\R^m$.
        We leave this line of investigation to future work.
    \end{remark}
    This method is numerically validated in Section~\ref{subsec:num_coarse_graining} below, and is applied to a molecular system in Section~\ref{subsec:diala}.
    
    \subsection{Optimization in the semiclassical limit}
        \label{sec:semiclassic}
    In this section, we briefly summarize a second approach to make the shape optimization problem tractable, based on low-temperature spectral asymptotic results obtained in~\cite[Section 2.5]{BLS24}.
    These results are proved in~\cite{BLS24} under a set of assumptions which we simplify here for the sake of clarity, while keeping the main ideas intact.
    We restrict ourselves to the case of a constant diffusion coefficient~$a = \mathrm{Id}$. The dynamics follows therefore a standard overdamped Langevin equation:
    \begin{equation}
        \label{eq:overdamped_langevin_sc}
        \d X_t^\beta = -\nabla V\left(X_t^\beta\right)\,\d t + \sqrt{\frac{2}{\beta}}\,\d W_t.
    \end{equation}
    We additionally note the dependence of the dynamics on $\beta$, which is inversely proportional to the temperature. In this section, it is an asymptotic parameter considered in the limit~$\beta\to+\infty$.

    The use of semiclassical techniques to approximate spectral properties of metastable diffusions is a well-established topic in the probabilistic literature, see for example~\cite{HS85,BGK05,HN04,LN15,DGLLPN19,DGLLPN20,LPN21,LRS24,BLS24} and references therein.
    \paragraph{Asymptotic shape optimization of eigenvalue functionals.}
    We consider the general problem of maximizing with respect to a shape~$\Omega$ a functional of the Dirichlet eigenvalues of~$-\cL_\beta$ on~$\Omega$:
    \begin{equation}
        \label{eq:shape_functional}
        J(\Omega) = \mathcal F(\lambda_{1}(\Omega),\dots,\lambda_{k}(\Omega)),
    \end{equation}
    where~$\mathcal F:\R^k\to\R$ is continuous. When~$d\gg 1$, the numerical optimization of~$J$ is generally numerically intractable, since the objective involves solving a high-dimensional boundary eigenvalue problem. 

    The low-temperature asymptotic approach to this problem consists in fixing a family of domains~$(\Omega_{\alpha,\beta})_{\beta>0,\alpha\in\mathcal S}$, whose boundary geometry is jointly parametrized by the asymptotic parameter~$\beta$, and a shape parameter~$\alpha$ in the design space~$\mathcal S$.
    Assume that the asymptotic behavior of~$J(\Omega_{\alpha,\beta})$ is dictated, at dominant order, only by~$\beta$ and~$\alpha$ in the limit~$\beta\to+\infty$:
            \begin{equation}
                \label{eq:shape_sensitive_asymptotics}
            \mathcal F\left(\lambda_1(\Omega_{\alpha,\beta}),\dots,\lambda_k(\Omega_{\alpha,\beta})\right)=\mathcal F_{\infty}(\alpha,\beta)(1+\smallo(1))
        \end{equation}
        for some function~$\mathcal F_\infty:\mathcal S \times \R_+^*\to \R$.
    At fixed~$\beta>0$, we say the domain~$\Omega_{\alpha^*_\beta,\beta}$ is asymptotically optimal if
    \begin{equation}
        \label{eq:asymptotic_opt_problem}
        \alpha^*_\beta \in \underset{\alpha \in \mathcal S}{\mathrm{Argmax}}\, \mathcal F_\infty(\alpha,\beta).
    \end{equation}

    The difficulty in this approach lies in computing spectral asymptotics for domains with temperature-dependent boundaries.
    In~\cite{BLS24}, we define a set of geometric assumptions under which these spectral asymptotics can be derived, computed in practice, and ultimately optimized to solve the asymptotic problem~\eqref{eq:asymptotic_opt_problem}.
    The derivation of these shape-sensitive asymptotic formulas relies on the construction of approximate eigenmodes (or quasimodes in the semiclassical terminology) for~$\cL_\beta$, which form the crux of identifying~$\mathcal F_\infty$ in~\eqref{eq:shape_sensitive_asymptotics}.

    \paragraph{Geometrical setting.}
        We now present a slightly simplified version of the geometrical setting used in~\cite{BLS24}, which will allow to express the asymptotic results as clearly as possible. We refer to~\cite{BLS24} for a weaker set of assumptions for which the asymptotic results remain valid.
        Throughout this section, we assume that~$V$ is a~$\mathcal C^\infty$ Morse function over~$\R^d$. This means that at each point~$z\in \R^d$ such that~$\nabla V(z)=0$, the Hessian matrix~$\nabla^2 V(z)$ is non-degenerate.
        For~$0\leq i < N$, we denote the eigenvalues of the Hessian~$\nabla^2 V(z_i)$ by
        \begin{equation}
            \label{eq:eigvals_hessian}
            \mathrm{Spec}(\nabla^2 V(z_i)) = \left\{\hessEigval{i}{1}, \hessEigval{i}{2}, \dotsm, \hessEigval{i}{d}\right\}.
        \end{equation}
        We make no assumption on the ordering of these eigenvalues, except that, if~$z_i$ is an index-1 saddle point, meaning that~$\nabla^2 V(z_i)$ has a unique negative eigenvalue, one has~$\hessEigval{i}{1}<0$ (i.e. the negative eigenvalue is the first one).
        
        The Morse property implies that $V$ has finitely many critical points in~$\mathcal K_\alpha$, which we enumerate as
        $(z_i)_{0\leq i < N}$ for some~$N>0$.
        Among the critical points of~$V$ in~$\mathcal K$, we distinguish the local minima and index-1 saddle points, respectively given by the sets
        \begin{equation}
            \label{eq:minima_saddles}
            \left\{z_i,\,0\leq i < N_0\right\},\qquad \left\{z_i,\,N_0\leq i<N_0+N_1\right\}.
        \end{equation}
        For a given~$x\in\R^d$, we denote by~$\basin{x}$ the basin of attraction of~$x$ for the steepest descent dynamics, i.e.
        \begin{equation}
            \label{eq:basin}
            \basin{x} = \left\{z\in\R^d:\,X(t)\xrightarrow{t\to\infty} x,\,X'(t)=-\nabla V(X(t)),\,X(0)=z\right\}.
        \end{equation}
        The set~$\basin{x}$ is non-empty if and only if~$\nabla V(x)=0$, and in this case~$\basin{x}$ is a~$d$-dimensional subset of~$\R^d$, where~$d$ is the number of positive eigenvalues of~$\nabla^2 V(x)$.

        We now introduce the parameter~$\alpha=\left(\epsLimit{i}\right)_{0\leq i < N}\in(-\infty,+\infty]^N := \mathcal S$, which controls the asymptotic geometry of the domains near critical points of $V$.
        The value of the parameter~$\alpha\in\mathcal S$ is fixed, its link with the domain geometry will be made explicit in Assumption~\ref{hyp:locally_flat} below. We first assume that the domains~$\Omega_{\alpha,\beta}$ are smooth, and uniformly bounded, i.e. there exists a compact set~$\mathcal K_\alpha\subset \R^d$ such that,~$\Omega_{\alpha,\beta} \subset \mathcal K_\alpha$ for all~$\beta>0$.
 
        \begin{hypothesis}
            \label{hyp:locally_flat}
            In a small neighborhood of each critical point~$z_i$ and for~$\beta$ sufficiently large, the domain~$\Omega_{\alpha,\beta}$ is shaped like a half-space:
            \begin{equation}
                \Omega_{\alpha,\beta} \cap B(z_i,\varepsilon) = z_i + \left\{x\in\R^d:\,(x-z_i)^\top\hessEigvec{i}{1} < \alpha^{(i)}/\sqrt\beta\right\},
            \end{equation}
            where~$\varepsilon>0$ is a fixed parameter which depends only on~$V$, and~$\hessEigvec{i}{1}$ is a unit eigenvector of~$\nabla^2 V(z_i)$ for the eigenvalue~$\hessEigval{i}{1}$ (pointing outward of~$\Omega_{\alpha,\beta}$ for~$\epsLimit{i}<+\infty$).
            \end{hypothesis}
        When~$\alpha^{(i)}<+\infty$, the orientation convention for~$\hessEigvec{i}{1}$ ensures that decreasing~$\alpha^{(i)}$ locally retracts the domain.
        When~$z_i$ is an index-1 saddle point, Assumption~\eqref{hyp:locally_flat} is physically motivated by the fact that~$\hessEigvec{i}{1}$ gives the direction of the minimum energy path through~$z_i$ connecting a local minimum in the domain with a local minimum outside the domain (that is, the gradient flow lines joining the minima lying on both sides of the saddle point~$z_i$). Informally, the parameter~$\alpha$ encodes the position of the boundary along these paths, on the length scale~$1/\sqrt\beta$.
        
        The second assumption is that there is only one local minimum far from the boundary, in the following sense.
        \begin{hypothesis}
            \label{hyp:one_min}
            The point~$z_0$ is the only local minimum of~$V$ in~$\mathcal K$ such that~$\epsLimit{0} = +\infty$.
        \end{hypothesis}
        Informally, this assumption forces the QSD inside~$\Omega_{\alpha,\beta}$ to be unimodal, and to concentrate around~$z_0$ in the limit~$\beta\to\infty$.
            
        In order to state the last hypothesis, we introduce the sets 
        \begin{equation}
            \label{eq:low_energy_saddles}
            \begin{aligned}
                &\mathrm{SSP}(z_0) = \left\{z_i: N_0\leq i < N_0+N_1,\,\exists{m\neq z_0}\text{ a local minimum of }V\text{ s.t. } z_i\in \overline{\basin{z_0}}\cap\overline{\basin{m}}\right\},\\
                &I_{\min} = \left\{N_0\leq i < N_0+N_1:\, z_i\in\underset{i\in \mathrm{SSP}(z_0)}{\mathrm{Argmin}}\, V\right\},\qquad \Vstar = \underset{i\in \mathrm{SSP}(z_0)}{\min}\, V(z_i).
            \end{aligned}
        \end{equation}
        The set~$\mathrm{SSP}(z_0)$ corresponds to so-called separating saddle points, which lie on the boundary of the basin of attraction of~$z_0$, and the boundary of the basin of attraction of some other local minimum. Physically, these points correspond to the lowest-energy transition states on the boundary of~$\basin{z_0}$.
        The set~$I_{\min}$ contains the indices of these low-energy separating saddle points, and the associated minimal energy is given by~$\Vstar$.

        The final assumption is the following.
        \begin{hypothesis}
            \label{hyp:energy_condition}
            There exists~$c>0$ such that, for~$\beta$ sufficiently large, it holds
        \begin{equation}
            \basin{z_0}\cap \left\{V < \Vstar + c\right\}\subset \Omega_{\alpha,\beta} \setminus \bigcup_{i\in I_{\min}} B(z_i,\varepsilon).
        \end{equation}
        \end{hypothesis}
        This assumption ensures that the boundary of~$\Omega_{\alpha,\beta}$ does not enter below the energy level~$\Vstar$, except perhaps near low-energy separating saddle points. The role of this assumption is to avoid the introduction of spurious low-energy transition states, corresponding to local minima of~$V$ on the boundary which have no relation to the physically relevant transition pathways.
        Assumption~\ref{hyp:energy_condition} ensures that these so-called generalized saddle points are higher in energy than the low-energy transition states, and do not pollute the dominant asymptotic behavior of the metastable exit time. This assumption is crucial in ensuring that the asymptotics are, at dominant order, only a function of~$\beta$ and~$\alpha$, as in the desideratum~\eqref{eq:shape_sensitive_asymptotics}.
        However, it is expected in~\cite{BLS24} that a similar analysis can be performed even if Assumption~\eqref{hyp:energy_condition} does not hold, but at the cost of introducing a global counterpart to the local geometric Assumption~\eqref{hyp:locally_flat}. Relaxing Assumption~\ref{hyp:energy_condition} therefore leads once again to a high-dimensional (if not infinite-dimensional) design space~$\mathcal S$, and besides cannot improve upon the maximizers of~\eqref{eq:asymptotic_opt_problem} in the case~$\mathcal F(\lambda_1,\lambda_2)=(\lambda_2-\lambda_1)/\lambda_1$, which is why we enforce it.

    \paragraph{Harmonic approximation of the spectral gap.}
   The first main result of~\cite{BLS24} gives a quantitative and computable estimate of the spectral gap of the Dirichlet generator on~$\Omega_{\alpha,\beta}$, in the limit~$\beta\to+\infty$. In fact, it more generally shows that, for each~$k\geq 1$, the $k$-th eigenvalue $\lambda_{k,\beta}(\Omega_{\alpha,\beta})$ converges to the $k$-th eigenvalue of a temperature-independent operator, the so-called harmonic approximation.
    \begin{theorem}
        \label{thm:harmonic}
        Under Assumption~\ref{hyp:locally_flat}, it holds
        \begin{equation}
            \label{eq:harmonic_approx}
            \lambda_{k,\beta}(\Omega_{\alpha,\beta})\xrightarrow{\beta\to+\infty} \lambda_{k,\alpha}^{\mathrm{H}},
        \end{equation}
        where~$\lambda_{k,\alpha}^{\mathrm{H}}$ is the~$k$-th eigenvalue of the operator
        \begin{equation}
            \label{eq:harmonic_approximation}
            K_\alpha = \bigoplus_{i=0}^N K^{(i)}_{\alpha^{(i)}},\qquad K_{\alpha^{(i)}}^{(i)} = -\Delta + \frac14x^\top \mathcal D^{(i)}x -\frac{\Delta V(z_i)}2,
        \end{equation}
        with~$\mathcal D^{(i)} = \mathrm{diag}\left(\nu_j^{(i)2}\right)_{j=1}^d$, and where the operator~$K_{\alpha^{(i)}}^{(i)}$ is the Dirichlet realization of a quantum harmonic oscillator acting on the half-space~$(-\infty,\alpha^{(i)})\times \R^{d-1}$.
    \end{theorem}
    The local operators~$K_{\alpha^{(i)}}^{(i)}$ serve as (appropriately rescaled) local models for the action of~$-\cL_\beta$ near critical points of $V$.
    The proof of Theorem~\ref{thm:harmonic} relies on a variational argument similar to the one used in~\cite[Theorem 11.1]{CFKS87} or~\cite{S83}.
    Using the eigenmodes of~$K_\alpha$, we construct variational test families for $\cL_\beta$, so-called harmonic quasimodes. The convergence~\eqref{eq:harmonic_approx} follows from localization estimates on these quasimodes and the Courant--Fischer principle.

    Crucially, the geometric assumptions outlined in the previous paragraph ensure that the eigenvalues~$\lambda_{k,\alpha}^{\mathrm{H}}$ can be explicitly computed, as they belong to the spectrum of one of the local oscillators~$K_{\alpha^{(i)}}^{(i)}$ for some~$0\leq i<N$.
    Indeed (see~\cite[Section 4.2]{BLS24}), the spectrum of $K_{\alpha^{(i)}}^{(i)}$ can be enumerated (with multiplicities) by
    \begin{equation}
        \label{eq:harmonic_l2}
        \mathrm{Spec}\,K_{\alpha^{(i)}}^{(i)}=\left(|\hessEigval{i}{1}|\mu_{n_1,\alpha^{(i)}\left(|\hessEigval{i}{1}|/2\right)^{1/2}}-\frac{\hessEigval{i}{1}}{2}+\sum_{j=2}^d\left[|\hessEigval{i}{j}|(n_j+1/2)-\frac{\hessEigval{i}{j}}{2}\right]\right)_{n\in\N^d},
    \end{equation}
    where~$\mu_{n,a}$ is the~$(n+1)$-th eigenvalue of the one-dimensional quantum oscillator~$\frac12(-\partial_x^2-x^2)$ acting on~$L^2(-\infty,a)$ with Dirichlet boundary conditions.
    The particular values~$\mu_{n,\infty}=n+1/2$ and~$\mu_{n,0}=2n+3/2$ are well-known, so that the spectrum of the harmonic approximation is fully explicit in terms of eigenvalues of the Hessian~$\nabla^2 V(z_i)$ in the case all the critical points~$z_i$ of~$V$ in~$\mathcal K$ lie either on the boundary (i.e.~$\epsLimit{i}=0$) or~$\varepsilon$-inside~$\partial\Omega_{\alpha,\beta}$ (i.e.~$\epsLimit{i}=+\infty$) for all~$\beta>0$. Otherwise, one generally has to compute the values of~$\mu_{n,a}$ numerically. The (nonincreasing) functions~$a\mapsto \mu_{n,a}$ can be computed once and for all with high precision for a range of integers~$n$.

    The value of~$\lambda_{k,\alpha}^{\mathrm{H}}$ can then be easily obtained by taking the~$k$-th largest element from the union with multiplicity (i.e. the multiset union) of each of the sets~$\mathrm{Spec}\,K_{\alpha^{(i)}}^{(i)}$.
    
    \paragraph{Modified Eyring--Kramers formula for the metastable exit rate.}
    When~$z_i$ is a local minimum of~$V$ such that~$\epsLimit{i}=+\infty$, the bottom eigenvalue of~$K_{\epsLimit{i}}^{(i)}$ is~$0$. Thus, the harmonic approximation predicts a metastable exit rate of~$0$, which calls for finer asymptotics.
    The following result fulfills this need.
    \begin{theorem}
        \label{thm:eyring_kramers}
        Under Assumptions~\ref{hyp:locally_flat},~\ref{hyp:one_min} and~\ref{hyp:energy_condition},
        it holds, in the limit~$\beta\to\infty$:
        \begin{equation}
            \label{eq:eyring_kramers}
            \lambda_{1,\beta}(\Omega_{\beta}) = \e^{-\beta\left(\Vstar-V(z_0)\right)} \left[\sum_{i\in I_{\min}} \frac{|\hessEigval{i}{1}|}{2\pi\Phi\left(\sqrt{|\hessEigval{i}{1}|}\alpha_i\right)}\sqrt{\frac{\det \nabla^2 V(z_0)}{\left|\det \nabla ^2 V(z_i)\right|}}\right]\left(1 +\bigo(\beta^{-\frac12})\right),
        \end{equation}
        where~$\Phi(x)=(2\pi)^{-\frac12}\int_{-\infty}^x\e^{-\frac{t^2}2}\,\d t$, and~$\hessEigval{i}{1}$ is the unique negative eigenvalue of the Hessian matrix~$\nabla^2 V(z_i)$ at the saddle point~$z_i$.
    \end{theorem}
    For a full proof of this result, see the proof of~\cite[Theorem 5]{BLS24}. It relies on the construction of a precise approximation~$\psi_\beta$ of the principal Dirichlet eigenmode~$u_1(\Omega_\beta)$.
    
    Roughly speaking,~$\psi_\beta$ is constructed by combining a smoothed indicator of the set~$\basin{z_0}\cap \{V<\Vstar\}$ with, near each low-energy saddle points~$(z_i)_{i\in I_{\min}}$, a finer construction based on formal eigenmodes for the linearization of the dynamics~\eqref{eq:overdamped_langevin_sc}, which corresponds to an unstable Ornstein--Uhlenbeck process.
    Due to the geometric structure of the domain~$\Omega_{\alpha,\beta}$ given by Assumption~\ref{hyp:locally_flat}, one can separate variables in the unstable direction, leading to an explicit expression for these formal eigenmodes in terms of the unstable coordinate
    $$\xi_\beta^{(i)}(x) = \sqrt{\beta}(x-z_i)^\top \hessEigvec{i}{1}.$$
    The quasimode is then projected onto the principal eigenspace~$\mathrm{Span}(u_1(\Omega_{\alpha,\beta}))$, yielding~$\lambda_{1,\beta}(\Omega_{\alpha,\beta})$ as a Rayleigh quotient associated with the projected quasimode. Quantitative estimates based on a modification of Laplace's method and a resolvent estimate then allows to bound the projection error, which is sufficiently small to give sharp estimates on~$\lambda_{1,\beta}(\Omega_{\alpha,\beta})$.

    \paragraph{Application to the separation of timescales.}
    We briefly discuss the implications of Theorems~\ref{thm:harmonic} and~\ref{thm:eyring_kramers} for the problem of maximizing the separation of timescales~\eqref{eq:separation_of_timescales}. We refer to~\cite[Section 3.3]{BLS24} for additional details.

    The first point of interest is that Theorem~\ref{thm:harmonic} gives a quantitative estimate of the spectral gap~$\lambda_2(\Omega_{\alpha,\beta})-\lambda_1(\Omega_{\alpha,\beta})$ for large~$\beta$, and, in view of~\eqref{eq:decorrelation_rate}, of the asymptotic rate of convergence to the QSD.
    This estimate is solely a function of the asymptotic shape parameter~$\alpha$, and of the eigenvalues of the Hessian~$\nabla^2 V$ of the potential at some critical points.
    As such, it can be used to choose decorrelation times in Algorithm~\ref{alg:parrep} for highly metastable systems. Explicitly, the second harmonic eigenvalue is given by:
    \begin{equation}
        \label{eq:lambda_2}
        \lambda_2^{\mathrm{H}}=\min\left\{\underset{1\leq j\leq d}{\min}\,\hessEigval{0}{j},\underset{1\leq i<N}{\min}\,\left[\abs{\hessEigval{i}{1}}\mu\left(\alpha^{(i)}\sqrt{\hessEigval{i}{1}/2}\right)-\frac{\hessEigval{i}{1}}2 + \sum_{2\leq j\leq d}\abs{\hessEigval{i}{j}}\1_{\hessEigval{i}{j}<0}\right]\right\},
    \end{equation}
    where we set~$\mu(\theta):=\mu_{1,\theta}$ to be the principal eigenvalue of the Dirichlet harmonic oscillator on~$(-\infty,\theta)$. The limiting eigenvalue~$\lambda_2^{\mathrm{H}}$ is positive under Assumption~\ref{hyp:one_min}.
    Interestingly, this estimate is not always in agreement with standard numerical practice, which relies on a harmonic approximation of the energy basin at the local minimum to set the decorrelation time~(see for instance~\cite{PUV15}). This approximation neglects the possible effect of higher-order saddle points. It can be shown to fail when the Hessian~$\nabla^2 V$ has sufficiently soft modes around such critical points, and these are low enough in energy to be visited during decorrelation to the QSD.

    Theorem~\ref{thm:eyring_kramers} provides a quantitative estimate of the exit time starting from the QSD as a function of~$\alpha$. Combined with the previous estimate, we therefore obtain an estimate for the separation of timescales~\eqref{eq:separation_of_timescales} as a function of~$\alpha$ and~$\beta$. In view of~\eqref{eq:harmonic_l2} and~\eqref{eq:eyring_kramers}, one finds that the asymptotic separation of timescales, for the class of domains satisfying Assumptions~\ref{hyp:locally_flat}--\ref{hyp:energy_condition}, is of order~$\e^{\beta(\Vstar-V(z_0))}$, with a~$\beta$-independent prefactor~$C(\alpha)$ (neglecting error terms).
    Therefore, the asymptotic shape optimization problem~\eqref{eq:asymptotic_opt_problem} for the separation of timescales can only hope to improve on the prefactor.
    
    We show in~\cite[Section 3.3]{BLS24} that there exist asymptotically optimal domains, and indeed many in general. Qualitatively, these optimal domains are found to have the following properties. Firstly, they spill out beyond low-energy separating saddle points, on geometric scales of the order~$1/\sqrt\beta$ in the unstable direction. This means that one should wait for the system to reach an energy level lower than~$\Vstar$ (in the next basin of attraction) by a multiple of the characteristic thermal fluctuation~$\beta$ before declaring that a transition has occurred. The value of the multiplicative constant depends on the geometry of the energy landscape, but can be computed numerically.
    Secondly, one can show that they can never absorb other local energy minima, in the sense that the asymptotic separation of timescales necessarily decreases when continuously growing the domain so as to include any other minima far (i.e. at distances~$\gg1/\sqrt\beta$) inside the domain. This gives some theoretical indication that there indeed exist local shape optima surrounding basins of attraction of local minima for the steepest descent dynamics.

    We present validations of Theorems~\ref{thm:harmonic} and~\ref{thm:eyring_kramers} in Section~\ref{subsec:semiclassical_val} below, and connect the asymptotic problem~\eqref{eq:asymptotic_opt_problem} to a shape-optimization problem in one spatial dimension.
\section{Numerical experiments}
\label{sec:numerical}
In this section, we present various numerical experiments to illustrate and validate the results and methodology presented in Sections~\ref{sec:practical_opt}, \ref{sec:coarse_graining}, and~\ref{sec:semiclassic}. In Section~\ref{subsec:num_coarse_graining}, we verify, on a model two-dimensional situation that, given a suitable choice of CV, the coarse-grained Dirichlet eigenvalues provide a good approximation for the true eigenvalues of the Dirichlet generator.
In Section~\ref{subsec:semiclassical_val}, we show how the results of Section~\ref{sec:semiclassic} can be used to approximate the shape optimization problem in the semiclassical limit, and verify in particulars the spectral asymptotics given by Theorems~\ref{thm:eyring_kramers} and~\ref{thm:harmonic}.
In Section~\ref{subsec:diala}, we finally apply the coarse-grained shape optimization methodology to a realistic molecular system, and estimate the gain in the separation of timescales in the practical setting of underdamped Langevin dynamics.

The code used to generate the numerical results of this paper are publicly available in the paper repository~\cite{github}. Data generated from the various simulations and optimization runs can moreover be obtained from the repository~\cite{data}.

\subsection{Validation of the coarse-graining approximation}
\label{subsec:num_coarse_graining}
In this section, we demonstrate numerically that, for an appropriate choice of CV, the coarse-grained Dirichlet eigenvalues defined in Section~\ref{sec:coarse_graining} provide a good approximation for the lowest eigenvalues of the Dirichlet generators. As such, they can be used as a proxy to optimize the separation of timescales.
\paragraph{Two-dimensional system and collective variables.}
We consider, for a parameter~$\varepsilon>0$, the following family of potential functions defined on the configurational space~$\R^2\setminus\{0\}$:
\begin{equation}
    \label{eq:two_d_pot}
    V_{\varepsilon}(x,y) = (x^2-1)^2 + \frac{1}{\varepsilon}(x^2+y^2-1)^2 + \frac1{\sqrt{x^2+y^2}}.
\end{equation}
The potential~$V_{\varepsilon}$ is the sum of a quartic double-well potential in the variable~$x$, and of a harmonic energy in the squared radial coordinate $r^2=x^2+y^2$, whose sharpness is modulated by~$\varepsilon$, confining the dynamics to the unit circle. The additional repulsion term~$1/r$ ensures that the effective diffusion coefficient~$a_{\xi_1}$ is well-defined, as discussed below.
The potential is depicted in Figure~\ref{fig:pot2d} for the three values of~$\varepsilon$ we consider in this experiment.
\begin{figure}
    \begin{subfigure}{\linewidth}
    \center
        \includegraphics[width=0.49\linewidth]{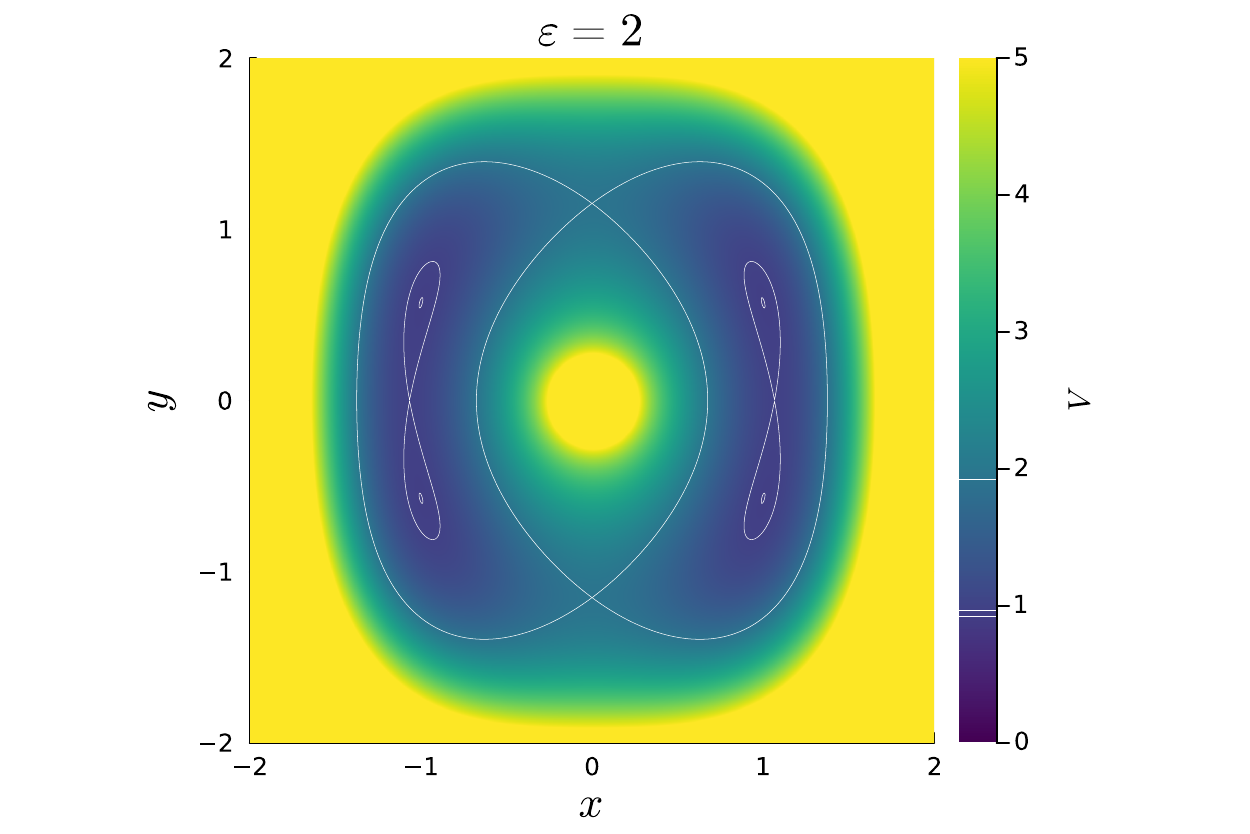}
        \includegraphics[width=0.49\linewidth]{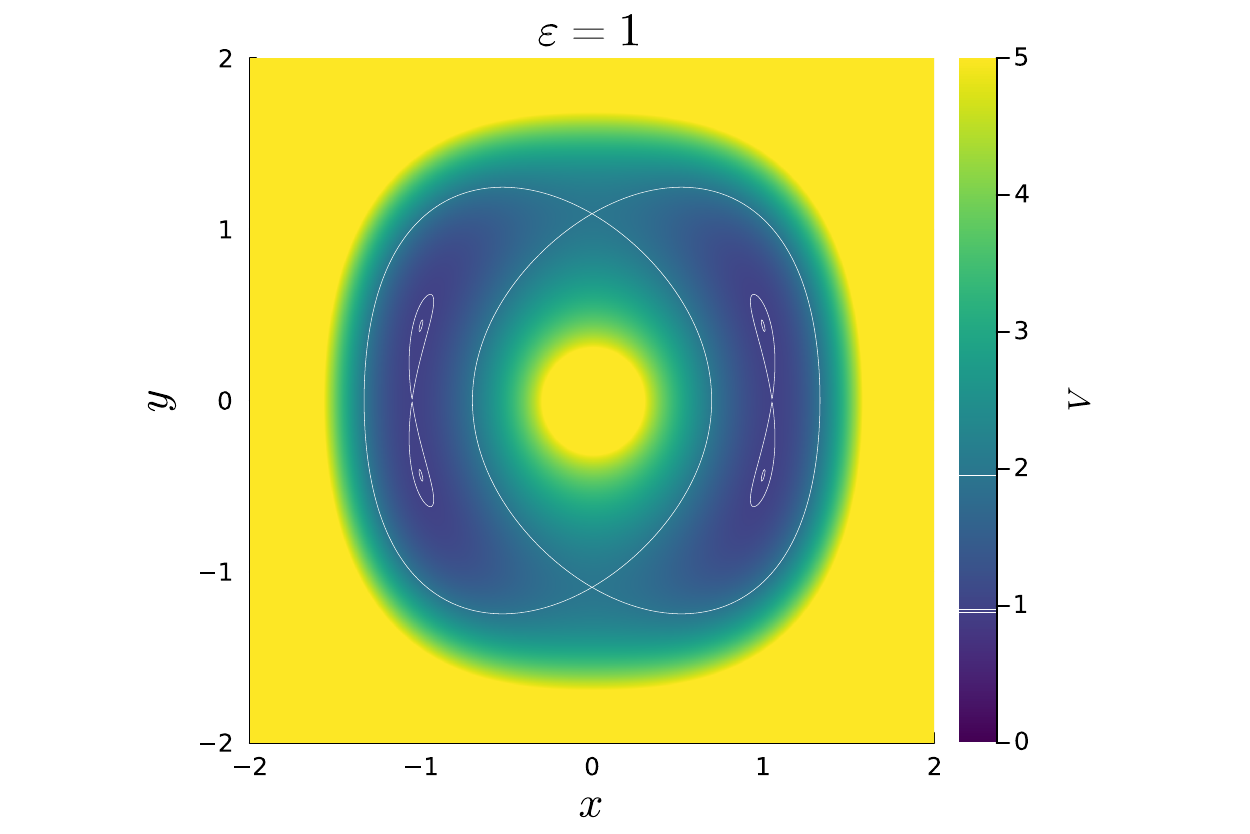}
    \end{subfigure}
    \begin{subfigure}{\linewidth}
    \center
        \includegraphics[width=0.49\linewidth]{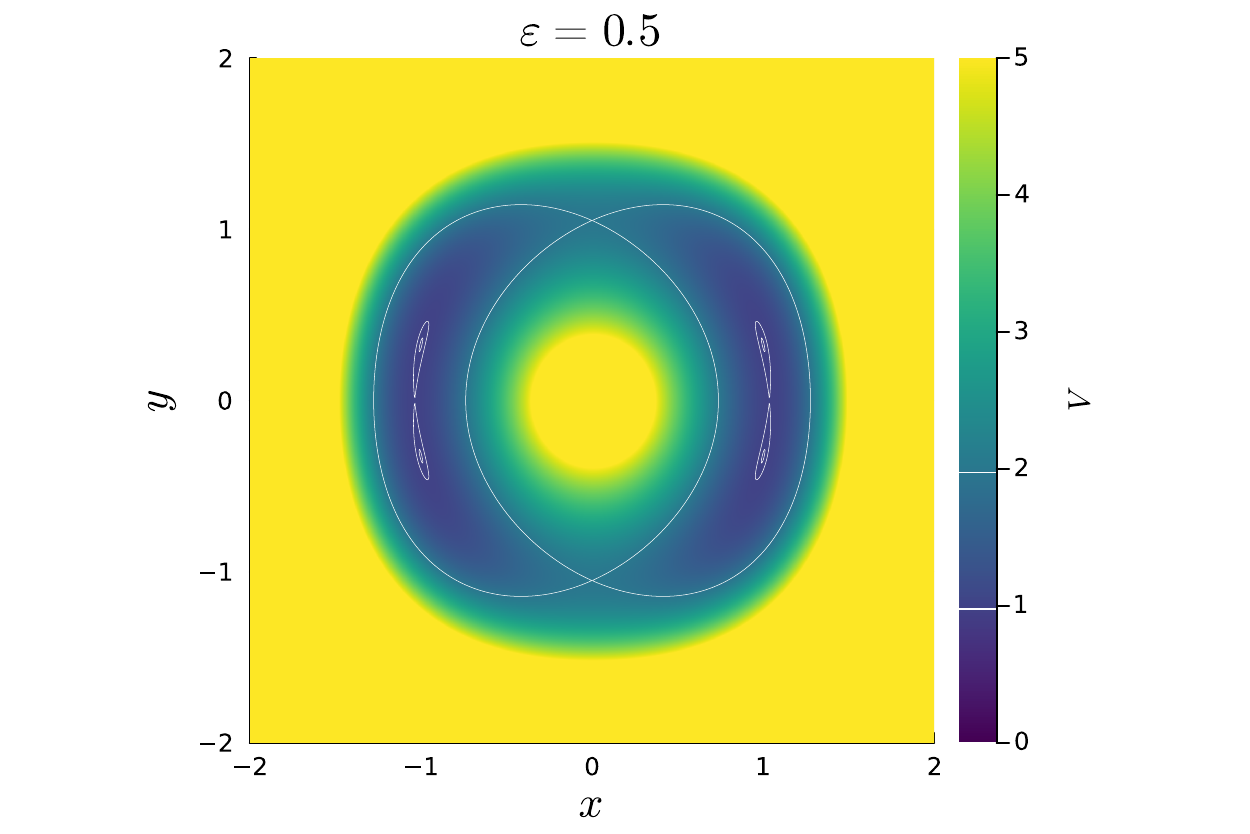}
    \end{subfigure}
\caption{Two-dimensional potentials~\eqref{eq:two_d_pot}, for decreasing values of the parameter~$\varepsilon$. In each case, the potential has a local minimum in each quadrant of the plane, and two saddles on each axis. The saddles on the~$y$-axis separate two deep energy basins, while the saddle points on the~$x$-axis form shallow energy barriers inside these basins. Some energy level sets, in thin white lines, highlight the well structure.} 
\label{fig:pot2d}
\end{figure}

We compare  the two following CVs:
\begin{equation}
    \label{eq:reaction_coordinates}
    \xi_1(x,y) = \frac 2\pi \atan\left(\frac{y}{x+\sqrt{x^2+y^2}}\right),\qquad \xi_2(x,y)=x.
\end{equation}
The variable~$\xi_1$ is equal to~$\theta/\pi$, where~$(r,\theta)$ is the image of~$(x,y)$ via a polar change of variables. In particular, the CV~$\xi_1$ takes values in the compact interval~$[-1,1]$, while~$\xi_2$ is unbounded. In the limit~$\varepsilon\to 0$, we expect the effective dynamics through~$\xi_1$ to provide a good one-dimensional description of the original dynamics, and~$\xi_2$, while able to resolve the main energy barrier, is blind to the shape of the energy minima (e.g. the shallow energy barriers separating the two rightmost local minima), leading to a poor model for the local decorrelation inside the rightmost well. 

For each of these functions, the value of the free-energy and diffusion coefficient, given respectively by~\eqref{eq:free_energy} and~\eqref{eq:diffusion_tensor}, are computed at values of~$z$ corresponding to~$N=1000$ points on a regular grid (on the interval~$[-2,2]$ for~$\xi_2$),  by numerical quadrature (using the Gauss--Kronrod rule as implemented in the Julia package Cubature.jl) on the manifold~$\Sigma_z$, which for both our choices of CV~\eqref{eq:reaction_coordinates} have a simple linear parametrization. The resulting free-energy and diffusion profiles are presented in Figure~\ref{fig:fd_plot}.

\begin{figure}
    \centering
    \includegraphics[width=0.49\textwidth]{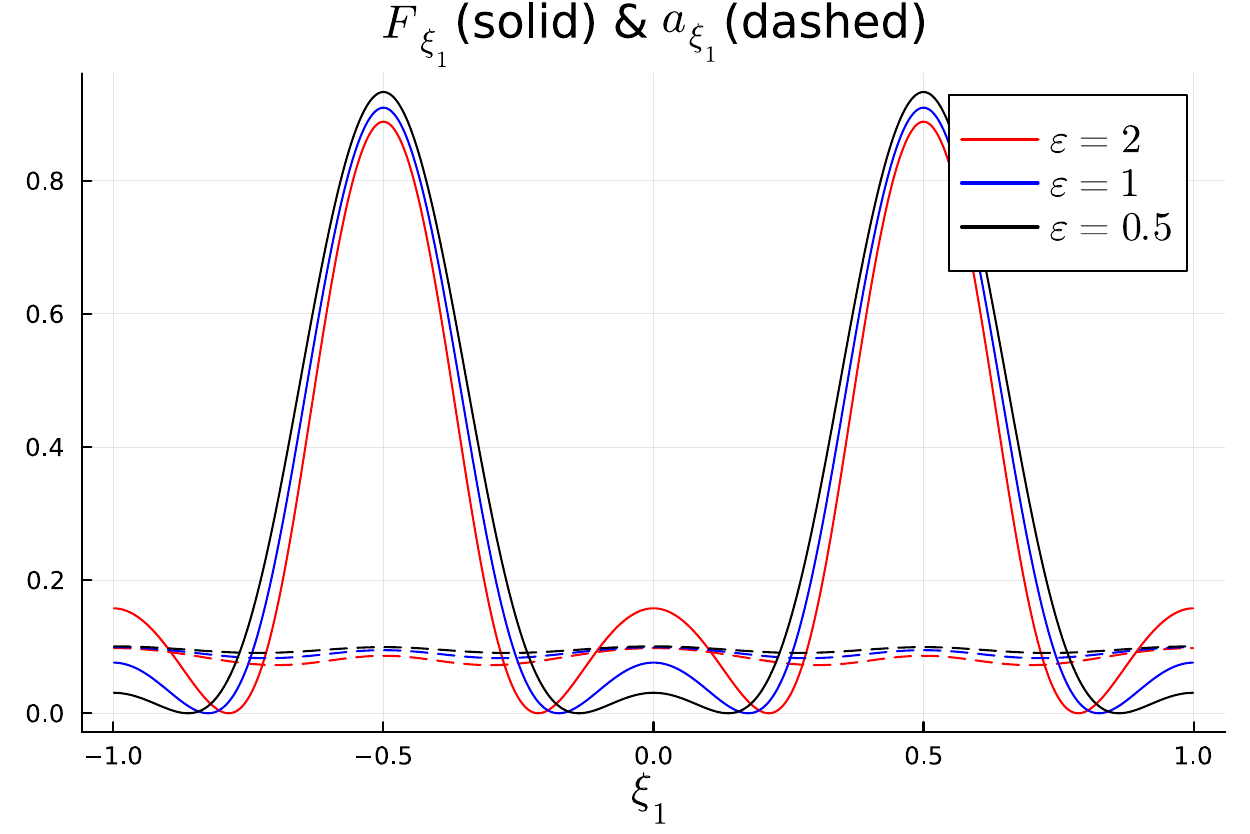}
    \includegraphics[width=0.49\textwidth]{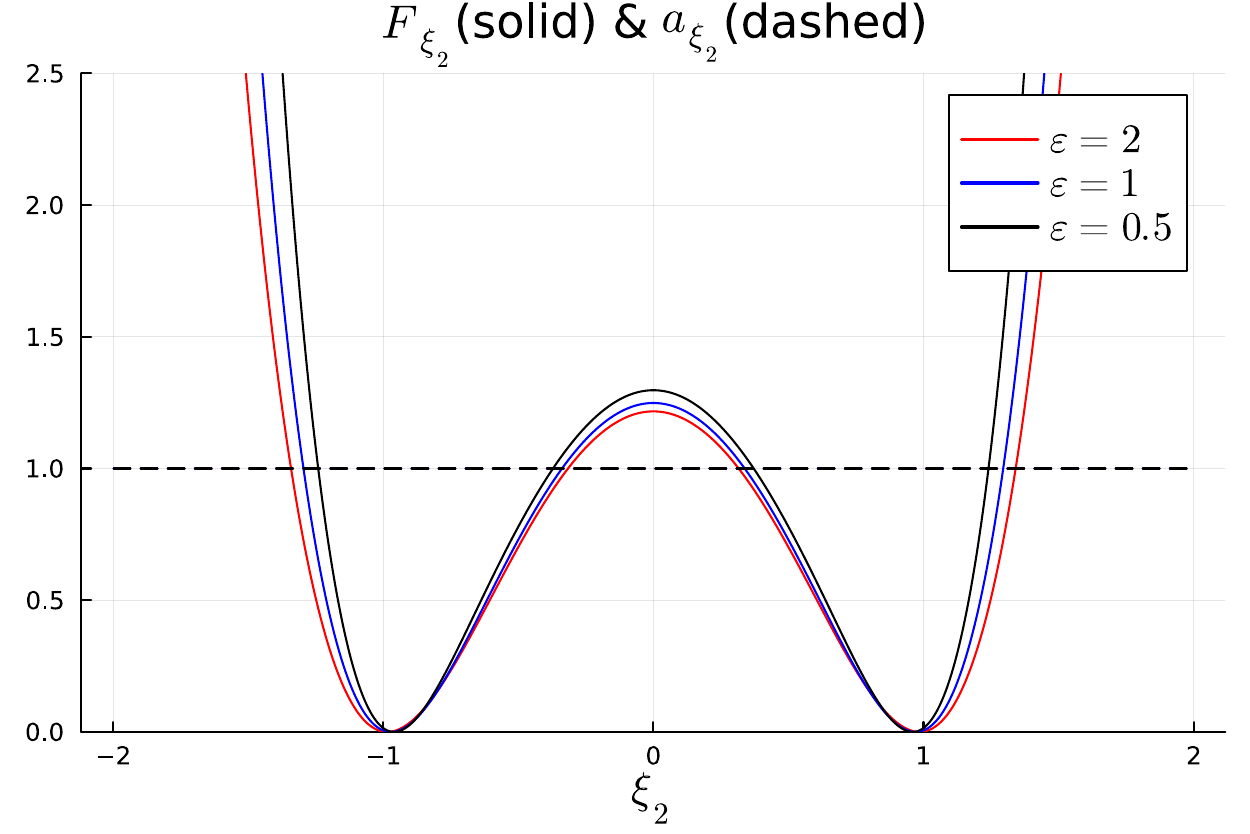}
    \caption{Free energy profiles and effective diffusion coefficients for the CVs~$\xi_1$ and~$\xi_2$ defined in~\eqref{eq:reaction_coordinates}, and for the potential~\eqref{eq:two_d_pot}. Various values of the parameter~$\varepsilon$ are color-coded. Free energy profiles are depicted in solid lines and effective diffusion coefficients are plotted in dashed lines.}
    \label{fig:fd_plot}
\end{figure}

\paragraph{Computation of the coarse-grained Dirichlet eigenvalues.}
For~$\xi\in\left\{\xi_1,\xi_2\right\}$, we discretize the effective generator~$\cL_{\beta}^\xi$ as the generator of a reversible jump process on a regular grid~$(z_i)_{i\in \mathbb L_N}$ in collective variable space, where~$\mathbb L_N$ is either a periodic lattice~$\mathbb{L}_N=\mathbb Z/N\mathbb Z$ if~$\xi=\xi_1$ or~$\mathbb{L}_N=\left\{0,\dots,N-1\right\}$ if~$\xi=\xi_2$.
In both cases, we set~$N=1000$.
The grid points are defined by~$z_i=\xi_{\max}\left(\frac{2i+1}{N}-1\right)$, where~$\xi_{\max}=1$ for~$\xi=\xi_1$, and~$2$ for~$\xi=\xi_2$.
The jump rates are only positive for nearest neighbors:
\begin{equation}
    \label{eq:lattice_jump_generator}
    \cL_{N,\beta,ij}^\xi = \left(\beta(2\xi_{\max}/N)^2\right)^{-1}\e^{-\frac{\beta}2\left(F_{j,\xi}-F_{i,\xi}\right)}\left(\frac{a_{i,\xi}+a_{j,\xi}}{2}\right),\qquad \forall\,|i-j|_{\mathbb{L}_N}=1,
\end{equation}
where~$F_{i,\xi} = F_\xi(z_i)$~$a_{i,\xi}=a_\xi(z_i)$ for any~$i\in \mathbb Z^m$, and~$|\cdot|_{\mathbb L_N}$ is the nearest-neighbor graph metric on~$\mathbb L_N$.
A simple computation shows that the jump process~\eqref{eq:lattice_jump_generator} is reversible for the on-site Boltzmann distribution, defined by
\begin{equation}\mu_{N}\left(\{z_i\}\right) =\frac{\e^{-\beta F_{i,\xi}}}{\sum_{j\in\mathbb L}\e^{-\beta F_{j,\xi}}},\qquad \forall\,i\in\mathbb L.\end{equation}

The factor~$\left(\beta(2\xi_{\max}/N)^2\right)^{-1}$ ensures that~\eqref{eq:lattice_jump_generator} is a consistent approximation of the generator associated with the SDE~\eqref{eq:effective_dynamics}.
Given an open domain~$\Omega\subset[-\xi_{\max},\xi_{\max}]$, the effective eigenvalues are approximated by computing the eigenvalues of the generator for the process killed outside~$\Omega$:
\begin{equation}
    \label{eq:submarkovian_generator}
    \left(\cL_{N,\beta,ij}^\xi\right)_{i,j\in I_{\Omega}},\qquad I_{\Omega}=\left\{i\in\mathbb L:\,z_i\in\Omega\right\}.
\end{equation}
The eigenvalues of the sparse matrix~\eqref{eq:submarkovian_generator} were numerically computed using the Julia interface to the Arpack module~\cite{LSY98}.
\paragraph{Validation of the approximation.}
In Figure~\ref{fig:coarse_grained_val}, we compare the approximation obtained from~\eqref{eq:submarkovian_generator} with Dirichlet eigenvalues of the generator~$\cL$ approximated using Algorithm~\ref{alg:rayleigh_ritz} in FreeFem++~\cite{freefem}.
The FreeFem++ implementation, including the parameters we used for geometry parametrization and meshing (which are the default parameters in the provided code), are available on GitHub~\cite{github}.
We compute the values of the first four Dirichlet eigenvalues for domains of the form~$\Omega(b)=\left(a,b\right)$, for a fixed value of~$a$ and for a range of values of~$b$, and for several values of the parameter~$\varepsilon$ (see Figure~\ref{fig:pot2d}).
We compare these eigenvalues to those of the effective generator, using the jump-process approximation~\eqref{eq:lattice_jump_generator}. We observe that, for~$\xi_1$, even for relatively large values of~$\varepsilon$, the effective eigenvalues give a good approximation to the true eigenvalues of the generator, across a wide range of boundary conditions.
The error appears to decrease for low values of~$\varepsilon$, as expected. However, the effective eigenvalues for~$\xi_2$ significantly depart from the true eigenvalues. This is especially true for the higher eigenvalues, confirming that the effective diffusion through~$\xi_2$ is unable to correctly model the decorrelation inside the energy wells.

These results suggest that~$\xi_1$ may be used for the purpose of shape optimization of the separation of timescales~$N^*$ defined in~\eqref{eq:separation_of_timescales}.
In Figure~\ref{fig:opt_2D}, we compare the (locally) optimal domain of the effective generator~\eqref{eq:lattice_jump_generator} with the (locally) optimal domain of the true generator in the class of domains defined in terms of the CV~$\xi_1$, and for the value~$\varepsilon=0.5$.
These optima were found by a full grid-search over the set of domains of the form~$(a,b)$ for~$-\xi_{\mathrm{max}}<a<b<\xi_{\mathrm{max}}$ in the case of the effective generator, and an iteratively refined grid search over domains of the form~$\xi_1^{-1}(a,b)$ for the case of the FEM generator, for the same range of~$a$ and~$b$.
The iterative refinement procedure consisted in searching for optimal domains for values~$(a,b)$ on a regularly spaced~$6\times 6$ grid, and iterating this procedure, restricting the search at the next iteration to the cell of the maximizer and its nearest-neighbors. The procedure stopped once a target grid resolution of~$\delta\xi=0.01$ was reached. 
We find the result of both these optimization procedures to give almost indistinguishable optimal values of~$a$ and~$b$, showcasing the usefulness of the effective generator, whose Dirichlet eigenvalues are significantly cheaper to compute.
\begin{figure}
    \centering
    \begin{subfigure}{\linewidth}
    \includegraphics[width=0.49\textwidth]{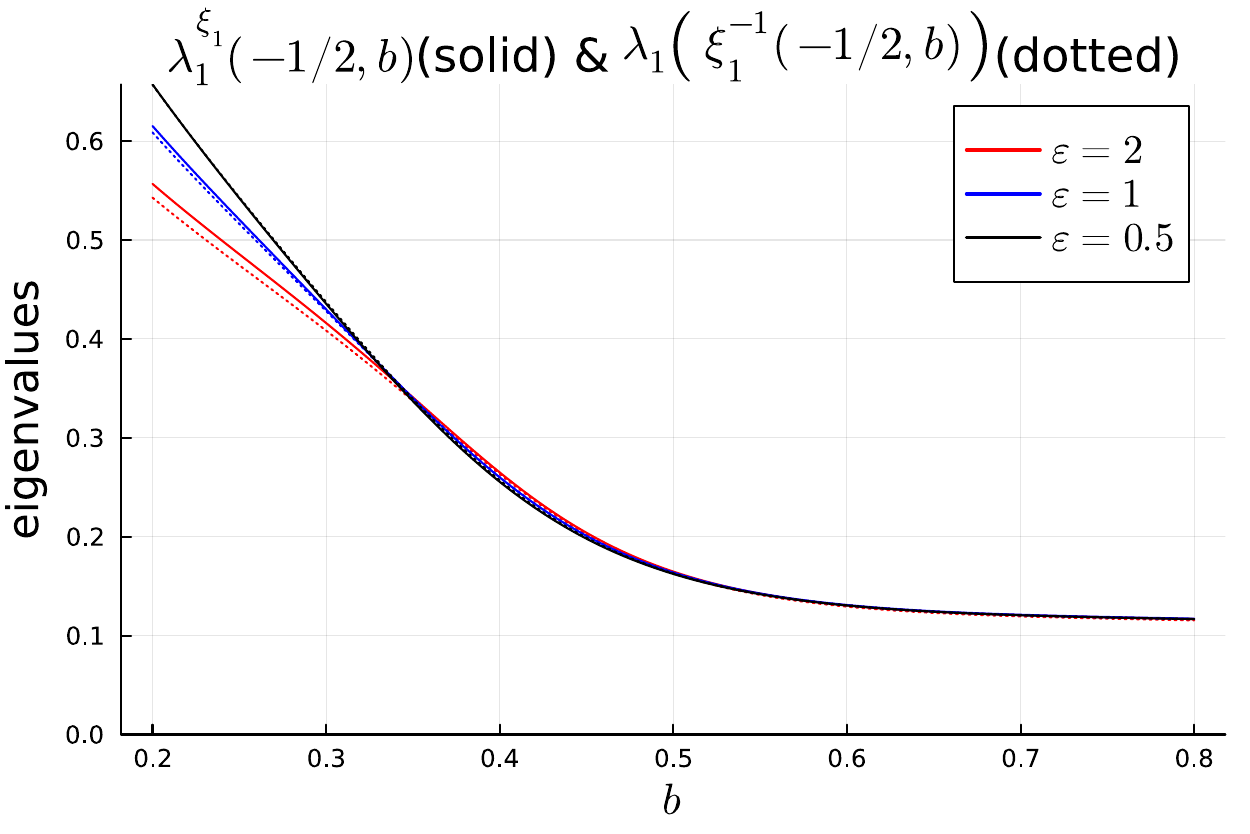}
    \includegraphics[width=0.49\textwidth]{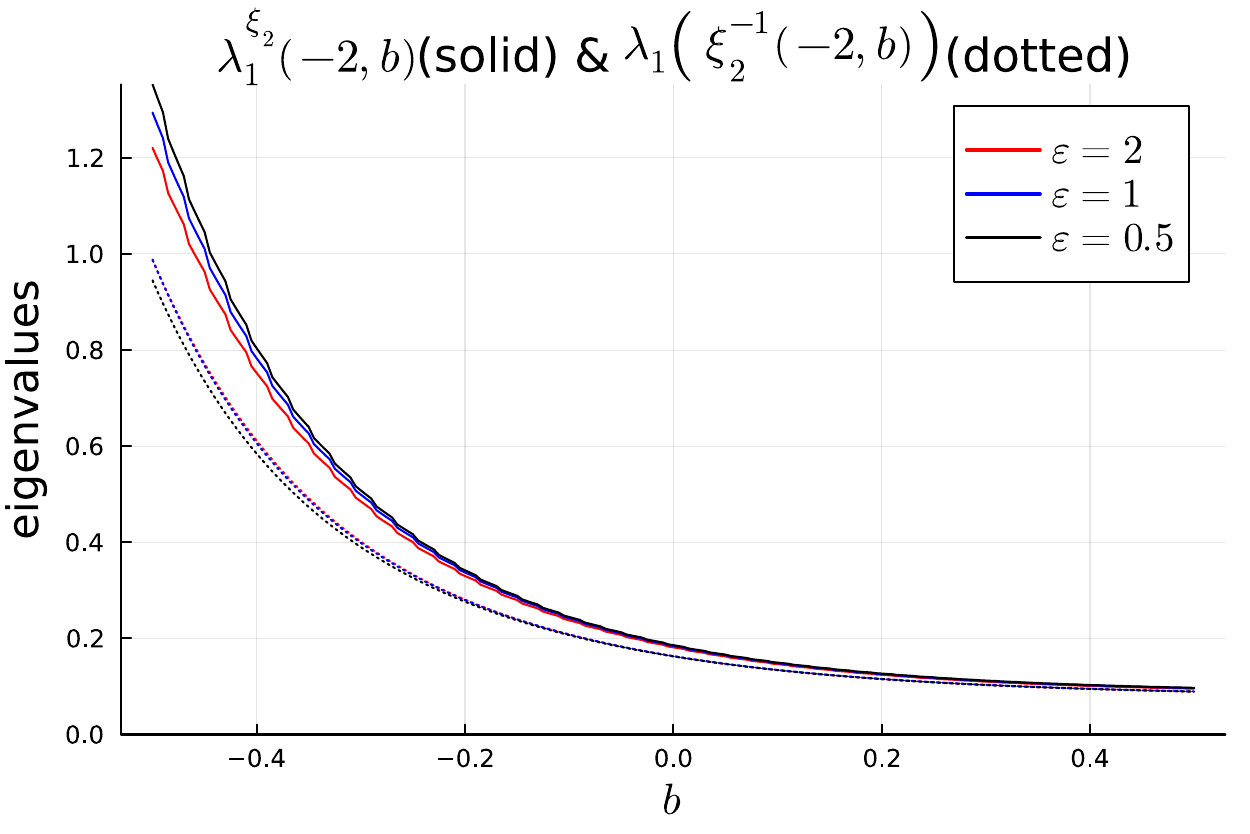}
    \caption{Approximation of the first Dirichlet eigenvalue.}
    \end{subfigure}
\begin{subfigure}{\linewidth}
    \includegraphics[width=0.49\textwidth]{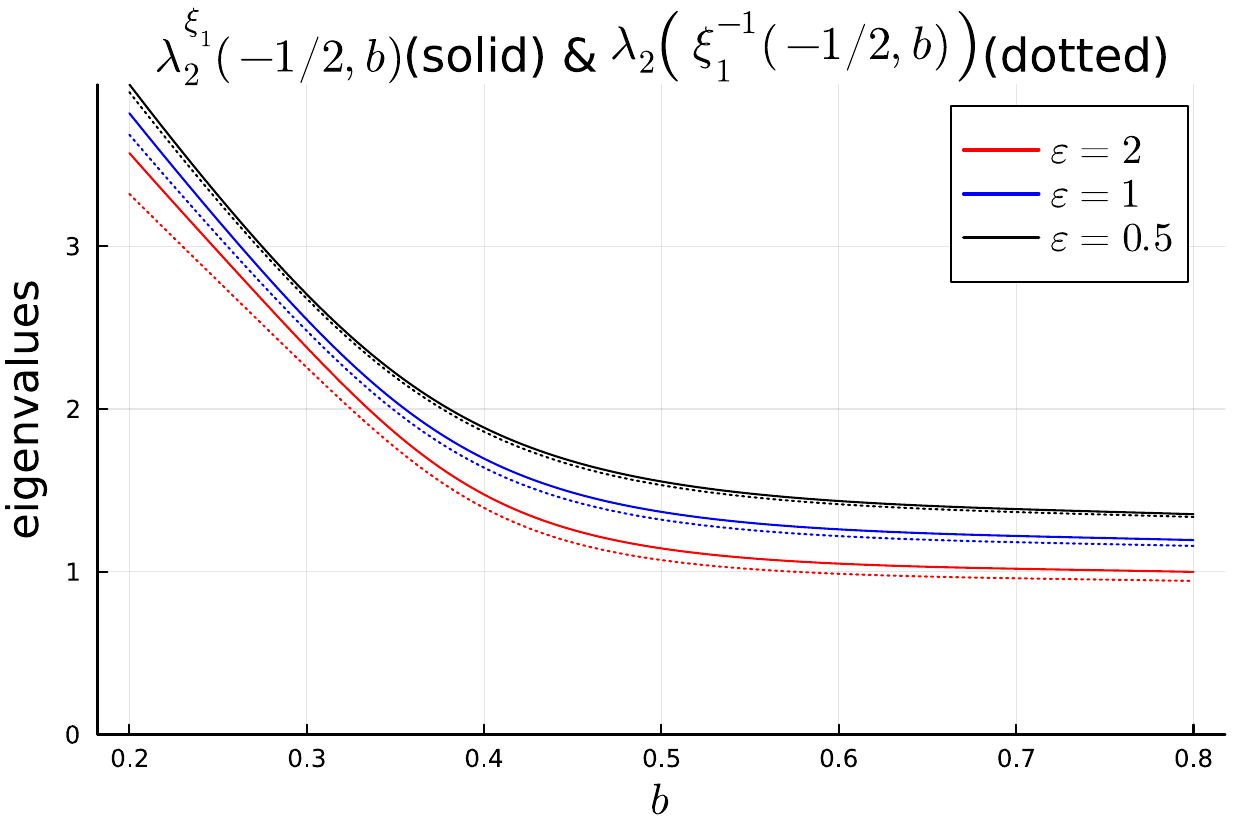}
    \includegraphics[width=0.49\textwidth]{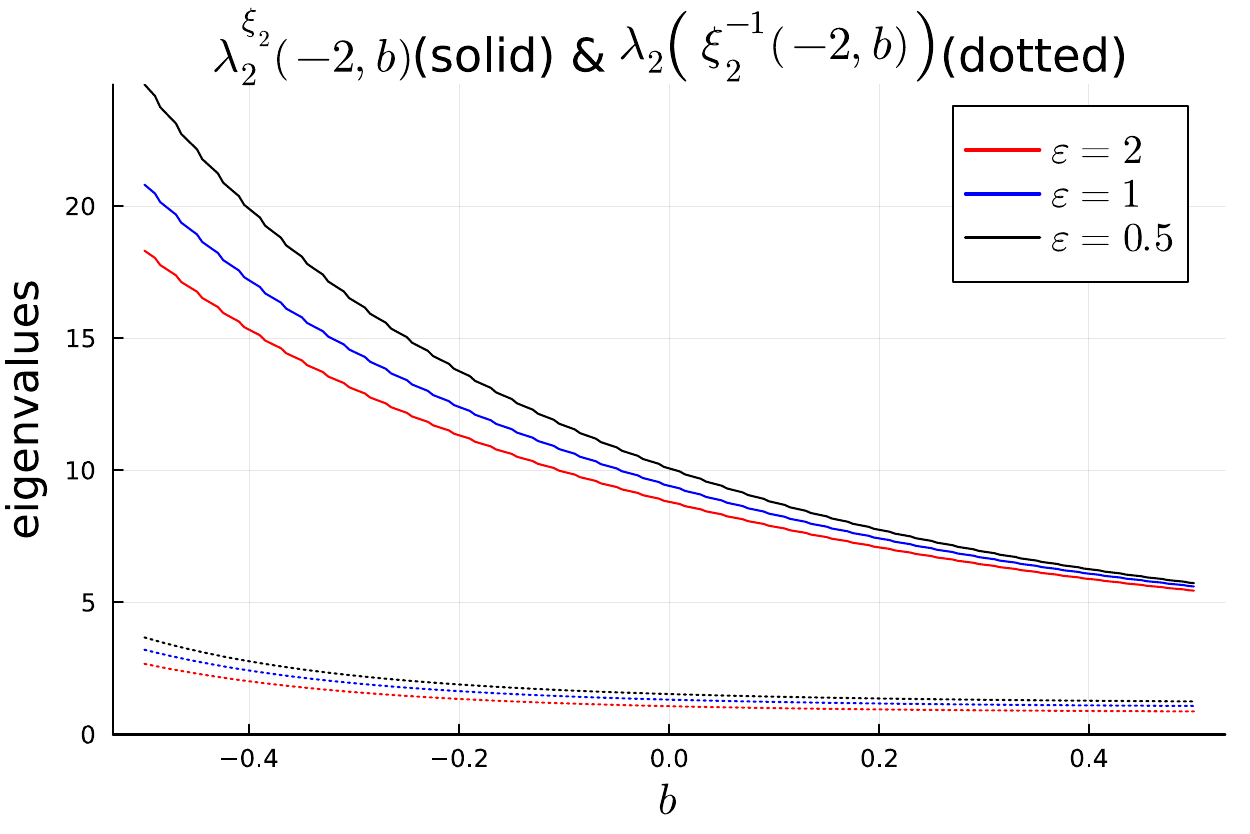}
    \caption{Approximation of the second Dirichlet eigenvalue.}
\end{subfigure}
\begin{subfigure}{\linewidth}
    \includegraphics[width=0.49\textwidth]{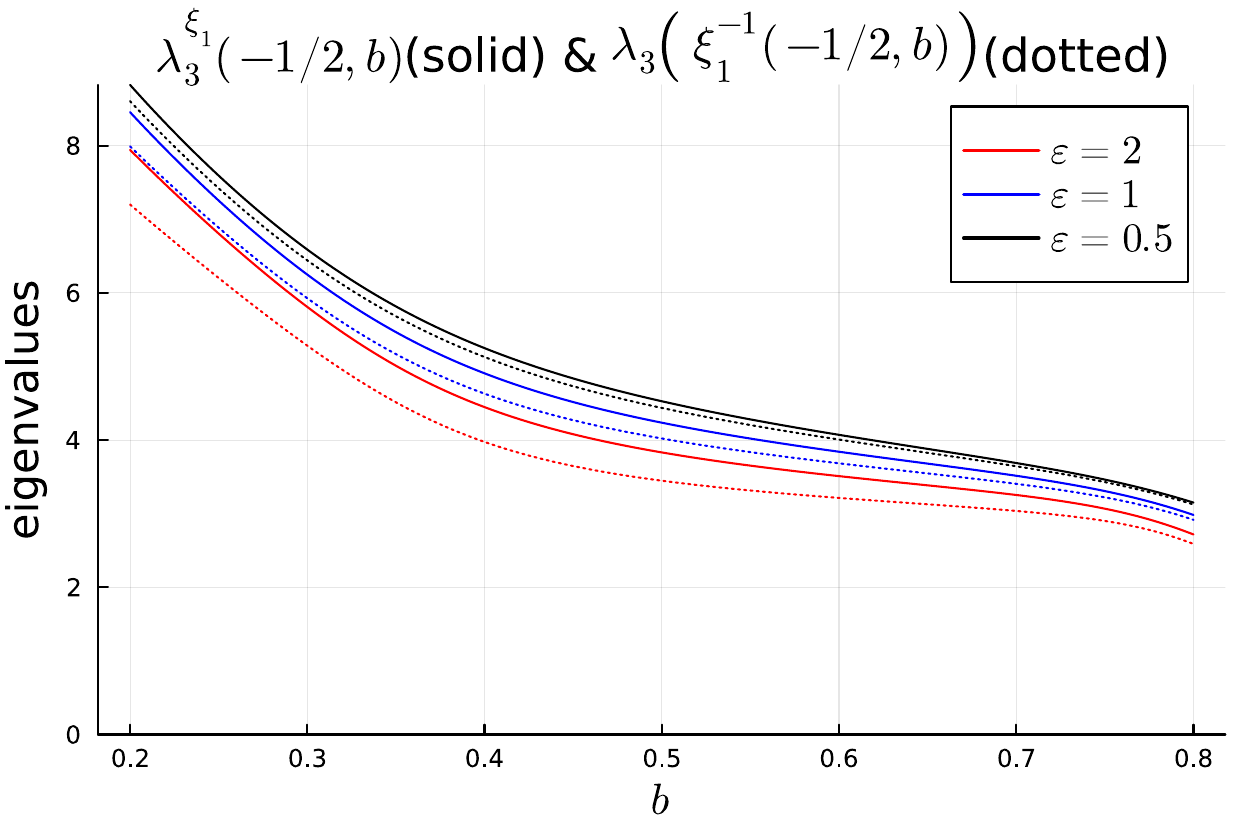}
    \includegraphics[width=0.49\textwidth]{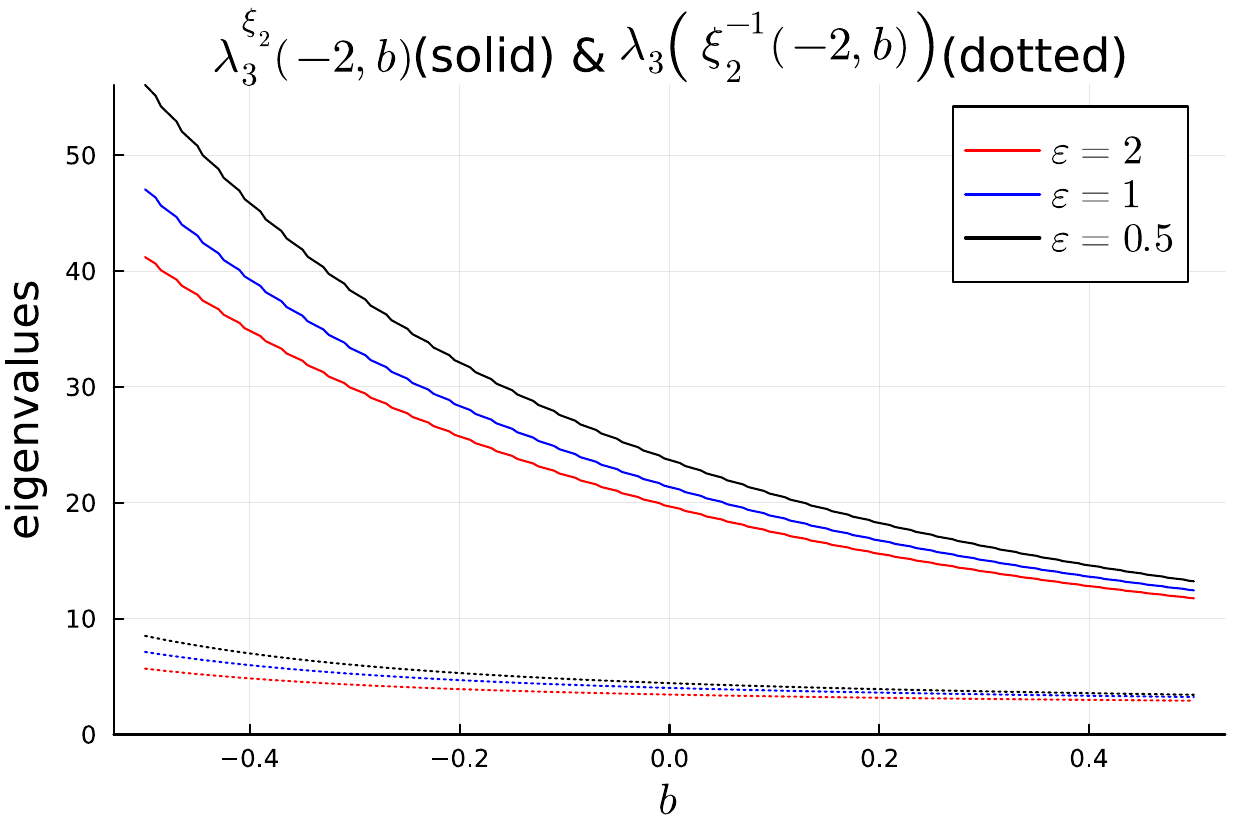}
    \caption{Approximation of the third Dirichlet eigenvalue.}
\end{subfigure}
\begin{subfigure}{\linewidth}
    \includegraphics[width=0.49\textwidth]{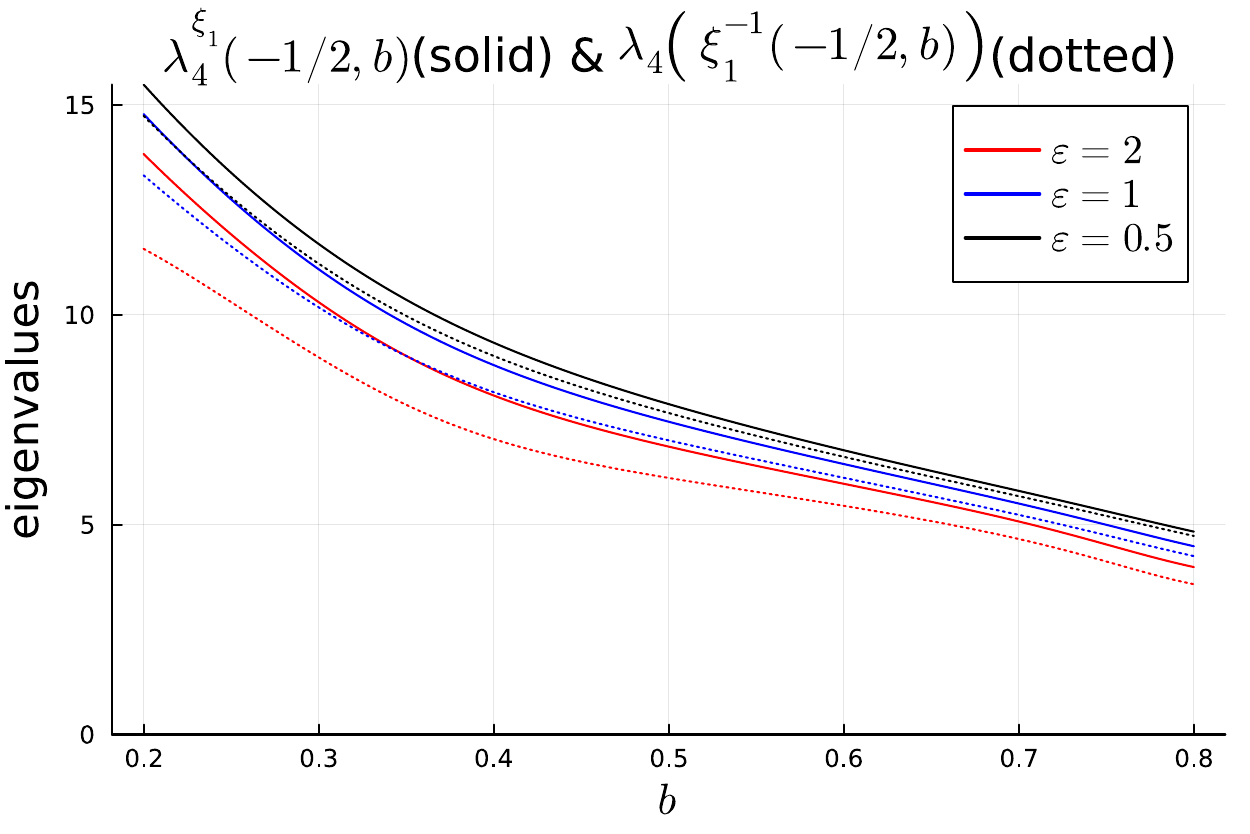}
    \includegraphics[width=0.49\textwidth]{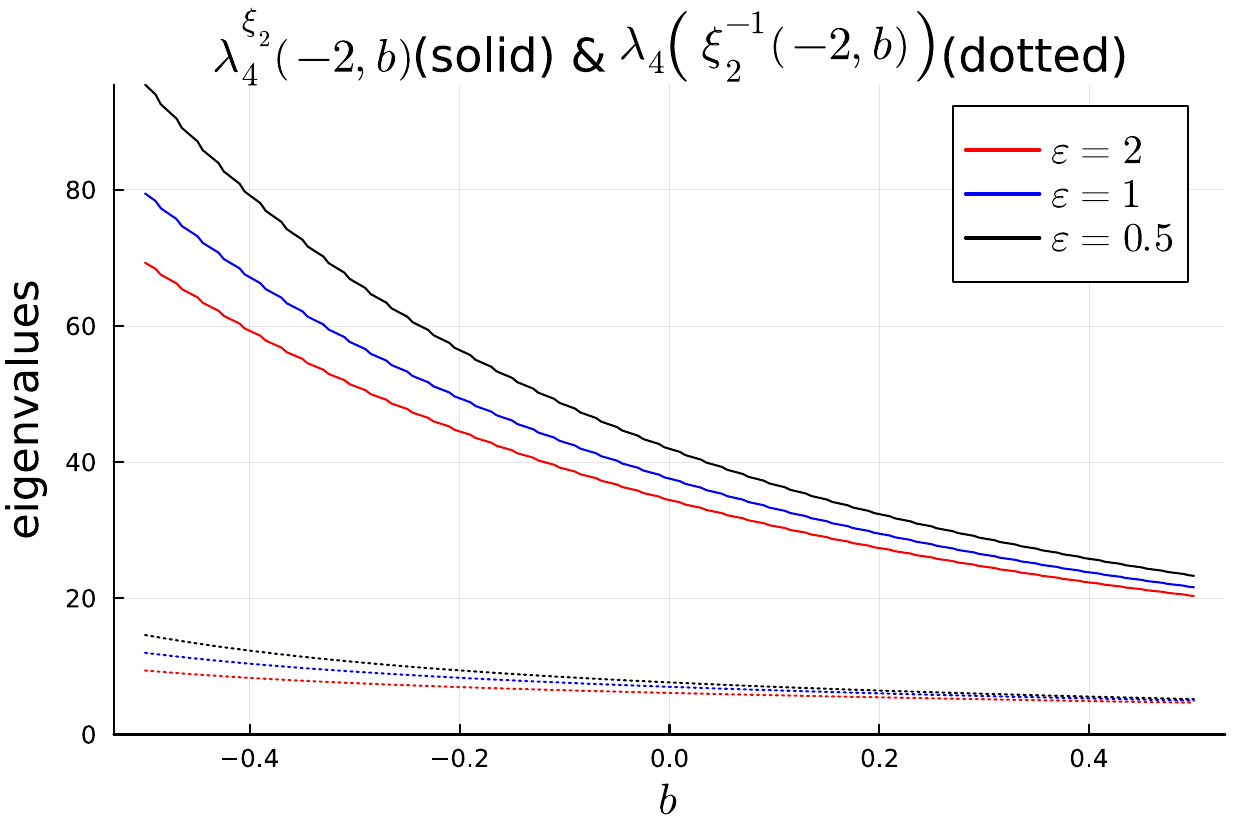}
    \caption{Approximation of the fourth Dirichlet eigenvalue.}
\end{subfigure}
\caption{Domain-dependent eigenvalues (dotted lines) and their coarse-grained approximations (dashed lines), for parametric families of domains defined in CV space.}
\label{fig:coarse_grained_val}
\end{figure}

\begin{figure}
    \centering
    \includegraphics[width=0.8\textwidth]{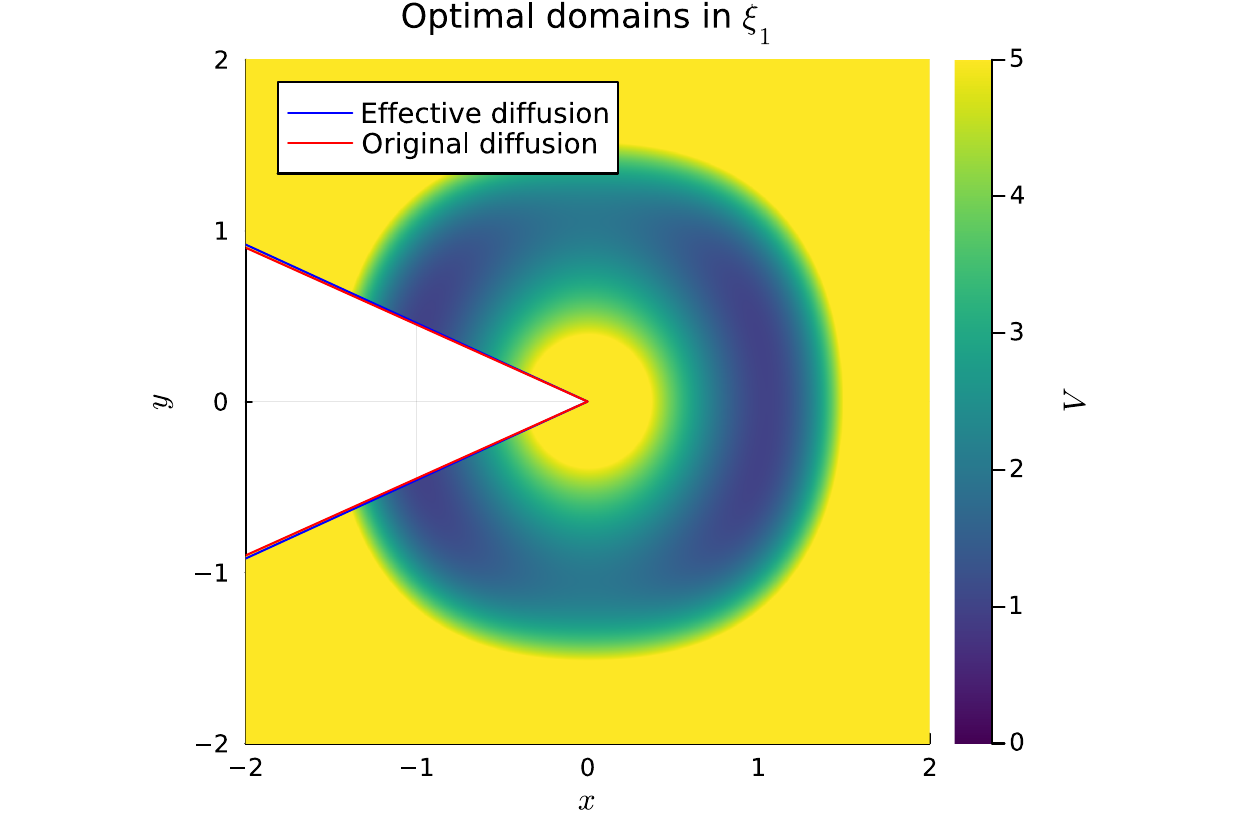}
    \caption{Optimal domain for the effective dynamics, and optimal domain for the original generator, in the class of domains defined in terms of~$\xi_1$, for the value of the parameter $\varepsilon=0.5$.  Points outside both of these domains lie in the white region. The optimized domains are almost indistinguishable.}
    \label{fig:opt_2D}
\end{figure}

\subsection{Validation of the semiclassical asymptotics}
\label{subsec:semiclassical_val}
In this section, we give a numerical verification of the semiclassical results obtained in~\cite{BLS24} (which corresponds to Theorems~\ref{thm:harmonic} and~\ref{thm:eyring_kramers} here), and assess their usefulness for the state definition problem, in a model one-dimensional situation.
\paragraph{Definition of the toy system.}
The potential~$V$ is defined by
\begin{equation}
    \label{eq:one_d_pot}
    V(x) = \epsilon\left(1-\cos\frac{x}{\sigma} + \exp\left(-\frac12\left(\frac{x}{\sigma}-1\right)^2\right) + \ell x\right),
\end{equation}
where~$(\epsilon,\sigma) = (0.7,1/4)$ are energy and scale parameters, and~$\ell \approx 0.01293$ is a constant factor chosen so that~$V$ has two index-one saddle points at~$z_1 \approx -0.7824$ and~$z_2\approx 0.8286$, satisfying~$V(z_1)=V(z_2)=\Vstar$, so that~$I_{\min}=\{1,2\}$.
Additionally~$V$ admits a local minimum at~$z_0\approx 0.1166$. The corresponding eigenvalues of the Hessian are given by~$(\nu^{(0)}_1,\nu^{(1)}_1,\nu^{(2)}_1)\approx (16.9532,-11.2348,-14.3845)$.
We consider, for a parameter~$\alpha = (\epsLimit{1},\epsLimit{2})\in \R^2$, temperature-dependent domains defined by
\begin{equation}
    \label{eq:temperature_dependent_domains_1D}
    \Omega_{\alpha,\beta} = \left(z_1-\frac{\epsLimit{1}}{\sqrt\beta},z_2+\frac{\epsLimit{2}}{\sqrt\beta}\right),
\end{equation}
which satisfy the assumptions of Theorems~\ref{thm:harmonic} and~\ref{thm:eyring_kramers}. The potential and domains (for a fixed value of~$\beta$) are depicted in Figure~\ref{fig:pot}.
\begin{figure}
    \includegraphics[width=0.8\textwidth]{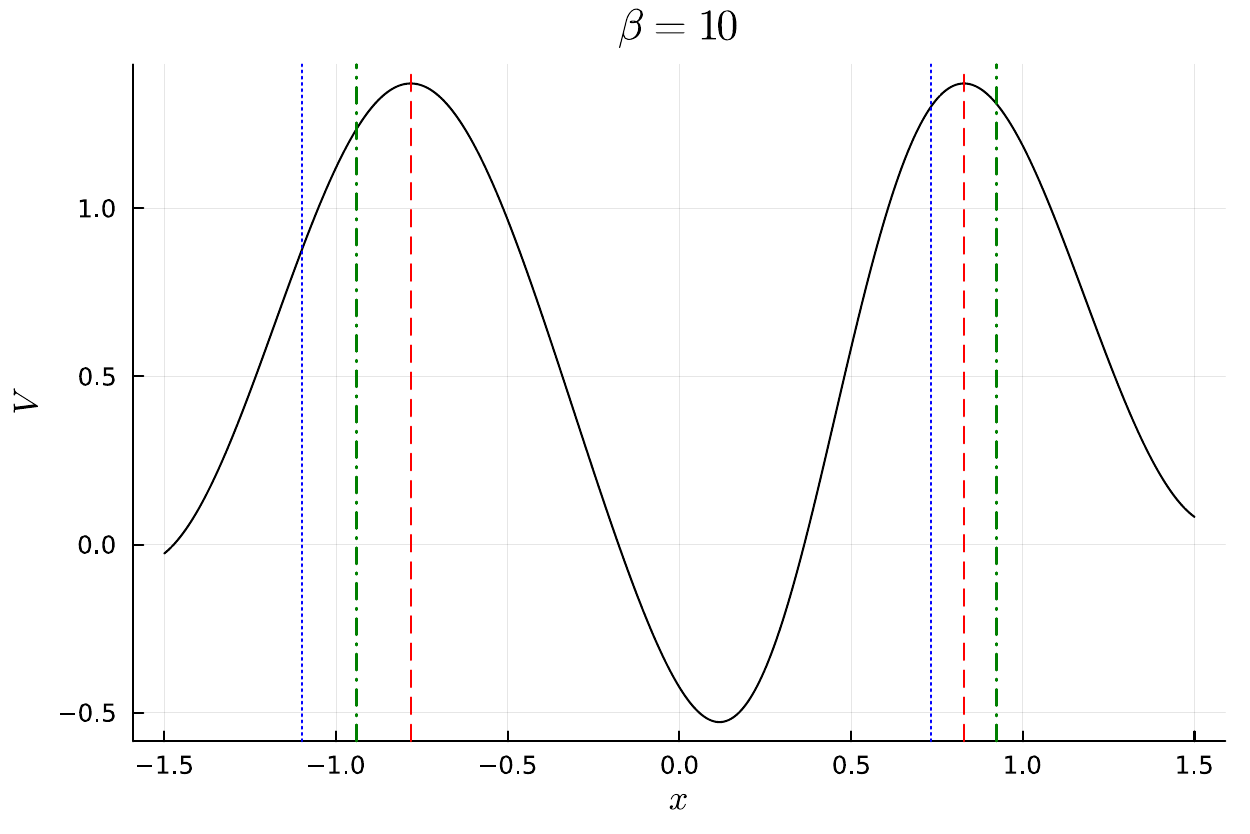}
    \caption{Potential landscape and domains~$\Omega_{\alpha,\beta}$ used in Figures~\ref{fig:ek} and~\ref{fig:harm}, as defined by~\eqref{eq:temperature_dependent_domains_1D}, at the fixed value of the temperature parameter~$\beta=10$. The color coding is the same as that used in Figure~\ref{fig:semiclassic}, i.e.~$\alpha=(0.5,0.3)$ in green,~$\alpha=(1.0,-0.3)$ in blue and~$\alpha=(0.0,0.0)$ in red, which correponds to the basin of attraction~$\basin{z_0}$.}
    \label{fig:pot}
\end{figure}

We aim to maximize
$$\frac{\lambda_{2,\beta}(\Omega_{\alpha,\beta})-\lambda_{1,\beta}(\Omega_{\alpha,\beta})}{\lambda_{1,\beta}(\Omega_{\alpha,\beta})}$$ with respect to~$\alpha$.
Fixing~$\beta$, it is equivalent to maximize the quantity
\begin{equation}
    \label{eq:reduced_obj}
    J_\beta(\alpha) = \frac{\lambda_{2,\beta}(\Omega_{\alpha,\beta})\lambda_{1,\beta}(\basin{z_0})}{\lambda_{1,\beta}(\Omega_{\alpha,\beta})\lambda_{2,\beta}(\basin{z_0})},
\end{equation}
where we recall~$\basin{z_0} = \Omega_{0,\beta}$ is the basin of attraction for the local minimum~$z_0$, see~\eqref{eq:basin}. The interest of considering this objective~$J_\beta$ is that, according to Theorems~\ref{thm:harmonic} and~\ref{thm:eyring_kramers},~$J_\beta~\xrightarrow{\beta\to+\infty}J_\infty$ pointwise, where
\begin{equation}
    \label{eq:lim_obj}
        J_\infty(\alpha) = \frac{\lambda_{2,\alpha}^{\mathrm{H}}C(0)}{\lambda_{2,0}^{\mathrm{H}}C(\alpha)},\qquad C(\alpha)=\sum_{i\in I_{\min}}\frac{|\hessEigval{i}{1}|}{2\pi\Phi\left(\sqrt{|\hessEigval{i}{1}|}\alpha_i\right)}\sqrt{\frac{\det \nabla^2 V(z_0)}{\left|\det \nabla ^2 V(z_i)\right|}},
\end{equation}
where~$C(\alpha)$ is the pre-exponential factor in~\eqref{eq:eyring_kramers}. Substituting the expression~\eqref{eq:lambda_2} in~\eqref{eq:lim_obj}, we find explicitly:
\begin{equation}
    \label{eq:lim_obj_explicit}
    J_\infty(\alpha)=\resizebox{.9\hsize}{!}{$2\displaystyle\frac{\min\left\{\hessEigval{0}{1},|\hessEigval{1}{1}|\left(\mu\left(\sqrt{|\hessEigval{1}{1}|/2}\epsLimit{1}\right)+\displaystyle\frac12\right),|\hessEigval{2}{1}|\left(\mu\left(\sqrt{|\hessEigval{2}{1}|/2}\epsLimit{2}\right)+\displaystyle\frac12\right)\right\}\left(\sqrt{|\hessEigval{1}{1}|}+\sqrt{|\hessEigval{2}{1}|}\right)}{\min\left\{\hessEigval{0}{1},2|\hessEigval{1}{1}|,2|\hessEigval{2}{1}|\right\}\left(\displaystyle\frac{\sqrt{|\hessEigval{1}{1}|}}{\Phi\left(\sqrt{|\hessEigval{1}{1}|}\epsLimit{1}\right)}+\displaystyle\frac{\sqrt{|\hessEigval{2}{1}|}}{\Phi\left(\sqrt{|\hessEigval{2}{1}|}\epsLimit{2}\right)}\right)}$},
\end{equation}
where we recall that~$\mu(\theta)$ is the principal Dirichlet eigenvalue of the one-dimensional Dirichlet harmonic oscillator~$\frac12(-\partial_x^2+x^2)$ on~$(-\infty,\theta)$.

\paragraph{Numerical results.}
We approximate the generator~$\cL_\beta$ using the same procedure as for the effective generator in Section~\ref{subsec:num_coarse_graining}.
In Figure~\ref{fig:ek}, we illustrate the validity of the modified Eyring--Kramers formula. The~$\alpha$-dependent prefactor correctly predicts fine effects of the boundary geometry near the saddle points. The asymptotic regime is reached for relatively small values of~$\beta$.
In Figure~\ref{fig:harm}, we illustrate the harmonic approximation of Theorem~\ref{thm:harmonic}. Eigenvalues appear to converge to the prediction of the harmonic approximation in the limit~$\beta\to\infty$. For domains in which the second harmonic eigenvalue corresponds to a local model around an index-1 saddle point (the blue and green domains in Figure~\ref{fig:pot}), this convergence appears to occur faster, though we have no explanation for why this should be the case.

\begin{figure}
    \centering
    \begin{subfigure}{0.9\textwidth}
        \includegraphics[width=0.9\textwidth]{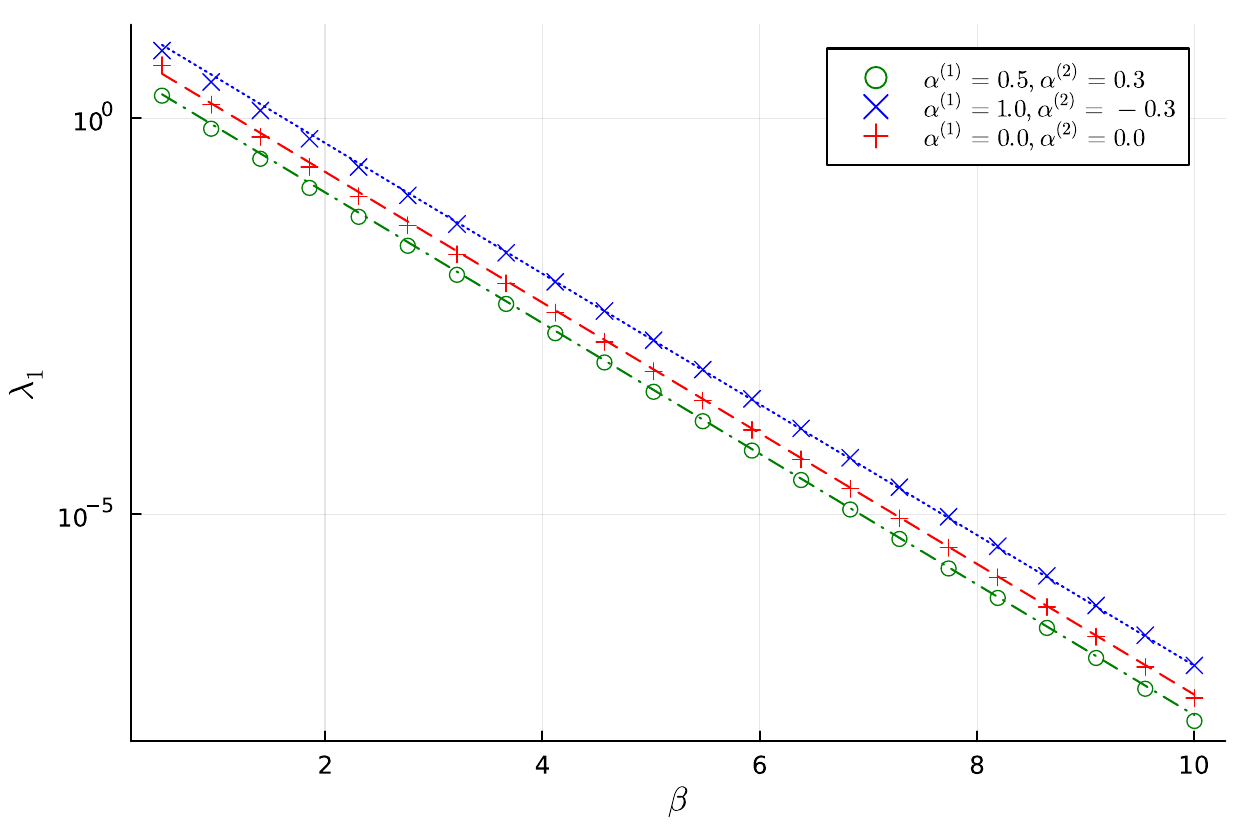}
        \caption{Principal Dirichlet eigenvalue of~$\cL_\beta$ on~$\Omega_{\alpha,\beta}$, for various values of the shape parameter~$\alpha$. The theoretical leading-order asymptotic of Theorem~\ref{thm:eyring_kramers} is represented with a dotted line.}
        \label{fig:ek}
    \end{subfigure}
    \hfill
\begin{subfigure}{0.9\textwidth}
    \includegraphics[width=0.9\textwidth]{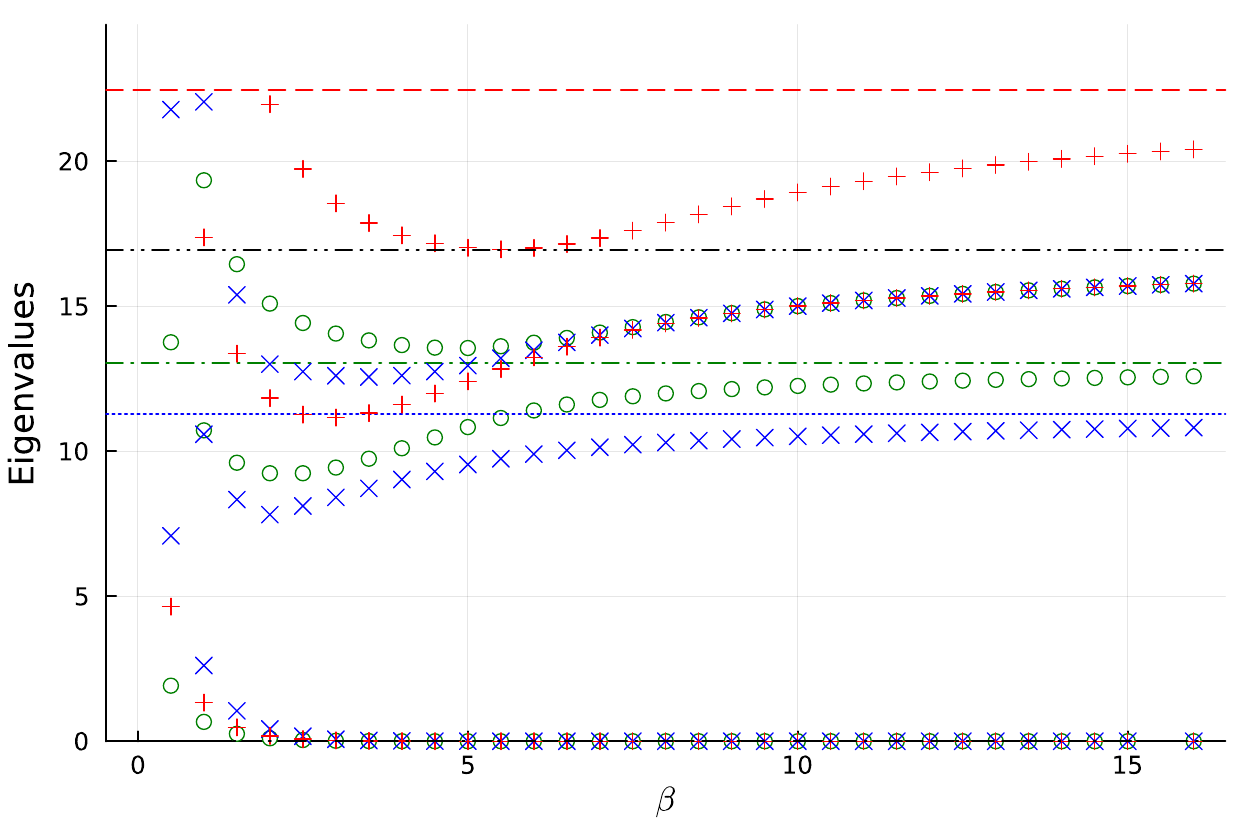}
    \caption{First three Dirichlet eigenvalues of~$-\cL_\beta$ on~$\Omega_{\alpha,\beta}$, for the three values of~$\alpha$ from Figure~\ref{fig:ek}. Horizontal lines correspond to the theoretical limiting values from Theorem~\ref{thm:harmonic}. The black line (--$\cdot\cdot$--) corresponds to a harmonic eigenvalue shared between all the domains. Missing values failed to converge. We observe convergence to the limiting regime, with eigenvalues corresponding to a lower asymptotic value appearing to converge faster.}
    \label{fig:harm}

\end{subfigure}
\caption{Numerical validation of the low-temperature asymptotics of Theorems~\ref{thm:eyring_kramers} and~\ref{thm:harmonic} from~\cite{BLS24}, for the one-dimensional potential depicted in Figure~\ref{fig:pot}.}
\label{fig:semiclassic}
\end{figure}
In Figure~\ref{fig:opt_1D}, we compare two quantities, for a fixed value of~$\beta=10$ (see Figure~\ref{fig:pot} for examples of corresponding domains): the low-temperature approximation~$J_\infty$ to the shape-optimization landscape defined in~\eqref{eq:lim_obj_explicit}, and the actual optimization landscape obtained by numerically approximating the reduced objective~\eqref{eq:reduced_obj}.
The low-temperature approximation and the true objective agree, making the low-temperature approximation an acceptable surrogate objective in the low-temperature regime, which can be maximized at a much smaller computational cost.
\begin{figure}
    \centering
    \begin{subfigure}{0.9\textwidth}
    \includegraphics[width=0.9\textwidth]{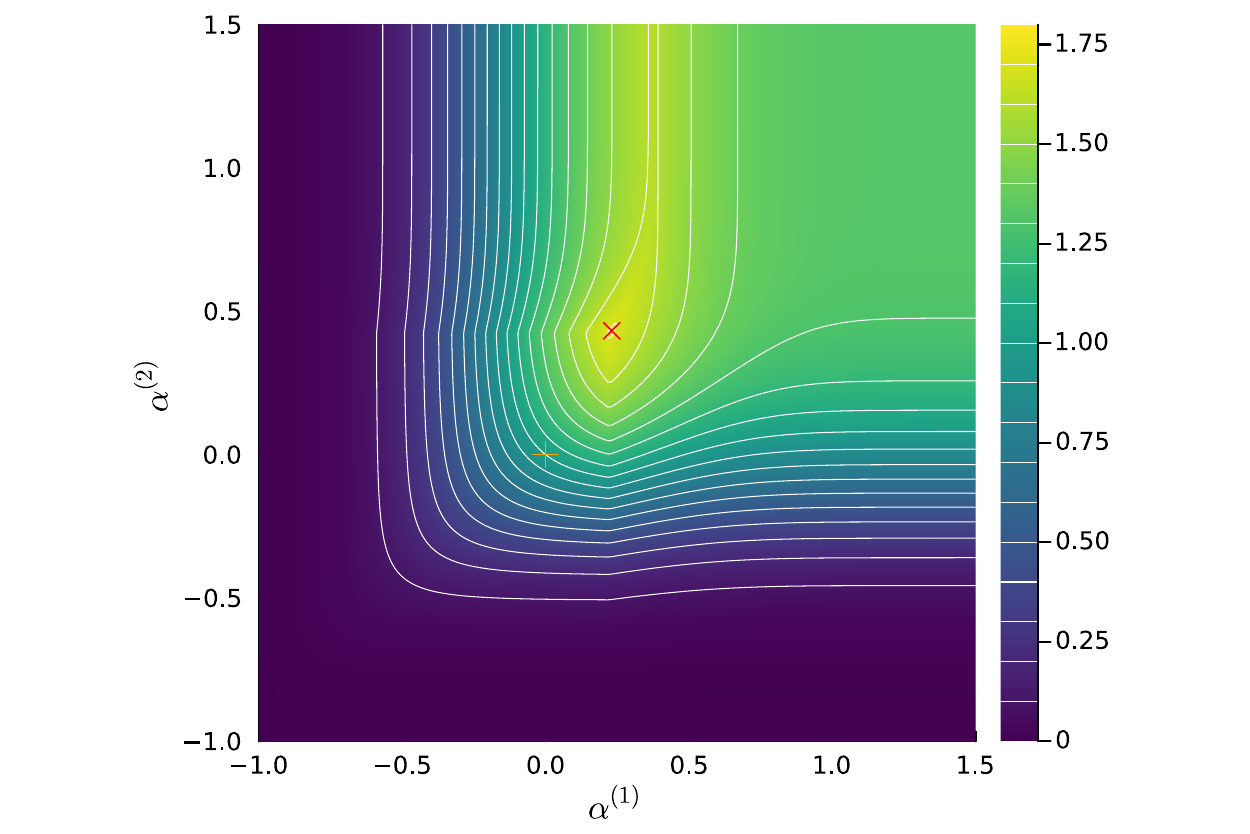}
        \caption{Semiclassical approximation of the shape optimization landscape. The limiting objective~$J_\infty(\alpha)$ defined in~\eqref{eq:lim_obj} is plotted for the potential~$V$ defined in~\eqref{eq:one_d_pot} and depicted in Figure~\ref{fig:pot}. The optimal~$\alpha^\star_\infty$ is marked by~$\boldsymbol{\color{red}{\times}}$, and the basin of attraction~$\basin{z_0}$ is marked by~$\boldsymbol{\color{orange}{+}}$. Ridge-like features are discernible, and correspond to the loci of eigenvalue crossings for the harmonic approximation~$K_\alpha$ defined in~\eqref{eq:harmonic_approximation}. The optimal value is attained for~$\alpha^\star_{\infty}\approx(0.23116,0.43216)$ with~$J_\infty(\alpha^\star_\infty)\approx 1.71$.}
        \label{fig:shape_landscape_sc}
    \end{subfigure}
    \begin{subfigure}{0.9\textwidth}
        \includegraphics[width=0.9\textwidth]{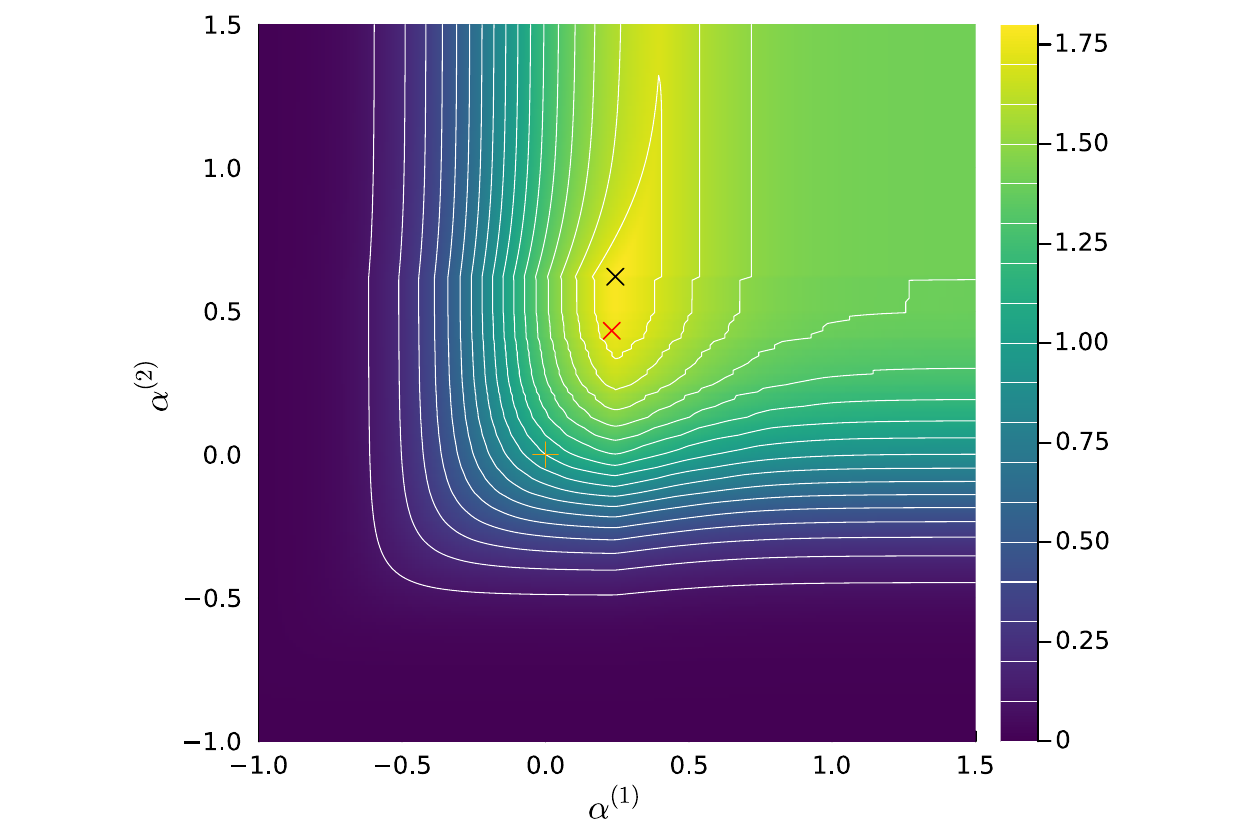}
            \caption{Shape-optimization landscape for the reduced objective~$J_\beta(\alpha)$ defined in~\eqref{eq:reduced_obj} for the value~$\beta=10$. The optimal shape~$\alpha^\star_\beta$ is marked by~$\boldsymbol{\times}$, the basin of attraction~$\basin{z_0}$ is marked by~$\boldsymbol{\color{orange}{+}}$ and the semiclassical prescription~$\alpha^\star_\infty$ is marked by~$\boldsymbol{\color{red}{\times}}$. The optimal value is attained for~$\alpha^\star_{\beta}\approx(0.24372,0.6206)$ with~$J_{\beta}(\alpha^\star_{\beta})\approx 1.81$. By comparison~$J_\beta(\alpha^\star_\infty)\approx 1.76$.}
            \label{fig:shape_landscape_10}
        \end{subfigure}
        \caption{Asymptotic approach to the shape optimization problem for the potential~\eqref{eq:one_d_pot} and the objective~\eqref{eq:reduced_obj}. At low temperature, the semiclassical approximation~(Figure~\ref{fig:shape_landscape_sc}) faithfully captures the features of the true optimization landscape~(Figure~\ref{fig:shape_landscape_10}). In particular, the semiclassical optimizer is close, both in argument and value of the objective function, to the true optimizer.}
    \label{fig:opt_1D}
\end{figure}

When optimizing functionals of eigenvalues, degenerate eigenvalues are commonly encountered at the optimal value of the design parameter, see~\cite[Theorems 8.4.11 and Theorem 8.4.14]{H06} or~\cite[Section 3.2.1]{LPRSS24} for examples of this phenomenon. Moreover, in view of the expression~\eqref{eq:harmonic_l2} for the second harmonic eigenvalue as the minimum of finitely many functions of~$\alpha$, one expects to see degenerate asymptotic eigenvalues for the optimizer of the asymptotic problem. Indeed, this is the case in our toy example, as seen in Figure~\ref{fig:shape_landscape_sc}, where the optimizer lies at the intersection of three ridges, corresponding to the loci of eigenvalue crossings for the harmonic approximation.
We question whether this degeneracy of eigenvalues at the optimum also holds at the level of the finite-temperature shape optimization problem. In Figure~\ref{fig:ev_degen}, we compute the second and third Dirichlet eigenvalues, for local optimizers of the reduced objective~\eqref{eq:reduced_obj}, and plot their relative difference.
We observe that~$\lambda_2$ is close to~$\lambda_3$ at the optimal domain, but not degenerate for finite values of~$\beta>0$. The relative difference converges to zero as~$\beta\to+\infty$, in accordance with the semiclassical prediction.

\begin{figure}
\centering
\includegraphics[width=0.75\textwidth]{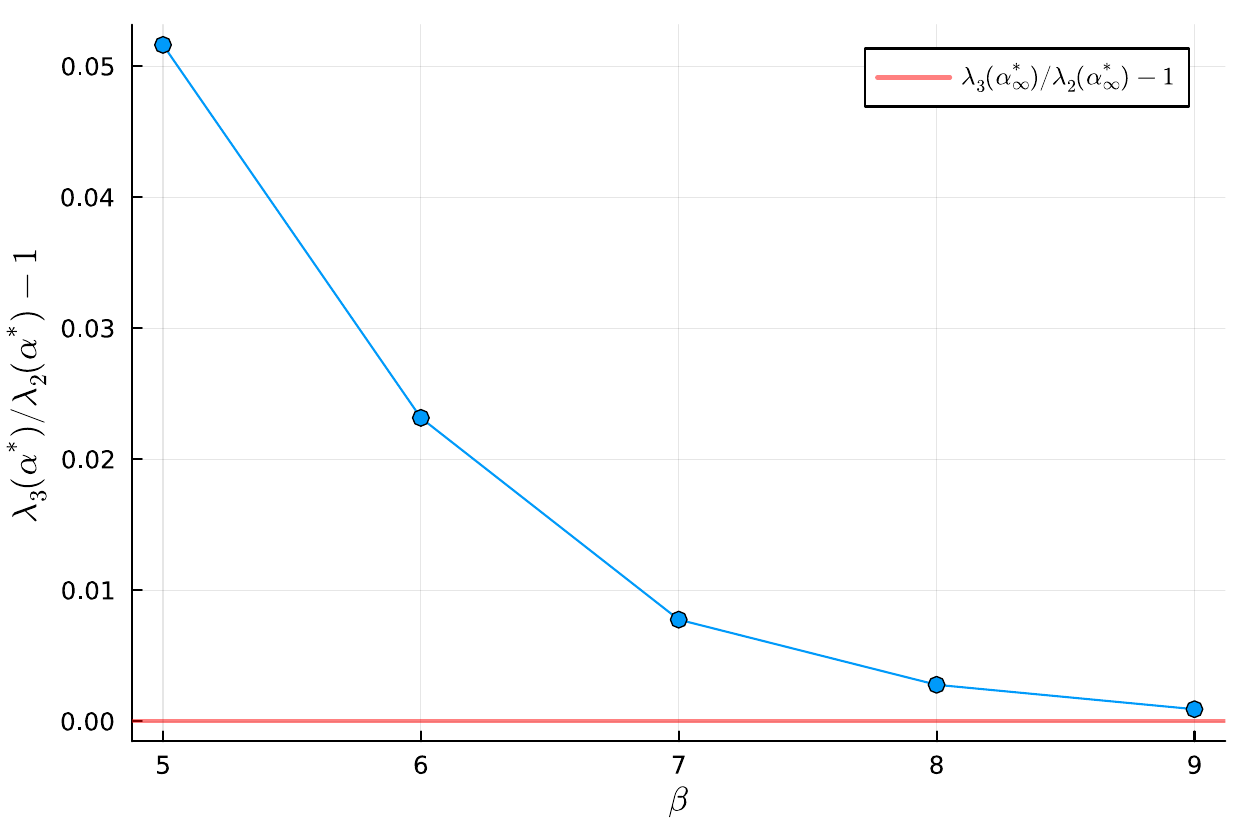}
\caption{Relative difference between the second and third Dirichlet eigenvalues at the optimum, for the reduced optimization problem defined by~\eqref{eq:reduced_obj}, as a function of the temperature parameter~$\beta$.}
\label{fig:ev_degen}
\end{figure}

\subsection{Application to a molecular system}
\label{subsec:diala}
In this section, we apply our shape-optimization method to the energy landscape of a small molecule commonly used to benchmark methods in MD, namely alanine dipeptide solvated in water. The system is composed of~$N=d/3=619$ atoms, in fact~$22$ atoms in the peptide chain and~$199$ water molecules. Atomic positions are restricted to a periodic cubic box of length~$L= 18.643\,\mathring{\mathrm{A}}$. As a collective variable, we use the dihedral angles (a standard choice, see~\cite{BDC00}),
$$\xi=(\phi,\psi).$$
The values and gradients of~$\phi$ and~$\psi$ are available through the Tinker-HP~\cite{Tinker18} interface to the Colvars library~\cite{FKH13}.

\paragraph{Simulation parameters.}
All simulation runs were performed using a modified version of Tinker-HP~\cite{Tinker18} allowing to simulate the Fleming--Viot process (see~\ref{alg:fv} below) inside an arbitrary domain defined in CV space.

Unless otherwise specified, simulations of the underdamped Langevin dynamics~\eqref{eq:underdamped_langevin} (with~$\Gamma=M$) were performed at~$T=300\,\mathrm{K}$ ($\beta=1.677\,\mathrm{mol}\cdot\mathrm{kcal}^{-1}$) and discretized with the BAOAB scheme, setting the time step to~$\Delta t = 2$ fs, using the AMBER-ff99 interaction potential, and the SHAKE method~\cite{RCB77} to fix the geometry of the solvent molecules.

Experiments were performed across a range of friction parameters,~$\gamma\in\{1,2,5,10\}\,\mathrm{ps}^{-1}$, to assess the effectiveness of the methodology in various dynamical settings.
Since our methodology requires a low-dimensional reversible diffusion~\eqref{eq:overdamped_langevin} as input, we use the effective dynamics~\eqref{eq:effective_dynamics} associated with the Kramers--Smoluchowski approximation~\eqref{eq:overdamped_langevin} of the underdamped Langevin dynamics (where~$a=M^{-1}$). In other words, the (rescaled by~$\gamma$) effective generator whose eigenvalues we optimize is given by
\begin{equation}
    \label{eq:eff_generator_num}
    \cL_\beta^\xi=\frac1{\gamma\beta}\e^{\beta F_\xi}\div\left(\e^{-\beta F_\xi}a_\xi\nabla\cdot\right),\qquad a_{\xi}(z) = \int_{\Sigma_z}\nabla\xi^\top M^{-1}\nabla\xi\,\d\mu_z.
\end{equation}

It has been observed in previous studies of realistic molecular systems (see for instance~\cite[Sections 4.2.2 and 4.3.2]{NN24}) that the dynamical rates inferred by the Kramers--Smoluchowski approximation often differ greatly from those associated with the underlying underdamped Langevin dynamics, even when accounting for rescaling by the friction parameter~$\gamma$. Therefore we shall not use our reduced model to directly infer timescales for the original dynamics, but merely as a proxy to define good metastable states.
The effectiveness of these states, in the sense of maximizing the separation of timescales, will therefore be assessed at the level of the original dynamics, and not of the reduced model.

\paragraph{Free energy landscape and effective diffusion.}
We first compute the free energy~$F_\xi$ and effective diffusion tensor~$a_\xi$ entering in the definition of the effective dynamics~\eqref{eq:effective_dynamics}. The free-energy landscape is represented in Figure~\ref{fig:pmf}, and was precomputed using a multiple-replica adaptive biasing force dynamics~(see~\cite{CGHLPC15}), with four replicas, for a total of~$t=600$ ps of simulation time.
\begin{figure}
    \center
    \includegraphics[width=0.75\linewidth]{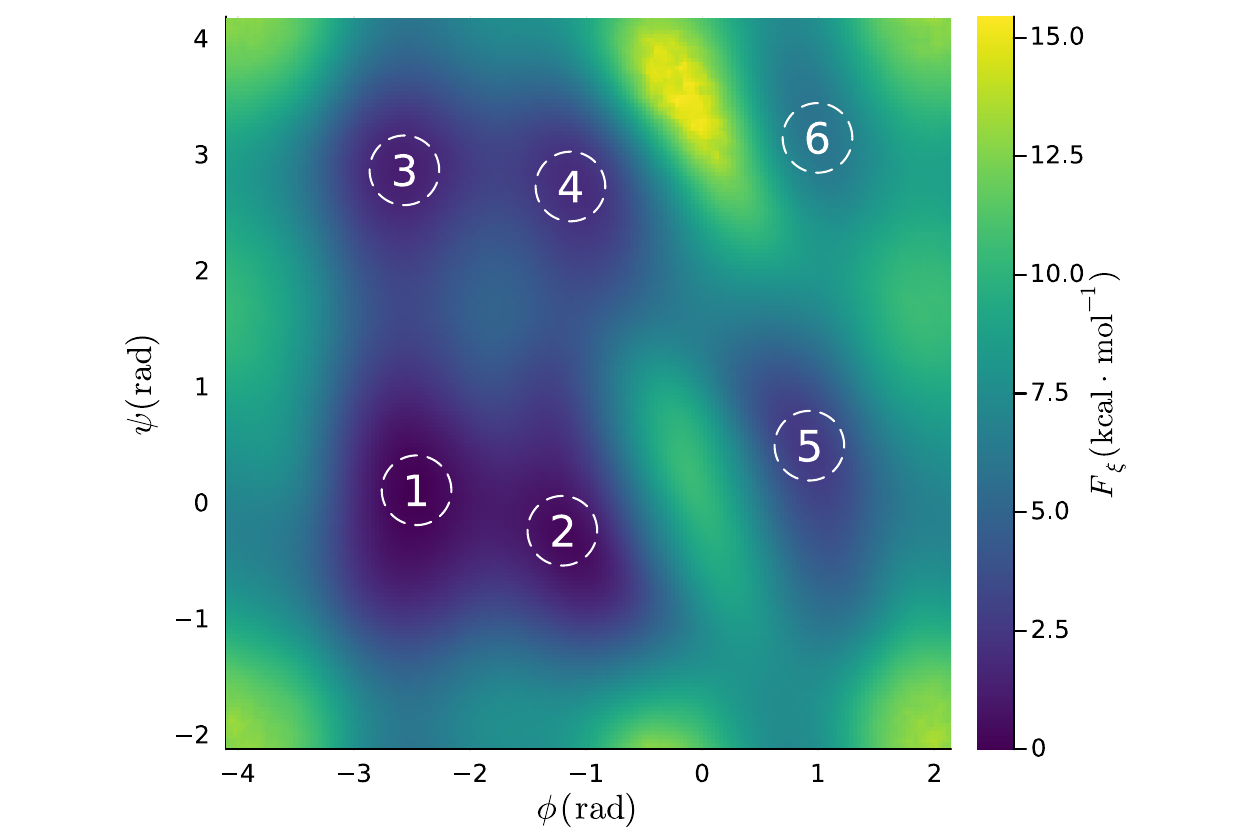}
    \caption{Free energy landscape in the dihedral angles~$(\phi,\psi)$. We identify and label six local minima.}
    \label{fig:pmf}
\end{figure}
The effective diffusion tensor was estimated using an importance sampling scheme using a family of harmonically biased potentials.
More precisely, the collective variable space~$(-\pi,\pi)^2$ was divided into a set~$\mathcal W$ of square-shaped windows of side-length~$\Delta\phi_{\mathcal W}=\Delta\psi_{\mathcal W}=\pi/36\,\mathrm{rad}$.
For each window~$w\in\mathcal W$, we performed a biased simulation of the underdamped Langevin dynamics~\eqref{eq:underdamped_langevin} using a harmonic biasing potential
\begin{equation}
    \label{eq:umbrella_potential}
   V^w = V+U^w,\qquad U^w(q) = \frac1{2\eta}|\xi(q)-z_w|^2,
\end{equation} 
where~$z_w$ is the center of the window~$w$, and~$\eta = 40\,\mathrm{mol}\cdot\mathrm{kcal}^{-1}$ is the inverse force constant.
For~$z\in(-\pi,\pi)^2$, we use the estimator
\begin{equation}
    \label{eq:approx_diffusion_tensor}
\widehat{a}_\xi(z) = \sum_{w\in\mathcal W}\rho_w(z)\frac{\displaystyle\sum_{k=1}^{N_{\mathrm{sim}}}\nabla\xi(X_k^w)^\top M^{-1}\nabla\xi(X_k^w)\e^{\beta U^w(X_k^w)}\1_{|\xi(X_k^w)-z|_{\infty}<h/2}}{\displaystyle\sum_{k=0}^{N_{\mathrm{sim}}}\e^{\beta U^w(X_k^w)}\1_{|\xi(X_k^w)-z|_\infty <h/2}},
\end{equation}
where~$(X_k^w)_{k=1,\dots,N_{\mathrm{sim}}}$ are sample points of the numerical trajectory for the biased dynamics in the window~$w\in\mathcal W$,~$h=\pi/90\,\mathrm{rad}$ is the histogram resolution and~$\rho_w$ is a weighting function chosen so that~$\sum_{w\in\mathcal W}\rho_w(z)=1$ for all~$z$. For simplicity, we chose~$\rho_w(z)$ to give uniform weight to each window for which the ratio in~\eqref{eq:approx_diffusion_tensor} was well-defined.

The initial condition~$X_0^w$ was prepared by running a harmonically steered-MD simulation from a reference configuration toward the value~$\xi=z_w$, followed by a~$5\,\mathrm{ps}$ equilibration run, both with a value of the friction parameter~$\gamma=1\,\mathrm{ps}^{-1}$. The values of the CV, biasing energy and instantaneous tensor~$\nabla\xi^\top M^{-1}\nabla\xi$ were recorded every~$10\,\mathrm{fs}$. The overall computation can be straightforwardly parallelized, as the estimators within each window are independent of one another.
The results of the computation of the effective diffusion tensor are shown in Figure~\ref{fig:diffusion_tensor_diala}.
\begin{figure}
    \begin{subfigure}{\linewidth}
        \includegraphics[width=0.49\linewidth]{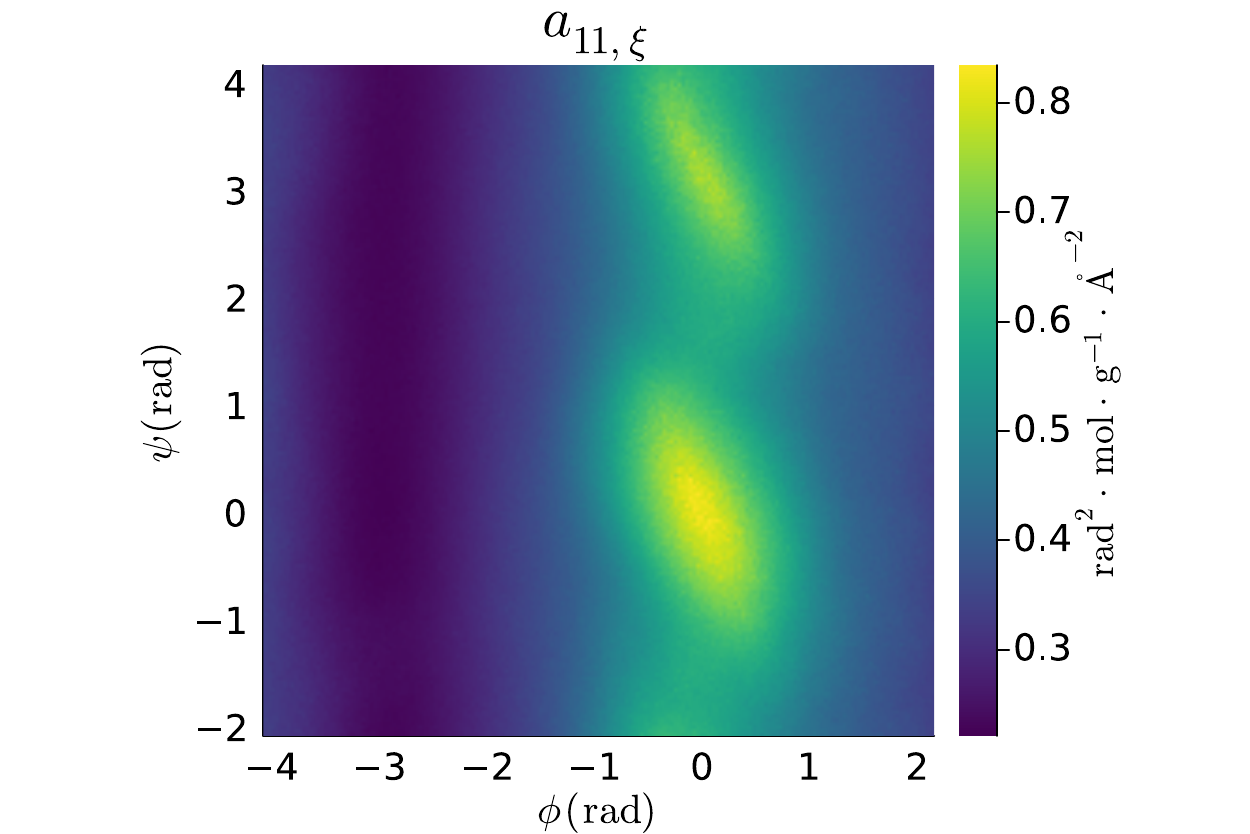}
        \includegraphics[width=0.49\linewidth]{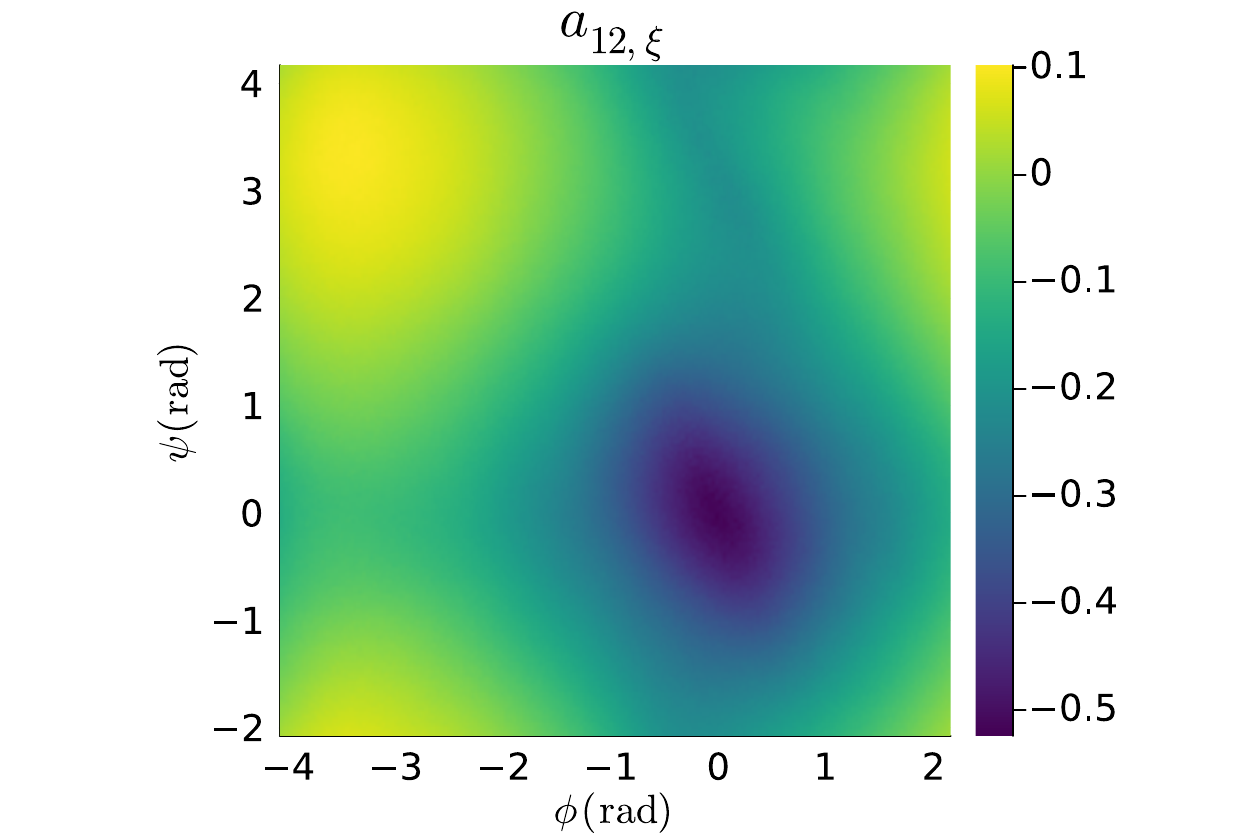}
        \includegraphics[width=0.49\linewidth]{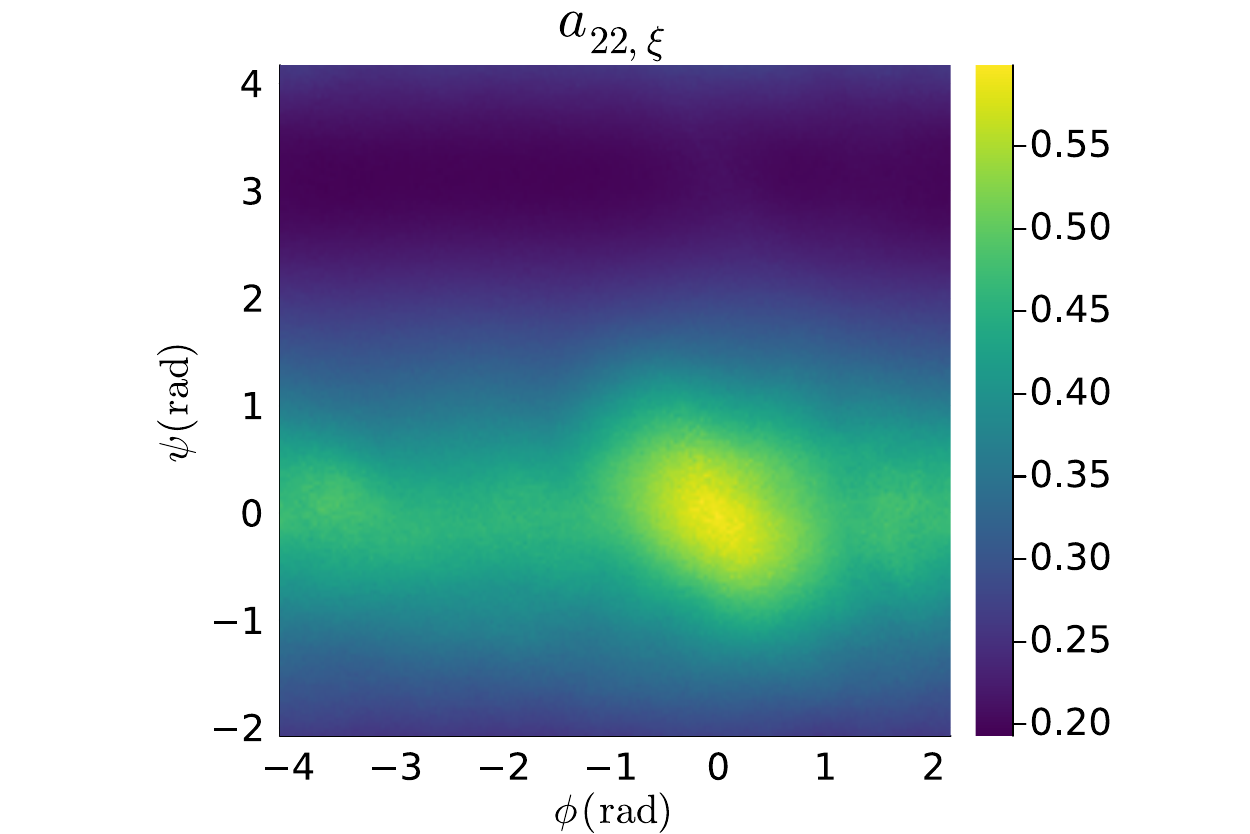}
        \includegraphics[width=0.49\linewidth]{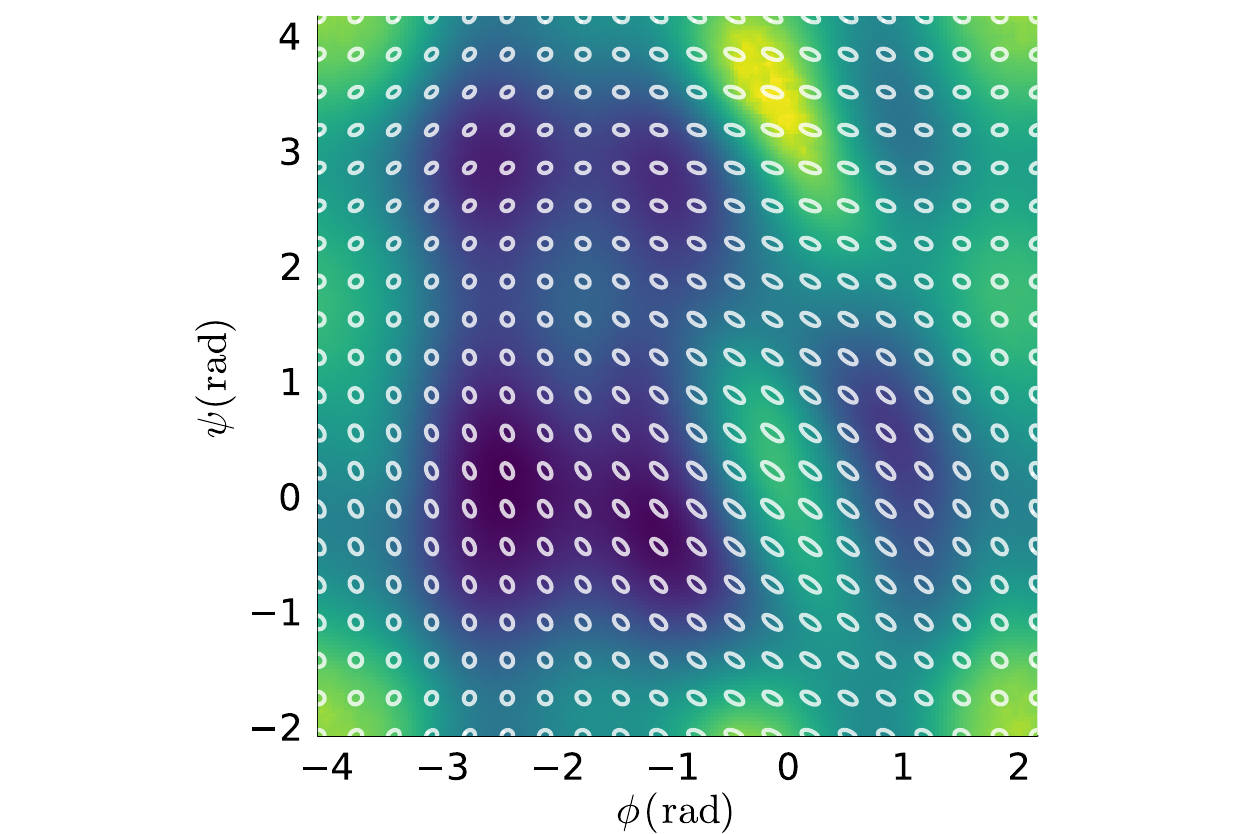}
    \end{subfigure}
    \caption{Components of the effective diffusion tensor~$a_\xi$ (top row and left-bottom row), and corresponding ellipsoid glyph representation (right-bottom row).}
    \label{fig:diffusion_tensor_diala}
\end{figure}

\paragraph{Shape optimization of eigenvalues for the effective dynamics.}
We apply Algorithm~\ref{alg:ascent} to obtain optimized domains in the two-dimensional space of dihedral angles~$(\phi,\psi)$, using the thermodynamic quantities computed in the previous paragraph and~Corollary~\ref{prop:directional_derivative_xi} for shape-variation formulas.
Algorithm~\ref{alg:ascent} was implemented in FreeFem++. Its code is available in the paper repository~\cite{github}.

The algorithm was run six times, each time initialized with~$\Omega_0 = B\left((\phi_0,\psi_0),0.3\right)$ in CV space, where~$(\phi_0,\psi_0)$ ranged across the six free-energy local minima displayed in Figure~\ref{fig:pmf}. All optimization runs were performed with the parameters~$\varepsilon_\reg=\sqrt{0.1}$,~$\varepsilon_{\mathrm{degen}}=0.01$,~$m_{\max}=2$,~$\eta_{\max}=0.004$,~$\alpha=0.8$,~$\varepsilon_{\mathrm{term}}=0.005$,~$M_{\mathrm{grad}}=2$ and~$N_{\mathrm{search}}=1000$, except for the optimization of state 2, for which a value~$\eta_{\max}=0.001$ was necessary to achieve convergence.
The mesh adaptation procedure~$\mathcal A$ from step C. of Algorithm~\ref{alg:ascent} enforced a maximal cell width of~$h_{\max}=0.03$ throughout the runs.

The initial domains are plotted alongside the corresponding numerically optimized domains in Figure~\ref{fig:opt_domains_diala}, together with the associated QSDs for the effective dynamics~\eqref{eq:effective_dynamics} in Figure~\ref{fig:eff_qsds}.
\begin{figure}
    \begin{subfigure}{0.5\linewidth}
        \center
            \includegraphics[width=\linewidth]{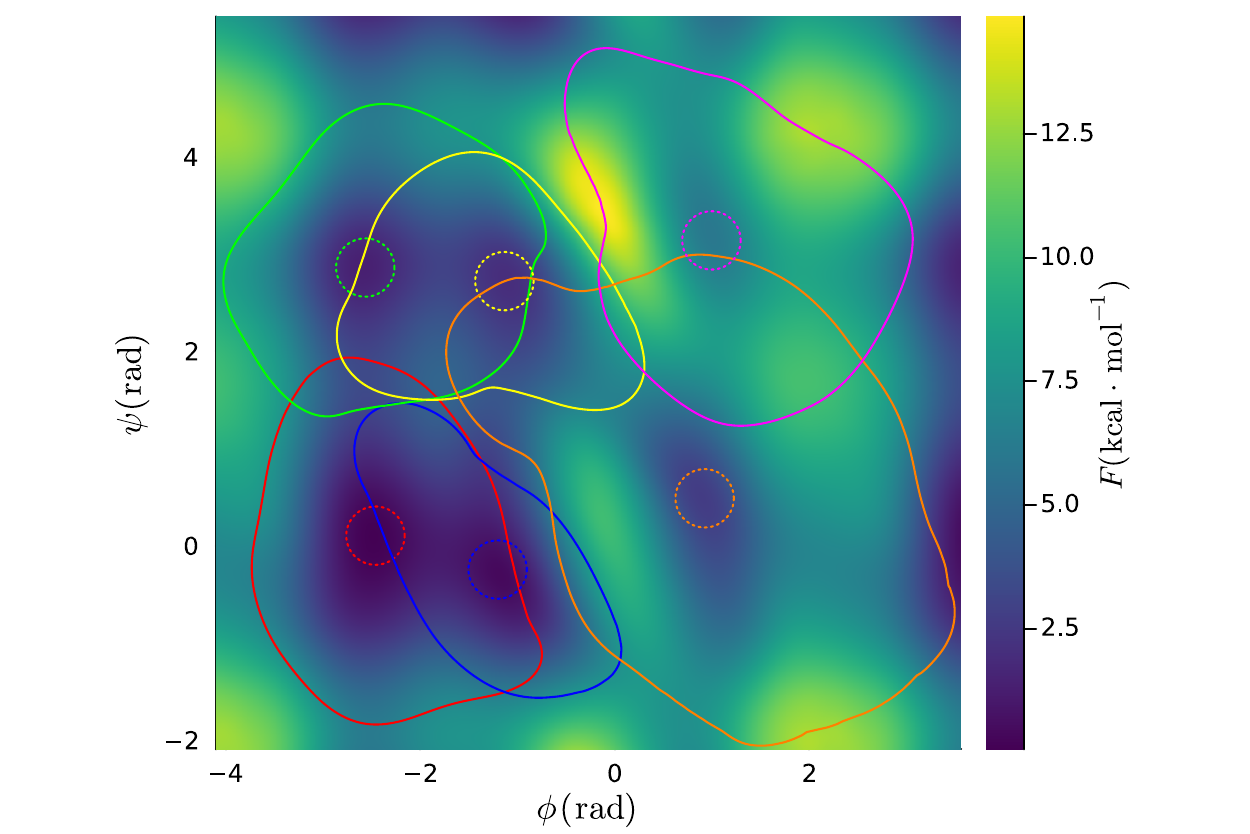}
            \caption{Numerically optimized metastable states for alanine dipeptide using Algorithm~\ref{alg:ascent}.}
            \label{fig:opt_domains_diala}
        \end{subfigure}  
        \begin{subfigure}{0.5\linewidth}  
            \includegraphics[width=\linewidth]{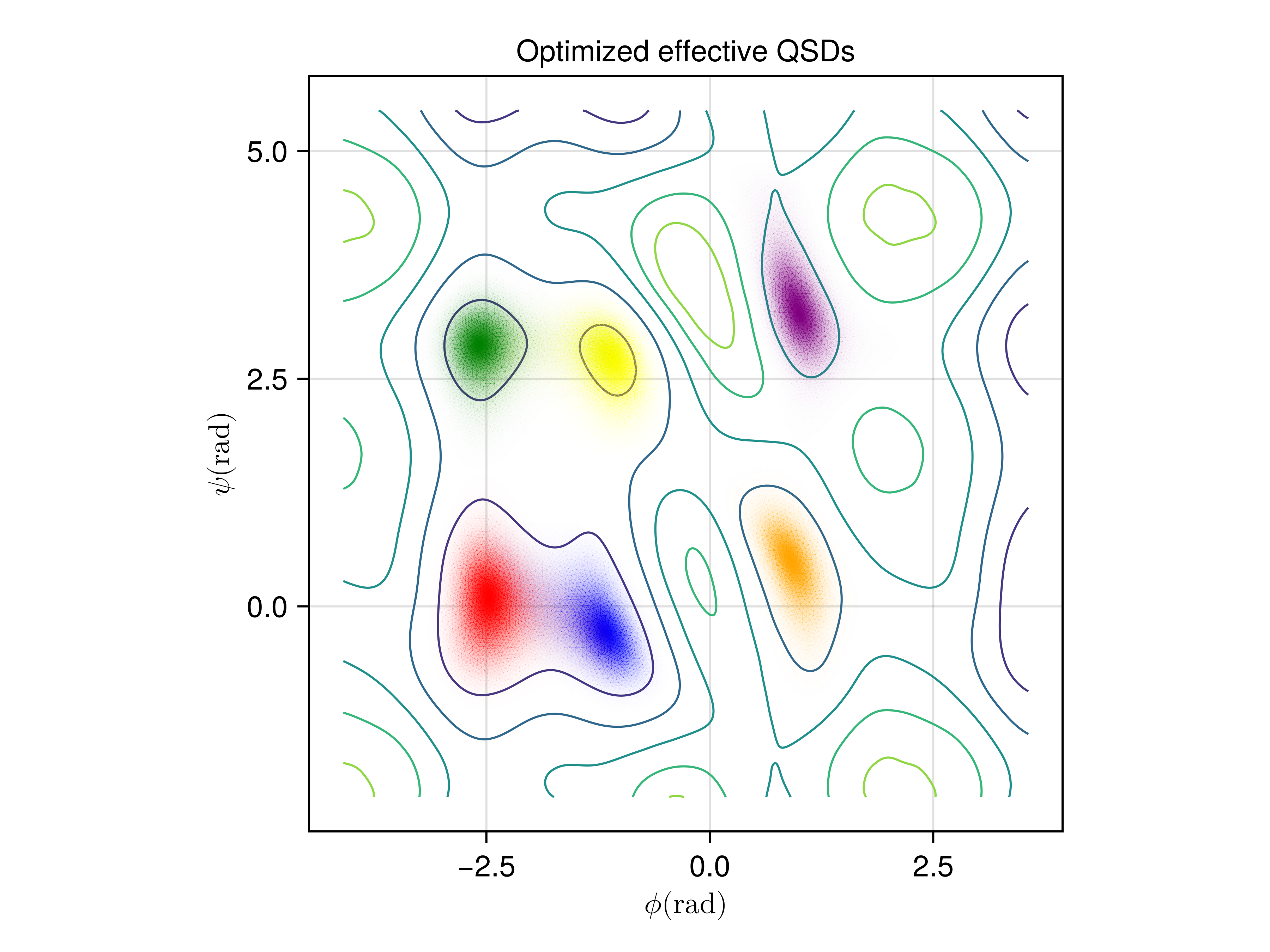}
            \caption{QSDs for the effective dynamics~\eqref{eq:effective_dynamics}. }
            \label{fig:eff_qsds}
        \end{subfigure}
        \caption{In Figure~\ref{fig:opt_domains_diala}, solid lines correspond to the boundaries of the optimized domains, with corresponding initial domains in dotted lines. In Figure~\ref{fig:eff_qsds}, higher densities map to lower transparency values, with the same color-coding as in Figure~\ref{fig:opt_domains}. QSDs have been normalized in~$L^\infty$. In both figures, the free-energy landscape from Figure~\ref{fig:pmf} is plotted for reference.}
        \label{fig:opt_domains}
\end{figure}

In Figure~\ref{fig:obj_conv_diala}, we plot the evolution of the effective separation of timescales during the optimization process, for the six runs of Algorithm~\ref{alg:ascent}. State 5 is the most locally metastable state for the effective diffusion, with an effective separation of timescales of nearly~$500$. 
\begin{figure}
    \center
    \includegraphics[width=1\linewidth]{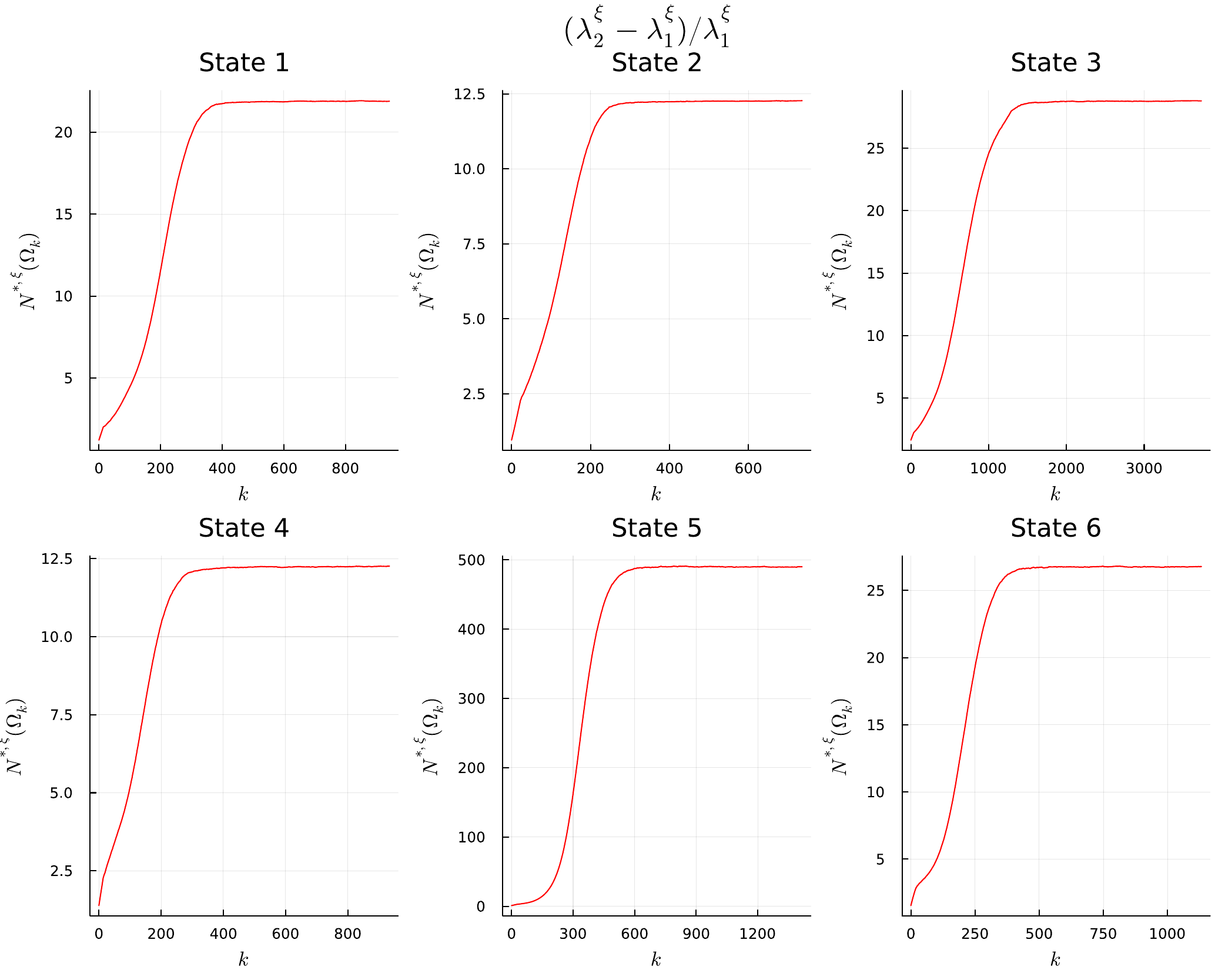}
    \caption{Effective separation of timescales throughout six runs of Algorithm~\ref{alg:ascent}, initialized with coresets around the six free-energy minima depicted in Figure~\ref{fig:pmf}.}
    \label{fig:obj_conv_diala}
\end{figure}

In Figure~\ref{fig:lambdas_diala}, we plot the evolution of the first four Dirichlet eigenvalues of the effective generator during the six runs of Algorithm~\ref{alg:ascent}, showcasing frequent eigenvalue crossings. In all cases, we observe that the second and third Dirichlet eigenvalues coalesce during an early phase of the optimization process, which suggests that encountering degenerate eigenvalues is the rule rather than the exception. For the purpose of fixing a timescale in Figure~\ref{fig:lambdas_diala}, we (somewhat arbitrarily) set~$\gamma=5\,\mathrm{ps}^{-1}$ in the definition of the effective generator~\eqref{eq:eff_generator_num}.
\begin{figure}
    \center
    \includegraphics[width=0.95\linewidth]{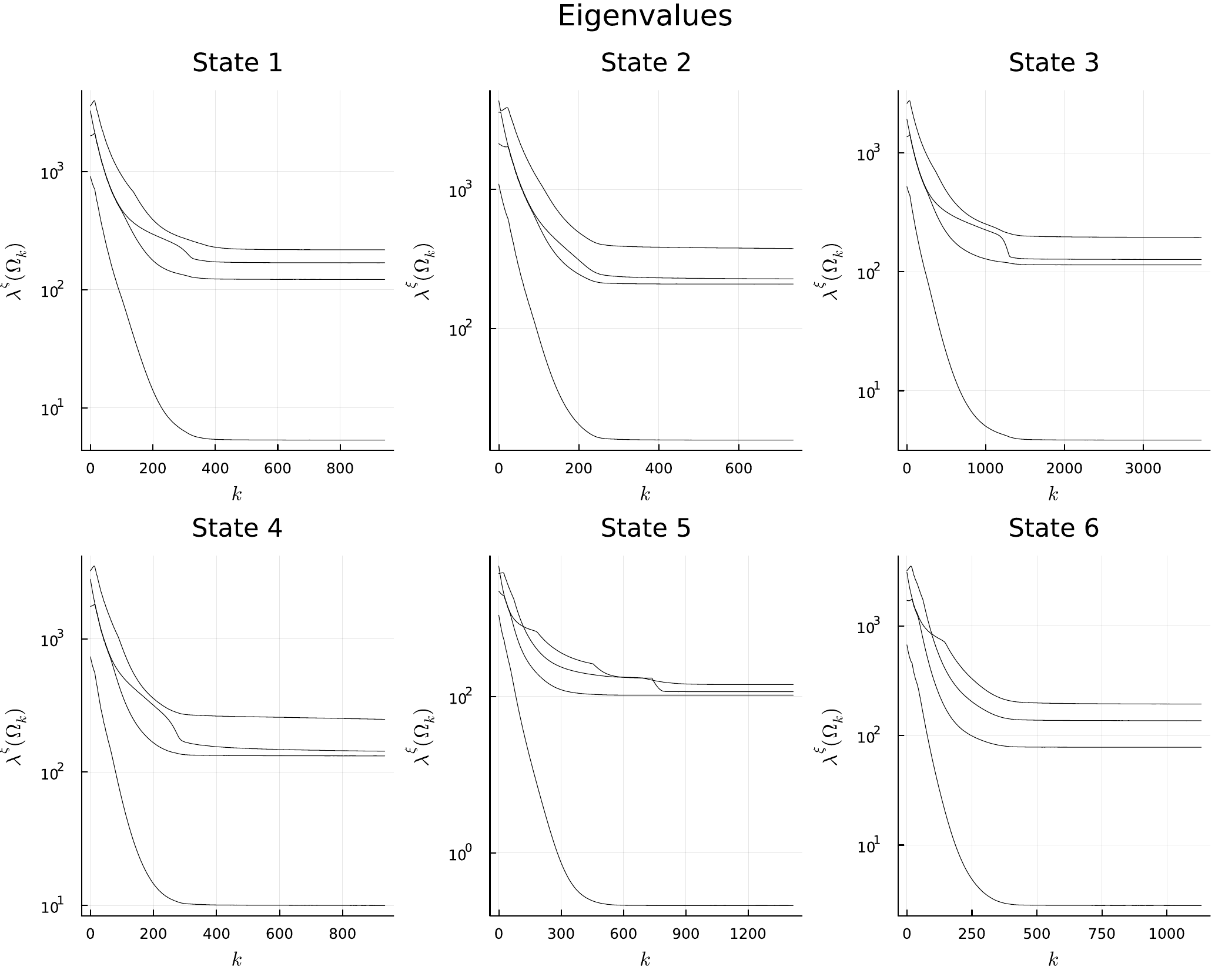}
    \caption{Behavior of the four smallest Dirichlet eigenvalues throughout the six runs of Algorithm~\ref{alg:ascent} depicted in Figure~\ref{fig:obj_conv_diala}, displaying frequent eigenvalue coalescence and crossings.}
    \label{fig:lambdas_diala}
\end{figure}

In Figure~\ref{fig:oscillations}, we illustrate the usefulness of the choice of ascent direction in Algorithm~\ref{alg:ascent} using the numerical degeneracy parameter~$\varepsilon_{\mathrm{degen}}$. Omitting the numerical degeneracy parameter and trusting the non-degenerate shape gradients may lead to oscillations in the objective function, due to non-differentiable features of the objective landscape near points of near-degeneracy. The algorithm adapting the choice of ascent direction in the case of approximately degenerate eigenvalues leads to a significant improvement in the speed of increase of the objective function, and successfully suppresses eigenvalue crossings and oscillations in the value of the objective.
\begin{figure}
    \begin{subfigure}{0.49\linewidth}
        \center
            \includegraphics[width=\linewidth]{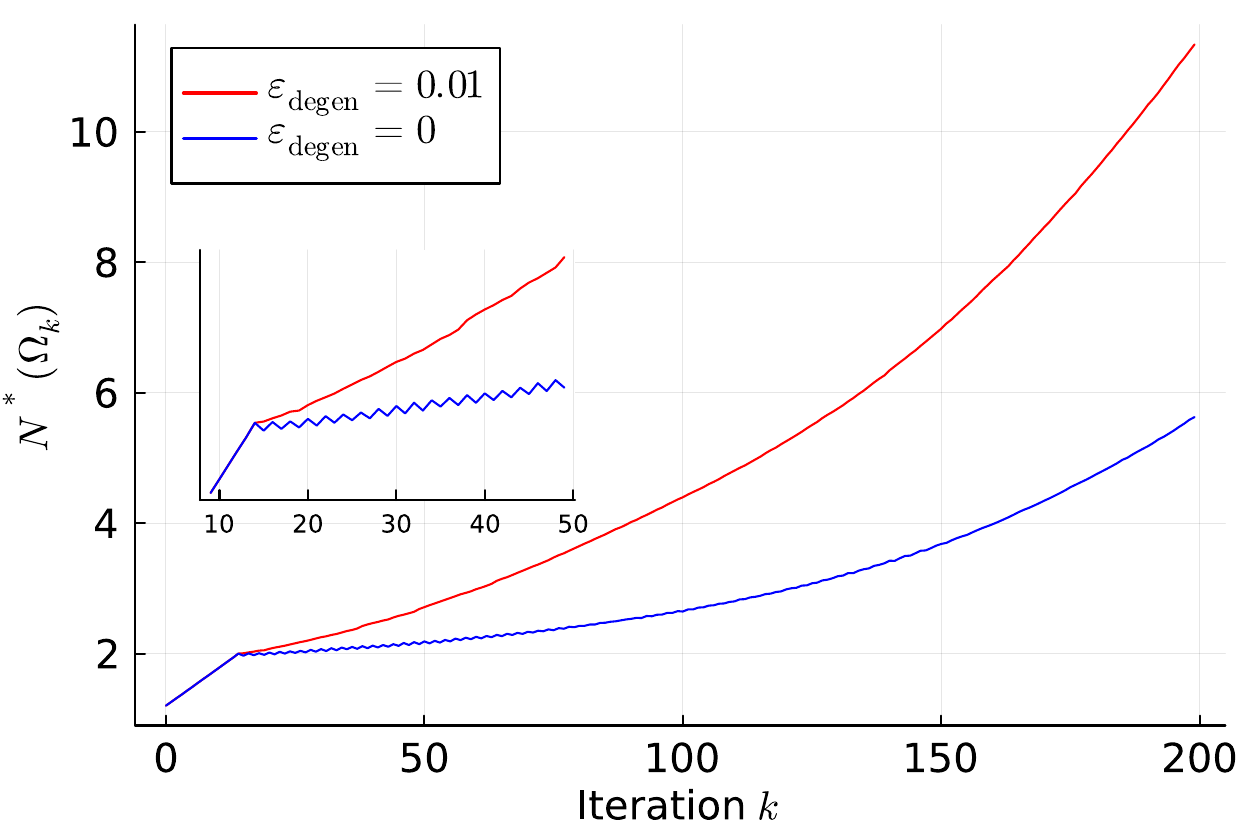}
            \caption{Value of the objective function versus number of iterations, with zoom on iterations~$10$--$50$.}
            \label{fig:obj_mix}
        \end{subfigure}
        \begin{subfigure}{0.49\linewidth}
        \center
            \includegraphics[width=\linewidth]{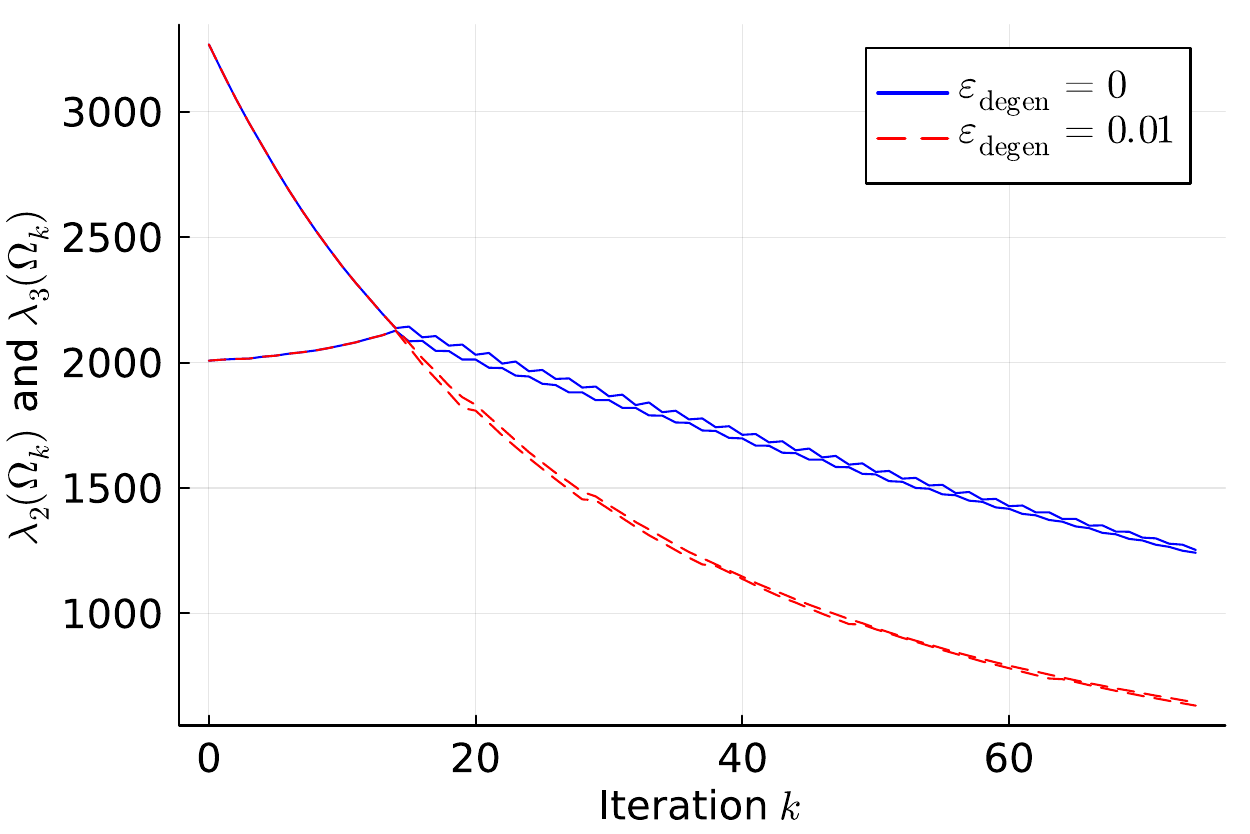}
            \caption{Second and third Dirichlet eigenvalue versus number of iterations.}
            \label{fig:oscillations_mix}
        \end{subfigure}

    \caption{Effect of the numerical degeneracy parameter~$\varepsilon_{\mathrm{degen}}$ in Algorithm~\ref{alg:ascent}, during the optimization of state 1 (see Figure~\ref{fig:pmf}). Setting~$\varepsilon_{\mathrm{degen}}>0$ ensures local ascent of the objective function, leading to an overall improvement in the convergence behavior (Figure~\ref{fig:obj_mix}), and effectively suppresses eigenvalue crossings (Figure~\ref{fig:oscillations_mix}).}
    \label{fig:oscillations}
\end{figure}

\paragraph{Parametrization of the states.}
The boundary vertices~$(\phi_i,\psi_i)_{1\leq i\leq N_{\mathrm{V}}}$ of the optimized mesh for state 5 were transformed from~$(\phi,\psi)$-space into~$(r_i,\theta_i)_{1\leq i\leq N_{\mathrm{V}}}$ for a system~$(r,\theta)$ of polar coordinates centered at the free-energy minimum inside state 5. A finite Fourier series
\begin{equation}
    R(t):=\sum_{k=0}^{N_{\mathrm{modes}}}\left[a_k \cos(kt) + b_k\sin(kt)\right]
\end{equation}
was then fitted to these points via ordinary least squares with~$N_{\mathrm{modes}}=20$, and the final definition we took for the optimized state was
\begin{equation}
    \Omega =\left\{(r,\theta):r<R(\theta)\right\}.
\end{equation}
This definition assumes that the domain is star-shaped around the minimum, which is indeed the case here.
The boundary of the free-energy basin, which was computed using finite-difference gradient descent on the estimated free-energy (see Figure~\ref{fig:pmf}), was similarly fitted with a Fourier series.

\paragraph{The Fleming--Viot process.}
To quantify the performance gained from using optimized definitions of metastable states at the level of the original high-dimensional dynamics, we must quantify the separation of timescale directly.
We focus on state 5, the most locally metastable domain according to the effective dynamics, and argue numerically that the state optimized with the surrogate coarse-grained objective leads to a significant improvement in the separation of timescales, when compared with a reference domain given by the basin of attraction of the local minimum in state~5, for a steepest-descent dynamics on the free-energy landscape.

We achieve this by using a Fleming--Viot process~(see for instance~\cite{DMM00}), which allows to infer both metastable timescales of interest, namely the exit rate starting from the QSD and the convergence rate to the QSD. 

\begin{algorithm}[Discrete-time Fleming--Viot process]
    \label{alg:fv}
    One step of the Fleming--Viot process with hard-killing, given a domain~$\Omega\subset\R^d$, a stride length~$\Delta t_{\mathrm{FV}}$, and given a number~$N_{\mathrm{proc}}\geq 1$ of replicas in state $(X_{t_0}^{(i)})_{1\leq i\leq N_{\mathrm{proc}}}$, consists in iterating the following procedure from $k=0$.
    \begin{enumerate}[A.]
        \item{At step $k$: evolve each replica with independent Brownian motions for a physical time~$\Delta t_{\mathrm{FV}}$ using a discretization of the underdamped Langevin dynamics.}
        \item{For any~$1\leq i\leq N_{\mathrm{proc}}$, if~$X_{t_0 + k\Delta t_{\mathrm{FV}}}^{(i)}\not\in\Omega$, kill this replica, and branch it in the next step from the state of a replica chosen uniformly at random among the survivors~(that is, the set of replicas which did not exit~$\Omega$ in step $k$).}
        \item{Set~$k\leftarrow k+1$ and proceed from step A.}
    \end{enumerate}
\end{algorithm}
This algorithm corresponds to the Fleming--Viot process with hard-killing for the discrete-time Markov chain obtained by subsampling the numerical trajectories in time at integer multiples of~$\Delta t_{\mathrm{FV}}$.
Algorithm~\ref{alg:fv} should be understood as a particle approximation of the dynamics conditioned on remaining inside~$\Omega$, in the sense that the empirical distribution of replicas at time~$t$ converges to the conditional distribution~$\mu_{t,X_0}$ (recall~\eqref{eq:yaglom_limit}) as~$N_{\mathrm{proc}}\to \infty$, see~\cite[Theorem 2.2]{V14}. This convergence can in some cases be controlled uniformly in time, see~\cite[Theorem 3.1]{R06} for an early approach, and~\cite[Theorem 2]{JM22} for a recent result in the overdamped case.
In particular, the empirical stationary distribution of the Fleming--Viot process approaches the QSD as~$N\to\infty$.

The time evolution of a single particle from the Fleming--Viot process (say~$X^{(1)}$), resembles a~$\nu$-return process (where~$\nu$ is the QSD): it evolves according to the dynamics until it reaches the boundary of the state, and is then instantly resurrected according to the empirical distribution of the Fleming--Viot process whose invariant measure approximates~$\nu^{\otimes N_{\mathrm{proc}}}$.
This approximation underpins the estimation of the exit rate from~$\Omega$ starting from~$\nu$, and also step C. of Algorithm~\ref{alg:parrep}.

For each value of the friction parameter~$\gamma$ and the two definitions of the state, we sample~$N_{\gamma}$ independent trajectories of the Fleming--Viot process (starting from a random initial condition~$X_0$ which we make precise below), lasting~$t_{\mathrm{sim}}=60\,\mathrm{ps}$ in total. The first~$t_{\mathrm{eq}}=30\,\mathrm{ps}$ were used to probe the decorrelation behavior to the QSD, and the last~$30\,\mathrm{ps}$ were used to sample the QSD (or an approximation thereof), and stationary exit events.

\paragraph{Estimation of the exit rate.}
To estimate the metastable exit rate~$\lambda_1(\Omega)$, we compute the empirical stationary exit rate for the Fleming--Viot process by counting the number~$N_{\mathrm{exit},\Omega}(t)$ of branching events recorded after time~$t$.
The exit rate is estimated (for each value of~$\gamma$ and definition of the state) as
\begin{equation}
    \label{eq:empirical_exit_rate}
    \widehat{\lambda}_1(\Omega) = \frac{N_{\mathrm{exit},\Omega}(t_{\mathrm{sim}})-N_{\mathrm{exit},\Omega}(t_{\mathrm{eq}})}{N_{\mathrm{proc}}\left(t_{\mathrm{sim}}-t_{\mathrm{eq}}\right)}.
\end{equation}
Under the approximation that the stationary Fleming--Viot process is a collection of~$N_{\mathrm{proc}}$ independent~$\nu$-return processes, the counting process~$N_{\mathrm{exit},\Omega}$ is a Poisson process with rate measure~$\lambda_1(\Omega)N_{\mathrm{proc}}\,\d t$, which motivates the choice of estimator~\eqref{eq:empirical_exit_rate}.
Confidence intervals for this exit rate were constructed from the independent realizations of the Fleming--Viot process.

\paragraph{Estimation of the convergence rate to the QSD.}
We assess the convergence of the conditional measure~$\mu_{t,X_0}$ (see~\eqref{eq:yaglom_limit}) to the QSD~$\nu$ at the level of convergence in total variation distance for their~$\phi$ and~$\psi$ marginals (and not their~$(\phi,\psi)$-marginals, due to data scarcity). This choice is motivated by the assumption that the CVs~$\phi$ and~$\psi$ correspond to the ``slow'' variables in the system, meaning that other degrees of freedom should have relaxed to their quasi-stationary state by the time~$\phi$ and~$\psi$ do. We first approximate the conditional law~$\mu_{t,X_0}$ and the QSD~$\nu$ with empirical approximations $\widehat{\mu}_{t,X_0}$ and~$\widehat{\nu}$.
The approximations~$\widehat{\nu}$ (or rather their~$\phi$ and~$\psi$ marginal histograms) were constructed by aggregating samples of the CV values recorded over all realizations of the Fleming--Viot process and the last~$60\,\mathrm{ps}$ of their trajectories.
The approximations~$\widehat{\mu}_{k\Delta t_{\mathrm{hist}}}$ (rather, their histograms) were constructed for~$k\geq 1$ at regular time intervals of length~$\Delta t_{\mathrm{hist}}=0.2\,\mathrm{ps}$ by aggregating samples of the CV values across realizations, and on the time interval~$\left((k-1)\Delta t_{\mathrm{hist}},k\Delta t_{\mathrm{hist}}\right]$.

We estimate the total variation distances between marginals (where~$f\sharp\mu$ denotes the pushforward of the measure~$\mu$ by the function~$f$):
$$\|\phi\sharp\mu_{k\Delta t_{\mathrm{hist}}}-\phi\sharp\nu\|_{\mathrm{TV}},\qquad\|\psi\sharp\mu_{k\Delta t_{\mathrm{hist}}}-\psi\sharp\nu\|_{\mathrm{TV}}$$
by considering the~$L^1$-distances between the one-dimensional histograms, constructed using~$50$ regular bins. We denote by~$e_{k\Delta t_{\mathrm{hist}}}(\phi)$,~$e_{k\Delta t_{\mathrm{hist}}}(\psi)$ the corresponding estimators, and define~$e_t(\phi),\,e_t(\psi)$ for any~$t\geq \Delta t_{\mathrm{hist}}$ by linear interpolation.

Values of the ``mixing-time'' at level~$\varepsilon=0.05$, defined as~$\MT_{\varepsilon}(f) = \inf\,\{t\geq \Delta t_{\mathrm{hist}}:\, e_t(f)<\varepsilon\}$, for~$f\in\{\phi,\psi\}$ were computed. Additionally, we inferred a ``decorrelation rate'', by performing an affine fit on~$\log\e_t(f)$ on~$t\in [1,\MT_{0.1}(f)]\,\mathrm{ps}$ if~$\MT_{0.1}-1\geq 4\,\mathrm{ps}$. Otherwise, no fit was performed.
We give an example of convergence curves for the value of the friction parameter~$\gamma = 10\,\mathrm{ps}^{-1}$ in Figures~\ref{fig:decorr_basin} and~\ref{fig:decorr_opt_domain}, for the free-energy basin and the optimized state, respectively. The horizontal line correspond to the value~$\varepsilon=0.05$ of the tolerance threshold for the mixing time. The regression line corresponding to the smallest decorrelation rate is also plotted. Error curves are color-coded according to the procedure with which the initial configuration is sampled, as made precise in the next paragraph.
\begin{figure}
        \centering
        \begin{subfigure}{0.49\textwidth}
        \includegraphics[width=\textwidth]{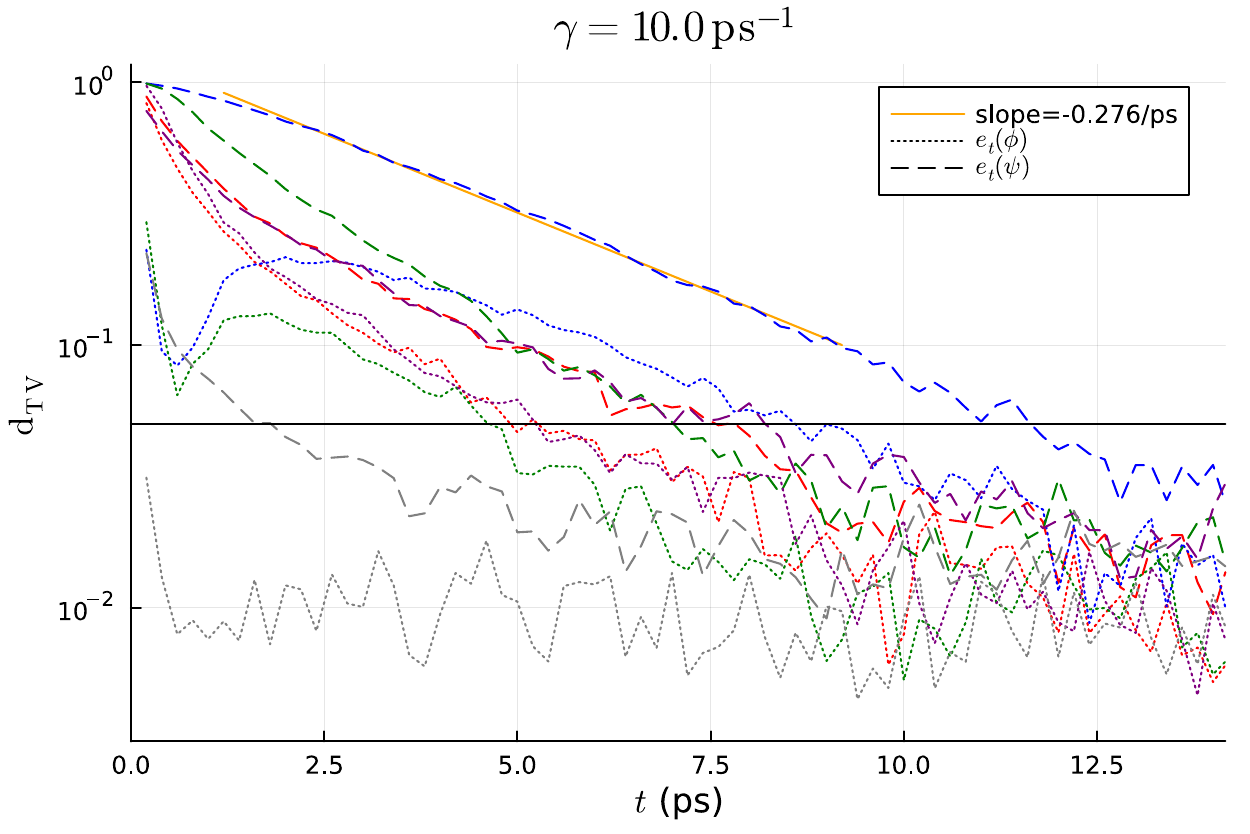}
            \caption{Free-energy basin.}
            \label{fig:decorr_basin}
        \end{subfigure}
        \begin{subfigure}{0.49\textwidth}
            \includegraphics[width=\textwidth]{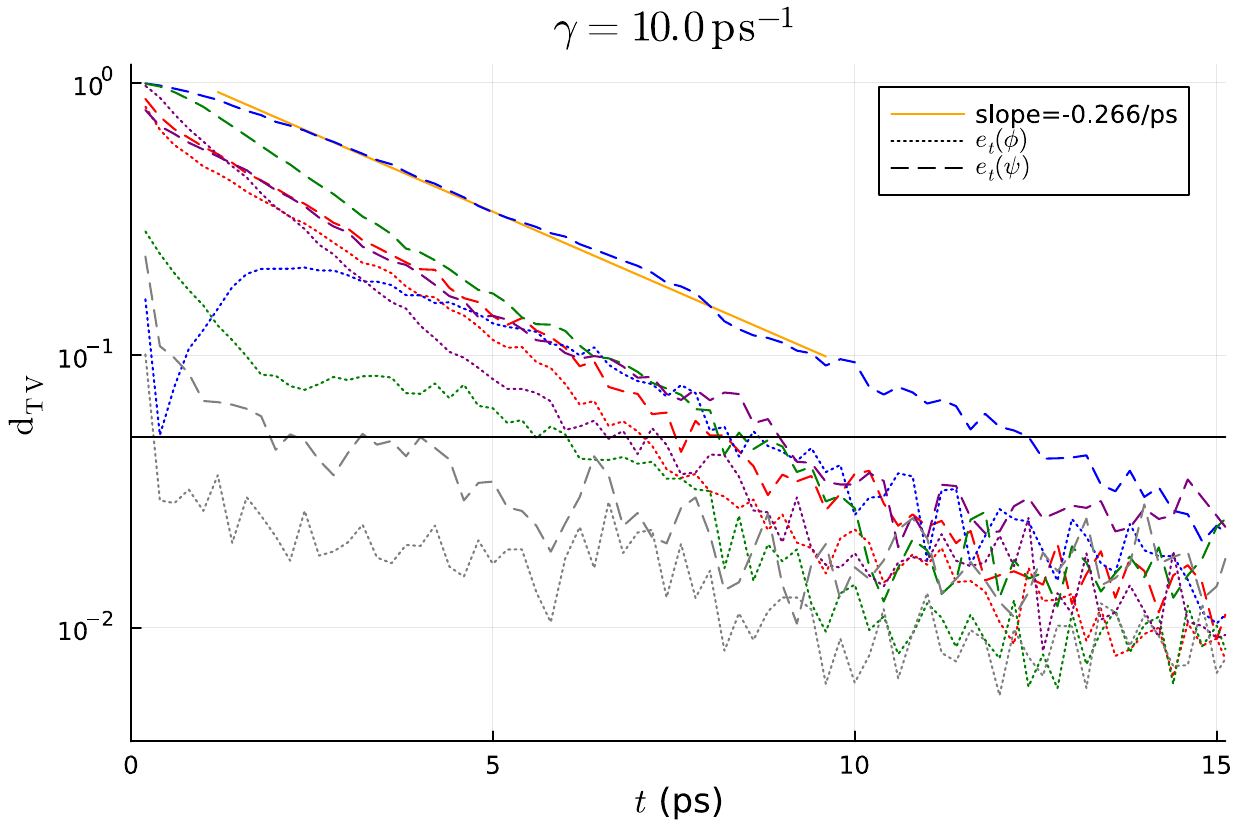}
                \caption{Optimized domain.}
                \label{fig:decorr_opt_domain}
            \end{subfigure}
            \caption{Convergence of the marginals of the Fleming--Viot process to the corresponding quasi-stationary marginals for~$\gamma=10\,\mathrm{ps}^{-1}$. The dependence on the initial condition is color-coded as in Figure~\ref{fig:qsds}, except for the gray curves which correspond to initial conditions close to the free-energy minimum. We observe, after a short transient phase, exponential convergence to the quasi-stationary marginals, and also a slight decrease in the decorrelation rate for the optimized state.}
        \label{fig:decorrelation_fv}
\end{figure}

\paragraph{Sampling of initial configurations.}
To assess the dependence of the decorrelation errors~$e_t(\phi)$,~$e_t(\psi)$ on the initial configuration of the system, we compute a realization of~$e_t(\phi),\,e_t(\psi)$ for various distributions of initial configurations~$X_0$, each one corresponding to a critical point of the free-energy.
\begin{itemize}
    \item{Four distributions corresponding to the four free-energy saddle points surrounding state 5~(see Figure~\ref{fig:pmf}). First, a steered MD simulation was performed to bring the system close to the target critical point, following which a harmonically restrained simulation was performed, with a biasing potential centered at the critical point. Only initial conditions with ``entering'' velocities were considered. We mean by this that, for the purposes of this experiment, we discarded samples which moved away from the free-energy minimum in CV space during an equilibrium MD simulation of~$2\,\mathrm{fs}$, or which had a final configuration outside of the free-energy basin.}
    \item{One distribution corresponding to the local free-energy minimum in state 5. Again, a steered MD simulation was performed to bring the system close to the free-energy minimum, followed by a harmonically restrained simulation. However, no ``velocity check'' was performed in this case.}
\end{itemize}
In both cases, the steering phase was performed for~$1\,\mathrm{ps}$, and the harmonically restrained phase for~$5\,\mathrm{ps}$, both with an inverse force constant of~$\eta=40\,\mathrm{mol}\cdot\mathrm{kcal}^{-1}$.
Timesteps of~$0.5\,\mathrm{fs}$ and~$1\,\mathrm{fs}$ were used respectively for the steering phase and the harmonically-restrained equilibration phase.

These two families of initial conditions correspond roughly to two natural definitions of the core-set~$\mathcal C$ from Algorithm~\ref{alg:fv} associated with~$\Omega$. Initial conditions associated with free-energy saddle points correspond to a coreset~$\mathcal C=\basin{z_5}\cap\Omega$, where~$z_5$ is the free-energy minimum associated with $\Omega$, and~$\basin{z_5}$  denotes the corresponding free-energy basin. Initial conditions steered towards the free-energy minimum correspond to the definition~$\mathcal C = B(z_5,r_{\mathrm{c}})$ for some small~$r_{\mathrm{c}}>0$ in CV space.
In the case where~$\Omega =\basin{z_5}$, the first core-set corresponds to the state itself: $\mathcal C=\Omega$, which is the standard situation in ParRep.

In Figure~\ref{fig:qsds}, we show the empirical stationary~$\xi$-marginal~$\xi_*\widehat{\nu}$ of the Fleming--Viot process, for the two state definitions we compare, and the value of the friction parameter~$\gamma=5\,\mathrm{ps}^{-1}$. Additionally, the sampled initial values of the collective variable are scattered, and color-coded according to the associated free-energy saddle point. The color coding is the same as in Figures~\ref{fig:decorr_basin} and~\ref{fig:decorr_opt_domain}.
\begin{figure}
\centering
\includegraphics[width=0.49\linewidth]{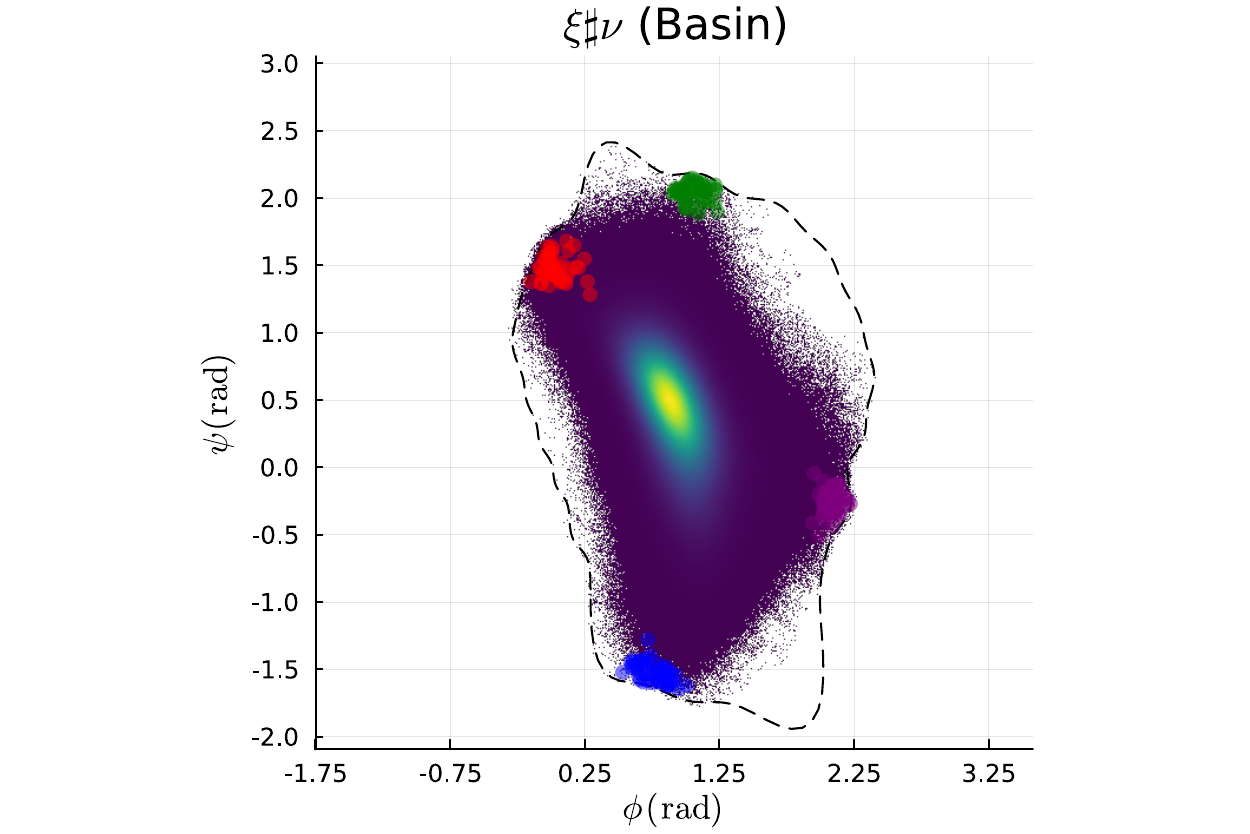}
\includegraphics[width=0.49\linewidth]{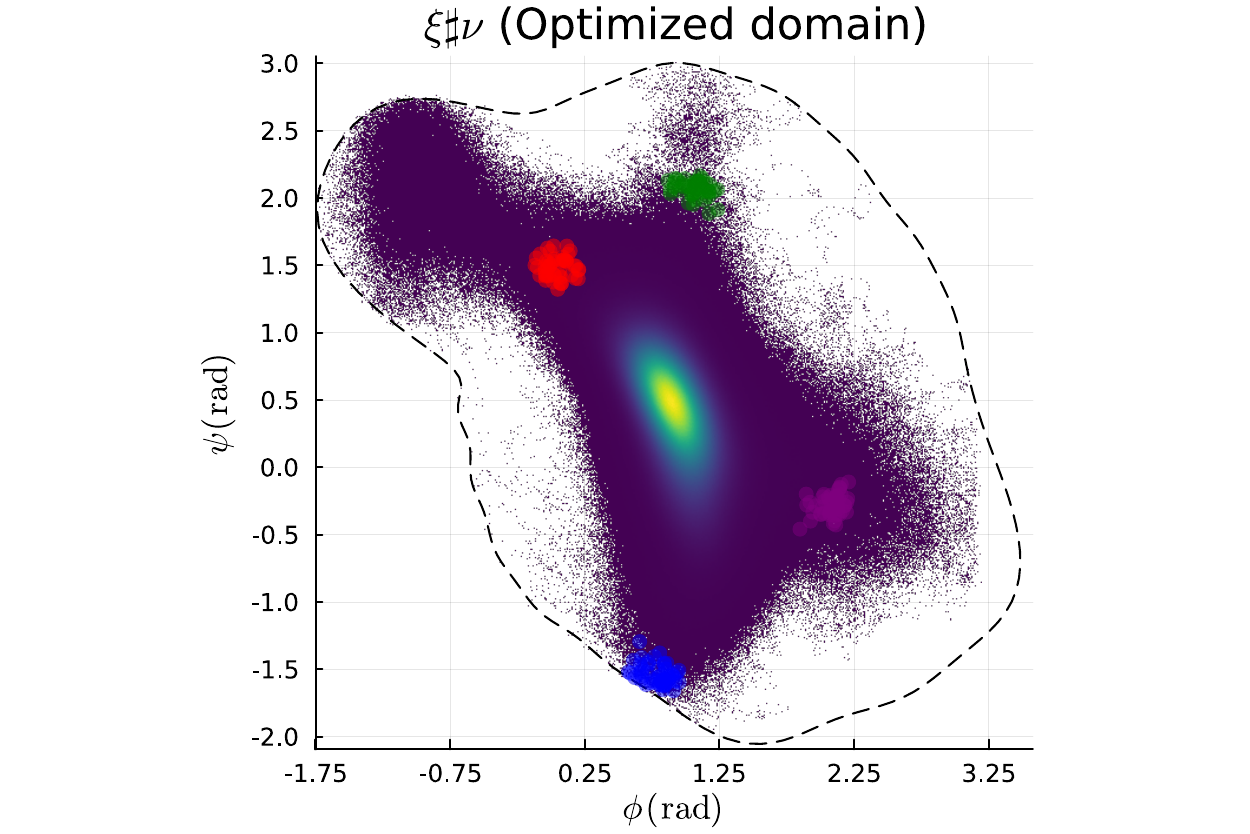}
\caption{Empirical $\xi$-marginal for the stationary Fleming--Viot process for~$\gamma=10\,\mathrm{ps}^{-1}$. Left: free-energy basin. Right: numerically optimized domain. On both figures, sampled initial configurations for the Fleming--Viot process are overlaid on the stationary histogram, and distinguished by color according to the corresponding free-energy saddle point. Initial conditions sampled around the free-energy minimum are not depicted.}
\label{fig:qsds}
\end{figure}

\paragraph{Results.}
We present the results in Table~\ref{fig:benchmark_opt}: for each of the states and values of~$\gamma\in\{1,2,5,10,20\}\,\mathrm{ps}^{-1}$, we report the estimated exit rate ($\ER$)~$\widehat{\lambda}_1(\Omega)$, in~$\mathrm{ps}^{-1}$, as well as various metrics quantifying the speed of convergence to the QSD.
\begin{itemize}
    \item{The decorrelation rate ($\DR$) in~$\mathrm{ps}^{-1}$, defined as the least infered decorrelation rate among the observables~$\phi,\psi$ and ensembles of initial configurations.}
    \item{The mixing time from saddle points ($\MT^{\mathrm{s}}$) in~$\mathrm{ps}$, defined as the largest mixing time~$\MT_{0.05}(f)$ for~$f\in\{\phi,\psi\}$ and initial conditions steered towards one of the four free-energy saddle points according to the procedure described in the previous paragraph.}
    \item{The mixing time from the minimum ($\MT^{\mathrm{m}}$) in~$\mathrm{ps}$, defined as the largest mixing time~$\MT_{0.05}(f)$ for~$f\in\{\phi,\psi\}$, and initial conditions steered towards the free-energy minimum.}
\end{itemize}
To each of these metrics, we associate a corresponding measure of the separation of metastable timescales, namely the respective ratios~$\DR/\ER$,~$1/(\MT^{\mathrm{s}}\cdot\ER)$ and~$1/(\MT^{\mathrm{m}}\cdot\ER)$ (where the inverse mixing times are interpreted as ``mixing rates'').
The full results are given in Tables~\ref{tab:free_energy_basin} and~\ref{tab:optimized_state}, for the free-energy basin and optimized state, respectively. The timescale ratios are also plotted for visual comparison in Figure~\ref{fig:benchmark_opt}. We consistently observe a gain in timescale separation when using the optimized state, especially for higher values of the friction parameter, where the gain is estimated to be about~$\times 3$ for the optimized state, across all measures of timescale separation. At lower values of the friction parameter, the gain is less pronounced, but still substantial.
The improvements in timescale separation are reported in Figure~\ref{tab:ratios}. The various timescale separation metrics are generally in agreement about this improvement.
\begin{figure}[ht]
  \centering
  \small
  \begin{subtable}{0.9\textwidth}
    \centering
    \caption{Free-energy basin.}
    \begin{tabular}{l|lllllll}
      \hline
      $\gamma$ & $\ER$ & $\DR$ &  $\DR/\ER$ & $\MT^{\mathrm{s}}$ & $(\ER\cdot\MT^{\mathrm{s}})^{-1}$ & $\MT^{\mathrm{m}}$  & $(\ER\cdot\MT^{\mathrm{m}})^{-1}$ \\
      \hline
$1$ & $(4.31 \pm 0.20) \times 10^{-3}$ & $0.48$ & $111.7$ & $6.8$ & $34.1$ & $0.8$ & $290.0$ \\
$2$ & $(4.21 \pm 0.21) \times 10^{-3}$ & $0.47$ & $110.5$ & $7.2$ & $33.0$ & $1.0$ & $237.0$ \\
$5$ & $(3.69 \pm 0.21) \times 10^{-3}$ & $0.36$ & $97.1$ & $9.4$ & $28.8$ & $1.8$ & $150.6$ \\
$10$ & $(3.26 \pm 0.19) \times 10^{-3}$ & $0.28$ & $84.9$ & $11.8$ & $26.0$ & $2.0$ & $153.6$ \\
$20$ & $(2.52 \pm 0.16) \times 10^{-3}$ & $0.19$ & $72.4$ & $16.8$ & $23.6$ & $3.0$ & $132.1$ \\
      \hline
    \end{tabular}
    \label{tab:free_energy_basin}
  \end{subtable}
  \hfill
  \begin{subtable}{0.9\textwidth}
    \centering
    \caption{Optimized domain.}
    \begin{tabular}{l|lllllll}
      \hline
      $\gamma$ & $\ER$ & $\DR$ &  $\DR/\ER$ & $\MT^{\mathrm{s}}$ & $(\ER\cdot\MT^{\mathrm{s}})^{-1}$ & $\MT^{\mathrm{m}}$  & $(\ER\cdot\MT^{\mathrm{m}})^{-1}$ \\
      \hline
$1$ & $(1.86 \pm 0.15) \times 10^{-3}$ & $0.25$ & $134.0$ & $8.6$ & $62.5$ & $1.6$ & $336.0$ \\
$2$ & $(1.70 \pm 0.13) \times 10^{-3}$ & $0.40$ & $237.0$ & $7.6$ & $77.6$ & $1.2$ & $491.0$ \\
$5$ & $(1.43 \pm 0.12) \times 10^{-3}$ & $0.26$ & $184.0$ & $13.6$ & $51.4$ & $1.4$ & $500.0$ \\
$10$ & $(1.01 \pm 0.11) \times 10^{-3}$ & $0.27$ & $265.0$ & $12.6$ & $78.9$ & $2.0$ & $497.0$ \\
$20$ & $(7.55 \pm 0.92) \times 10^{-4}$ & $0.17$ & $228.0$ & $17.0$ & $77.9$ & $3.0$ & $442.0$ \\
      \hline
    \end{tabular}
    \label{tab:optimized_state}
  \end{subtable}

    \begin{subtable}{0.9\textwidth}
    \centering
    \caption{Improvement of the optimized domain over the free-energy basin in timescale separation metrics.}
    \begin{tabular}{l|lll}
      \hline
      $\gamma$ & $\DR/\ER$ & $(\ER\cdot\MT^{\mathrm{s}})^{-1}$ & $(\ER\cdot\MT^{\mathrm{m}})^{-1}$ \\
      \hline 
$1$ & $1.2$ &  $1.83$ &  $1.16$\\
$2$ &  $2.14$ &   $2.35$ &  $2.07$\\
$5$ &  $1.89$ &$1.78$ &$3.32$\\
$10$ & $3.12$ & $3.03$ & $3.24$\\
$20$ &  $3.15$ &  $3.3$ & $3.35$\\
    \hline
\end{tabular}
\label{tab:ratios}
\end{subtable}
  \caption{Results of the Fleming--Viot simulations, showing that the optimized state consistently outperforms the free-energy basin. Reported errors are at the level~$\pm1.96\sigma$.}
  \label{fig:benchmark_opt}
  \normalsize
\end{figure}

\begin{figure}
\centering
\includegraphics[width=0.75\linewidth]{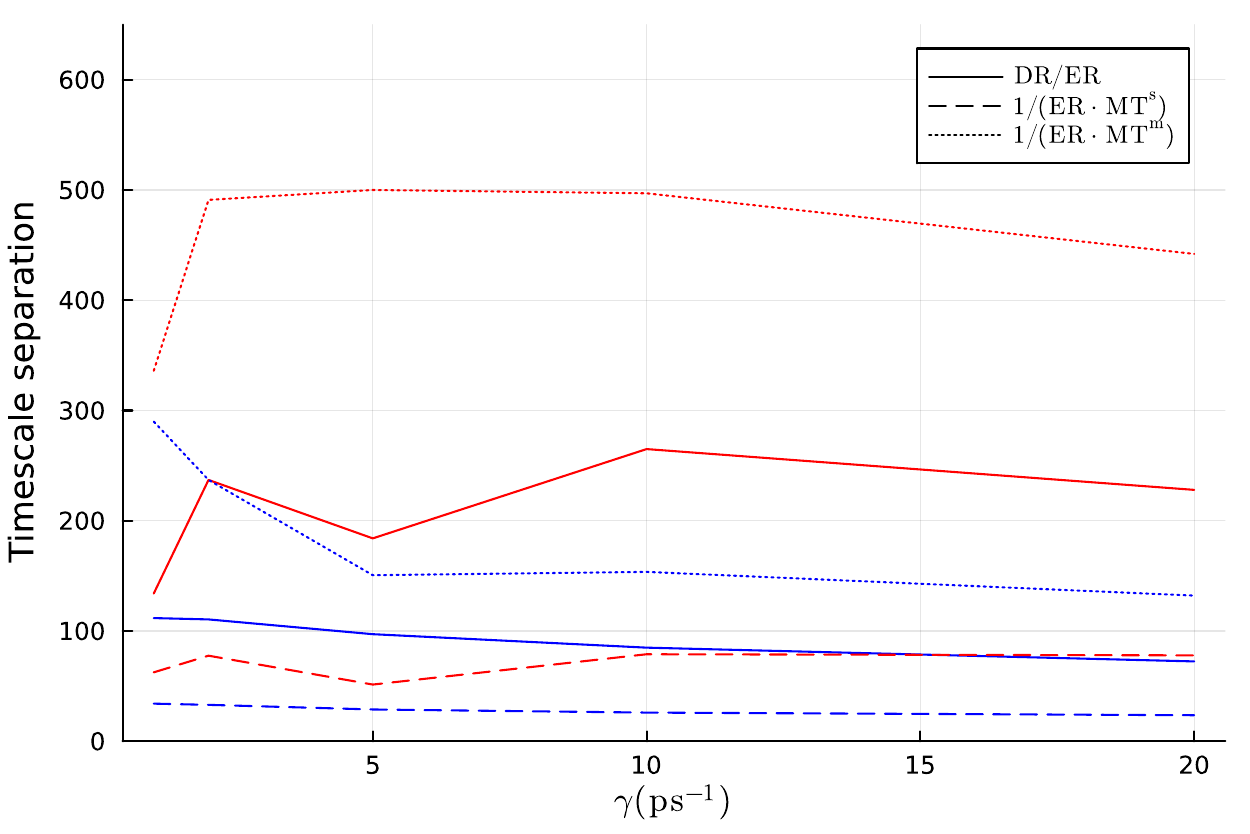}
\caption{Timescale separation ratios. Blue lines correspond to the free-energy basin, and red lines to the optimized domain. Across  all metrics, the optimized domain outperforms the free-energy basin.}
\label{fig:benchmark_visual}
\end{figure}

\section{Conclusions and perspectives}
This work raises a number of perspectives which could prove interesting for future research.
\begin{itemize}
    \item The most salient point is the extension of the shape perturbation results of Theorem~\ref{thm:gateaux_differentiability} to the case of non-reversible and/or hypoelliptic diffusions. We expect that, due to the non-symmetry and/or non-ellipticity of the generator~\eqref{eq:generator}, this represents a significant theoretical endeavour.
    \item A standing question would be how to systematically optimize the definition of the core-sets in Algorithm~\ref{alg:parrep} (see Appendix~\ref{sec:parrep} below), either numerically or in some limiting asymptotic regime. This question is related to the search for quantitative estimates of the prefactor~$C(x)$ in the error estimate~\eqref{eq:decorrelation_rate}. Another idea would be to jointly optimize a collection of states to maximize a weighted sum of separation of timescales within each state.
    \item At this point, a convergence proof for the method described in Algorithm~\ref{alg:ascent} is lacking. It would be interesting to obtain consistency results with respect to the various approximation parameters.
    \item The direct shape optimization method, due to the FEM discretization, is limited to settings where low-dimensional ($m=2$ or $3$) representations of the dynamics (i.e. good CVs) are available. To go beyond this limitation, a natural approach would be to follow a parametric approach, setting~$\Omega_\theta=\Phi_\theta(\mathcal C)$ for some reference domain~$\mathcal C\subset \R^d$, where~$\theta\mapsto \Phi_\theta$ is a parametric family of homeomorphisms, represented for instance using a neural network. The main question becomes how to define a neural architecture for which the Dirichlet eigenvalue problem associated with the transported operator~\eqref{eq:transported_operator} is solvable, and for which perturbations of the eigenvalues are tractable.
    \item Instead of computing the thermodynamic coefficients~$F_\xi$,~$a_\xi$ (see Equations~\eqref{eq:free_energy},~\eqref{eq:diffusion_tensor}) associated with the effective dynamics~\eqref{eq:effective_dynamics}, one could seek kinetically-tuned effective dynamics. One approach would be to train a parametric model of a dynamics of the form~\eqref{eq:overdamped_langevin} on CV trajectory data, using EDMD-like techniques or neural SDEs in the hope of obtaining a dynamical description which is more robust to a suboptimal choice of CV. Note that the results of Theorem~\ref{thm:gateaux_differentiability}, since they are currently limited to the case of reversible dynamics, place a constraint on the class of allowable models.
    \item Recent techniques (see for example~\cite{MPWN18,WN20,ZS23,TLK25}) target the learning of collective variables approximating the dominant eigenvectors of Markov processes. In view of Proposition~\ref{prop:rates_estimates}, a similar strategy is also natural in our setting. However, the problem is more subtle in the case of Dirichlet eigenvalues, since the collective variable defines the state, which in turn defines the eigenvectors. This leads to a bilevel optimization problem in~$(\xi,\Omega)$, for which various methods could be considered. A key practical challenge in this setting is to avoid a costly recomputation of the free energy and diffusion tensor at each update of the collective variable.
\end{itemize}
\label{sec:ccl}

\appendix
\section{Proof of Theorem~\ref{thm:gateaux_differentiability}}
\label{sec:proof}
We prove Theorem~\ref{thm:gateaux_differentiability} in this appendix. The proof relies on the transport of the variational formulation of the generalized eigenvalue problem on~$\Omega_\theta$ to the reference domain~$\Omega$.
This leads to the definition of a family of generalized eigenproblems associated with bilinear forms parametrized by~$\theta$. The corresponding eigenvalues are exactly the Dirichlet eigenvalues of~$-\cL_\beta$ on~$\Omega_\theta$.
One then proves the Fréchet-differentiability of these bilinear forms, or equivalently, by polarization, of the associated quadratic forms. Since the first-order perturbations are themselves unbounded quadratic forms, this regularity only holds with respect to the topology of relatively bounded perturbations of the reference quadratic forms.
Once this has been established, classical results of perturbation theory from~\cite{K13} can be leveraged to show the Fréchet differentiability of the inverse operator, and finally the Gateaux-differentiability of multiple eigenvalues.
\begin{remark}
    By adapting the approach based on the implicit function theorem discussed in~\cite[Section 5.7]{HP05} for the Dirichlet Laplacian (which corresponds to the special case~$a\equiv\Id$ and~$V\equiv 0$ in our setting), one can also show that the map~$\theta\mapsto \lambda_k(\Omega_\theta)$ is~$\mathcal C^1$ around~$\theta=0$ in a somewhat less technical manner. 
    However, this approach is only adapted to the case of simple eigenvalues. Since the main purpose of~Theorem~\ref{thm:gateaux_differentiability} is to identify ascent directions for functionals of the Dirichlet eigenvalues of~$\cL_\beta$ with respect to the perturbation~$\theta$,
    and since multiple eigenvalues have been noted to occur in eigenvalue shape optimization problems (see e.g.~\cite[Theorem 2.5.10]{H06} or~\cite[Section 4.5]{O04}), including in our own numerical experiments (see Figure~\ref{fig:lambdas_diala} above), it is of practical interest to devise numerical strategies adapted to this reality.
\end{remark}

\begin{proof}[Proof of Theorem~\ref{thm:gateaux_differentiability}]
    At various places, we assume that
    $$\|\theta\|_{\Winf} < h_0$$
    for some~$h_0>0$ whose value will be reduced several times.
    We also fix a regular open and bounded set~$\mathcal D\subset\R^d$, sufficiently large so that~$\bigcup_{x\in \Omega} B_{\R^d}(x,h_0)\subset \mathcal D$. This ensures in particular that~$\Omega_\theta\subset \mathcal D$ for all~$\|\theta\|_{\Winf}< h_0$.

    We say that an estimate of the form
    \begin{equation}
        J(\Omega,\theta) \leq C(\mathcal D)h(\theta),
    \end{equation}
    with~$J : \mathcal B(\R^d)\times \Winf \to \R$,~$h:\Winf\to\R$ and $C(\mathcal D)>0$, holds ``uniformly inside~$\mathcal D$'' if it holds for all pairs~$(\Omega,\theta)$ with~$\Omega\subset \mathcal D$ an open Lipschitz set and~$\theta\in B_{\Winf}(0,h_0)$.

    \paragraph{Transport of the variational formulation.}
    Introduce, for~$\theta\in\Winf$, the bilinear forms
    \begin{equation}
        \label{eq:bilinear_forms_0}
        \forall\,u,v\in H_0^1(\Omega_\theta),\qquad \fa_0(u,v;\Omega_\theta) = \frac1\beta\int_{\Omega_\theta}\nabla u^\top a \nabla v \,\e^{-\beta V},\qquad \fb_0(u,v;\Omega_\theta) = \int_{\Omega_\theta}uv\,\e^{-\beta V}.
    \end{equation}
    For~$\|\theta\|_{\Winf}<h_0$ sufficiently small, the map~$\Phi_\theta(x) = x + \theta(x)$ is a bi-Lipschitz homeomorphism of~$\R^d$, and using the Lebesgue change of variables formula, it holds
    \begin{equation}
        \begin{aligned}
        \fa_0(u,v;\Omega_\theta) &= \frac1\beta \int_{\Omega}\left(\nabla u^\top a \nabla v\,\e^{-\beta V}\right)\circ \Phi_\theta \left|\det\,\nabla \Phi_\theta\right|\\
        &=\frac1\beta\int_\Omega \nabla\left(u\circ \Phi_\theta\right)^\top \nabla \Phi_\theta^{-\top}a\circ \Phi_\theta \nabla \Phi_\theta^{-1}\nabla\left(v\circ\Phi_\theta\right)\e^{-\beta V\circ\Phi_\theta}\left|\det\,\nabla\Phi_\theta\right|\\
        &:=\fa_\theta(u\circ \Phi_\theta,v\circ\Phi_\theta;\Omega),
        \end{aligned}
    \end{equation}
    where we used~$\nabla\left(u\circ\Phi_\theta\right)=\nabla \Phi_\theta \left(\nabla u\right)\circ\Phi_\theta$ in the penultimate line, and~$\nabla \Phi_\theta^{-1}$ denotes the matrix inverse of~$\nabla \Phi_\theta$.
    Similarly,
    \begin{equation}
        \fb_0(u,v;\Omega_\theta) = \int_\Omega \left(u\circ\Phi_\theta\right) \left(v\circ\Phi_\theta\right) \e^{-\beta V\circ \Phi_\theta}\left|\det\,\nabla \Phi_\theta\right| := \fb_\theta(u\circ\Phi_\theta,v\circ\Phi_\theta;\Omega).
    \end{equation}
    From now on, all bilinear forms act on the fixed domain~$\Omega$, which we therefore omit in the notation for the bilinear forms, i.e. we define for all~$u,v\in H_0^1(\Omega)$,
    \begin{equation}
        \label{eq:bilinear_forms}
        \fa_\theta(u,v) = \frac1\beta\int_\Omega \nabla u^\top \nabla \Phi_\theta^{-\top}a\circ \Phi_\theta \nabla \Phi_\theta^{-1}\nabla v\e^{-\beta V\circ\Phi_\theta}\left|\det\,\nabla\Phi_\theta\right|,\qquad \fb_\theta(u,v) = \int_\Omega uv\e^{-\beta \circ \Phi_\theta}\left|\det\,\nabla\Phi_\theta\right|.
    \end{equation}

    \paragraph{Spectral properties.}
    Now, considering an eigenpair~$(\lambda_\theta,u_\theta)$ for~$-\cL_\beta$ on~$L^2_\beta(\Omega_\theta)$, it holds, for all~$v\in H_{0,\beta}^1(\Omega_\theta)=H_0^1(\Omega_\theta)$, that
    \[\fa_\theta\left(u_\theta\circ\Phi_\theta,v\circ\Phi_\theta\right) = \lambda_\theta \fb_\theta(u\circ\Phi_\theta,v\circ\Phi_\theta).\]
    Since, using the isomorphism~\eqref{eq:composition}, any function~$v\in H_0^1(\Omega)$ can be written under the form~$v\circ \Phi_{\theta}^{-1}\circ \Phi_\theta$ with~$v\circ\Phi_{\theta}^{-1}\in H_0^1(\Omega_\theta)$, the transported eigenvector~$w_\theta = u_\theta\circ \Phi_\theta$ satisfies
    \begin{equation}
        \fa_\theta(w_\theta,v) = \lambda_\theta \fb_\theta(w_\theta,v)\qquad \forall v\in H_0^1(\Omega).
    \end{equation}
    In other words, $w_\theta$ is a generalized eigenvector for $(\fa_\theta,\fb_\theta)$. Let us make this statement precise.

    We first introduce the following estimates, which hold, for~$h_0<1$, uniformly inside~$\mathcal D$ for some~$C_1(\mathcal D),C_2(\mathcal D)>0$:
    \begin{equation}
        \label{eq:grad_phi_estimates}
        \begin{aligned}
            \|\nabla\Phi_\theta^{-1} - \Id\|_{L^\infty(\Omega;\mathcal M_d)} &\leq \frac{\|\theta\|_{\Winf}}{1-\|\theta\|_{\Winf}} \leq C_1(\mathcal D) \|\theta\|_\Winf,\\
            \|\nabla\Phi_\theta^{-1}-\left(\Id-\nabla \theta\right)\|_{L^\infty(\Omega;\mathcal M_d)}&\leq \frac{\|\theta\|^2_{\Winf}}{1-\|\theta\|_{\Winf}} \leq C_2(\mathcal D)\|\theta\|^2_{\Winf}.
        \end{aligned}
    \end{equation}
    These follow by expanding~$\nabla\Phi_\theta(x)^{-1}=\left(\Id+\nabla\theta(x)\right)^{-1}$ into a Neumann series and estimating the (submultiplicative)~$\mathcal M_d$-norm of the first and second partial remainders respectively.
    In fact we can take $C_1(\mathcal D)=C_2(\mathcal D)=(1-h_0)^{-1}$ here, but we nevertheless distinguish these constants for the sake of clarity.
    
    Secondly, by Jacobi's formula for the Fréchet derivative of the determinant of a $d\times d$ matrix, it holds almost everywhere in~$\R^d$ (by Rademacher's theorem) that
    \begin{equation}
        \left|\det \nabla \Phi_\theta(x)\right| = \det(1+\nabla \theta(x)) = 1 + \Tr\,\nabla \theta(x) + \bigo(|\nabla\theta(x)|^2_{\mathcal M_d});
    \end{equation}
    whence, uniformly inside~$\mathcal D$ for some constants~$C_3(\mathcal D),C_4(\mathcal D)>0$, it holds
    \begin{equation}
        \label{eq:divergence_estimate}
        \begin{aligned}
            \left\|\det \nabla\Phi_\theta-1\right\|_{L^\infty(\Omega)} &\leq C_3(\mathcal D)\left\|\theta\right\|_{\Winf},\\
            \left\|\det \nabla\Phi_\theta-1-\div\,\theta\right\|_{L^\infty(\Omega)} &\leq C_4(\mathcal D)\left\|\theta\right\|_{\Winf}^2.
        \end{aligned}
    \end{equation}
    Note that we used~$\Tr\nabla \theta = \div\,\theta$ and~$|\Tr M|\leq d|M|_{\mathcal M_d}$ for all~$M\in \mathcal M_d$.

    From the estimates~\eqref{eq:grad_phi_estimates} and~\eqref{eq:divergence_estimate}, we deduce that the symmetric bilinear form~$\fa_\theta$, with domain~$H_0^1(\Omega)\subset L^2(\Omega)$, satisfies the following upper bound uniformly inside~$\mathcal D$:
    \begin{equation}
        \label{eq:quadratic_form_upper_bound}
        \fa_\theta(u,u) \leq \frac1\beta\|a\|_{L^\infty\left(\mathcal D;\mathcal M_d\right)}\|\e^{-\beta V}\|_{L^\infty\left(\mathcal D\right)}(1+C_3(\mathcal D) h_0)(1+C_1(\mathcal D)h_0)^2\|\nabla u\|^2_{L^2(\Omega)},
    \end{equation}
     as well as the lower bound
    \begin{equation}
        \label{eq:quadratic_form_lower_bound}
        \begin{aligned}
            \fa_\theta(u,u) &\geq \frac1\beta\varepsilon_a(\overline{\mathcal D}) m_V(\mathcal D) (1-C_1(\mathcal D)h_0)^2(1-C_3(\mathcal D)h_0)\|\nabla u\|_{L^2(\Omega)}^2\\
            &\geq \frac1\beta\varepsilon_a(\overline{\mathcal D}) m_V(\mathcal D) (1-C_1(\mathcal D)h_0)^2(1-C_3(\mathcal D)h_0) \mu_0(\Omega)\|u\|^2_{L^2(\Omega)}\\
            &\geq \frac1\beta\varepsilon_a(\overline{\mathcal D}) m_V(\mathcal D) (1-C_1(\mathcal D)h_0)^2(1-C_3(\mathcal D)h_0) \mu_0(\mathcal D)\|u\|^2_{L^2(\Omega)},
        \end{aligned}
    \end{equation}
    with~$\mu_0>0$ is the principal Dirichlet eigenvalue of the Laplacian, and where we use~$\mu_0(\Omega)\geq \mu_0(\mathcal D)$ in the last line~(see for instance~\cite[Proposition 16]{BLS24}), and recall the definition~\eqref{eq:a_ellipticity} of~$\varepsilon_a(\overline{\mathcal D})$.
    In~\eqref{eq:quadratic_form_lower_bound}, we define
    $$ m_V(\mathcal D) := \exp\left(-\beta \left[\underset{\mathcal D}{\mathrm{ess}\sup}\, V\right]\right)>0.$$
    Note that the lower bound in~\eqref{eq:quadratic_form_lower_bound} is positive for~$h_0$ sufficiently small, therefore~$\fa_{\theta}$ is~$H^1(\Omega)$-coercive and $L^2(\Omega)$-bounded from below, uniformly inside~$\mathcal D$.
    Moreover, it follows from~\eqref{eq:quadratic_form_upper_bound} and~\eqref{eq:quadratic_form_lower_bound} that~$\fa_\theta$ is closed, since the squared form norm~$\|u\|_{\fa_\theta}^2 = \fa_\theta(u,u) + \|u\|^2_{L^2(\Omega)}$ on~$H_0^1(\Omega)$ is equivalent to the~squared $H^1(\Omega)$-norm.
    Therefore, by a representation result for positive symmetric closed forms~\cite[Theorem VI.2.6]{K13}, there exists a self-adjoint operator~$A_\theta$ satisfying
    \begin{equation}
        \label{eq:def_a_theta}
        \fa_\theta(u,v) = \left\langle A_\theta u,v\right\rangle_{L^2(\Omega)},\qquad \forall (u,v)\in \mathcal D(A_\theta)\times L^2(\Omega),
    \end{equation}
    with~$\mathcal D(A_\theta)\subset H_0^1(\Omega)$ and $A_\theta$ being $L^2(\Omega)$-bounded from below, with the same lower bound as~$a_\theta$:
    \begin{equation}
        \label{eq:A_lower_bound}
        \langle A_\theta u,u\rangle_{L^2(\Omega)} \geq \frac1\beta\varepsilon_a(\overline{\mathcal D}) m_V(\mathcal D) (1-C_1(\mathcal D)h_0)^2(1-C_3(\mathcal D)h_0) \mu_0(\mathcal D)\|u\|^2_{L^2(\Omega)}.
    \end{equation}
    In particular, the resolvent~$A_\theta^{-1}$ is bounded and compact in view of the compact embedding~$H_0^1(\Omega)\subset L^2(\Omega)$ given by the Rellich--Kondrachov theorem.

    Note that, by integration by parts,~$A_\theta$ extends the positive (for~$h_0$ sufficiently small) operator
    \begin{equation}
        \label{eq:transported_operator}
        -{\widetilde\cL}_{\beta,\theta} \varphi = -\frac1\beta \div\left(\left|\det \nabla \Phi_\theta\right|\e^{-\beta V\circ \Phi_\theta} \nabla \Phi_\theta^{-\top}a\circ \Phi_\theta \nabla \Phi_\theta^{-1}\nabla \varphi\right)\qquad \forall\varphi\in \mathcal D({\widetilde\cL}_{\beta,\theta}) = \mathcal \testfuncs(\Omega),
    \end{equation}
    therefore, the closed operator~$A_\theta$, which extends~$-{\widetilde\cL}_{\beta,\theta}$, also corresponds to its Friedrichs extension (see~\cite[Section VI.2.3]{K13}).
    For the sake of consistency, we also write~$A_0$ for the operator~$-\cL_\beta$. 
    Similarly,~$\fb_\theta$ has the representation~$\fb_\theta(u,v) = \left\langle B_\theta u,v\right\rangle_{L^2(\Omega)}$, where $B_\theta$ is the bounded, positive linear operator given by multiplication by~$\e^{-\beta V\circ \Phi_\theta}\left|\det\,\nabla \Phi_\theta\right|$, which is both bounded from above and from below uniformly inside~$\mathcal D$:
    \begin{equation}
        \label{eq:b_bound}
        m_{V}(\mathcal D)(1-C_3(\mathcal D)h_0)\|u\|^2_{L^2(\Omega)}\leq \fb_\theta(u,u) = \left\langle B_\theta u,u\right\rangle_{L^2(\Omega)}\leq \|\e^{-\beta V}\|_{L^\infty(\mathcal D)}(1+C_3(\mathcal D)h_0)\|u\|^2_{L^2(\Omega)}.
    \end{equation}
    It follows (see~\cite[Proposition 1]{HR80b} or the discussion in~\cite[Section VII.6.1]{K13}) that the reciprocals of the eigenvalues of the compact, positive operator~$A_\theta^{-1}B_\theta$ on~$L^2(\Omega)$ (which is also self-adjoint for the topologically equivalent scalar product~$\left\langle B_\theta\cdot,\cdot\right\rangle_{L^2(\Omega)}$) are the solutions to
    \begin{equation}
        \label{eq:eigenvalue_problem_operator_form}
        A_\theta w_\theta = \lambda_\theta B_\theta w_\theta,\quad \lambda_\theta >0,\quad w_\theta \in \mathcal D(A_\theta).
    \end{equation}
    In fact it is more convenient than solving the latter generalized eigenvalue problem to consider the spectrum of the compact operator~$A_\theta^{-1}B_\theta$, which is composed of positive, isolated eigenvalues of finite multiplicity.

    \paragraph{Perturbation estimates.}
    Let us define the first-order perturbations of the linear forms~\eqref{eq:bilinear_forms}. More precisely, we define, for~$u,v\in H_0^1(\Omega)$, the symmetric bilinear forms
    \begin{equation}
        \label{eq:perturbations}
        \begin{aligned}
        \d \fa_0(\theta)(u,v) &= \frac1\beta\int_\Omega \nabla u^\top \left(\nabla a^\top \theta -a\nabla \theta -\nabla\theta^\top a\right)\nabla v\,\e^{-\beta V}+ \frac1\beta \int_\Omega \nabla u^\top a \nabla v\,\div\left(\theta\e^{-\beta V}\right),\\
        \d \fb_0(\theta)(u,v) &= \int_\Omega uv\,\div(\theta \e^{-\beta V}).
        \end{aligned}
    \end{equation}
    Writing, for~$u\in H_0^1(\Omega)$,
    \begin{equation}
        \begin{aligned}
            \label{eq:remainders}
            r_{\fa}(\theta,u) &= \fa_\theta(u,u) - \fa_0(u,u)-\d \fa_0(\theta)(u,u),\\
            r_{\fb}(\theta,u) &= \fb_\theta(u,u) - \fb_0(u,u) - \d \fb_0(\theta)(u,u),
        \end{aligned}
    \end{equation}
    we next show that the following bounds hold for all~$u\in H_0^1(\Omega)$ and~$\theta\in B_{\Winf}(0,h_0)$:
    \begin{equation}
        \label{eq:perturbation_estimates}
        \begin{aligned}
        |\d \fa_0(\theta)(u,u)| &\leq C_{\fa,1}(\mathcal D) \|\theta\|_{\Winf}\fa_0(u,u), &|\d \fb_0(\theta)(u,u)| &\leq C_{\fb,1}(\mathcal D) \|\theta\|_{\Winf}\fb_0(u,u),\\
        |r_{\fa}(\theta,u)|&\leq C_{\fa,2}(\mathcal D)\|\theta\|^2_{\Winf}\fa_0(u,u), &|r_{\fb}(\theta,u)&\leq C_{\fb,2}(\mathcal D)\|\theta\|_{\Winf}^2 \fb_0(u,u),
        \end{aligned}
    \end{equation}
    where~$C_{\fa,1}(\mathcal D),C_{\fa,2}(\mathcal D),C_{\fb,1}(\mathcal D),C_{\fb,2}(\mathcal D)$ are positive constants.
    These estimates, together with the linearity of the maps~$\theta\mapsto \d \fa_0(\theta)$ and~$\theta\mapsto \d \fb_0(\theta)$, establish the Fréchet differentiability of the bilinear forms~$\fa_\theta,\fb_\theta$, in the topology of relative~$\fa_0$-form-boundedness and~$\fb_0$-form-boundedness respectively, at~$\theta=0$. Note that the Kato--Rellich theorem~(see for instance~\cite[Theorem 6.4]{T14}) then implies that~$\mathcal D(A_\theta) = \mathcal D(A_0)$ in a~$\Winf$-neighborhood of~$\theta=0$.
    Therefore, we may assume that~$h_0$ is sufficiently small so that~$\mathcal D(A_\theta)=\mathcal D(A_0)$ for all~$\theta\in B_{\Winf}(0,h_0)$.

    The expressions~\eqref{eq:perturbations} are motivated by formal first-order expansions in~$\theta$ in the expressions~\eqref{eq:bilinear_forms}.
    In order to establish them, we first note that the following estimates hold uniformly inside~$\mathcal D$:
    \begin{equation}
        \label{eq:taylor_expansions}
        \begin{aligned}
            \left\|\nabla \Phi_\theta^{-1}\right\|_{L^\infty(\Omega;\mathcal M_d)} &\leq 1+C_1(\mathcal D)h_0,\\
            \left\|\det \nabla \Phi_\theta\right\|_{L^\infty(\Omega)} &\leq 1+C_3(\mathcal D)h_0,\\
            \left\|a\circ\Phi_\theta\right\|_{L^\infty(\Omega;\mathcal M_d)}&\leq \left\|a\right\|_{L^\infty(\mathcal D;\mathcal M_d)},\\
            \left\|\e^{-\beta V\circ\Phi_\theta}\right\|_{L^2\infty(\Omega)}&\leq \left\|\e^{-\beta V}\right\|_{L^\infty(\mathcal D)},\\
            \left\|\nabla\Phi_\theta^{-1} - (\Id - \nabla\theta)\right\|_{L^\infty(\Omega;\mathcal M_d)}&\leq C_2(\mathcal D)\|\theta\|^2_{\Winf},\\
            \left\|\det\nabla\Phi_\theta - 1 - \div\,\theta\right\|_{L^\infty(\Omega)}&\leq C_4(\mathcal D)\|\theta\|^2_{\Winf},\\
            \left\|a\circ\Phi_\theta - a- \nabla a^\top\theta\right\|_{L^\infty(\Omega;\mathcal M_d)} & \leq \frac12\left\|\nabla^2 a\right\|_{L^\infty(\mathcal D;\mathcal M_d\otimes \mathcal M_d)}\left\|\theta\right\|^2_{L^\infty(\R^d;\R^d)},\\
            \left\|\e^{-\beta V\circ\Phi_\theta}-\left(\e^{-\beta V}- \beta\nabla V^\top \theta \e^{-\beta V}\right) \right\|_{L^\infty(\Omega)} &\leq \frac12\left\|\nabla^2 \left(\e^{-\beta V}\right)\right\|_{L^\infty(\mathcal D;\mathcal M_d)}\left\|\theta\right\|^2_{L^\infty(\R^d;\R^d)},\\
            \|\theta\|_{L^\infty(\Omega;\R^d)},\|\nabla\theta\|_{L^\infty(\Omega;\mathcal M_d)}&\leq h_0,\\
            \left\|\div\,\theta\right\|\leq d h_0 .
        \end{aligned}
    \end{equation}
    The two first estimates in~\eqref{eq:taylor_expansions} follow immediately from~\eqref{eq:grad_phi_estimates} and~\eqref{eq:divergence_estimate}. The third and fourth follow from the inclusion~$\Omega_\theta\subset \mathcal D$, the fifth and sixth are already given in~\eqref{eq:grad_phi_estimates} and~\eqref{eq:divergence_estimate}. The seventh and eighth follow from the regularities~$V\in \cW^{2,\infty}(\mathcal D)$,~$a\in \cW^{2,\infty}(\mathcal D;\mathcal M_d)$ given by Assumption~\eqref{eq:coeff_regularity}, and the inclusion~$\Omega_\theta\subset\mathcal D$.
    The penultimate estimate is clear, and the last one follows from~$|\Tr M|\leq d|M|_{\mathcal M_d}$.

    From the estimates~\eqref{eq:grad_phi_estimates},~\eqref{eq:divergence_estimate} and~\eqref{eq:taylor_expansions}, we obtain, by the Leibniz rule in the Banach algebra~$L^\infty(\Omega;\mathcal M_d)$, that the map
    \begin{equation}
        \label{eq:matrix_expr}
        \alpha:\left\{\begin{aligned}
            \Winf&\rightarrow  L^\infty(\Omega;\mathcal M_d),\\
            \theta & \mapsto\nabla \Phi_\theta^{-\top} a\circ \Phi_\theta \nabla \Phi_\theta^{-1}|\det \nabla \Phi_\theta| \e^{-\beta V\circ \Phi_\theta},
        \end{aligned}\right.
    \end{equation}
    is Fréchet differentiable at~$\theta = 0$, with
    \begin{equation}
        D\alpha(0)\theta = -\nabla\theta^\top a \e^{-\beta V} + \nabla a^\top\theta \e^{-\beta V} -a\nabla\theta\e^{-\beta V} + a\,\div\left(\theta\e^{-\beta V}\right).
    \end{equation}
    Moreover, there exist~$M_{\fa,1}(\mathcal D),M_{\fa,2}(\mathcal D)>0$ such that, uniformly inside~$\mathcal D$, it holds
    \begin{equation}
        \label{eq:dalpha_estimates}
        \begin{aligned}
            \|D\alpha(0)\theta\|_{L^\infty(\Omega;\mathcal M_d)} &\leq M_{a,1}(\mathcal D)\|\theta\|_{\Winf},\\
            \|\alpha(\theta)- \alpha(0)-D\alpha(0)\theta\|_{L^\infty(\Omega;\mathcal M_d)} & \leq M_{a,2}(\mathcal D)\|\theta\|_{\Winf}^2.
        \end{aligned}
    \end{equation}

    By a similar argument, the map
    \begin{equation}
        \label{eq:jacob_expr}
        \gamma:\left\{\begin{aligned}
            \Winf&\rightarrow L^\infty(\Omega),\\
            \theta & \mapsto |\det \nabla \Phi_\theta| \e^{-\beta V\circ \Phi_\theta},
        \end{aligned}\right.
    \end{equation}
    is Fréchet differentiable at~$\theta = 0$, with
    \begin{equation}
        D\gamma(0)\theta = \div\left(\theta\e^{-\beta V}\right),
    \end{equation}
    and the estimates
    \begin{equation}
        \label{eq:dgamma_estimates}
        \begin{aligned}
            \|D\gamma(0)\theta\|_{L^\infty(\Omega)} &\leq M_{b,1}(\mathcal D)\|\theta\|_{\Winf},\\
            \|\gamma(\theta)- \gamma(0)-D\gamma(0)\theta\|_{L^\infty(\Omega)} & \leq M_{b,2}(\mathcal D)\|\theta\|_{\Winf}^2,
        \end{aligned}
    \end{equation}
    hold uniformly inside~$\mathcal D$ for some positive constants~$M_{\fb,1}(\mathcal D),M_{\fb,2}(\mathcal D)>0$.

    We now show~\eqref{eq:perturbation_estimates}. It holds
    \begin{equation}
        \begin{aligned}
            \d \fa_0(\theta)(u,u) &= \frac1\beta\int_\Omega \nabla u^\top D\alpha(0)\theta \nabla u,\\
            r_{\fa}(\theta,u) &= \frac1\beta \int_{\Omega} \nabla u^\top\left(\alpha(\theta) - \alpha(0) - D\alpha(0)\theta\right)\nabla u,\\
            \d {\fb}_0(\theta)(u,u) &= \int_\Omega u^2 D\gamma(0)\theta,\\
            r_{\fb}(\theta,u) &= \int_\Omega u^2 \left(\gamma(\theta)-\gamma(0)-D\gamma(0)\theta\right),
        \end{aligned}
    \end{equation}
    whence, uniformly inside~$\mathcal D$:
    \begin{equation}
        \begin{aligned}
            |\d {\fa}_0(\theta)(u,u)| &\leq \frac{M_{{\fa},1}(\mathcal D)}{\beta}\|\theta\|_{\Winf}\|\nabla u\|^2_{L^2(\Omega;\R^d)},\\
            |r_{\fb}(\theta,u)| & \leq\frac{M_{{\fa},2}(\mathcal D)}{\beta}\|\theta\|^2_{\Winf}\|\nabla u\|^2_{L^2(\Omega;\R^d)},\\
            |\d {\fb}_0(\theta)(u,u)| & \leq M_{{\fb},1}(\mathcal D)\|\theta\|_{\Winf}\|u\|^2_{L^2(\Omega)},\\
            |r_{\fb}(\theta,u)|&\leq M_{{\fb},2}(\mathcal D)\|\theta\|_{\Winf}^2\|u\|^2_{L^2(\Omega)}.
        \end{aligned}
    \end{equation}
    Using~\eqref{eq:quadratic_form_lower_bound} and likewise the lower bound in~\eqref{eq:b_bound}, it follows that~\eqref{eq:perturbation_estimates} holds with constants
    \begin{equation}
        \begin{aligned}
            C_{{\fa},1}(\mathcal D) &= \frac{M_{{\fa},1}(\mathcal D)}{\varepsilon_a(\overline{\mathcal D}) m_V(\mathcal D)(1-C_1(\mathcal D)h_0)^2(1-C_3(\mathcal D)h_0)},&C_{{\fb},1}(\mathcal D) &= \frac{M_{{\fb},1}(\mathcal D)}{m_V(\mathcal D)(1-C_3(\mathcal D)h_0)},\\
            C_{{\fa},2}(\mathcal D) &= \frac{M_{{\fa},2}(\mathcal D)}{\varepsilon_a(\overline{\mathcal D}) m_V(\mathcal D)(1-C_1(\mathcal D)h_0)^2(1-C_3(\mathcal D)h_0)},&C_{{\fb},2}(\mathcal D) &= \frac{M_{{\fb},2}(\mathcal D)}{m_V(\mathcal D)(1-C_3(\mathcal D)h_0)}.
        \end{aligned}
    \end{equation}

    At this point, we have obtained the necessary estimates casting the problem in the form treated in~\cite{HR80a,HR80b}, using abstract arguments of perturbation theory. We next largely follow the approach of these works, but nevertheless give a full proof below, not only for the sake of self-completeness but also because we require stronger intermediate regularity results than those obtained in~\cite{HR80a} in order to prove the third item in Theorem~\ref{thm:gateaux_differentiability}.
    \paragraph{Continuous Fréchet differentiability of the inverse operator.}
    As previously noted,~$\lambda_k(\Omega_\theta)$ is the reciprocal of the~$k$-th largest eigenvalue of the operator
    $$S(\theta) := A_\theta^{-1}B_\theta.$$
    To obtain the results of Theorem~\ref{thm:gateaux_differentiability}, it is then sufficient to study the regularity of the eigenvalues of~$\theta\mapsto S(\theta)$.
    Assuming these are Gateaux-semi-differentiable, in order to obtain the second item in Theorem~\ref{thm:gateaux_differentiability}, we may indeed write, for~$0\leq \ell<m$:
    \begin{equation}
        \label{eq:lambda_quotient_rule}
    \left.\frac{\d}{\d t}\lambda_{k+\ell}(\Omega_{t\theta})\right|_{t=0^+} = -\lambda_k(\Omega)^2 \left.\frac{\d}{\d t}\frac{1}{\lambda_{k+\ell}(\Omega_{t\theta})}\right|_{t=0^+},
    \end{equation}
    where one recognizes right-Gateaux-derivatives of the eigenvalues of~$S$ at~$0$ on the right-hand side of this equality. A similar observation holds for Fréchet-differentiability.

    The first step is to show that~$\theta\mapsto A_\theta^{-1}B_\theta$ is~$\mathcal C^1$ in a~$\Winf$-neighborhood of~$\theta=0$ for the~$L^2_\beta(\Omega)$ operator norm. 
    
    From the estimates~\eqref{eq:perturbation_estimates}, it holds from the representation result~\cite[Lemma VI.3.1]{K13} that there exists~$L^2(\Omega)$-bounded operator-valued maps~$\theta \mapsto D_{A_0}^{(1)}\theta,R_{A_0}(\theta),D_{B_0}^{(1)}\theta,R_{B_0}(\theta) \in \mathcal B(L^2(\Omega))$ such that
    \begin{equation}
        \label{eq:representation_formulae}
        \begin{aligned}
            \d \fa_0(\theta)(u,v) &= \left\langle D_{A_0}^{(1)}\theta A_0^{1/2}u,A_0^{1/2}v\right\rangle,&\d \fb_0(\theta)(u,v) &= \left\langle D_{B_0}^{(1)}\theta B_0^{1/2}u,B_0^{1/2}v\right\rangle,\\
            r_{\fa}(\theta,u,v) &= \left\langle R_{A_0}(\theta)A_0^{1/2}u,A_0^{1/2}v\right\rangle,&r_{\fb}(\theta,u,v) &=  \left\langle R_{B_0}(\theta)B_0^{1/2}u,B_0^{1/2}v\right\rangle,
        \end{aligned}
    \end{equation}
    where~$A_0^{1/2}$ is the positive self-adjoint operator defined on~$\mathcal D(A_0^{1/2}) = H_0^1(\Omega)$ (the form domain of~$A_0$) by functional calculus, such that~$A_0^{1/2}A_0^{1/2} = A_0$ on~$\mathcal D(A_0)$,
    and where the bilinear forms~$r_a(\theta,\cdot,\cdot),\,r_b(\theta,\cdot,\cdot)$ are defined by polarization from the expressions~\eqref{eq:remainders}.
    Moreover, the operators~$D_{A_0}^{(1)},D_{B_0}^{(1)}$ are clearly linear, and the bounds
    \begin{equation}
        \label{eq:perturbation_estimates_ops}
        \begin{aligned}
        \left\|D_{A_0}^{(1)}\theta\right\|_{\mathcal B(L^2(\Omega))} &\leq C_{\fa,1}(\mathcal D) \|\theta\|_{\Winf},&\left\|D_{B_0}^{(1)}\theta\right\|_{\mathcal B(L^2(\Omega))} &\leq C_{\fb,1}(\mathcal D) \|\theta\|_{\Winf},\\
        \left\|R_{A_0}(\theta)\right\|_{\mathcal B(L^2(\Omega))} &\leq C_{\fa,2}(\mathcal D)\|\theta\|^2_{\Winf},&\left\|R_{B_0}(\theta)\right\|_{\mathcal B(L^2(\Omega))} &\leq C_{\fb,2}(\mathcal D)\|\theta\|^2_{\Winf}
        \end{aligned}
    \end{equation}
    are satisfied uniformly inside~$\mathcal D$, with the constants appearing in~\eqref{eq:perturbation_estimates}.

    Let~$f\in L^2(\Omega)$, and set
    \begin{equation}
        \label{eq:inv_A_ansatz}
        u := A_0^{-1/2}(\Id+D_{A_0}^{(1)}+R_{A_0}(\theta))^{-1}A_0^{-1/2}f\in \mathcal D(A_0^{1/2}),
    \end{equation}
    where the inverse is well-defined for~$\|\theta\|_{\Winf}<h_0$ sufficiently small, by the estimates~\eqref{eq:perturbation_estimates_ops}.
    By the representation result~\cite[Theorem VI.3.1]{K13}, it holds for any~$v\in\cD(A_0^{1/2})=H_0^1(\Omega)$, that
    \begin{equation}
        a_\theta(u,v) = \left\langle (\Id + D_{A_0}^{(1)}\theta + R_{A_0}(\theta))A_0^{1/2}u,A_0^{1/2}v\right\rangle_{L^2(\Omega)} = \left\langle A_0^{-1/2}f,A_0^{1/2}v\right\rangle_{L^2(\Omega)} = \left\langle f,v\right\rangle_{L^2(\Omega)}.
    \end{equation}
    By the representation theorem for symmetric positive closed forms~\cite[Theorem VI.2.1]{K13}, it holds~$u\in\cD(A_\theta)$ and~$A_\theta u = f$, whence
    \begin{equation}
        A_\theta^{-1} = A_0^{-1/2}\left(\Id + D_{A_0}^{(1)}\theta+R_{A_0}(\theta)\right)^{-1}A_0^{-1/2},
    \end{equation}
    and writing the Neumann series expansion, it then holds that
    \begin{equation}
        A_\theta^{-1} = A_0^{-1/2}\left(\Id-D_{A_0}^{(1)}\theta\right){A_0}^{-1/2} + \widetilde{R}_{A_0^{-1}}(\theta),
    \end{equation}
    with quadratically bounded remainder:~$\|\widetilde{R}_{A_{0}^{-1}}(\theta)\|_{\mathcal B(L^2(\Omega))} \leq C_{A_0,2}(\mathcal D) \|\theta\|_{\Winf}^2$ uniformly inside~$\mathcal D$ for some constant~$C_{A_0,2}(\mathcal D)>0$ and some operator-valued map~$\widetilde{R}_{A_0^{-1}}:\Winf\to\mathcal B(L^2(\Omega))$.
    Since~$D_{A_0}^{(1)}$ is controlled uniformly inside~$\mathcal D$ in the~$L^2(\Omega)$-operator norm, we only need to check that this is also the case for~$A_0^{-1/2}$, but this follows from the lower bound~\eqref{eq:quadratic_form_lower_bound}, which is uniform inside~$\mathcal D$.

    Therefore,~$\theta\mapsto A_\theta^{-1}$ is Fréchet-differentiable at~$\theta=0$, with
    \[D A_0^{-1}\theta = - A_0^{-1/2}D_{A_0}^{(1)}\theta A_0^{-1/2}.\]
    Since~$\theta\mapsto B_\theta$ is Fréchet-differentiable at~$\theta=0$ from~\eqref{eq:perturbation_estimates_ops}, the inverse operator~$\theta\mapsto S(\theta)$ is also Fréchet-differentiable at~$\theta=0$, with
    \begin{equation}
        \label{eq:derivative_solution_operator}
        D S(0)\theta = -A_0^{-1/2}D_{A_0}^{(1)}\theta A_0^{-1/2}B_0 + A_0^{-1}B_0^{1/2}D_{B_0}^{(1)}\theta B_0^{1/2},
    \end{equation}
    and with quadratically bounded remainder:
    \begin{equation}
        \label{eq:solution_expansion}
        S(\theta) = S(0) + D S(0)\theta + R_{S(0)}(\theta),\qquad \|R_{S(0)}(\theta)\|_{\mathcal B(L^2(\Omega))} \leq C_{S_0,2}(\mathcal D)\|\theta\|^2_{\Winf}
    \end{equation}
    uniformly inside~$\mathcal D$ for some~$C_{S_0,2}(\mathcal D)>0$ and some operator-valued map~$R_{S(0)}:\Winf\to\mathcal B(L^2(\Omega))$.

    At this point, we note that, due to the uniformity of the estimates~\eqref{eq:quadratic_form_lower_bound} and~\eqref{eq:taylor_expansions} inside~$\mathcal D$, the same analysis shows that, for~$\|\theta\|_{\Winf}$ sufficiently small,
    $S$ is Fréchet differentiable at~$\theta$ and the expansion
    \begin{equation}
        \label{eq:solution_expansion_nonzero}
        S(\theta+\delta\theta) = S(\theta) + D S(\theta)\delta\theta + R_{S(\theta)}(\delta\theta),\qquad \|R_{S(\theta)}(\delta\theta)\|_{\mathcal B(L^2(\Omega))} \leq C_{S_0,2}(\mathcal D)\|\delta\theta\|^2_{\Winf},
    \end{equation}
    is valid, with crucially the same constant as in~\eqref{eq:solution_expansion}, and some other operator-valued map~$R_{S(\theta)}:\Winf\to\mathcal B(L^2(\Omega))$.
    
    Indeed, the previous argument applies upon replacing~$\Omega$ by the bounded Lipschitz domain~$\Omega_\theta \subset \mathcal D$ as long as~$\|\Phi_{\delta\theta}\circ \Phi_\theta - \Id\|_{\Winf}<h_0$. A simple computation shows that~$\|\Phi_{\delta\theta}\circ \Phi_\theta - \Id\|_{\Winf}\leq \|\theta\|_{\Winf}+\|\delta\theta\|_{\Winf}+\|\theta\|_{\Winf}\|\delta\theta\|_{\Winf}$, so that taking~$\|\theta\|_{\Winf},\|\delta\theta\|_{\Winf}<\sqrt{1+h_0}-1$ suffices.

    Therefore, upon further reducing~$h_0$, we assume from now on that~$S$ is Fréchet-differentiable inside~$B_{\Winf}(0,h_0)$.
    In fact, the uniformity with respect to~$\theta$ of the remainder in~\eqref{eq:solution_expansion_nonzero} implies that~$S$ is~$\mathcal C^1$ in a $\Winf$-neighborhood~$\mathcal N_S$ of~$\theta=0$ for the~$L^2(\Omega)$-operator norm, which we now show.

    Let~$\theta_1,\theta_2\in\Winf(\R^d)$ be sufficiently small, and write~$\theta_2 = \theta_1 + \delta\theta_1$. Additionally, take~$w\in \Winf$ with~$\|w\|_{\Winf}=1$, and write~$\delta\theta_2 = \|\delta\theta_1\|_{\Winf}w$.
    By the expansion~\eqref{eq:solution_expansion_nonzero}, it holds
    \begin{equation}
        \label{eq:solution_expansion_2pts}
        \begin{aligned}
            S(\theta_2+\delta\theta_2) &= S(\theta_2) + DS(\theta_2)\delta\theta_2 + R_{S(\theta_2)}(\delta\theta_2),\\
            S(\theta_2+\delta\theta_2) &= S(\theta_1) + DS(\theta_1)(\delta\theta_1+\delta\theta_2) + R_{S(\theta_1)}(\delta\theta_1+\delta\theta_2).
        \end{aligned}
    \end{equation}
    Substituting the further expansion
    \begin{equation}
        S(\theta_2)=S(\theta_1)+DS(\theta_1)\delta\theta_1 + R_{S(\theta_1)}(\delta\theta_1)
    \end{equation}
    in the first line of~\eqref{eq:solution_expansion_2pts}, we find after simplification
    \begin{equation}
    \left[DS(\theta_2)-DS(\theta_1)\right]\delta\theta_2 = R_{S(\theta_1)}(\delta\theta_1+\delta\theta_2)-R_{S(\theta_1)}(\delta\theta_1)-R_{S(\theta_2)}(\delta\theta_2).
    \end{equation}
    Estimating the~$L^2(\Omega)$-operator norm using~\eqref{eq:solution_expansion_nonzero}, we find
    \begin{equation}
        \begin{aligned}
            \|\left[DS(\theta_2)-DS(\theta_1)\right]w\|_{L^2(\Omega)}\|\delta\theta_1\|_{\Winf} &\leq C_{S_0,2}\left(\|\delta\theta_1+\delta\theta_2\|_{\Winf}^2+\|\delta\theta_1\|_{\Winf}^2+\|\delta\theta_2\|_{\Winf}^2\right)\\
        &\leq 6C_{S_0,2}\|\delta\theta_1\|_{\Winf}^2,
        \end{aligned}
    \end{equation}
    since~$\|\delta\theta_1\|_{\Winf}=\|\delta\theta_2\|_{\Winf}$ and therefore~$\|\delta\theta_1+\delta\theta_2\|_{\Winf}\leq 2\|\delta\theta_1\|_{\Winf}$. Therefore,
    \begin{equation}
        \|[DS(\theta_2)-DS(\theta_1)]w\|_{L^2(\Omega)} \leq 6C_{S_0,2}\|\theta_2-\theta_1\|_{\Winf},
    \end{equation}
    which, upon taking the supremum over~$\{w\in\Winf,\|w\|_{\Winf}=1\}$, shows that~$DS$ is Lipschitz (and in particular continuous) for the~$L^2(\Omega)$-operator norm in a~$\Winf$-neighborhood of~$\theta=0$.

    We now show the first item in Theorem~\ref{thm:gateaux_differentiability}. We have already proved that~$\theta\mapsto S(\theta)$ and~$\theta\mapsto B(\theta)$ are~$C^1$ in a~$\Winf$-neighborhood of~$\theta=0$. From the bounds~\eqref{eq:b_bound}, the same regularity holds for~$\theta\mapsto B(\theta)^{\pm1/2}$. Therefore, the map
    \begin{equation}
        \left\{\begin{aligned}\Winf(\R^d)&\to \mathcal K_{\mathrm{sa}}(L^2(\Omega))\\
            \theta&\mapsto B_\theta^{1/2}S(\theta)B_\theta^{-1/2}\end{aligned}\right.
    \end{equation}
    is~$\mathcal C^1$, at~$\theta=0$, hence Lipschitz on some neighborhood~$\widetilde{\mathcal N}_{S}\subset \Winf$ of $0$, where~$\mathcal K_{\mathrm{sa}}(L^2(\beta))$ denotes the subspace of compact self-adjoint operators in~$\mathcal B(L^2(\Omega))$.

    A well-known consequence of the Courant--Fischer principle (the so-called Weyl perturbation inequality, see for example~\cite[Section 1.3.3]{T12} for the analogous case of Hermitian matrices) implies that, given a Hilbert space~$\mathcal H$, for any~$j\geq 1$, the eigenvalue map
    \begin{equation}
        \left\{\begin{aligned} \mathcal K_{\mathrm{sa}}(\mathcal H)&\to \R\\
            A&\mapsto \mu_j(A)\end{aligned}\right.
    \end{equation}
    is $1$-Lipschitz in the~$\mathcal H$-operator norm, where~$\mu_k(A)$ denotes the~$k$-th largest eigenvalue of~$A$ (counted with multiplicity). By composition, for any~$j\geq 1$, the map~$\theta\mapsto \mu_j(B_\theta^{1/2}S(\theta)B_\theta^{-1/2}) = \mu_j(S(\theta))$ is also Lipschitz on~$\widetilde{\mathcal N}_S$.
    The claim then easily follows since~$\theta\mapsto\left(\lambda_{k+\ell}(\Omega_\theta)\right)_{0\leq \ell <m} = \left(1/\mu_{k+\ell}(S(\theta))\right)_{0\leq \ell <m}$ with~$\mu_{k+\ell}(S(\theta))>0$ for any~$\theta$, and the map~$x\mapsto \left(1/x_i\right)_{1\leq i\leq m}$ is locally Lipschitz on $\left(0,+\infty\right)^m$.
    
    From now on we view~$S(\theta)$, for~$\theta\in\Winf$ sufficiently small, as an operator on~$L_\beta^2(\Omega)$, stressing that, in this setting,~$S(0)$ is a compact self-adjoint operator, although for~$\theta\neq 0$, these operators are generally non-symmetric, but still compact with real spectrum (since the~$S(\theta)$ are conjugate to self-adjoint operators on~$L^2(\Omega)$).
    
    \paragraph{Finite-dimensional reduction around eigenvalues clusters.}
    We recall that, by assumption,~$\lambda_k(\Omega)$ has multiplicity~$m\geq 1$.
    By compactness of the family~$S(\theta)$ and the continuity of its eigenvalues, there exists a complex, positively oriented contour~$\Gamma:[0,1]\to \mathbb C$ separating~$1/\lambda_k(\Omega)$ from the eigenvalues of~$S(0)$ different from~$1/\lambda_k(\Omega)$, and~$h_0>0$ such that, for any~$\theta\in B_{\Winf}(0,h_0)$,~$S(\theta)$ has exactly~$m$ eigenvalues inside~$\Gamma$, counted with multiplicity.
    We denote the Riesz projector by
    \begin{equation}
        \label{eq:spectral_projector}
        \Pi_\theta = -\frac1{2{\mathrm{i}}\pi}\int_{\Gamma} \mathcal R_{\zeta}(\theta)\,\d\zeta,
    \end{equation}
     where we define the resolvent of~$S(\theta)$ as
    \begin{equation}
        \label{eq:resolvent}
        \mathcal R_\zeta(\theta) = \left(S(\theta) - \zeta\right)^{-1} = B_\theta^{-1}(A_\theta^{-1}-\zeta B_\theta^{-1})^{-1}.
    \end{equation}
    Note that~$\Pi_\theta$ is a projector onto the~$S(\theta)$-invariant subspace
    $$\mathrm{Span}\left\{u_{k+\ell}(\Omega_\theta),0\leq \ell < m\right\},$$
    and is~$L^2_\beta(\Omega)$-orthogonal when~$\theta=0$.
    We next show that~$\theta\mapsto\Pi_\theta$ is~$\mathcal C^1$ in a~$\Winf$-neighborhood of~$\theta=0$ for the~$L^2_\beta(\Omega)$-operator norm.

    By continuity of the eigenvalues of~$S(\theta)$ with respect to~$\theta$, we may choose~$h_0$ sufficiently small and~$C(\mathcal D)>0$ so that, uniformly inside~$\mathcal D$ and for all~$\zeta\in \Gamma$, it holds
    \begin{equation}
        \label{eq:resolvent_uniform_bound}
        \|\mathcal R_\zeta(\theta)\|_{\mathcal B(L^2(\Omega))} \leq C(\mathcal D).
    \end{equation}
    We furthermore assume~$h_0$ to be sufficiently small so that the expansion~\eqref{eq:solution_expansion_nonzero} holds.
    Let~$\theta,\delta\theta \in B_{\Winf}(0,h_0)$.
    The second resolvent identity states that, for any~$\zeta\in \Gamma$,
    \begin{equation}
        \mathcal R_\zeta(\theta+\delta\theta) - \mathcal R_\zeta(\theta) = \mathcal R_\zeta(\theta+\delta\theta)(S(\theta)-S(\theta+\delta\theta))\mathcal R_\zeta(\theta),
    \end{equation}
    so that, rearranging, we obtain
    \begin{equation}
        \mathcal R_\zeta(\theta) = \mathcal R_\zeta(\theta+\delta\theta)\left(\Id +\left[S(\theta+\delta\theta) - S(\theta)\right]\mathcal R_\zeta(\theta)\right),
    \end{equation}
    whence, for~$\left\|S(\theta+\delta\theta)-S(\theta)\right\|_{\mathcal B(L^2(\Omega))}\leq \left\|\mathcal R_\zeta(\theta)\right\|_{\mathcal B(L^2(\Omega))}^{-1}$, we have the expression
    \begin{equation}
        \label{eq:neumann_series}
        \mathcal R_\zeta(\theta+\delta\theta) = \mathcal R_\zeta(\theta)\left(\sum_{k=0}^\infty (-1)^k\left[\left(S(\theta+\delta\theta) - S(\theta)\right)\mathcal R_\zeta(\theta)\right]^k \right).
    \end{equation}
    Then, by the expansion~\eqref{eq:solution_expansion_nonzero} and the uniform bound~\eqref{eq:resolvent_uniform_bound}, one can find~$h_0,K(\mathcal D)>0$ such that, uniformly inside~$\mathcal D$, and for any~$\zeta\in\Gamma,\,\delta\theta\in\Winf$ with~$\|\theta+\delta\theta\|_{\Winf}<h_0$, it holds
     \begin{equation}
        \label{eq:neumann_series_expansion}
        \mathcal R_\zeta(\theta+\delta\theta) = \mathcal R_\zeta(\theta) - \mathcal R_\zeta(\theta) D S(\theta)\delta\theta\mathcal R_\zeta(\theta) + Q(\theta,\delta\theta,\zeta),\quad\|Q(\theta,\delta\theta,\zeta)\|\leq K(\mathcal D)\|\delta\theta\|_{\Winf}^2.
    \end{equation}
    Therefore~$\mathcal R_\zeta$ is Fréchet-differentiable at~$\theta$, and its Fréchet derivative is given by~$D \mathcal R_\zeta(\theta)\delta\theta = -\mathcal R_\zeta(\theta)DS(\theta)\delta\theta\mathcal R_\zeta(\theta)$, which is continuous in~$\theta$ in a~$\Winf$-neighborhood of~$\theta=0$ for the~$L^2_\beta(\Omega)$-operator norm, owing to the~$\mathcal C^1$-regularity of $S$ and the continuity of~$\mathcal R_\zeta$.
    By dominated convergence in~\eqref{eq:spectral_projector} using the bound~\eqref{eq:resolvent_uniform_bound} and the fact that~$D S(\theta)$ is bounded in the~$L^2_\beta(\Omega)$-operator norm uniformly inside~$\mathcal D$, it follows that~$\theta\mapsto \Pi_\theta$ is also~$\mathcal C^1$ in the~$L^2_\beta(\Omega)$-operator norm in a~$\Winf$-neighborhood of~$\theta=0$. We denote its Fréchet derivative by~$\delta\theta\mapsto D \Pi_\theta\delta\theta$.

    The last key step is to connect the invariant~$m$-dimensional subspaces~${\Pi_\theta L^2_\beta(\Omega)}$ and~${\Pi_0 L^2_\beta(\Omega)}$ via a linear isomorphism which is Fréchet differentiable with respect to~$\theta$ at~$\theta=0$.
    This will allow to relate eigenvalue variations of~$S(\theta)$ to those of a conjugated operator~$\widehat S(\theta)$ acting on the fixed $m$-dimensional Hilbert space~${\Pi_0 L^2_\beta(\Omega)}$ on which perturbation results are readily available.
    This follows the general construction discussed in~\cite[Section I.4.6]{K13} for continuous families of projectors, to which we refer for additional details.
    Introduce the bounded operators
    \begin{equation}
        Q(\theta) = (\Pi_\theta-\Pi_0)^2,\qquad U(\theta) = \left(\Pi_\theta\Pi_0+(\Id-\Pi_\theta)(\Id-\Pi_0)\right)(\Id-Q(\theta))^{-1/2},
    \end{equation}
    where~$(\Id-Q(\theta))^{-1/2}$ can be defined via the following expansion for~$\|\Pi_\theta-\Pi_0\|<1$:
    $$(\Id-Q(\theta))^{-1/2} = \sum_{k=0}^\infty {{-1/2}\choose{k}} (-Q(\theta))^k.$$

    Note that~$U(0)=\Id$. The definition of~$U(\theta)$ is motivated by the observation
    \begin{equation}
        C_{\theta,0}D_{0,\theta} = C_{\theta,0}D_{0,\theta} = \Id-Q(\theta),
    \end{equation}
    where
    $$C_{\theta,0} = \Pi_0\Pi_\theta + (\Id-\Pi_0)(\Id-\Pi_\theta):\Pi_\theta L^2_\beta(\Omega)\to \Pi_0 L^2_\beta(\Omega),$$
    and~$D_{0,\theta}$ is the analogous operator obtained by exchanging the roles of~$\theta$ and~$0$ in the definition of~$C_{\theta,0}$.
    Simple computations (see also the discussion in~\cite[Section I.4.6]{K13}) then show that~$U(\theta):\Pi_\theta L_\beta^2(\Omega)\to\Pi_0 L^2_\beta(\Omega)$ is an isomorphism,
    and that~$\Pi_\theta,\Pi_0$ are conjugated via
    \begin{equation}
        \label{eq:derivative_conjugation}
        \Pi_\theta = U(\theta)^{-1}\Pi_0 U(\theta),\qquad\text{where }U(\theta)^{-1} = \left(\Pi_0\Pi_\theta+(\Id-\Pi_0)(\Id-\Pi_\theta)\right)(\Id-Q(\theta))^{-1/2}.
    \end{equation}
    Setting
    \begin{equation}
        \label{eq:conjugation}
        \widehat{S}(\theta)=U(\theta) S(\theta)U(\theta)^{-1},
    \end{equation}
    it holds, since~$\Pi_\theta$ commutes with~$S(\theta)$ and~$U(\theta)$ is bijective, that
    \begin{equation}
        \widehat S(\theta){\Pi_0 L^2_\beta(\Omega)} \subset {\Pi_0 L^2_\beta(\Omega)},
    \end{equation}
    so that~$\widehat S(\theta)|_{{\Pi_0 L^2_\beta(\Omega)}}$ is a well-defined linear map. The bounded operator~$S(\theta)|_{{\Pi_\theta L^2_\beta(\Omega)}}$ is diagonalizable, as it is conjugate to the operator
    $$\left.B_\theta^{1/2}S(\theta)B_\theta^{-1/2}\right|_{B_\theta^{1/2}{\Pi_\theta L^2_\beta(\Omega)}},$$
    which is self-adjoint for the~$L^2(\Omega)$ inner product. Therefore, the conjugate operator~$\widehat S(\theta)|_{{\Pi_0 L^2_\beta(\Omega)}}$ is also diagonalizable, and the spectra of these two operators are identical, counting with multiplicity.

    Moreover, due to the~$\mathcal C^1$ regularity of~$\theta\mapsto \Pi_\theta$, the map~$U(\theta)$ is also~$\mathcal C^1$ in a~$\Winf$-neighborhood of~$\theta=0$, and since~$DQ(0)= 0$, it also holds~$\left.D(\Id-Q(\theta))^{-1/2}\right|_{\theta=0}=0$, whence
    \[D U(0)\theta = \left(D\Pi_0\theta\right)\Pi_0 - \left(D\Pi_0\theta\right)(\Id - \Pi_0) = 2\left(D\Pi_0\theta\right)\Pi_0 - D\Pi_0\theta = 0,\]
    since the last expression is the Fréchet differential of~$\Pi_\theta^2-\Pi_\theta=0$ at~$\theta=0$. Similarly,~$D U^{-1}(0)\theta=0$, so that
    $\theta\mapsto \widehat S(\theta)$ is~$\mathcal C^1$ as a map~$\Winf(\R^d;\R^d)\to\mathcal B({\Pi_0 L^2_\beta(\Omega)})$ in a~$\Winf$-neighborhood of~$\theta=0$, with
    \begin{equation}
        D\widehat{S}(0)\theta = D S(0)\theta.
    \end{equation}
    We stress that~$\Pi_0 L^2_\beta(\Omega)$ is a~$m$-dimensional vector space, on which~$\widehat S(\theta)$ defines a $\mathcal C^1$-family of diagonalizable endomorphisms.
    In particular, for fixed~$\theta\in \Winf$, there exists~$t_\theta>0$ such that the map~$t\mapsto \widehat S(t\theta)|_{{\Pi_0 L^2_\beta(\Omega)}}$ is differentiable on~$(-t_\theta,t_\theta)$, so that from finite-dimensional perturbation theory (see~\cite[Section II.5.4, Theorem 5.4 and Remark 5.5 and Section II.5.5, Theorem 5.6]{K13}), and since~$1/\lambda_k(\Omega)$ is semisimple in the sense of~\cite[Section I.4]{K13} (as~$\widehat S(0)$ is diagonalizable on~$\Pi_0 L_\beta^2(\Omega)$), there exist~$m$ maps~$\left(\mu_\ell\right)_{1\leq \ell\leq m}$, differentiable on~$(-t_\theta,t_\theta)$ and satisfying~\eqref{eq:multiset}, such that
    \begin{equation}
        \label{eq:gateaux_derivative_eigenvalue}
        \forall\,1\leq \ell\leq m,\qquad \mu_\ell'(0) \in \mathrm{Spec}\left(\Pi_0 D S(0) \theta \Pi_0\right),
    \end{equation}
    where~$\Pi_0 D S(0)\theta\Pi_0$ is viewed as a linear map on~${\Pi_0 L^2_\beta(\Omega)}$.
    \paragraph{Computation of the Gateaux derivatives.}
    It remains to show the formula~\eqref{eq:derivative_volume_expression}. This reduces to computing the components of the matrix representation of~$\Pi_0 D S(0) \theta \Pi_0$ for the~$L^2_\beta(\Omega)$ scalar product, in the given~$L^2_\beta(\Omega)$-orthonormal basis~$\left\{u_k^{(\ell)}(\Omega),\, 1\leq \ell \leq m\right\}$.
    For convenience, we denote by
    \[\forall 1\leq \ell\leq m,\qquad u_\ell = u_k^{(\ell)}(\Omega),\qquad\text{ and}\qquad \lambda = \lambda_k(\Omega).\]
    Recall the relation~\eqref{eq:lambda_quotient_rule}. Setting
    \begin{equation}
        M_{ij}(\theta) = -\lambda^2\left\langle \Pi_0 D S(0) \theta \Pi_0 u_i,u_j\right\rangle_{L^2_\beta(\Omega)},
    \end{equation}
    and using~\eqref{eq:derivative_solution_operator}, we find, since~$\Pi_0$ is~$L^2_\beta(\Omega)$-self-adjoint and~$\Pi_0 u_i = u_i$ for each~$1\leq i\leq m$,
    \begin{equation}
        \begin{aligned}
            \left\langle \Pi_0 D S(0) \theta \Pi_0 u_i,u_j\right\rangle_{L^2_\beta(\Omega)} &= \left\langle D S(0) \theta \Pi_0 u_i,\Pi_0 u_j\right\rangle_{L^2_\beta(\Omega)}\\
            &=\left\langle B_0 D S(0) \theta u_i,u_j\right\rangle_{L^2(\Omega)}\\
             &= \left\langle B_0\left( -A_0^{-1/2} D_{A_0}^{(1)}\theta A_0^{-1/2}B_0 + A_0^{-1}B_0^{1/2}D_{B_0}^{(1)}\theta B_0^{1/2}\right)u_i, u_j\right\rangle_{L^2(\Omega)}\\
            &= \left\langle \left(-A_0^{1/2}D_{A_0}^{(1)}\theta A_0^{1/2}A_0^{-1}B_0 + B_0^{1/2}D_{B_0}^{(1)}\theta B_0^{1/2}\right) u_i,A_0^{-1} B_0 u_j\right\rangle_{L^2(\Omega)}\\
            &= -\lambda^{-2}\left\langle A_0^{1/2}D_{A_0}^{(1)}\theta A_0^{1/2}u_i,u_j\right\rangle_{L^2(\Omega)} + \lambda^{-1}\left\langle B_0^{1/2}D_{B_0}^{(1)}\theta B_0^{1/2} u_i,u_j\right\rangle_{L^2(\Omega)},
        \end{aligned}
    \end{equation}
    taking adjoints of the~$L^2(\Omega)$-self-adjoint operators~$A_0^{-1},B_0$ in the fourth line, and using the eigenrelation~$A_0^{-1}B_0 u_i = u_i/\lambda$ for all~$1\leq i\leq m$ in the last line. It follows that
    \begin{equation}
        \begin{aligned}
            M_{ij}(\theta) &= \left\langle A_0^{1/2}D_{A_0}^{(1)}\theta A_0^{1/2} u_i,u_j\right\rangle_{L^2(\Omega)} - \lambda\left\langle B_0^{1/2}D_{B_0}^{(1)}\theta B_0^{1/2} u_i,u_j\right\rangle_{L^2(\Omega)}\\
            &= \d \fa_0(\theta)(u_i,u_j) - \lambda\d \fb_0(\theta)(u_i,u_j),
        \end{aligned}
    \end{equation}
    where we used the representation formulas~\eqref{eq:representation_formulae}. Substituting in the expressions for the first-order perturbations~\eqref{eq:perturbations} finally yields~\eqref{eq:derivative_volume_expression}.
    \paragraph{Fréchet differentiability for simple eigenvalues.}
    We now assume that~$\lambda_k(\Omega)$ is a simple eigenvalue.
    For~$\|\theta\|_{\Winf}$ sufficiently small, it holds from the conjugation~\eqref{eq:derivative_conjugation} that~$\rank\Pi_\theta=1$, and
    \begin{equation}
        \frac1{\lambda_k(\Omega_\theta)} = \left\langle \widehat{S}(\theta)u_k(\Omega),u_k(\Omega)\right\rangle_{L^2_\beta(\Omega)}.
    \end{equation}
    Since~$\theta\mapsto \widehat{S}(\theta)$ is~$\mathcal C^1$ in a~$\Winf$-neighborhood of~$\theta=0$ and~$\lambda_k(\Omega)>0$, the map~$\theta\mapsto\lambda_k(\Omega_\theta)$ is also~$\mathcal C^1$ in a~$\Winf$-neighborhood of~$\theta=0$ by the chain rule.
    A closed-form for the Fréchet derivative~$D\lambda_k(\Omega_0)\theta$ at~$\theta=0$ is then given by~$M_{11}(\theta)$, or by the boundary form
    \begin{equation}
        D\lambda_k(\Omega_0)\theta = -\frac1\beta\int_{\partial\Omega} \left(\frac{\partial u_k(\Omega)}{\partial\n}\right)^2\n^\top a \n \theta^\top \n\,\e^{-\beta V}
    \end{equation}
    in the case~$\partial\Omega$ is~$\mathcal C^{1,1}$ or if~$\Omega$ is convex, as in Corollary~\ref{cor:boundary_expression}.
\end{proof}
\section{The Parallel Replica algorithm and its efficiency}
\label{sec:parrep}

In this appendix, we motivate the shape-optimization objective~\ref{eq:shape_optimization_problem} by discussing its relevance to a class of accelerated MD methods, the so-called Parallel Replica class of algorithms.

    The maximization of~\eqref{eq:separation_of_timescales} is motivated by algorithms in accelerated molecular dynamics, where the separation of timescales is key in ensuring the efficiency of the Parallel Replica method (ParRep); see~\cite[Section 6.2]{SL13} or~\cite[Section 2.7]{PUV15}, where the authors already discuss the influence of the domain definition on the metric~\eqref{eq:separation_of_timescales}.
    In this context, the quantity defined in~\eqref{eq:separation_of_timescales} is called the scalability metric, and is directly related to the efficiency of ParRep~\cite{V98}.
    While many ParRep-like methods have been proposed (see for instance~\cite{USV07,BLS15,A19,PCWKV16}), we present in this section one of the simplest versions, for which the objective of maximizing~\eqref{eq:separation_of_timescales} with respect to~$\Omega$ is most easily motivated.
    
    At its core, ParRep provides a way, given a metastable domain~$\Omega\subset\R^d$, to trade some details of the dynamics inside~$\Omega$ against a kinetically correct sample of the exit from~$\Omega$ (in the sense that both the exit time and the exit point are unbiased), coming at a lower cost in wall-clock time, using parallel computing resources.
    Given a good coverage of the configuration space by a set of good metastable states~$(\Omega_\alpha)_{\alpha\in I}$, one can then effectively parallelize in time the sampling of a long, spatially coarse-grained trajectory.

    A major advantage of ParRep compared to other accelerated MD methods (see~\cite{V97,SV00}) is that it is largely agnostic to the form of the dynamics and therefore applies to a wide range of Markov processes.
    The theoretical justification of the method, however, requires proving the existence and uniqueness of the QSD, see~\cite{LBLLP12,RLR22} for results on the Langevin dynamics in the overdamped and underdamped settings, respectively.
    
    We now describe the Parallel Replica method. While the original formulation~\cite{V98} of the algorithm used disjoint metastable states, defined as basins of attraction for the steepest descent dynamics on the energy landscape,
    we formulate a variant which is more general, in the sense that it accommodates metastable states which may overlap, and whose union does not necessarily cover the whole configuration space.
    
    We first introduce a number of hyperparameters.
    \noindent
{
\begin{center}
\begin{tabular}{ll}
\toprule
\textbf{Parameter} & \textbf{Description} \\
\midrule
$N_{\mathrm{proc}}\in\N^*$ & the number of replicas,\\
$(\Omega_i)_{i\in I}$,& a set of metastable states, and for each~$i\in I$: \\
$\mathcal C_i\subset \Omega_i,$ & an associated core-set, \\
$t_{\corr}(\Omega_i)>0,$ & a decorrelation timescale, and\\
$t_{\phase}(\Omega_i)>0\,$ & a dephasing timescale. \\
\bottomrule
\end{tabular}

\noindent
Input parameters for Algorithm~\ref{alg:parrep}.
\end{center}
}
We assume that the core-sets are pairwise disjoint:
\begin{equation}
    \label{eq:disjoint_coresets}
    \forall\,i,j\in I,\,i\neq j,\qquad \mathcal C_i \cap \mathcal C_j =\varnothing.
\end{equation}
This condition ensures that there is no ambiguity as to which state is entered in step~{\bf{A}} of the algorithm below.
    \begin{algorithm}[ParRep with rejection and core-sets.]
        \label{alg:parrep}
        The algorithm proceeds by iterating the following steps:
        \begin{enumerate}[A.]
            \item[\bf{A}]{Initialization: run the dynamics until it enters a core-set~$\mathcal C_i$, at time $t_0$, for some $i\in I$. Denote~$\tau =\inf\{t\geq t_0: X_t\not\in\Omega_i\}$ the next exit time from the corresponding state~$\Omega_i$.}
            \item[\bf{B1}]{Decorrelation (successful case): if the dynamics remains for a time $t_{\corr}(\Omega_i)$ inside~$\Omega_i$, it is presumed to be distributed according to the QSD~$\nu^i$ in~$\Omega_i$. This introduces a bias, but which decays quickly with~$t_{\corr}(\Omega_i)$ according to~\eqref{eq:decorrelation_rate}, provided~$\lambda_2(\Omega_i)-\lambda_1(\Omega_i)$ is large.}
            \item[\bf{B2}]{Decorrelation (unsuccessful case): if the dynamics exits at $\tau < t_0 + t_{\corr}(\Omega_i)$, record the exit event $(\tau,X_\tau)$, and proceed from step $\text{\bf A}$.}
            \item[\bf{C}]{Dephasing: Simulate~$N_{\mathrm{proc}}$ independent copies~$\left(X^{(i)}\right)_{1\leq i\leq N_{\mathrm{proc}}}$ of the dynamics starting from $X_0^{(i)}=X_{t_0+t_{\corr}(\Omega_i)}$, for a time~$t_{\phase}(\Omega_i)>0$.}
            \item[\bf{D}]{Conditioning: reject replicas which exited $\Omega_i$ during step $\text{\bf C}$. Denote by~$N \leq N_{\mathrm{proc}}$ the random variable counting the number of surviving replicas. The~$(X_{t_{\phase}(\Omega_i)}^{(i)})_{1\leq i\leq N}$ are now presumed to be i.i.d. according to~$\nu^i$. Again, this is correct up to some bias decaying quickly with~$t_{\phase}(\Omega_i)$.}
            \item[\bf{E}]{Parallel exit: evolve the~$N$ replicas independently until the first exits from~$\Omega_i$, say~$X_{\tau^{(j)}}^{(j)}\not\in\Omega_i$, i.e.~$\tau^{(j)} = \underset{1\leq i\leq N}{\min}\,\tau^{(j)}$. According to the property~\eqref{eq:exit_event}, the equality~$$\left(t_0+t_{\corr}(\Omega_i) + N\left[\tau^{(j)}-t_{\phase}(\Omega_i)\right],X^{(j)}_{\tau^{(j)}}\right)\overset{\mathrm{law}}{=}(\tau,X_{\tau})$$ holds in law.}
            \item[\bf{F}]{Set~$X_{t_0+t_{\corr}(\Omega_i) + N(\tau^{(j)}-t_{\phase}(\Omega_i))} = X_{\tau^{(j)}}^{(j)}$ and proceed from step {$\text{\bf A}$}.}
        \end{enumerate}
    \end{algorithm}
    Let us make a few remarks about Algorithm~\ref{alg:parrep}. Steps~{\bf{C}} and~{\bf{E}} can be run on a parallel computer with~$N_{\mathrm{proc}}$ processors. Assuming synchronized MD engines, these two steps therefore only cost~$t_{\phase}(\Omega_i)$ and $\tau^{(j)}-t_{\phase}(\Omega_i)$ respectively in wall-clock time. Since~$\tau^{(j)}-t_{\phase}(\Omega_i)\sim\mathcal E(N/\lambda_1(\Omega_i))$ conditionally on~$N$, this provides to a large acceleration if~$N$ is large, at the cost of an overhead~$t_{\phase}(\Omega_i)$ in step~{\bf{C}}, which does not correspond to a physical time evolution. 

    Because exit events sampled during step~{\bf{B2}} are driven by the original dynamics, they are unbiased. Therefore, ParRep differs from other accelerated MD methods in that it correctly samples the full distribution of exit events, including those corresponding to short, correlated exit times.
    
    Step~{\bf{B}} can also be performed in parallel to step~{\bf{C}}, and this is often done in practice. In this variant of Algorithm~\ref{alg:parrep}, the replicas are initialized at~$X_0^{(i)}=X_{t_0}$ for~$1\leq i\leq N_{\mathrm{proc}}$, and one usually chooses~$t_{\corr}(\Omega_i)=t_{\phase}(\Omega_i)$. Moreover, in the case~{\bf{B2}}, the exit of the reference dynamics $\Omega_i$ kills the replicas running in step~{\bf{C}}.
    
    The path obtained by concatenating~ the segments
    $$(X_t)_{t_0\leq t< t_0+t_{\corr}(\Omega_i)},\left\{(X_t^{(i)})_{t\leq t_{\phase}(\Omega_i)<\tau^{(j)}},\,1\leq j \leq N\right\}\text{ and }(X_t^{(j)})_{t_{\phase}(\Omega_i)\leq t<\tau^{(j)}}$$
    has, in law, the same length as the path~$(X_t)_{t_0\leq t < \tau}$ under the sequential dynamics (neglecting the bias in steps~{\bf{B}} and~{\bf{C}}). Therefore, Algorithm~\ref{alg:parrep} can be understood as sampling a discontinuous modification to the original dynamics, which jumps $N$ times from~$\nu^i$ to~$\nu^i$ in the quasistationary portion~$(X_t)_{t_0+t_{\corr}(\Omega_i)\leq t < \tau}$ of the trajectory.
    
    The core-set~$\mathcal C_i$ encode how one defines an \textit{entrance} into~$\Omega_i$, while the set~$\Omega_i$ encodes how one defines an \textit{exit} from~$\Omega_i$. We argue that the latter is the most important parameter as it impacts all the steps~{\bf B}--{\bf E}, with~$\mathcal C_i$ is only involved in step~{\bf{A}}. The sets~$\mathcal C_i$ can be defined using physical intuition. In our numerical experiments (see~Section~\ref{subsec:diala} below), we consider two natural definitions of these core-sets, namely small balls around free-energy minima, or the intersection of the associated free-energy basin with the state~$\Omega_i$.
    An outstanding question, which we leave for future work, is whether one can \textit{optimize} the definition of the core-sets~$\mathcal C_i$, given definitions for the states~$\Omega_i$, to make Algorithm~\ref{alg:parrep} efficient. Heuristically, the set~$\bigcup_{i\in I}\mathcal C_i$ should be visited often by the dynamics, and starting from~$\partial\mathcal C_i$, convergence to~$\nu^i$ should be both likely and fast (so as to minimize the time spent in step~{\bf B}). This question is related to the minimization of the pre-exponential factor in the error estimate~\eqref{eq:decorrelation_rate}.

    A pathology may occur in the event no replica survives in step {\bf{C}}. This possibility can be assumed to be rare provided~$\Omega_i$ is locally metastable and~$N_{\mathrm{proc}}$ is large, for reasonable choices of~$\mathcal C_i$.
    Nevertheless, the rejection sampling performed in step~{\bf{D}} can be replaced by a branching mechanism known as the Fleming--Viot process (see Algorithm~\ref{alg:fv} below), which has the advantage of ensuring~$N = N_{\mathrm{proc}}$ replicas survive, at the cost of introducing additional (small) correlations between replicas at the end of step~{\bf{D}}, which therefore induces some bias in step~{\bf{E}}.
\begin{figure}
\begin{tikzpicture}[x=0.75pt,y=0.75pt,yscale=-0.75,xscale=0.75]

\draw  [color={rgb, 255:red, 255; green, 0; blue, 0 }  ,draw opacity=1 ][line width=1.5]  (126.33,93) .. controls (222.33,112) and (329.33,33) .. (386.33,139) .. controls (443.33,245) and (381.33,416) .. (247.33,394) .. controls (113.33,372) and (226.67,297) .. (112,284) .. controls (-2.67,271) and (30.33,74) .. (126.33,93) -- cycle ;
\draw  [color={rgb, 255:red, 0; green, 118; blue, 255 }  ,draw opacity=1 ][line width=1.5]  (381.33,79.33) .. controls (383.33,-23.67) and (674.33,134.33) .. (592.33,231.33) .. controls (510.33,328.33) and (627.33,384) .. (475.33,386.33) .. controls (323.33,388.67) and (344.33,364.33) .. (285,273) .. controls (225.67,181.67) and (379.33,182.33) .. (381.33,79.33) -- cycle ;
\draw  [color={rgb, 255:red, 255; green, 0; blue, 0 }  ,draw opacity=1 ] (202.33,131.67) .. controls (261.67,98.67) and (289.33,81.67) .. (311.33,114.67) .. controls (333.33,147.67) and (334.33,185) .. (289.33,235) .. controls (244.33,285) and (174,276) .. (154,246) .. controls (134,216) and (143,164.67) .. (202.33,131.67) -- cycle ;
\draw  [color={rgb, 255:red, 0; green, 114; blue, 255 }  ,draw opacity=1 ] (434.33,97.33) .. controls (471.33,67.67) and (531.33,101.67) .. (555.33,123.33) .. controls (579.33,145) and (587,166.33) .. (541.33,223.67) .. controls (495.67,281) and (468,252) .. (448,222) .. controls (428,192) and (397.33,127) .. (434.33,97.33) -- cycle ;
\draw  [dash pattern={on 4.5pt off 4.5pt}]  (141.52,213.91) .. controls (218.99,273.34) and (185.5,208.06) .. (218.33,219) .. controls (251.33,230) and (219.33,114) .. (266.33,127) .. controls (313.33,140) and (275.33,162) .. (337.33,191) .. controls (399.02,219.86) and (386.46,294.25) .. (255.32,289.08) ;
\draw [shift={(253.33,289)}, rotate = 2.58] [fill={rgb, 255:red, 0; green, 0; blue, 0 }  ][line width=0.08]  [draw opacity=0] (8.93,-4.29) -- (0,0) -- (8.93,4.29) -- cycle    ;
\draw [shift={(142.67,214.79)}, rotate = 217.67] [fill={rgb, 255:red, 0; green, 0; blue, 0 }  ][line width=0.08]  [draw opacity=0] (8.93,-4.29) -- (0,0) -- (8.93,4.29) -- cycle    ;
\draw    (127.59,155.03) .. controls (79.02,222.27) and (113.67,249.15) .. (179.33,279) .. controls (245.33,309) and (232.33,384) .. (257.33,359) .. controls (282.33,334) and (294.33,187) .. (332.33,243) .. controls (369.19,297.32) and (348.66,317.77) .. (372.95,334.46) ;
\draw [shift={(375.33,336)}, rotate = 211.26] [fill={rgb, 255:red, 0; green, 0; blue, 0 }  ][line width=0.08]  [draw opacity=0] (8.93,-4.29) -- (0,0) -- (8.93,4.29) -- cycle    ;
\draw [shift={(126.87,156.03)}, rotate = 306.14] [fill={rgb, 255:red, 0; green, 0; blue, 0 }  ][line width=0.08]  [draw opacity=0] (8.93,-4.29) -- (0,0) -- (8.93,4.29) -- cycle    ;
\draw  [dash pattern={on 0.84pt off 2.51pt}]  (-22.67,-13) .. controls (17.33,-43) and (88.33,185) .. (140.33,213) ;
\draw  [dash pattern={on 0.84pt off 2.51pt}]  (375.33,336) .. controls (437.7,374.61) and (485.37,355.39) .. (495.05,257.98) ;
\draw [shift={(495.33,255)}, rotate = 95.14] [fill={rgb, 255:red, 0; green, 0; blue, 0 }  ][line width=0.08]  [draw opacity=0] (8.93,-4.29) -- (0,0) -- (8.93,4.29) -- cycle    ;
\draw  [dash pattern={on 4.5pt off 4.5pt}]  (529.82,92.2) .. controls (509.95,132.43) and (483.28,121.98) .. (481.33,160) .. controls (479.33,199) and (505.33,189) .. (495.33,255) ;
\draw [shift={(531.33,89)}, rotate = 114.44] [fill={rgb, 255:red, 0; green, 0; blue, 0 }  ][line width=0.08]  [draw opacity=0] (8.93,-4.29) -- (0,0) -- (8.93,4.29) -- cycle    ;
\draw  [dash pattern={on 0.84pt off 2.51pt}]  (531.33,89) .. controls (571.33,59) and (627.33,87) .. (667.33,57) ;

\draw (110,60) node [anchor=north west,scale=1.2][inner sep=0.75pt]    {$\partial \Omega_1 $};
\draw  (400,20) node [anchor=north west,scale=1.2][inner sep=0.75pt]    {$\partial \Omega_2 $};
\draw (150,120) node [anchor=north west,scale=1.2][inner sep=0.75pt]    {$\partial \mathcal C_1 $};
\draw (430,100) node [anchor=north west,scale=1.2][inner sep=0.75pt]    {$\partial \mathcal C_2 $};

\end{tikzpicture}

\caption{A trajectory sampled using Algorithm~\ref{alg:parrep}. Dotted lines correspond to step~{\bf{A}}, dashed lines to step~{\bf{B}}: a successful decorrelation~{\bf{B1}} in~$\Omega_1$, followed by a failed decorrelation~{\bf{B2}} in~$\Omega_2$. The solid line corresponds to the trajectory~$(X_t^{(j)})_{t_{\phase}(\Omega_1)\leq t <\tau^{(j)}}$ in step~{\bf{E}}. The discontinuity hides a (parallel) time evolution of length~$(N-1)(\tau^{(j)}-t_{\phase}(\Omega_1))$ in step~{\bf{D}}.}

\end{figure}
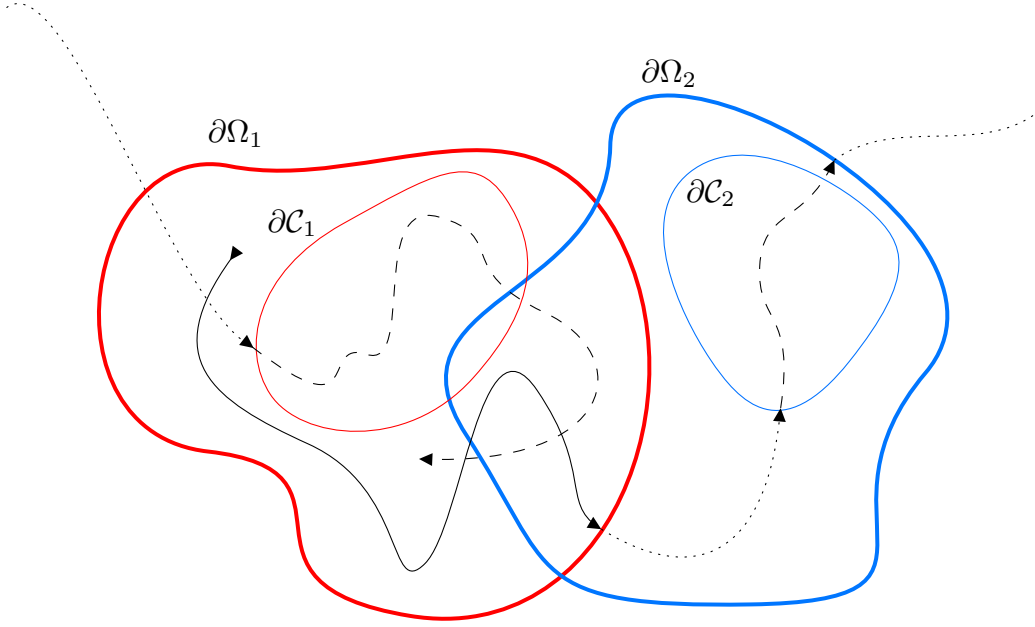

    Crucial hyperparameters are the decorrelation times~$t_{\corr}(\Omega_i)$, and dephasing times~$t_{\phase}(\Omega_i)$ for~$i\in I$. These should be valid, in the sense that the bias introduced in step {\bf B} and the correlations between replicas in step {\bf C} should be small, but setting~$t_{\corr}(\Omega_i),~t_{\phase}(\Omega_i)$ to large values will lead to excessive spending of wall-clock time in these two steps, leading to an overall decrease in the achieved speedup.
    A simple choice is to set
    \begin{equation}
        \label{eq:decorrelation_time}
        t_{\corr}(\Omega_i)=t_{\phase}(\Omega_i) = -\log\varepsilon_{\corr}/(\lambda_2(\Omega_i)-\lambda_1(\Omega_i))
    \end{equation} where~$0<\varepsilon_{\corr}<1$ is a small, domain-independent tolerance parameter.
    This choice, which has already been suggested (see~\cite{SL13,PUV15}), is motivated by taking logarithms in the estimate~\eqref{eq:decorrelation_rate}, and neglecting the contribution~$|\log C(x)|/(\lambda_2(\Omega_i)-\lambda_1(\Omega_i))$ to the bias coming from the prefactor, which depends on the initial condition~$x$.
    The choice~\eqref{eq:decorrelation_time} also motivates the need for quantitative estimates of the spectral gap~$\lambda_2(\Omega_i)-\lambda_1(\Omega_i)$, for which various strategies have been proposed, see~\cite[Section 3.3]{BLS24} for recent results in this direction.

    Let us fix~$i\in I$, and compare the expected wall clock-time to sample a metastable excursion inside~$\Omega_i$ using Algorithm~\ref{alg:parrep} to the expected wall-clock time using a sequential simulation.
    We assume successful decorrelation in step~{\bf B}, rejection sampling in step {\bf C} and the choice~\eqref{eq:decorrelation_time} where~$0<\varepsilon_{\corr}<1$ is sufficiently small so that the bias and correlations introduced in steps {\bf B1} and {\bf C} can be safely neglected.
    By~\eqref{eq:exit_event}, we have~$\E_\nu[N] = \e^{-\lambda_1(\Omega_i) t_{\phase}(\Omega_i)}N_{\mathrm{proc}}$ expected replicas at the end of step {\bf C}, i.i.d. according to the QSD.
    Replacing~$N$ by its expected value under~$\nu$ (and making a so-called annealed approximation in doing so), the expected wall-clock time in step {\bf E} is given by~$\e^{\lambda_1(\Omega_i)t_{\phase}(\Omega_i)}/(\lambda_1(\Omega_i)N_{\mathrm{proc}})$, by standard properties of exponential random variables. Therefore, the combined wall-clock time in steps {\bf B1}--{\bf E} is given by
    \begin{equation}
        \label{eq:wall_clock_time}
        t_{\mathrm{PR}}^{\text{{\bf B1}--{\bf E}}}(\Omega_i) = t_{\corr}(\Omega_i) + t_{\phase}(\Omega_i) + \e^{\lambda_1(\Omega_i)t_{\phase}(\Omega_i)}/(\lambda_1(\Omega_i)N_{\mathrm{proc}}).
    \end{equation}
    The second term accounts for the overhead of simulating~$N$ trajectories in step~B., which can be done in parallel.
    By contrast, the expected wall-clock time to simulate the same process using direct simulation is given by
    \begin{equation}
        \label{eq:wall_clock_time_direct}
        t_{\mathrm{DNS}}^{\text{{\bf B1}--{\bf E}}}(\Omega_i) = t_{\corr}(\Omega_i) + 1/\lambda_1(\Omega_i).
    \end{equation}
    Recalling the definition~\eqref{eq:separation_of_timescales} of~$N^*(\Omega_i)$, substituting in the definition~\eqref{eq:decorrelation_time} and rearranging, we find
    \begin{equation}
        \label{eq:efficiency_ratio}
        \frac{t_{\mathrm{DNS}}^{\text{{\bf B1}--{\bf E}}}(\Omega_i)}{t_{\mathrm{PR}}^{\text{{\bf B1}--{\bf E}}}(\Omega_i)} = \frac{N^*(\Omega_i)-\log\varepsilon_{\corr}}{(N^*(\Omega_i)/N_{\mathrm{proc}})\e^{-(\log\varepsilon_{\corr})/N^*(\Omega_i)}-2\log\varepsilon_\corr}.
    \end{equation}
    One can check that the right-hand side of~\eqref{eq:efficiency_ratio} is an increasing function of~$N^*(\Omega_i)$ for~$N_{\mathrm{proc}}>0$ and~$0<\varepsilon_{\corr}<1$.
    Therefore,~$N^*(\Omega_i)$ should be maximized to maximize the effectiveness of the ParRep algorithm. This objective is only reasonable if the bulk of the simulation time is captured by steps of type~{\bf B1},~{\bf C},~{\bf D} and~{\bf E} in Algorithm~\ref{alg:parrep}.
    That is, trajectories drawn from~\eqref{eq:overdamped_langevin} should spend most of the time inside metastable states, and not in excursions between them. This constraint is related to the intrinsic metastability of the system as a whole: in systems for which a significant portion of time is spent in non-metastable regions of phase space, accelerated MD methods are not expected to be efficient, regardless of the choice of states.

    We stress that the previous discussion is one of a number of possible ways to present the efficiency of ParRep and its variants, but the conclusion is always the same: one should maximize~$N^*(\Omega)$ with respect to~$\Omega$ to obtain maximal benefits from the algorithm inside the metastable state~$\Omega$.
    
    In Figure~\ref{fig:parrep}, we depict the objective~\eqref{eq:efficiency_ratio} as a function of~$N^*(\Omega)$ and the number~$N_{\mathrm{proc}}$ of available processors, as well as the parallel efficiency metric~$t_{\mathrm{DNS}}^{\text{{\bf B1}--{\bf E}}}(\Omega)/\left(N_{\mathrm{proc}}t_{\mathrm{PR}}^{\text{{\bf B1}--{\bf E}}}(\Omega)\right)$.
    This metric measures the wall-clock time speedup per number of processors, and therefore how effective Algorithm~\ref{alg:parrep} utilizes parallel computing resources for the purpose of acceleration. A simple estimate show that, in the regime~$N^*(\Omega)\gg 1$, ~$N_{\mathrm{proc}}$ should be chosen of the order of~$\bigo\left(N^*(\Omega)\right)$ to reach a target parallel efficiency~$0<\alpha<1$ for Algorithm~\ref{alg:parrep} inside~$\Omega$.
    In materials science applications, the separation of timescales~\eqref{eq:separation_of_timescales} is often much larger than the number of available processors, and parallel efficiency upwards of~$\alpha=0.95$ are often reported, see~\cite{PUV15}. The contour line of parallel efficiency~$\alpha=0.5$ is depicted on the right-hand side of Figure~\ref{fig:parrep}.

    \begin{remark}
        It would be somewhat more satisfactory, owing to~\eqref{eq:decorrelation_rate}, to take the spatially-dependent prefactor~$C(x)$ into account in the choice~\eqref{eq:decorrelation_time} of decorrelation time. An unresolved step in this direction is to obtain quantitative estimates of this prefactor, at least in limiting regimes or simple analytic cases. We leave this point for future work.
        At any rate, we expect the corresponding shape-optimization problem to be substantially more difficult to handle.
    \end{remark}

    Another family of methods (see~\cite{HL19}) attempt to estimate~$t_{\corr}(\Omega)+t_{\phase}(\Omega)$ ``on the fly'' using statistical information generated by a Fleming--Viot process in a combined step~({\bf B},{\bf C},{\bf D}).
    In the work~\cite{HL19}, this strategy is implemented using a Gelman--Rubin (non)-convergence diagnostic to estimate the decorrelation time. This opens up the possibility of applying ParRep to situations in which little a priori information is available on the timescales at hand, such as biological systems.
    However, some questions remain on how to optimally balance reliability and efficiency concerns in this context.

    \begin{figure}
    \centering
    \includegraphics[width=0.49\linewidth]{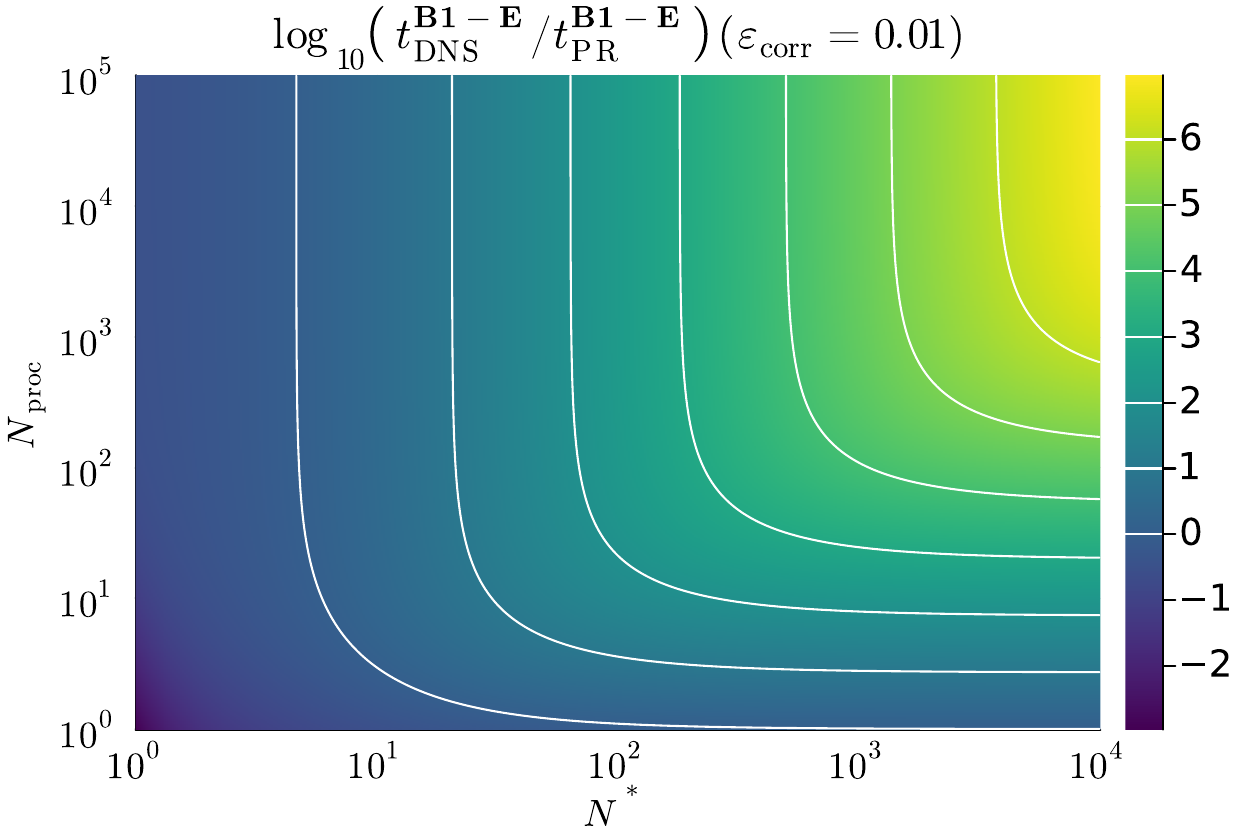}
    \includegraphics[width=0.49\linewidth]{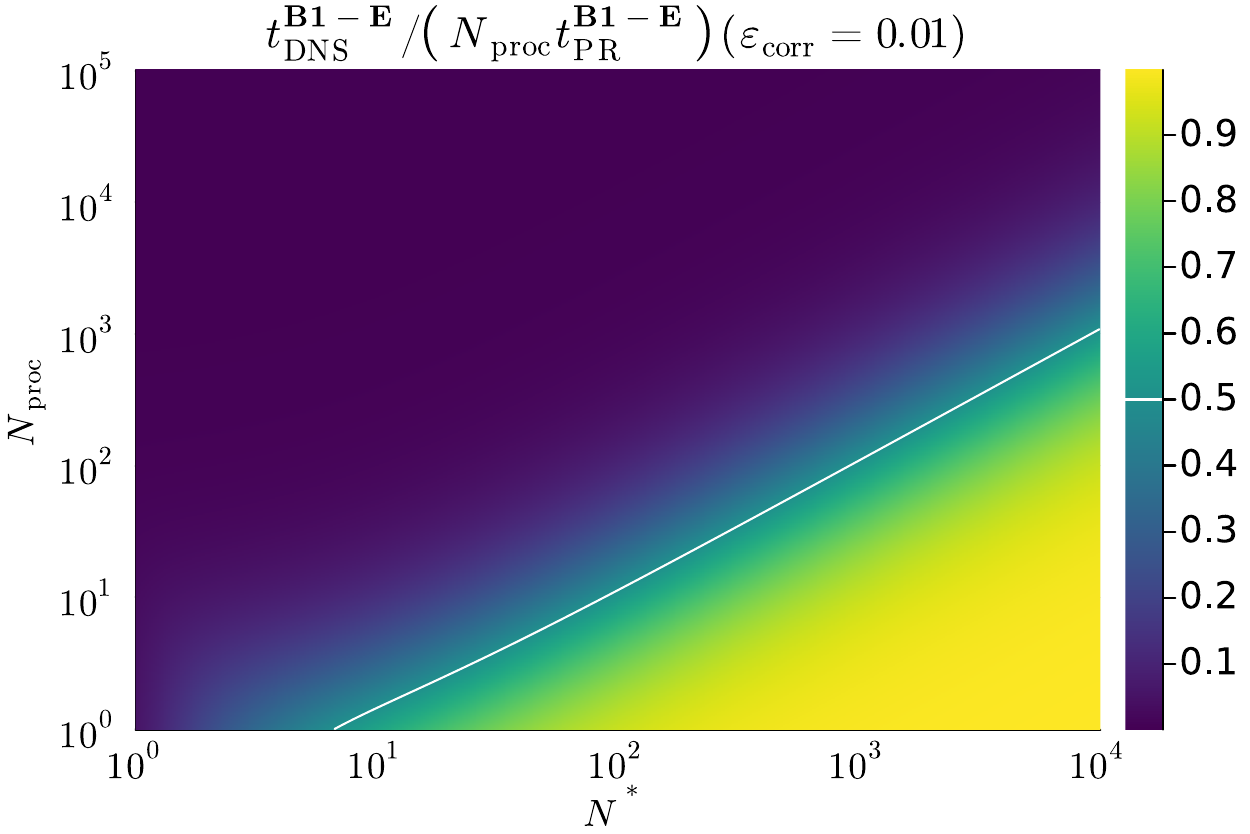}
    \caption{Effect of the separation of timescales~\eqref{eq:separation_of_timescales} and number~$N_{\mathrm{proc}}$ of processors on Algorithm~\ref{alg:parrep}. Left: wall-clock time speedup over direct simulation (Equation~\eqref{eq:efficiency_ratio}, with values on a $\log_{10}$-scale). Contours corresponding to ten-fold decreases in wall-clock time using ParRep are plotted in white, starting from the break-even contour below which direct simulation is faster. Right: parallel efficiency metric, with contour line~$\alpha=0.5$.}
    \label{fig:parrep}
    \end{figure}

\section{Shallow optima and penalized objective functions}
\label{sec:shallow_states}

As noted in Section~\ref{sec:intro} below the optimization problem~\eqref{eq:shape_optimization_problem}, this problem is not well-posed in general, and we seek in this work local maxima of~$N^*$ around candidate energy wells rather than a global optimum. Here, we show numerically that such local maxima may select wells separated by shallow energy barriers, which may or may not correspond to algorithmically relevant metastable states, in the sense of being useful in the Parallel Replica algorithm (see~Appendix~\ref{sec:parrep}).

We consider a one-dimensional toy model decribed by the probability measure corresponding to the following Gaussian mixture
\begin{equation}
    \label{eq:gaussian_mixture}
    \mu_h = \frac{1}{3}\left(\mathcal{N}(-h,1) + \mathcal{N}(0,1) + \mathcal{N}(h,1)\right),
\end{equation}
where~$\mathcal{N}(m,\sigma^2)$ denotes the Gaussian density with mean~$m$ and variance~$\sigma^2$. The associated potential energy reads~$V(x,h) = -\log \mu_h(x)$. For small~$h$, the potential has a single well centered at the origin. As~$h$ increases, three wells emerge, separated by increasingly pronounced energy barriers. Here, we restrict the design space to symmetric domains~$\Omega = (-\alpha, \alpha)$ surrounding the central well, and define~$\alpha^*(h)$ to be a local maximizer of~$\alpha \mapsto N^*((-\alpha,\alpha))$, when a local maximum exists.

In Figure~\ref{fig:shallow_bifurcation}~(top), we plot the potential landscape together with the locally optimal state boundaries~$\pm\alpha^*(h)$ and the loci of saddle points of~$x\mapsto V(x,h)$, as functions of~$h$. We observe that locally optimal domains appear as soon as the energy barrier emerges, even when the barrier height is small compared to the thermal energy scale~$\beta^{-1}=1$.

This numerical observation has algorithmic implications: in a complex free-energy landscape, Algorithm~\ref{alg:ascent} will typically identify local optima around every well. The resulting states, while locally optimal in the sense of~$N^*$, may have mean metastable exit times which are too short to be of physical or algorithmic interest.

A simple remedy is to penalize states with large exit rates, by replacing the objective with
\begin{equation}
    \label{eq:penalized_objective}
    N^*_\tau(\Omega) = N^*(\Omega) - \tau \lambda_1(\Omega),
\end{equation}
where~$\tau > 0$ is a user-defined timescale parameter. Since~$N^*_\tau$ is a smooth function of~$(\lambda_1(\Omega), \lambda_2(\Omega))$, with~$\partial_{\lambda_2} N^*_\tau = 1/\lambda_1 > 0$, the results of Section~\ref{sec:ascent} apply without modification. The parameter~$\tau$ is a timescale which can be chosen using some a priori information about the molecular system of interest.
More general penalized objectives of the form
\begin{equation}
    N^*_\ell(\Omega)=N^*(\Omega) - \ell(\lambda_1(\Omega))
\end{equation}
where~$\ell:\R_+\to\R_+$ is smooth can also be considered to target metastable states within a specific range of metastable exit timescales.

In Figure~\ref{fig:shallow_bifurcation}~(bottom), we plot the analogous bifurcation diagram for the penalized objective~$N^*_\tau$ with~$\tau = 50$. The shallow optima are suppressed: locally optimal domains only appear once the energy barrier is sufficiently high.

\begin{figure}[!h]
    \centering
    \includegraphics[width=0.75\linewidth]{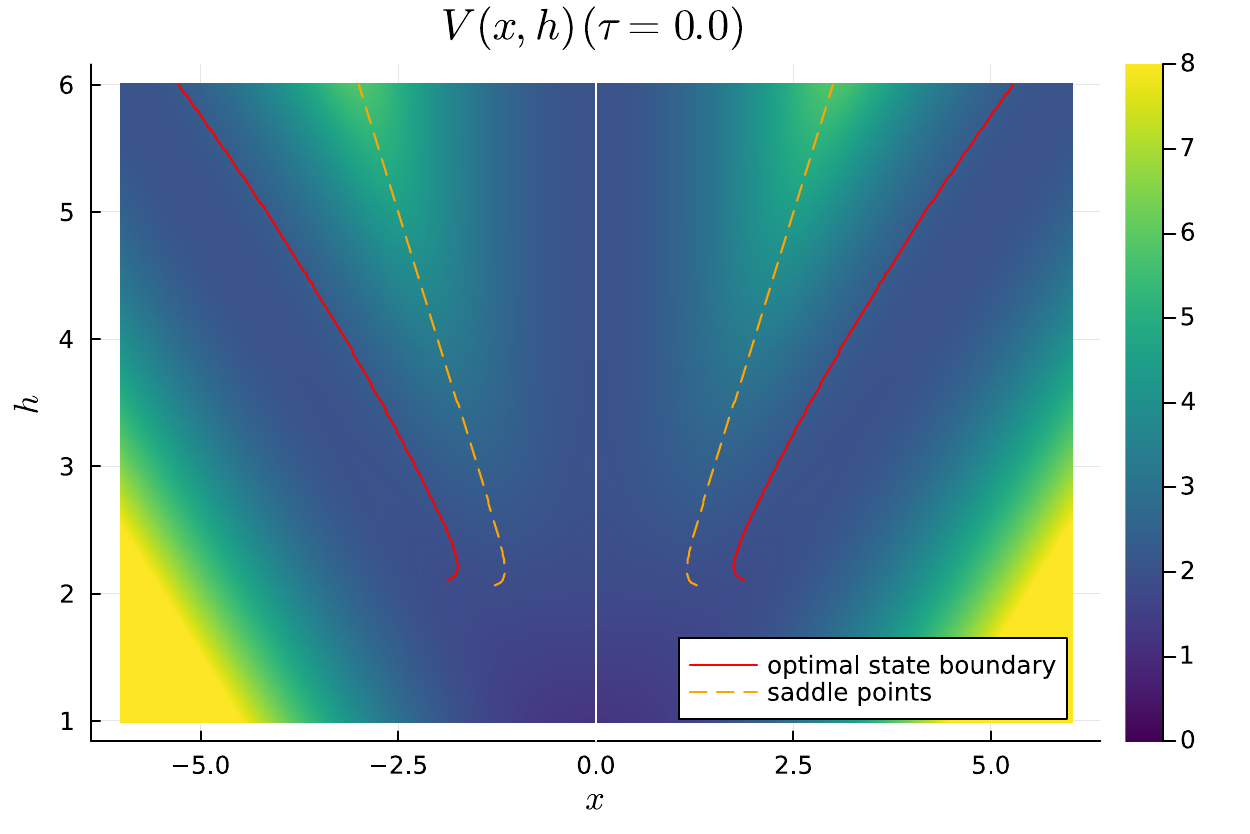}
    \includegraphics[width=0.75\linewidth]{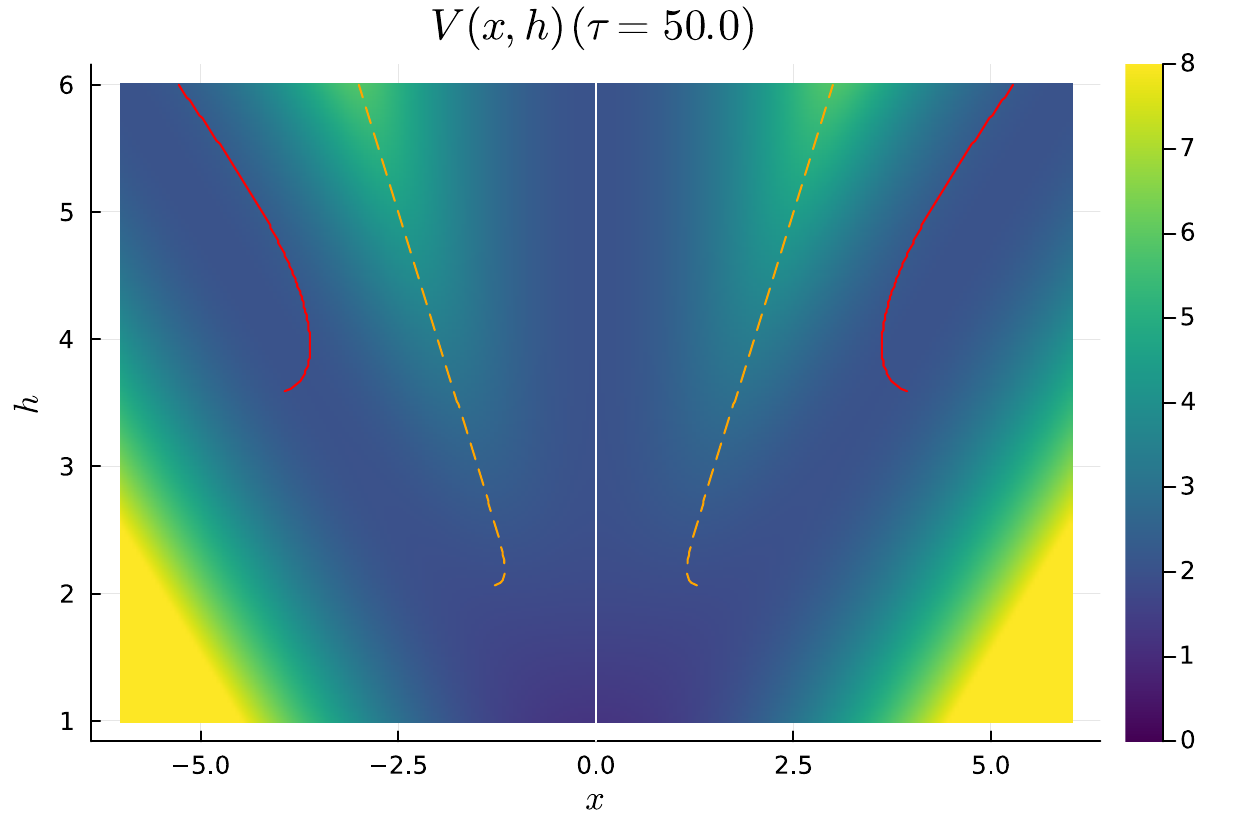}
    \caption{Locally optimal state boundaries~$\pm\alpha^*(h)$ (solid red) and loci of saddle points of~${x\mapsto V(x,h)}$ as functions of the mode separation parameter~$h$, for the Gaussian mixture potential~\eqref{eq:gaussian_mixture}. Top: unpenalized objective~$N^*$. Bottom: penalized objective~$N^*_\tau$ with~$\tau = 50$. The penalization suppresses local optima corresponding to shallow metastable states.}
    \label{fig:shallow_bifurcation}
\end{figure}

\section{Properties of the coefficients of the effective dynamics}
\label{sec:tech_reg}
    We give sufficient conditions for the regularity assumptions of Proposition~\ref{prop:directional_derivative_xi} using the following identities, proven for example in~\cite[Lemma 3.10]{LRS10} for a~$\mathcal C^\infty$ function~$\varphi:\R^d\to\R$, but which are still valid (with the same proof) under weaker regularity assumptions on~$\varphi$.
        Define the partial integration operator with respect to~$\xi$:
        \begin{equation}
            P_\xi\varphi(z) := \int_{\Sigma_z} \varphi\,\det(G_\xi)^{-1/2}\,\d \mathcal{H}_{\Sigma_z},
        \end{equation}
        which is continuous, for instance from~$L^1(\R^d)$ to~$L^1(\R^m)$ by the coarea formula. Then it holds
        \begin{equation}
            \nabla \left(P_\xi \varphi\right) = P_\xi \left(\nabla_\xi \varphi\right),\qquad \nabla_\xi\varphi := \div\left(\varphi\,G_\xi^{-1}\nabla\xi^\top\right),
        \end{equation}
        where in the last line,~$\div$ denotes the row-wise divergence applied to the~$m\times d$ matrix field~$\varphi\,G_\xi^{-1}\nabla\xi^\top$.
        In particular, for~$1\leq \alpha,\gamma\leq m$, it holds formally that
        \begin{equation}
            \partial^2_{\alpha\gamma} P_\xi \varphi = P_\xi\left[\left[M_\xi\right]_{\gamma}^\top\nabla\left(\left[M_\xi\right]_{\alpha}^\top \nabla \varphi + \varphi\,\div\,\left[M_\xi\right]_\alpha\right) + \div\left[M_\xi\right]_{\gamma}\left(\left[M_\xi\right]_{\alpha}^\top \nabla \varphi + \varphi\,\div\,\left[M_\xi\right]_\alpha\right)\right],
        \end{equation}
        where~$\left[M_\xi\right]_\alpha$ denotes the~$\alpha$-th row of the matrix~$G_\xi^{-1}\nabla\xi^\top$.
        From this identity, it follows that the mapping~$P_\xi:\mathcal W^{2,\infty}(\R^d)\to \mathcal W^{2,\infty}(\R^m)$ is continuous when~$M_\xi \in \cW^{2,\infty}(\R^m;\R^{m\times d})$,.
        In turn, this property is satisfied if~$\xi\in \cW^{3,\infty}(\R^d;\R^m)$ and if the condition~$\underset{x\in\R^d}{\inf} G_\xi(x) > 0$ holds in the sense of symmetric matrices.

        If this condition on~$\xi$ is satisfied, it is then easy to show that the conditions of~Proposition~\ref{prop:directional_derivative_xi} hold for instance if
        \begin{equation}
            V\in \cW^{2,\infty}(\R^d),\quad a\in \cW^{2,\infty}(\R^d;\mathcal M_d),\quad \exists\,\varepsilon_a>0:\,u^\top a(x)u\geq\varepsilon_a |u|^2\text{ for almost every }x\in\R^d,
        \end{equation}
        which are uniform versions of Assumptions~\eqref{eq:a_ellipticity} and~\eqref{eq:coeff_regularity}, accounting for the fact that~$\Sigma_z$ may not be compact for all~$z\in \R^m$.

        In practice, it is however often the case that both the dynamics~\eqref{eq:overdamped_langevin} and the CV~$\xi$ take values in compact manifolds, typically the tori~$L(\mathbb R/\mathbb Z)^d$ and $L_\xi(\mathbb R/\mathbb Z)^m$, corresponding respectively to a periodic simulation domain and a set of angular CVs.
        In this case, the regularity of~$F_\xi$ and~$a_\xi$ follows immediately from the smoothness of~$\xi$ and the condition~$\rank\,G_\xi=m$ everywhere.
\section*{}
\paragraph{Data availability.}
Code and trajectory data used in the production of the numerical results of Section~\ref{sec:numerical} are publicly available in the repositories~\cite{github} and~\cite{data} respectively.
\paragraph{Acknowledgments.}
    We thank Olivier Adjoua, Samuel Amstutz, Benjamin Bogosel, Nicola\"i Gouraud, Louis Lagard\`ere, Laurent Michel, Pierre Monmarch\'e, Feliks N\"uske, Danny Perez, Thomas Pl\'e and Julien Reygner for helpful discussions and insightful comments.
    We are grateful to the OPAL infrastructure from Universit\'e C\^ote d'Azur for providing resources and support.
    Some of the experiments presented in this paper were carried out using the Grid'5000 testbed, supported by a scientific interest group hosted by Inria and including CNRS, RENATER and several universities as well as other organizations (see \url{https://www.grid5000.fr}).
    This work received funding from the European Research Council (ERC) under the
    European Union's Horizon 2020 research and innovation program (project
    EMC2, grant agreement No 810367), and from the Agence Nationale de la
    Recherche, under grants ANR-19-CE40-0010-01 (QuAMProcs) and
    ANR-21-CE40-0006 (SINEQ).

\bibliographystyle{siam}
\bibliography{bibliography.bib}

\end{document}